\documentclass[a4paper,UKenglish,cleveref, autoref, thm-restate]{lipics-v2021}
\newboolean{arXiv}
\setboolean{arXiv}{true}

\newboolean{final}
\setboolean{final}{true}

\newboolean{debug}
\setboolean{debug}{false}

\ifthenelse{\boolean{arXiv}}{
    \pdfoutput=1 %
    \hideLIPIcs  %
}{}
\ifthenelse{\boolean{arXiv}}{ %
    \usepackage[bibliography=common]{apxproof}
}{
    \usepackage[bibliography=common,appendix=strip]{apxproof}
}

\usepackage{hyperref}
\usepackage{crossreftools}
\usepackage{stmaryrd} %

\usepackage[most]{tcolorbox} %
\usepackage{mathtools} %
\usepackage{environ} %
\usepackage{ebproof} %
\usepackage[normalem]{ulem} %
\usepackage{tikz}
\usetikzlibrary{shapes, arrows, arrows.spaced, arrows.meta, chains,positioning, cd,
    decorations, decorations.pathmorphing, decorations.markings, decorations.pathreplacing, automata, backgrounds, petri, matrix, fit, calc, graphs, quotes, bending}
\usepackage{mathrsfs} %
\usepackage{stackengine} %
\definecolor{Desire}{HTML}{eb3b5a} %
\definecolor{Boyzone}{HTML}{2d98da} %
\definecolor{Royal Blue}{HTML}{3867d6} %
\definecolor{NYC Taxi}{HTML}{f7b731} %
\definecolor{Beniukon Orange}{HTML}{fa8231}
\definecolor{Algal Fuel}{HTML}{20bf6b} %
\definecolor{Innuendo}{HTML}{a5b1c2} %
\definecolor{Twinkle Blue}{HTML}{d1d8e0} %
\definecolor{Blue Horizon}{HTML}{4b6584} %
\definecolor{Gloomy Purple}{HTML}{8854d0} %

\colorlet{cBlue}{Royal Blue}
\colorlet{cYellow}{NYC Taxi}
\colorlet{cOrange}{Beniukon Orange}
\colorlet{cGreen}{Algal Fuel}
\colorlet{cRed}{Desire}
\colorlet{cGrey}{Innuendo}
\colorlet{cDarkGrey}{Blue Horizon}
\colorlet{cLightGrey}{Twinkle Blue}
\colorlet{cPurple}{Gloomy Purple} %
\ifthenelse{\boolean{final}}{%
\usepackage[xcolor, cleveref, quotation, electronic]{knowledge}
}{%
\usepackage[xcolor, cleveref, quotation, composition]{knowledge}
}
\knowledgeconfigure{notion}
\usepackage{mathcommand}
\knowledgeconfigure{protect quotation={noquote,figure,tikzpicture,tikzcd,yoshiki}} %
\knowledgeconfigure{diagnose line=true, diagnose bar=true}

\definecolor{Dark Ruby Red}{HTML}{580507}
\definecolor{Dark Blue Sapphire}{HTML}{053641}
\definecolor{Dark Gamboge}{HTML}{be7c00}

\ifthenelse{\boolean{debug}}{
}{
\IfKnowledgePaperModeTF{
}{
    \knowledgestyle{intro notion}{color={Dark Ruby Red}, emphasize}
    \knowledgestyle{notion}{color={Dark Blue Sapphire}}
    \hypersetup{
        colorlinks=true,
        breaklinks=true,
        linkcolor={cGreen!80!black!100}, %
        citecolor={cGreen!80!black!100}, %
        filecolor={cGreen!80!black!100}, %
        urlcolor={cGreen!80!black!100},
    }
    \IfKnowledgeElectronicModeTF{
    }{
        \knowledgeconfigure{anchor point color={Dark Ruby Red}, anchor point shape=corner}
        \knowledgestyle{intro unknown}{color={Dark Gamboge}, emphasize}
        \knowledgestyle{intro unknown cont}{color={Dark Gamboge}, emphasize}
        \knowledgestyle{kl unknown}{color={Dark Gamboge}}
        \knowledgestyle{kl unknown cont}{color={Dark Gamboge}}
    }
}
} %

\crefname{section}{\S}{\S}
\crefname{appendix}{\S}{\S}
\crefname{paragraph}{\S}{\S}

\newtheoremrep{apxtheorem}[theorem]{Theorem}
\newtheorem*{theorem*}{Theorem}
\newtheoremrep{apxproposition}[theorem]{Proposition}
	\crefname{apxproposition}{proposition}{propositions} 
	\crefformat{apxproposition}{#2proposition~#1#3} 
	\Crefformat{apxproposition}{#2Proposition~#1#3} 
	\AddToHook{env/apxproposition/begin}{\crefalias{theorem}{proposition}}
	\AddToHook{env/apxpropositionrep/begin}{\crefalias{theorem}{proposition}}
\newtheoremrep{apxlemma}[theorem]{Lemma}
	\crefname{apxlemma}{claim}{claims} 
	\crefformat{apxlemma}{#2claim~#1#3} 
	\Crefformat{apxlemma}{#2Claim~#1#3} 
	\AddToHook{env/apxlemma/begin}{\crefalias{theorem}{lemma}}
	\AddToHook{env/apxlemmarep/begin}{\crefalias{theorem}{lemma}}
\newtheoremrep{apxclaim}[theorem]{Claim}
	\crefname{apxclaim}{claim}{claims} 
	\crefformat{apxclaim}{#2claim~#1#3} 
	\Crefformat{apxclaim}{#2Claim~#1#3} 
	\AddToHook{env/apxclaim/begin}{\crefalias{theorem}{claim}}
	\AddToHook{env/apxclaimrep/begin}{\crefalias{theorem}{claim}}
\newtheoremrep{apxcorollary}[theorem]{Corollary}
	\crefname{apxcorollary}{corollary}{corollaries} 
	\crefformat{apxcorollary}{#2corollary~#1#3} 
	\Crefformat{apxcorollary}{#2Corollary~#1#3} 
	\AddToHook{env/apxcorollary/begin}{\crefalias{theorem}{corollary}}
	\AddToHook{env/apxcorollaryrep/begin}{\crefalias{theorem}{corollary}}

\theoremstyle{definition}
\newtheoremrep{apxdefinition}[theorem]{Definition}
	\crefname{apxdefinition}{definition}{definitions} 
	\crefformat{apxdefinition}{#2definition~#1#3} 
	\Crefformat{apxdefinition}{#2Definition~#1#3} 
	\AddToHook{env/apxdefinition/begin}{\crefalias{theorem}{definition}}
	\AddToHook{env/apxdefinitionrep/begin}{\crefalias{theorem}{definition}}
\newtheoremrep{apxnotation}[theorem]{Notation}
	\crefname{apxnotation}{notation}{notation} 
	\crefformat{apxnotation}{#2notation~#1#3} 
	\Crefformat{apxnotation}{#2Notation~#1#3} 
	\AddToHook{env/apxnotation/begin}{\crefalias{theorem}{notation}}
	\AddToHook{env/apxnotationrep/begin}{\crefalias{theorem}{notation}}

\theoremstyle{remark}
\newtheoremrep{apxexample}[theorem]{Example}
	\crefname{apxexample}{example}{examples} 
	\crefformat{apxexample}{#2example~#1#3} 
	\Crefformat{apxexample}{#2Example~#1#3} 
	\AddToHook{env/apxexample/begin}{\crefalias{theorem}{example}}
	\AddToHook{env/apxexamplerep/begin}{\crefalias{theorem}{example}}
\newtheoremrep{apxremark}[theorem]{Remark}
	\crefname{apxremark}{remark}{remarks} 
	\crefformat{apxremark}{#2remark~#1#3} 
	\Crefformat{apxremark}{#2Remark~#1#3} 
	\AddToHook{env/apxremark/begin}{\crefalias{theorem}{remark}}
	\AddToHook{env/apxremarkrep/begin}{\crefalias{theorem}{remark}}

\AtBeginEnvironment{apxremark}{%
  \pushQED{\qed}%
}
\AtEndEnvironment{apxremark}{\popQED\endapxremark}

\AtBeginEnvironment{apxexample}{%
  \pushQED{\qed}%
}
\AtEndEnvironment{apxexample}{\popQED\endapxexample}

\renewcommand{\epsilon}{\varepsilon}
\renewcommand{\phi}{\varphi}

\ifthenelse{\boolean{final}}{
\usepackage[backgroundcolor=orange!20, textcolor={Dark Ruby Red}, textsize=tiny,
,disable %
]{todonotes}
}{
\usepackage[backgroundcolor=orange!20, textcolor={Dark Ruby Red}, textsize=tiny,
]{todonotes}
}
\definecolor{green}{RGB}{0,120,0}
\definecolor{hlyellow}{RGB}{250, 250, 190}
\definecolor{diegoeditcolor}{RGB}{210,210,255}
\definecolor{santieditcolor}{RGB}{210,255,210}
\definecolor{yoshikieditcolor}{RGB}{200,230,200}

\NewEnviron{sideyoshiki}{\todo[backgroundcolor=yoshikieditcolor, size=\tiny]{{\bf Y:} \BODY}}
\NewEnviron{yoshiki}[1][]{\todo[inline,backgroundcolor=yoshikieditcolor, size=\footnotesize, #1]{{\bf Y:} \BODY}} %
\NewEnviron{diegoenv}{\todo[inline,caption={{\bf D:} XXX},backgroundcolor=diegoeditcolor, size=\footnotesize]{{\bf D:} \BODY}} %
\NewEnviron{hidden}{} 
\NewEnviron{noquote}{\BODY} %

\ifthenelse{\boolean{final}}{
\usepackage[commandnameprefix=always,final]{changes}

\renewcommand{\chdeleted}[2][]{}
}{
\usepackage[commandnameprefix=always]{changes}
}
\definechangesauthor[name={Figueira Diego}, color=diegoeditcolor!80!black!100]{D}
\definechangesauthor[name={Figueira Santiago}, color=santieditcolor!80!black!100]{S}
\definechangesauthor[name={Yoshiki Nakamura}, color=yoshikieditcolor!80!black!100]{Y}

\newcommand{\proofcase}[1]{\noindent\colorbox{cLightGrey}{#1}~~}
\newcommand{\proofcasethin}[1]{\hypersetup{
	linkcolor={cGreen!50!black!100}, %
	citecolor={cGreen!50!black!100}, %
	filecolor={cGreen!50!black!100}, %
	urlcolor={cGreen!50!black!100},
}\noindent\setlength{\fboxsep}{1.pt}\colorbox{cLightGrey}{\normalcolor \textsf{#1}\hspace{.2em}}~~} %
\newcommand{\subproofcasethin}[1]{\hypersetup{
	linkcolor={cGreen!50!black!100}, %
	citecolor={cGreen!50!black!100}, %
	filecolor={cGreen!50!black!100}, %
	urlcolor={cGreen!50!black!100},
}\noindent\quad\setlength{\fboxsep}{1.pt}\colorbox{cLightGrey}{\textsf{#1}\hspace{.2em}}~~}

\knowledgenewrobustcmd{\freemu}[1][\mu]{\cmdkl{\textit{free}(#1)}}
\knowledgenewrobustcmd{\PhiTup}[2]{\cmdkl{\Phi^{#1}_{#2}}}
\tikzset{
    every edge/.append style = {
            line width = .3pt,
        },
    plab/.style={line width = 0.1pt, fill=#1, inner sep = .025cm, anchor=center, font = \fontsize{6pt}{0}},
    plab/.default= white,
    elab/.style={draw, rectangle, line width = 0.1pt, fill=#1, inner sep = .035cm, anchor=center, font = \footnotesize},
    elab/.default= white,
    tlab/.style={line width = 0.1pt, fill=#1, inner sep = .025cm, anchor=center, font = \fontsize{6pt}{0}\selectfont},
    tlab/.default= white
}
\tikzset{
    png export/.style={
            external/system call/.add={}{; convert -density 300 -transparent white "\image.pdf" "\image.png"},
            /pgf/images/external info,
            /pgf/images/include external/.code={
                    \includegraphics[width=\pgfexternalwidth,height=\pgfexternalheight]{##1.png}
                },
        }
}
\tikzset{apply style/.code={\tikzset{#1}}}
\tikzstyle{mynode} = [inner sep = 1.5pt, fill= gray!20, font=\footnotesize]
\tikzstyle{mysmallnode} = [inner sep = 1.pt, fill= gray!20]
\tikzstyle{vert} = [draw, circle, mynode]
\tikzset{earrow/.style={>={{[flex] Latex[length=.1cm, width=2.5pt]}}}}
\tikzstyle{elabel} = [inner sep = 1.pt, font = \scriptsize, opacity = 1]

\tikzstyle{dvert} = [vert, color = gray!20]
\tikzstyle{dearrow} = [earrow, color = gray!20]

\newtcolorbox{rulesbox}[1][]{ams align*, standard jigsaw, opacityback = 0.8, colframe=black!80,
boxrule=.3mm,left=.2mm, right=.2mm, #1}

\DeclareRobustCommand{\lipicsEnd}{%
}

\makeatletter
\newcommand{\subalign}[2][c]{%
  \if#1c\vcenter\else\vtop\fi{%
    \Let@ \restore@math@cr \default@tag
    \baselineskip\fontdimen10 \scriptfont\tw@
    \advance\baselineskip\fontdimen12 \scriptfont\tw@
    \lineskip\thr@@\fontdimen8 \scriptfont\thr@@
    \lineskiplimit\lineskip
    \ialign{\hfil$\m@th\scriptstyle##$&$\m@th\scriptstyle{}##$\hfil\crcr
      #2\crcr
    }%
  }%
}
\makeatother

\newcounter{mylabelcounter}
\makeatletter
\newcommand{\labeltext}[3]{%
#2\refstepcounter{mylabelcounter}%
\immediate\write\@auxout{%
  \string\newlabel{#1}{{1}{\thepage}{{\unexpanded{#3}}}{mylabelcounter.\number\value{mylabelcounter}}{}}%
}%
}
\makeatother

\DeclareMathSymbol{\shortminus}{\mathalpha}{operators}{`-}
\DeclarePairedDelimiter\set{\{}{\}} %
\DeclarePairedDelimiter\tup{(}{)} %

\newrobustcmd{\defeq}{\mathrel{\hat{=}}}
\newrobustcmd{\defiff}{\mathrel{\hat{\iff}}}
\newrobustcmd{\defleftrightarrow}{\mathrel{\hat{\leftrightarrow}}}

\newrobustcmd{\Nat}{\mathbb{N}}
\newrobustcmd\pset[1]{\+P(#1)} %

\newrobustcmd\smashxrightarrow[1]{%
  \raisebox{-.04em}{$%
    \xrightarrow{\smash{\raisebox{-.1em}{%
      \tiny{#1}%
    }}}%
  $}%
}%

\knowledgenewrobustcmd{\FV}{\cmdkl{\textup{\textsf{FV}}}} %
\knowledgenewrobustcmd{\toSet}[1]{#1} %
\knowledgenewrobustcmd{\TC}[3][0]{%
  \ifnum#1=2%
    \cmdkl{\Big[}#2\cmdkl{\Big]^*_{#3}}%
  \else\ifnum#1=1%
    \cmdkl{\big[}#2\cmdkl{\big]^*_{#3}}%
  \else%
    \cmdkl{[}#2\cmdkl{]}^{\cmdkl{*}}_{#3}%
  \fi\fi%
}
\knowledgenewrobustcmd{\kthTC}[4][0]{%
  \ifnum#1=2%
    \cmdkl{\Big[}#2\cmdkl{\Big]^{#3}_{#4}}%
  \else\ifnum#1=1%
    \cmdkl{\big[}#2\cmdkl{\big]^{#3}_{#4}}%
  \else%
    \cmdkl{[}#2\cmdkl{]}^{\cmdkl{#3}}_{#4}%
  \fi\fi%
}
\knowledgenewrobustcmd{\SQ}[3][0]{%
  \ifnum#1=2%
    \cmdkl{\Big[}#2\cmdkl{\Big]^{2}_{#3}}%
  \else\ifnum#1=1%
    \cmdkl{\big[}#2\cmdkl{\big]^{2}_{#3}}%
  \else%
    \cmdkl{[}#2\cmdkl{]}^{\cmdkl{2}}_{#3}%
  \fi\fi%
}
\knowledgenewrobustcmd{\SQgen}[4][0]{%
  \ifnum#1=2%
    \cmdkl{\Big[}#2\cmdkl{\Big]^{#3}_{#4}}%
  \else\ifnum#1=1%
    \cmdkl{\big[}#2\cmdkl{\big]^{#3}_{#4}}%
  \else%
    \cmdkl{[}#2\cmdkl{]}^{\cmdkl{#3}}_{#4}%
  \fi\fi%
}
\newrobustcmd{\struc}[1]{\mathbb{#1}} %
\newrobustcmd{\dom}[1]{\univ{\struc #1}} %

\knowledgenewrobustcmd{\arity}{\cmdkl{\textit{ar}}} %
\knowledgenewrobustcmd{\interpatom}[2]{#2^{#1}} %

\knowledgenewrobustcmd{\range}[2]{\cmdkl{[}#1\cmdkl{..}#2\cmdkl{]}} %
\knowledgenewrobustcmd{\1}[1]{\cmdkl{[}#1\cmdkl{]}} %

\knowledgenewrobustcmd{\ith}[1]{\cmdkl{[}#1\cmdkl{]}}

\knowledgenewrobustcmd{\Vars}{\cmdkl{\+V}} %

\knowledgenewrobustcmd{\Rels}{\cmdkl{\+R}} %

\newrobustcmd{\pto}{\rightharpoonup} %
\knowledgenewrobustcmd{\fdom}{\cmdkl{\operatorname{dom}}} %

\knowledgenewrobustcmd{\dcup}{\mathbin{\cmdkl{\sqcup}}} %
\newrobustcmd{\bigdcup}{\bigsqcup} %

\knowledgenewrobustcmd{\modelsass}[1]{\mathrel{\cmdkl{\models}}_{#1}}
\knowledgenewrobustcmd{\modelsnonass}{\mathrel{\cmdkl{\models}}}

\knowledgenewrobustcmd{\Phisep}{\cmdkl{\Phi_{\textit{sep}}}}
\knowledgenewrobustcmd{\isTup}{\cmdkl{\mathit{tup}}}
\knowledgenewrobustcmd{\isConst}{\cmdkl{\mathit{comp}}}
\knowledgenewrobustcmd{\PhiEQ}{\cmdkl{\Phi_{\EQrel}}}
\knowledgenewrobustcmd{\eqstar}{\mathop{\cmdkl{\overset{\smash{\raisebox{-0.4ex}{$\scriptstyle *$}}}{=}}}}
\knowledgenewrobustcmd{\RelsAr}{\cmdkl{\+R}_{\text{ar}}}
\knowledgenewrobustcmd{\EQrel}{\cmdkl{\textit{EQ}}}
\knowledgenewrobustcmd{\phiAtom}{\cmdkl{\phi[\EQrel/{=}]}}

\knowledgenewrobustcmd{\equiveq}[1]{\cmdkl{[}#1\cmdkl{]_{\eqstar}}}

\DeclarePairedDelimiter\ceil{\lceil}{\rceil}

\newrobustcmd{\const}[1]{\mathsf{#1}}
\newrobustcmd{\bl}{-}
\newcommand{\mul}[1]{\bar{#1}}

\newrobustcmd{\nat}{\Nat} %
\newrobustcmd{\znat}{\mathbb{Z}}
\newrobustcmd{\pnat}{\nat_{+}}
\knowledgenewrobustcmd{\card}{\cmdkl{\#}}

\newrobustcmd{\poly}{\operatorname{poly}}

\knowledgenewrobustcmd{\eps}{\cmdkl{\varepsilon}}
\NewDocumentCommand\word{O{1}}{%
    \ifcase#1 undefined
    \or w
    \else undefined \fi
}
\knowledgenewrobustcmd{\lepref}{\mathrel{\cmdkl{\le_{\mathrm{pref}}}}}

\newcommand{\back}{\shortminus}
\newrobustcmd{\dir}{d}
\knowledgenewrobustcmd{\yndirdom}[2]{\operatorname{\cmdkl{dom}}_{#1,#2}}
\NewDocumentCommand\bag{O{1}}{%
    \ifcase#1 undefined
    \or g
    \or h
    \else undefined \fi
}
\knowledgenewrobustcmd{\series}{\mathbin{\cmdkl{\diamond}}}

\NewDocumentCommand\rsym{O{1}}{%
    \ifcase#1 undefined
    \or R
    \or S
    \or T
    \else undefined \fi
}
\knowledgenewrobustcmd{\Erel}{\cmdkl{E}} %

\NewDocumentCommand\ynstruc{O{1}}{%
    \ifcase#1 undefined
    \or \struc A
    \or \struc B
    \else undefined \fi
}
\newrobustcmd{\vertex}{c}
\knowledgenewrobustcmd{\univ}[1]{\cmdkl{\mathrm{dom}}(#1)}

\newrobustcmd{\inter}{I} %
\knowledgenewrobustcmd{\intersubst}[2]{\cmdkl{[}#2/#1\cmdkl{]}} %

\knowledgenewrobustcmd{\allstruc}{\cmdkl{\mathsf{STR}}}
\knowledgenewrobustcmd{\allcountablestruc}{\cmdkl{\mathsf{STR}}}
\knowledgenewrobustcmd{\allfinstruc}{\cmdkl{\mathsf{STR}_{\mathrm{fin}}}}
\knowledgenewrobustcmd{\allatstruc}{\cmdkl{\mathsf{NSTR}}}

\NewDocumentCommand\ynmulstruc{O{1}}{\mul{\ynstruc[#1]}}

\NewDocumentCommand\afml{O{1}}{%
    \ifcase#1 undefined
    \or \alpha
    \or \beta
    \else undefined \fi
 }
\NewDocumentCommand\fml{O{1}}{%
    \ifcase#1 undefined
    \or \varphi
    \or \psi
    \or \rho
    \else undefined \fi
}
\NewDocumentCommand\fmlset{O{1}}{%
    \ifcase#1 undefined
    \or \Gamma
    \or \Delta
    \or \Lambda
    \else undefined \fi
}
\newcommand{\fmlclass}{\mathcal{L}}
\NewDocumentCommand\fmlsetclass{O{1}}{
    \ifcase#1 undefined
    \or \mathbb{F}
    \else undefined \fi
}

\newcommand{\EQ}{\mathbin{=}}

\knowledgenewrobustcmd{\fmlsubst}[2]{\cmdkl{[}#2/#1\cmdkl{]}} %
\knowledgenewrobustcmd{\renameeq}{\mathrel{\cmdkl{\equiv_{\mathrm{rn}}}}}

\knowledgenewrobustcmd{\seqlen}[1]{\cmdkl{\|}#1\cmdkl{\|}}
\knowledgenewrobustcmd{\fmllen}[1]{\cmdkl{\|}#1\cmdkl{\|}} %

\knowledgenewrobustcmd{\glue}{{\cmdkl{\odot}}} %
\knowledgenewrobustcmd{\strucbigdcup}[1]{\operatorname*{\cmdkl{\bigsqcup}}_{#1}} %
\knowledgenewrobustcmd{\quo}[1]{\cmdkl{[}#1\cmdkl{]}} %
\knowledgenewrobustcmd{\quorel}[1]{\mathrel{\cmdkl{\sim}}_{#1}} %

\knowledgenewrobustcmd{\fixedvertex}{\cmdkl{\const{\vertex}}} %
\newrobustcmd{\elem}{e} %
\newrobustcmd{\elemset}{E} %

\knowledgenewrobustcmd{\AD}{\cmdkl{\mathcal{D}}} %
\knowledgenewrobustcmd{\ADB}[2]{\cmdkl{\mathcal{D}}^{#1}_{#2}} %
\knowledgenewrobustcmd{\absfunc}[2]{\operatorname{\cmdkl{\const{Abs}}}^{#1}_{#2}} %
\knowledgenewrobustcmd{\conc}[2]{\operatorname{\cmdkl{\const{C}}}^{#1}_{#2}} %
\knowledgenewrobustcmd{\ND}[3]{\cmdkl{\mathcal{N}}^{#1}_{#2,#3}} %
\knowledgenewrobustcmd{\MD}[3]{\cmdkl{\mathcal{M}}^{#1}_{#2,#3}} %

\newrobustcmd{\ainter}{\mathscr{I}} %
\knowledgenewrobustcmd{\aintersubst}[2]{\cmdkl{[}#2/#1\cmdkl{]}} %

\NewDocumentCommand{\pri}{O{1}}{
    \ifcase#1 undefined
    \or p
    \or q
    \else undefined \fi    
} %
\NewDocumentCommand{\prizero}{O{1}}{0}
\NewDocumentCommand{\prione}{O{1}}{1}
\knowledgenewrobustcmd{\pathPri}[1]{\cmdkl{\Omega}_{#1}} %
\newrobustcmd{\Pri}{\Omega} %

\newrobustcmd{\mkstate}[5]{#1^{#2,#3}_{#4,#5}} %
\NewDocumentCommand{\mystate}{O{1}}{\dot{\fmlset[#1]}}
\knowledgenewrobustcmd{\absmodels}{\mathrel{\cmdkl{\models}}} %

\knowledgenewrobustcmd{\PBallfml}{\cmdkl{\mathscr{B}_{+}}}
\NewDocumentCommand{\doublefml}{O{1}}{\dot{\fml[#1]}}
\newcommand{\doubletrue}{\const{true}}
\newcommand{\doublefalse}{\const{false}}
\DeclareMathOperator{\bigdoublewedge}{{\stackMath\stackinset{c}{0pt}{c}{+0ex}{\scriptstyle\cdot}{\bigwedge}}}
\DeclareMathOperator{\bigdoublevee}{{\stackMath\stackinset{c}{0pt}{c}{-0.1ex}{\scriptstyle\cdot}{\bigvee}}}

\DeclareMathOperator*{\sbigdoublevee}{{\stackMath\stackinset{c}{0pt}{c}{-0.1ex}{\scriptstyle\cdot}{\bigvee}}}
\newcommand{\doublewedge}{\mathbin{\stackMath\stackinset{c}{0pt}{c}{-0.2ex}{\scriptstyle\cdot}{\wedge}}}
\newcommand{\doublevee}{\mathbin{\stackMath\stackinset{c}{0pt}{c}{-0.ex}{\scriptstyle\cdot}{\vee}}}
\knowledgenewrobustcmd{\lequiv}{\mathrel{\cmdkl{\equiv_{\const{L}}}}}
\knowledgenewrobustcmd{\lleqq}{\mathrel{\cmdkl{\leqq_{\const{L}}}}}
\knowledgenewrobustcmd{\lgeqq}{\mathrel{\cmdkl{\geqq_{\const{L}}}}}
\knowledgenewrobustcmd{\bfmllen}[1]{\cmdkl{\|}#1\cmdkl{\|}}

\knowledgenewrobustcmd{\allstates}{\cmdkl{\mathcal{Q}}}

\NewDocumentCommand\trace{O{1}}{%
    \ifcase#1 undefined
    \or \tau
    \or \sigma
    \or \rho
    \else undefined \fi
}

\knowledgenewrobustcmd{\leadstoUNFO}{\mathrel{{\cmdkl{\leadsto}}}} %
\knowledgenewrobustcmd{\vdashUNFO}{\mathrel{\cmdkl{\vdash}}} %
\knowledgenewrobustcmd{\nvdashUNFO}{\mathrel{\cmdkl{\nvdash}}} %
\knowledgenewrobustcmd{\vdashUNFOsub}{\mathrel{\cmdkl{\vdash}}} %
\knowledgenewrobustcmd{\nvdashUNFOsub}{\mathrel{\cmdkl{\nvdash}}} %
\knowledgenewrobustcmd{\clUNFO}{\operatorname{\cmdkl{\mathrm{cl}}}} %

\newrobustcmd{\allstatesUNFO}{\allstates_{\mathrm{\kl{UNFO}}}} %

\knowledgenewrobustcmd{\leadstoUNTC}{\mathrel{\cmdkl{\leadsto}}} %
\knowledgenewrobustcmd{\vdashUNTC}{\mathrel{\cmdkl{\vdash}}} %
\knowledgenewrobustcmd{\nvdashUNTC}{\mathrel{\cmdkl{\nvdash}}} %
\knowledgenewrobustcmd{\vdashUNTCsub}{\mathrel{\cmdkl{\vdash}}} %
\knowledgenewrobustcmd{\nvdashUNTCsub}{\mathrel{\cmdkl{\nvdash}}} %

\knowledgenewrobustcmd{\preclUNTC}{\operatorname{\cmdkl{\mathrm{ucl}}}} %
\knowledgenewrobustcmd{\prebiclUNTC}{\operatorname{\cmdkl{\mathrm{bicl}}}} %
\knowledgenewrobustcmd{\clUNTC}{\operatorname{\cmdkl{\mathrm{cl}}}} %

\knowledgenewrobustcmd{\psettwo}{\cmdkl{\+P_2}} %

\newrobustcmd{\allstatesUNTC}{\allstates_{\mathrm{\kl{UNTC}}}} %

\knowledgenewrobustcmd{\TBSAT}[1]{\ensuremath{\cmdkl{\mathbf{TB}(\mathrm{SAT})}_{#1 \times \mathrm{M}}}}

\knowledgenewrobustcmd{\PBinfallfml}{\cmdkl{\mathscr{B}^{\omega}_{+}}}

\NewDocumentCommand\automaton{O{1}}{%
    \ifcase#1 undefined
    \or \mathcal{A}
    \or \mathcal{B}
    \or \mathcal{C}
    \else undefined \fi
}
\knowledgenewrobustcmd{\automatonsize}[1]{\cmdkl{\|}#1\cmdkl{\|}}
\knowledgenewrobustcmd{\automatonlang}{\cmdkl{\mathcal{L}}}

\knowledgenewrobustcmd{\treedistance}{\cmdkl{\const{dist}}} %
\NewDocumentCommand\bagset{O{1}}{%
    \ifcase#1 undefined
    \or \mathcal{G}
    \or \mathcal{F}
    \else undefined \fi
}%
\knowledgenewrobustcmd{\addPri}[1]{#1^{\cmdkl{\complement}}}
\knowledgenewrobustcmd{\statePri}{\cmdkl{\Omega}} %
\knowledgenewrobustcmd{\iverson}[1]{\cmdkl{[}#1\cmdkl{]}}

\knowledgenewrobustcmd{\transitionUNFO}{\operatorname{\cmdkl{\delta}}} %
\knowledgenewrobustcmd{\paramconcUNFO}{\operatorname{\cmdkl{\const{C}}}} %
\knowledgenewrobustcmd{\distanceUNFO}{\cmdkl{\const{d}}} %
\knowledgenewrobustcmd{\paramUNFO}{\operatorname{\cmdkl{\natural}}}

\newrobustcmd{\concstate}[3]{#1^{#2}_{#3}}

\knowledgenewrobustcmd{\allsymbols}{\cmdkl{\const{S}}} %
\knowledgenewrobustcmd{\ADBsymbol}{\cmdkl{\underline{\mathcal{D}}}} %
\knowledgenewrobustcmd{\NDsymbol}{\cmdkl{\underline{\mathcal{N}}}} %
\knowledgenewrobustcmd{\MDsymbol}{\cmdkl{\underline{\mathcal{M}}}} %

\knowledgenewrobustcmd{\transitionUNTC}{\operatorname{\cmdkl{\delta}}} %
\knowledgenewrobustcmd{\paramconcUNTC}{\operatorname{\cmdkl{\const{C}}}} %
\knowledgenewrobustcmd{\distanceUNTC}{\cmdkl{\const{d}}} %
\knowledgenewrobustcmd{\paramUNTC}{\operatorname{\cmdkl{\natural}}}

\knowledgenewrobustcmd{\UNF}{\operatorname{\cmdkl{\const{UN}}}}

\knowledgenewrobustcmd{\mkbin}[2]{\cmdkl{(}#2\cmdkl{)}_{\cmdkl{2}}^{#1}}

\knowledgenewrobustcmd{\WHERE}[1]{\ensuremath{\cmdkl{\mathbf{Where}(\mathrm{EFO})_{#1}}}}
\knowledgenewrobustcmd{\WHEREtoFO}[1]{#1^{\cmdkl{\flat}}} %
\knowledgenewrobustcmd{\UNTCtoWHERE}{\operatorname{\cmdkl{Tr}}} %

\knowledgenewrobustcmd{\gdd}{\mathop{\cmdkl{\const{gdd}}}} %

\knowledge{text={resp.}, italic}
  | resp

\knowledge{text={i.e.}, italic}
  | ie

\knowledge{text={I.e.}, italic}
  | Ie

\knowledge{text={s.t.}, italic}
  | st

\knowledge{text={e.g.}, italic}
  | eg

\knowledge{text={vs.}, italic}
  | vs

\knowledge{text={w.r.t.}, italic}
  | wrt

\knowledge{text={a.k.a.}, italic}
  | aka

\knowledge{text={w.l.o.g.}, italic}
  | wlog

\knowledge{text={W.l.o.g.}, italic}
  | Wlog

\knowledge{text={cf.}, italic}
  | cf

\knowledge{text={Cf.}, italic}
  | Cf

\knowledge{text={iff}, italic}
  | iff

\knowledge{text={r.e.}, italic}
  | r.e.
  | re
\knowledge{wrap=\textsf}
  | NL

\knowledge{notion, text={paraNL}, wrap=\textsf}
  | para-NL
  | paraNL

\knowledge{notion, text={\textup{FPT}}, wrap=\textsf}
  | FPT

\knowledge{notion}
  | fixed-parameter tractable

\knowledge{text={ExpSpace}, wrap=\textsf}
  | EXPSPACE
  | ExpSpace

\knowledge{text={2ExpSpace}, wrap=\textsf}
  | 2EXPSPACE
  | 2ExpSpace

\knowledge{text={PSpace}, wrap=\textsf}
  | PSpace
  | PSPACE

  \knowledge{text={PTime}, wrap=\textsf}
  | PTime
  | ptime

\knowledge{text={FP}, wrap=\textsf}
  | FP

\knowledge{text={NP}, wrap=\textsf}
  | NP

  \knowledge{text={P}, wrap=\textsf}
  | P

  \knowledge{notion, text={\ensuremath{\textup{W}[1]}}, wrap=\textsf}
  | W1

  \knowledge{text={\ensuremath{\# \textsf{P}}}, wrap=\textsf}
  | shP

  \knowledge{text={\ensuremath{\# {\cdot}\textsf{NP}}}, wrap=\textsf}
  | shNP

  \knowledge{text={\ensuremath{\textsf{P}^{\# \textsf{P}}}}, wrap=\textsf}
  | PshP
  
\knowledge{text={\ensuremath{\textsf{FP}^{\# \textsf{P}}}}, wrap=\textsf}
  | FPshP

\knowledge{text={\ensuremath{\Pi^p_2}}, wrap=\textsf}
  | PiP2
  | Pi2
  | pi2
  | Pitwo
  | PiTwo
  | pitwo

\knowledge{url={https://en.wikipedia.org/wiki/Savitch\%27s_theorem}}
  | Savitch's Theorem

\knowledge{text={\ensuremath{\mathsf{ExpTime}}}}
  | EXPTIME
  | ExpTime

\knowledge{text={\ensuremath{\mathsf{2ExpTime}}}}
  | 2EXPTIME
  | 2ExpTime

\knowledge{text={\ensuremath{\mathsf{P}^{\mathsf{NP}}}}}
  | PNP

\knowledge{text={\ensuremath{\mathsf{P}^{\mathsf{NP}[\mathcal{O}(\log n)]}}}}
  | PNPlog

\knowledge{text={\ensuremath{\mathsf{P}^{\mathsf{NP}[\mathcal{O}(\log^2 n)]}}}}
  | PNPlog2

\knowledge{notion}
 | notion@notice
 | definition@notice

\knowledge{notion}
 | GNTC-UP

\knowledge{notion}
 | UNTC-UP
 
\knowledge{notion}
 | guarded negation
 | GNFO
 | guarded negation first-order logic

\knowledge{notion}
 | GNFO formula
 | GNFO formulas
 | formula@GNFO
 | formulas@GNFO

\knowledge{notion}
 | first-order logic
 | FO

\knowledge{notion}
  | FO formula
  | FO formulas
  | formula@FO
  | formulas@FO

\knowledge{notion}
  | FO sentence
  | FO sentences
  | sentence@FO
  | sentences@FO

\knowledge{notion}
 | free variables
 | free variable

\knowledge{notion}
 | GNFP-UP
 | guarded negation fixpoint logic with unguarded parameters

 \knowledge{notion}
 | ICPDL

 \knowledge{notion, text={UCPDL\ensuremath{{}^{\!+}}}}
  | UCPDL+

\knowledge{notion, text={UNFO\ensuremath{{}^{\text{reg}}}}}
  | UNFOreg

\knowledge{notion}
  | GNTC
  | guarded negation transitive closure logic

\knowledge{notion}
  | GNTC formula
  | GNTC formulas
  | formula@GNTC
  | formulas@GNTC

\knowledge{notion}
  | GNTC sentence
  | GNTC sentences
  | sentence@GNTC
  | sentences@GNTC

\knowledge{notion}
  | unary negation
  | UNFO

\knowledge{notion}
 | UNFO formula
 | UNFO formulas
 | formula@UNFO
 | formulas@UNFO

\knowledge{notion}
  | UNFO sentence
  | UNFO sentences
  | sentence@UNFO
  | sentences@UNFO

\knowledge{notion}
 | SAT-
 | satisfiability problem
 | satisfiability problems

\knowledge{notion}
 | finite-satisfiability problem

\knowledge{notion}
 | tuple elements
 | tuple element

\knowledge{notion}
 | component elements
 | component element

\knowledge{notion}
 | domain elements
 | domain element
 | element
 | elements

\knowledge{notion}
 | UNTC

\knowledge{notion}
 | UNTC formula
 | UNTC formulas
 | formula@UNTC
 | formulas@UNTC

\knowledge{notion}
| UNTC sentence
| UNTC sentences
| sentence@UNTC
| sentences@UNTC

\knowledge{notion}
 | satisfiable

\knowledge{notion}
 | finitely-satisfiable

\knowledge{notion}
 | arity

\knowledge{}
 | set
 | sets

\knowledge{notion}
 | cardinality

\knowledge{notion}
 | countable

\knowledge{}
 | subset
 | subsets

\knowledge{}
 | tuple
 | tuples

\knowledge{notion}
| domain@function

\knowledge{notion}
 | length@word

\knowledge{notion}
 | letter
 | letters

\knowledge{notion}
 | word
 | words

\knowledge{notion}
 | empty word
 | empty@word

\knowledge{notion}
 | concatenation

\knowledge{notion}
 | prefix

\knowledge{notion}
 | prefix-closed

 \knowledge{notion}
 | tree
 | trees

\knowledge{notion}
| binary
| Binary

\knowledge{notion}
 | direction
 | directions

\knowledge{notion}
 | child
 | children

\knowledge{notion}
 | path
 | paths

\knowledge{notion}
 | height

\knowledge{notion}
 | root
 | roots

\knowledge{notion}
 | leaf
 | leaves

\knowledge{}
 | graph
 | graphs

\knowledge{notion}
| isomorphic
| isomorphism

\knowledge{notion}
 | node
 | nodes

\knowledge{notion}
 | Full version
 | full version

\knowledge{notion}
 | relation name
 | relation names

\knowledge{notion}
 | relational structure
 | structure
 | structures
 | Structures

\knowledge{notion}
 | domain@structure

\knowledge{notion}
 | disjoint union@structure

\knowledge{notion}
 | semantics

\knowledge{notion}
| interpretation
| variable assignment

\knowledge{notion}
 | treewidth

\knowledge{notion}
 | tree decomposition
 | tree decompositions

\knowledge{notion}
 | size@treewidth

\knowledge{notion}
 | bag
 | bags

\knowledge{}
| subformula
| subformulas

\knowledge{notion}
| atomic formula
| atomic formulas

\knowledge{notion}
 | variable
 | variables

\knowledge{}
 | fresh

\knowledge{notion}
| size@formula
| sizes@formula

\knowledge{notion}
| bounded treewidth model property

\knowledge{notion}
| downward L{\"o}wenheim--Skolem property

\knowledge{notion}
 | non-emptiness problem

\knowledge{notion}
 | 2-way alternating parity tree automaton
 | 2-way alternating parity tree automata
 | 2APTA
 | 2APTAs

\knowledge{notion}
 | transition function

\knowledge{notion}
 | priority

\knowledge{notion}
 | accepting

\knowledge{notion}
 | run@2APTA
 | runs@2APTA

\knowledge{notion}
 | states@2APTA
 | state@2APTA

\knowledge{notion}
 | language
 | languages

\knowledge{notion}
| size@2APTA

\knowledge{notion}
 | positive Boolean formulas
 | positive Boolean formula

\knowledge{notion}
 | semantical equivalence relation
 | semantically equivalent

\knowledge{notion}
 | semantical entailment relation

\knowledge{notion}
 | dual

\knowledge{notion}
  | Boolean formula

\knowledge{notion}
  | truth value

\knowledge{notion}
  | Boolean variable
  | Boolean variables

\knowledge{notion}
 | gluing
 | gluing operator
 | glued structure

\knowledge{}
 | equivalence class
 | equivalence classes

\knowledge{}
 | equivalence relation

\knowledge{notion}
 | quotient structure

\knowledge{notion}
 | abstract interpretation
 | abstract interpretations

\knowledge{notion}
 | abstract domain
 
\knowledge{notion}
 | abstraction function

\knowledge{notion}
 | concretization function

\knowledge{notion}
 | abstraction

\knowledge{notion}
 | concretization
 | concretized

\knowledge{notion}
 | split

\knowledge{notion}
 | fixed names@LMC
 | fixed name@LMC

 \knowledge{notion}
| tree@LMC
| trees@LMC

\knowledge{notion}
 | states@LMC
 | state@LMC
 | instance@LMC
 | instances@LMC

\knowledge{notion}
 | priority@LMC

\knowledge{notion}
 | accepting@LMC

\knowledge{notion}
 | run@LMC
 | runs@LMC
 | derivation tree@LMC
 | derivation trees@LMC

\knowledge{notion}
 | clrun@LMC
 | clruns@LMC

\knowledge{notion}
 | subrun@LMC
 | subruns@LMC

\knowledge{notion}
 | subclrun@LMC
 | subclruns@LMC

\knowledge{notion}
| transition formula@LMC

\knowledge{notion}
| neighbouring set@LMC
| neighbouring set

\knowledge{notion}
| moving set@LMC
| moving set

\knowledge{notion}
 | local checker@UNFO
 | Local Checker@UNFO
 | LC@UNFO
 
\knowledge{notion}
 | rule@UNFO
 | rules@UNFO
 | Rule@UNFO
 | Rules@UNFO

\knowledge{notion}
 | closure@UNFO

\knowledge{notion}
 | concretization set@UNFO

\knowledge{notion}
 | parameter@UNFO

\knowledge{notion}
 | distance@UNFO

\knowledge{notion}
 | local checker@UNTC
 | Local Checker@UNTC
 | LC@UNTC

\knowledge{notion}
 | rule@UNTC
 | rules@UNTC
 | Rule@UNTC
 | Rules@UNTC

\knowledge{notion}
 | closure@UNTC

\knowledge{notion}
 | uni-closure@UNTC

 \knowledge{notion}
 | bi-closure@UNTC

\knowledge{notion}
 | concretization set@UNTC
\knowledge{notion}
 | parameter@UNTC
\knowledge{notion}
 | distance@UNTC

\knowledge{notion}
 | model of
\knowledge{notion}
 | width
\knowledge{notion}
 | guards
 | guard
\knowledge{notion}
 | guard@GNF(TC)
 | guards@GNF(TC)

\knowledge{notion}
 | atoms
 | atom

\knowledge{notion}
 | existentially quantified atoms
 | existentially quantified atom

\knowledge{notion}
 | guarded negation formula

\knowledge{notion}
 | negation formula
 | negation formulas

\knowledge{notion}
 | transitive closure logic
 | TC

\knowledge{notion}
| formula@TC
| formulas@TC

\knowledge{notion}
| sentence@TC
| sentences@TC

\knowledge{notion}
 | unparameterized guarded transitive closure formula
 | transitive closure formula
 | transitive closure formulas
 | TC-formula
 | TC-formulas
 | unparameterized guarded transitive closure
 | transitive closure

\knowledge{notion}
 | least fixpoint logic
 | LFP

\knowledge{notion}
  | LFP formula
  | LFP formulas
  | formula@LFP
  | formulas@LFP

\knowledge{notion}
 | GFO

\knowledge{notion}
 | guarded negation fixpoint logic
 | GNFP

\knowledge{notion}
 | PDL
 
 \knowledge{notion}
  | CQPDL

\knowledge{notion}
  | GNF(TC)

\knowledge{notion}
  | GNF(TC) formula
  | GNF(TC) formulas
  | formula@GNF(TC)
  | formulas@GNF(TC)

\knowledge{notion}
  | EPUNTC

\knowledge{notion}
  | genUNTC

\knowledge{notion}
  | genUNTC formula
  | genUNTC formulas
  | formula@genUNTC
  | formulas@genUNTC

\knowledge{notion}
  | genUNTC sentence
  | genUNTC sentences
  | sentence@genUNTC
  | sentences@genUNTC

\knowledge{notion}
  | genGNTC

\knowledge{notion}
  | genGNTC formula
  | genGNTC formulas
  | formula@genGNTC
  | formulas@genGNTC

\knowledge{notion}
  | genGNTC sentence
  | genGNTC sentences
  | sentence@genGNTC
  | sentences@genGNTC

\knowledge{notion}
 | UNFP

\knowledge{notion}
 | alternation-free

 \knowledge{notion}
 | clique-guarded negation fragment

\knowledge{notion}
  | UNSQ

\knowledge{notion}
  | UNSQ formula
  | UNSQ formulas
  | formula@UNSQ
  | formulas@UNSQ

\knowledge{notion}
  | UNSQ sentence
  | UNSQ sentences
  | sentence@UNSQ
  | sentences@UNSQ

\knowledge{notion}
  | genUNSQ

\knowledge{notion}
  | genUNSQ formula
  | genUNSQ formulas
  | formula@genUNSQ
  | formulas@genUNSQ

\knowledge{notion}
  | genUNSQ sentence
  | genUNSQ sentences
  | sentence@genUNSQ
  | sentences@genUNSQ

\knowledge{notion}
  | squaring formula
  | SQ-formula
  | squaring

\knowledge{notion}
  | QBF
  | quantified Boolean formula problem

\knowledge{notion}
  | quantified Boolean formula

\knowledge{notion}
  | intersection non-emptiness problem

\knowledge{notion}
  | DFAs
  | deterministic finite automata

\knowledge{notion}
  | states@DFA
  | state@DFA

\knowledge{notion}
  | EFO

\knowledge{notion}
  | EFO formula
  | EFO formulas
  | existential FO formula
  | formula@EFO
  | formulas@EFO

\knowledge{notion}
  | EFO sentence
  | EFO sentences
  | sentence@EFO
  | sentences@EFO

\knowledge{notion}
  | NNFO

\knowledge{notion}
  | NNFO formula
  | NNFO formulas
  | formula@NNFO
  | formulas@NNFO

\knowledge{notion}
  | NNFO sentence
  | NNFO sentences
  | sentence@NNFO
  | sentences@NNFO

\knowledge{notion}
  | negation depth

\knowledge{notion}
  | Fischer--Ladner closure

 \knowledge{notion}
  | derivatives

\knowledge{notion}
  | regular expressions

\knowledge{notion}
  | positive calculus of relations with transitive closure

\knowledge{notion}
  | model checking problem
  | model checking problems

\knowledge{notion}
  | expressive power

\knowledge{notion}
  | formula@where
  | formulas@where

\knowledge{notion}
  | TB-tree
  | TB-trees

\knowledge{notion}
  | block@TB-tree
  | blocks@TB-tree

\knowledge{notion}
  | width@TB-tree

\ifthenelse{\boolean{debug}}{%
\renewenvironment{toappendix}{
\textcolor{red}{$\downarrow\downarrow\downarrow$
BEGIN appendix (debug mode)
$\downarrow\downarrow\downarrow$}\\}
{
\mbox{}\\\textcolor{red}{$\uparrow\uparrow\uparrow$
END appendix (debug mode)
$\uparrow\uparrow\uparrow$}
} 
\renewenvironment{apxclaimrep}{\begin{claim}}{\end{claim}}
}{}

\title{
Guarded Negation Transitive Closure Logic
}
\titlerunning{Guarded Negation Transitive Closure Logic}

\author{Diego Figueira}{Univ. Bordeaux, CNRS, Bordeaux INP, LaBRI, UMR 5800, France}{diego.figueira@cnrs.fr}{https://orcid.org/0000-0001-7953-3755}{%
}
\author{Santiago Figueira}{University of Buenos Aires, Argentina \and CONICET, Buenos Aires, Argentina}{santiago@dc.uba.ar}{https://orcid.org/0000-0002-8055-397X}{}
\author{Yoshiki Nakamura}{Chiba University, Japan}{nakamura.yoshiki.ny@gmail.com}{https://orcid.org/0000-0003-4106-0408}{%
}
\authorrunning{D. Figueira, S. Figueira, and Y. Nakamura}
\Copyright{Diego Figueira, Santiago Figueira, and Yoshiki Nakamura}

\ccsdesc[500]{Theory of computation~Logic} %

\keywords{Transitive closure logic, Guarded negation, Unary negation, Satisfiability, Model checking}

\category{} %

\ifthenelse{\boolean{arXiv}}{
    \relatedversion{This is the full version of the LICS'26 homonymous paper \cite{thispaper}.} 
}{
    \relatedversion{\AP""Full Version"": \url{https://arxiv.org/abs/2501.15303}%
    } 
}

\funding{This work was supported by JSPS KAKENHI Grant Numbers JP21K13828, JP25K14985; by ANR grants INTENDED (ANR-19-CHIA-0014) and EXPAND (ANR-25-CE23-1215); and by French-Argentinian IRP
SINFIN.}

\acknowledgements{%
We thank anonymous reviewers for their helpful and insightful comments.
}%

\ifthenelse{\boolean{final}}{
    \nolinenumbers %
}{}

\EventEditors{Claudia Faggian and Joost-Pieter Katoen}
\EventNoEds{2}
\EventLongTitle{41st Annual Symposium on Logic in Computer Science (LICS 2026)}
\EventShortTitle{LICS 2026}
\EventAcronym{LICS}
\EventYear{2026}
\EventDate{July 20--23, 2026}
\EventLocation{Lisbon, Portugal}
\EventLogo{}
\SeriesVolume{380}
\ArticleNo{43}

\begin{document}

\ifthenelse{\boolean{final}}{}{
\thispagestyle{empty}
\listoftodos[Comment list]
\clearpage
\thispagestyle{empty}
\listofchanges[show=all]
\clearpage
\setcounter{page}{1}
}

\maketitle

\begin{abstract}
We study the guarded negation fragment of transitive closure logic (\emph{GNTC}).
We show that the satisfiability problem for GNTC is 2ExpTime-complete, by establishing the following reductions:
(i) a polynomial-time reduction from
the satisfiability problem for GNTC to
the satisfiability problem for the unary negation fragment \emph{UNTC} of GNTC, and 
(ii) a direct exponential-time reduction from the satisfiability problem for UNTC to the non-emptiness problem for 2-way alternating parity tree automata.
Furthermore, we show that
the model checking problem for GNTC is $\mathsf{P}^{\mathsf{NP}[\mathcal{O}(\log^2 n)]}$-complete in combined complexity.
Our result implies
$\mathsf{P}^{\mathsf{NP}[\mathcal{O}(\log^2 n)]}$-completeness
for both UNTC and $\mathrm{UNFO}^{\mathrm{reg}}$, which were left open in previous works.
 \end{abstract}

\section{Introduction}
\label{sec:introduction}
A successful approach to obtaining decidable fragments of (fixpoint) first-order logics is to add syntactic restrictions on the use of quantification or negation.
The ``guardedness approach'' leads to restricting its use via adding atoms acting as guards, protecting the syntax from any dangerous use of quantification or negation that may lead to undecidability.
The restriction on quantification leads to ``guarded FO'' (or \AP""GFO"") \cite{AndrekaNB98}, where quantification can only be used in the form $\forall \bar x  ~ \alpha(\bar x \bar y \bar z) \rightarrow \phi(\bar x \bar y)$ or $\exists \bar x  ~ \alpha(\bar x \bar y \bar z) \land \phi(\bar x \bar y)$, where $\alpha$ is a (positive) atom. 
The restriction on negation results in the strictly more expressive fragment of ``guarded negation FO'' (or \reintro{GNFO}) \cite{baranyGuardedNegation2015}, where the negation is restricted to be of the form $\alpha(\bar x \bar y) \land \lnot \phi(\bar x)$ and no universal quantifiers are allowed. 
Both sorts of restrictions ("GFO" and "GNFO") enjoy desirable properties, in particular they capture modal logics, and retain some of their features such as decidability of satisfiability and finite model property.

Adding \emph{recursion} to these logics is a natural next step since many important inductively-defined concepts such as reachability cannot be expressed in FO. 
In this regard, the guardedness restrictions can also be extended to least fixpoint operators while preserving decidability.
The extension of "GNFO" named \reintro{guarded negation fixpoint logic} (\AP""GNFP"") \cite{baranyGuardedNegation2015},
is then the fragment of (first-order) least fixpoint logic (\AP""LFP""\phantomintro{LFP formula}) obtained by also requiring that fixpoint variables be guarded (in addition to the guarded negation restriction).
"GNFP" is a highly expressive logic, in particular extending the modal $\mu$-calculus with backward modalities \cite{vardiReasoningTwowayAutomata1998}. "GNFP" still enjoys a decidable, "2ExpTime"-complete, "satisfiability problem"  \cite[Theorem 4.4]{baranyGuardedNegation2015} -- just like "GNFO" -- while it no longer has the finite model property.

However, a limitation of "GNFP" is that it cannot express the transitive closure of formulas \cite[Proposition 2]{benediktStepExpressivenessDecidable2016}.\footnote{The proof of this fact can be found in {\cite[Appendix A.1]{benediktStepExpressivenessDecidable2016full}}.}
The ability to use transitivity to ``navigate'' the structure is a desirable feature 
in many scenarios,
and arguably the most basic form of recursion.
But expressing the transitive closure $\TC{\phi}{x y}(x,y)$ of a binary "GNFO" formula $\phi$ (testing that there is a path from $x$ to $y$ in the graph induced by the interpretation of $\phi$) would require the use of \emph{unguarded parameters} inside fixpoint operators, which is forbidden in "GNFP".

In view of this, \cite{benediktStepExpressivenessDecidable2016} studies "guarded negation fixpoint logic with unguarded parameters" (\AP""GNFP-UP""), which extends "GNFP" while preserving decidability, and in particular captures transitive closure formulas.
Yet, the expressive power that "GNFP-UP" adds to "GNFP" comes at a high computational price, as the complexity of the "satisfiability problem" increases from "2ExpTime" to non-elementary \cite[Theorem 20]{benediktStepExpressivenessDecidable2016}.
The non-elementary procedure has also matching lower bounds; however, such bounds make use of formulas that go well beyond the expressive power of transitive closure formulas.

This begs the question of whether the addition of transitive closure formulas to "GNFO" already shows an inherent non-elementary behaviour, or if it can be rather better behaved. 
To investigate this, we consider \reintro{GNTC}, the extension of "GNFO" with transitive-closure formulas,%
\footnote{%
In "GNTC", transitive closure formulas do not contain \emph{unguarded parameters}.
For the extension of "GNTC" with unguarded parameters in transitive closure formulas (""GNTC-UP""), the "satisfiability problem" remains decidable (since "GNTC-UP" is contained in "GNFP-UP") but it is non-elementary.
This is because the containment problem for (existential) positive first-order logic with unary transitive closure (called \textrm{PFO-TC1}) is already non-elementary \cite[Theorem~4]{bourhisHowBestNest2014} (see also \cite{bourhisReasonableHighlyExpressive2015}).
For the same reason, "UNTC" with unguarded parameters in transitive closure formulas (""UNTC-UP"")
is also decidable but non-elementary.
} 
which is incomparable with "GNFP" in terms of expressive power, but contained in "GNFP-UP".
Concretely, we address the following question: \emph{Is the satisfiability of "GNTC" elementary?} We answer this question affirmatively.

\subparagraph{Contributions.}
We show that the satisfiability problem for "GNTC" is "2ExpTime"-complete, matching "GNFO", whose lower bound holds even on fixed binary signatures.

We first show that the satisfiability problem for "GNTC" is equivalent -- modulo polynomial-time reductions -- to the restriction on its \emph{unary-negation} fragment, or "UNTC" (\Cref{thm:GNTC-UNTC}). Concretely, \reintro{UNTC} is defined 
by restricting negations to be of the form $\lnot \phi(x)$, and transitive closure to be applied only to formulas with two free variables. %
This reduction preserves "satisfiability@satisfiable" and "finite-satisfiability@finitely-satisfiable".

Once this reduction is in place, the "2ExpTime" upper bound for "GNTC" follows from the (unpublished) result of \cite[Theorem 8.5]{figueiraCommonAncestorPDL2025} establishing a "2ExpTime" bound for the satisfiability of "UNTC".
The proof of \cite{figueiraCommonAncestorPDL2025} -- heavily relying on prior work \cite{gollerPDLIntersectionConverse2009} -- consists of a long chain of non-trivial reductions to the non-emptiness of "2-way alternating parity tree automata" ("2APTA") from which it is hard to derive any intuition.
Here we propose an alternative proof for the satisfiability of "UNTC" based on a \emph{direct reduction} to the non-emptiness of "2APTA", from which the tight complexity upper bound follows. That is, we show an exponential-time reduction from the "satisfiability problem" for "UNTC" to the non-emptiness problem for "2APTA", and since the latter is in "ExpTime" \cite[\S4]{vardiReasoningTwowayAutomata1998}, the "2ExpTime" bound follows (\Cref{corollary: 2-EXPTIME GNTC}).
While our proof is not simple, it is arguably simpler than the previous sequence of reductions.
It also has the advantage of being direct, shorter, and self-contained. 
Further, we show how our reduction can be intuitively interpreted via elementary derivation rules of a ``local'' %
(model/satisfiability)
checker running on tree decompositions.

The proof we exhibit also gives an alternative direct proof of known results on the satisfiability for extensions of Propositional Dynamic Logics -- known as "ICPDL" and "UCPDL+" \cite{figueiraPDLSteroidsExpressive2023,gollerPDLIntersectionConverse2009} --, since they admit trivial polynomial translations into "UNTC" \cite[Proposition 8.3]{figueiraCommonAncestorPDL2025}.

Finally, we also study the model checking problem, "ie", the problem of, given a (finite) structure $\struc{A}$ and a "GNTC" sentence $\phi$, whether $\struc{A} \models \phi$. We show that, in combined complexity, this problem sits in the same complexity class as "GNFO", namely "PNPlog2"-complete\footnote{The class of decision problems solvable in polynomial time with $\+O(\log^2(n))$ calls to an "NP" oracle.} (\Cref{theorem: model checking}).
This is in contrast with the complexity of "GNFP", which is hard for "PNP", even when restricted to its unary negation fragment \cite[Theorem 5.5]{segoufinUnaryNegation2013}.
Our results also establish "PNPlog2"-completeness of model checking for "UNFOreg" and "UNTC", which were left open in previous works \cite[Open Problem 2 and Footnote 16 of version v1]{figueiraCommonAncestorPDL2025}.

\subparagraph{Related work.}
"GNTC" is a syntactic fragment of (first-order) \reintro{transitive closure logic} (\reintro{TC}).
The "satisfiability problem" for full "TC" is highly undecidable 
\cite[Corollary~6.7]{gradel1999undecidability} (see also
\cite{gradelGuardedFixedPoint1999,immermanBoundaryDecidabilityUndecidability2004}). Apart from the already mentioned work on the highly expressive "GNFP-UP" \cite{benediktStepExpressivenessDecidable2016} which captures "GNTC" and implies its decidability (with a non-elementary bound), we are not aware of other decidable logics capturing "GNTC".

Recent ISO standards for querying property graphs, SQL/PGQ \cite{iso/iecjtc1/sc32ISOIEC9075162023} and GQL \cite{iso/iecjtc1/sc32ISOIEC390752024},
are at least as expressive as "FO" and can express transitive closure \cite{gheerbrantGQLSQLPGQ2025}.
In particular, a large fragment of GQL (``restrictor-free'') is known to be in "TC" \cite[Lemma 1]{FigueiraLP26}.
Further, SQL/PGQ (with view creation) \cite[Corollary 6.3]{rotschieldExpressivenessLanguagesQuerying2025}, as well as the earlier language GraphLog \cite[Theorem 3.3]{consensGraphLogVisualFormalism1990}, precisely capture the expressive power of "TC" for graph database queries.

Another relevant line of work consists in considering extensions of Propositional Dynamic Logics, or ""PDL"" \cite{fischerPropositionalDynamicLogic1979} (also known as $\mathcal{ALC}_{\textup{reg}}$). The logic \AP""ICPDL"" adds intersection and converse to "PDL" modalities (``programs'') and \cite{gollerPDLIntersectionConverse2009,lange2ExpTimeLowerBounds2005} show its satisfiability to be "2ExpTime"-complete. The work \cite{figueiraPDLSteroidsExpressive2023} pushes this even further by adding also conjunctive queries over "PDL" programs. As shown in \cite{figueiraCommonAncestorPDL2025}, a slight extension of the logic of \cite{figueiraPDLSteroidsExpressive2023}, coined \AP""UCPDL+"", is expressive-equivalent to "UNTC" (while in principle exponentially less succinct).%

"UNTC" (or, equivalently, "UCPDL+")  is also known to capture several other logics
such as \AP""UNFOreg"" \cite{jungQueryingUnaryNegation2018}, which is the result of adding regular expressions over binary relations to "UNFO", CQPDL \cite[p.\ 4]{benediktStepExpressivenessDecidable2016}, the \intro*\kl{positive calculus of relations with transitive closure} \cite{pousPositiveCalculusRelations2018,nakamuraPartialDerivativesGraphs2017,nakamuraDerivativesGraphsPositive2025}, as well as several graph database query languages, including Conjunctive Regular Path Queries (CRPQs) and  Regular Queries \cite{reutterRegularQueriesGraph2017} (see \cite[\S 1.1 \& Figure~1]{figueiraCommonAncestorPDL2025} for more examples).

An important relevant related result is that of \cite[Theorem 8.5]{figueiraCommonAncestorPDL2025}, showing a "2ExpTime" upper bound for "UNTC" satisfiability.
One major disadvantage of the proof of \cite[Theorem 8.5]{figueiraCommonAncestorPDL2025} is that it is based on a long chain of reductions: an exponential translation from "UNTC" to "UCPDL+", which in turn has an exponential-time reduction to the satisfiability of "ICPDL", which is shown in \cite[Theorem 4.8]{gollerPDLIntersectionConverse2009} via an exponential-time reduction to two models of automata (in particular, a word automaton that runs on paths labeled by states of a tree automaton), which can then be reduced to the non-emptiness of "2-way alternating parity tree automata" ("2APTA").
Moreover, it is necessary to carefully analyze the reductions to derive an optimal complexity bound.\footnote{%
    It requires showing: 
        that the exponential translation "UNTC" $\leadsto$ "UCPDL+" \cite[Proposition 8.1]{figueiraCommonAncestorPDL2025} produces formulas of polynomial so-called ``conjunctive width''; 
        that while the reduction "UCPDL+" $\leadsto$ "ICPDL" \cite[Lemma 7.11]{figueiraCommonAncestorPDL2025} is exponential (which would lead to a 3"ExpTime" satisfiability procedure for "UNTC") it nevertheless produces formulas of polynomial so-called `intersection-width';
        and that while the reduction "ICPDL" $\leadsto$ "2APTA" is exponential, it yields automata with a number of states which is only exponential in the `intersection-width' \cite[Lemma 3.7]{gollerPDLIntersectionConverse2009}. 
    Thus, the combined reduction of \cite[Theorem 8.5]{figueiraCommonAncestorPDL2025} gives a doubly-exponential procedure by \cite[Lemma 3.7]{gollerPDLIntersectionConverse2009}.%
}
Combining these non-trivial reductions "UNTC" $\leadsto$ "UCPDL+" $\leadsto$ "ICPDL" $\leadsto$ "2APTA" we obtain a correct but rather long and intricate proof, from which it is difficult to gain intuition. The present work can be seen as a simplification and streamlining of this complex chain of reductions.

Finally, the recently introduced Guarded Fragment with Regular Guards (RGF) \cite{BednarczykKieronski-kr25} extends the expressive power of "GFO" and "ICPDL" by allowing to use "ICPDL"-programs as guards. RGF is incomparable with "GNTC" (or "UNTC") in terms of expressive power.

\subparagraph{Organization.} 
After some preliminary definitions in Section~\ref{sec:prelim}, we show the reduction from "GNTC" to "UNTC" in Section~\ref{sec:gntc-untc}. The following Sections \ref{sec:untc-proof-approach}--\ref{section: UNTC} are dedicated to showing decidability for the satisfiability of "UNTC", starting with a high-level description of the proof in Section~\ref{sec:untc-proof-approach}. Finally, we study the model checking problem in Section~\ref{section: model checking} and conclude with Section~\ref{sec:conclusion}.
\ifthenelse{\boolean{arXiv}}{}{All missing proofs can be found in the "full version".}

\section{Preliminaries}
\label{sec:prelim}

\AP We write $\znat$, $\nat$, and $\pnat$ for the \kl{sets} of integers, non-negative integers and positive integers, respectively.
For $l, r \in \znat$, we write $\intro*\range{l}{r}$ for the \kl{set} $\{i \in \znat \mid l \le i 
\le r\}$.
In particular, for $n \in \nat$, we write $\intro*\1 n$ for $\range{1}{n}$.
For a \kl{set} $X$,
we write $\intro*\card X$ for the \intro*\kl{cardinality} of $X$
and $\pset{X}$ for the \kl{set} of \intro*\kl{subsets} of $X$.
We use $\intro*\dcup$ to denote that the set union $\cup$ is disjoint.

\AP
For a "set" $X$ of \intro*\kl{letters}, we write $X^*$ for the "set" of \intro*\kl{words} over $X$.
We write $\intro*\eps$ for the \intro*\kl{empty word} and
$\word[1] \word[1]'$ for the \intro*\kl{concatenation} of \kl{words} $\word[1]$ and $\word[1]'$.
We write $\word[1] \intro*\lepref \word[1]'$ if $\word[1]$ is a \kl{prefix} of $\word[1]'$.
For a sequence $\mul{a} = a_1 \dots a_n$, the \intro*\kl(word){length} \AP$\intro*\seqlen{\mul{a}}$ is $n$, and we write \AP$\mul{a}\intro*\ith{i}$ for $a_i$.

\AP\nointro{node}\nointro{root}\nointro{child}\nointro{leaf}\nointro{path}\nointro{height}\nointro{prefix}\nointro{prefix-closed}%
For a "set" $X$ of \kl{letters}, an $X$-labeled (rooted) \intro*\kl{tree} is a partial function $T \colon \pnat^* \pto X$ such that
its ""domain@@function"" $\intro*\fdom(T)$ is \kl{prefix-closed} and non-empty.
We say that $T$ is \intro*\kl{binary} if $\fdom(T) \subseteq \set{1, 2}^{*}$.
For $\word[1], \word[1]' \in \fdom(T)$ and $\dir \in \znat$, we say that $\word[1]'$ is in the \intro*\kl{direction} $\dir$ from $\word[1]$ if
$\word[1] = \word[1]'$ when $\dir = 0$, 
$\word[1]\dir \lepref \word[1]'$ when $\dir > 0$, and
$\word[1] \not\lepref \word[1]'$ when $\word[1] = \word[1]''(\back \dir)$ for some $\word[1]'' \in \pnat^*$.

We write $\AP\intro*\yndirdom{\word}{\dir}(T)$ %
for the set of all elements in the \kl{direction} $\dir$ from $\word$.
For $\word \in \fdom(T)$ and $\dir \in \znat$,
\AP
we define $\word \intro*\series \dir \in \fdom(T)$ ("cf", \cite[p.\ 285]{gollerPDLIntersectionConverse2009}) as
the "node" adjacent to $\word$ and in the \kl{direction} $\dir$ from $\word$ if it exists, and as undefined otherwise; see \Cref{figure: direction} for an illustration.\footnote{For economy of space, trees depicted in this manuscript grow from left to right.}
\begin{figure}[ht]
    \centering
\begin{tikzpicture}
        \tikzstyle{vert}=[draw = black, circle, fill = gray!10, inner sep = 2pt, minimum size = 1pt, font = \scriptsize]
        \tikzstyle{local}=[draw = blue, line width = 1.pt]
        \tikzstyle{edge}=[draw = gray!30]
        \tikzstyle{zone}=[rounded corners, fill = blue!30, draw = blue, fill opacity=.3, line width = 1.pt]
        \tikzstyle{bag}=[draw=gray, circle]
        \node[vert, xshift = -2.cm, yshift = 0cm](b-1){};
        \node[above = .3ex of b-1, font = \footnotesize] at (b-1.north) {$\word \series (\back \dir)$};
        \draw[zone] ($(b-1) + (-2.6,.4)$) -- ($(b-1) + (.3, .15)$) -- ($(b-1) + (.3,-.15)$) -- ($(b-1) + (-2.6,-.4)$) -- cycle;
        \node[font= \footnotesize] at ($(b-1) + (-1.7, .1)$){\kl{direction} $\back \dir$};
        \node[vert, draw = blue, line width = 1pt, xshift = -0.5cm, yshift = 0cm](b){};
        \node[above ,align =center] at (b.north) {$\word$};
        \node[below right = 1.7ex and -1.7em of b, font= \footnotesize](bn) {\kl{direction} $0$};
        \draw[draw = gray, densely dashed] ($(bn)+(0,.15)$) -- ($(b)+(0,-.15)$);
        \node[vert, xshift = 1.6cm, yshift = .3cm](b1) {};
        \draw[zone] ($(b1) + (2.6,.5)$) -- ($(b1) + (-.3,.15)$) -- ($(b1) + (-.3,-.15)$) -- ($(b1) + (2.6,-.2)$) -- cycle;
        \node[font= \footnotesize] at ($(b1) + (1.8, .18)$){\kl{direction} $1$};
        \node[above = .3ex of b1, font = \footnotesize] at (b1.north) {$\word \series 1$};
        \node[vert, xshift = 1.6cm, yshift = -.3cm](b2) {};
        \draw[zone] ($(b2) + (2.6,.2)$) -- ($(b2) + (-.3,.15)$) -- ($(b2) + (-.3,-.15)$) -- ($(b2) + (2.6,-.5)$) -- cycle;
        \node[font= \footnotesize] at ($(b2) + (1.8, -.18)$){\kl{direction} $2$};
        \node[below = .3ex of b2, font = \footnotesize] at (b2.south) {$\word \series 2$};
        \graph[use existing nodes, edges={color=black, pos = .5, earrow}, edge quotes={inner sep=1pt,font= \scriptsize}]{
            (b-1) ->["$d$"{auto, font = \scriptsize}, gray, pos = .5] b;
            (b) ->["$1$"{auto, font = \scriptsize}, gray, pos = .5] (b1);
            (b) ->["$2$"{auto, font = \scriptsize, below}, gray, pos = .5] (b2);
        };
\end{tikzpicture}
     \caption{Illustration of \kl{directions} (from $\word$) and the operator $\series$.}
    \label{figure: direction}
\end{figure}

\AP 
Let $\intro*\Rels$ be an infinite set of ""relation names"".
We write $\intro*\arity(\rsym)$ for the ""arity"" of a "relation name" $\rsym \in \Rels$.
\AP A \intro*\kl{relational structure} $\ynstruc$ (henceforth \reintro*\kl{structure}) 
consists of a non-empty ""countable""%
\footnote{Both \kl{GNTC} and \kl{UNTC} are semantic fragments
of \kl{least fixpoint logic} (\kl{LFP}),
which has the \intro*\kl{downward L{\"o}wenheim--Skolem property} \cite{flumInfiniteModelTheory1998,gradelGuardedFixedPoint2002}.
Thus, \kl{countable} \kl{structures} suffice for "satisfiability@satisfiable".} %
""domain@@structure"" $\intro*\univ{\ynstruc}$
of ""elements@domain element"" and,
for each "relation name" $\rsym \in \Rels$, a relation \AP$\intro*\interpatom{\ynstruc}{\rsym} \subseteq \univ{\ynstruc}^{\arity(\rsym)}$.
We write \AP$\intro*\allstruc_{\kappa}$ for the class of all \kl{structures} of \kl{cardinality} at most $\kappa$, and let $\intro*\allcountablestruc \defeq \allstruc_{\aleph_0}$ where $\aleph_0$ is the countably infinite cardinal and let $\intro*\allfinstruc \defeq \bigcup_{n \in \pnat} \allstruc_{n}$.

\AP
We use $\ynmulstruc[1], \ynmulstruc[2], \dotsc$ to denote $\allfinstruc$-labeled "binary" "trees".
An element $\bag \in \fdom(\ynmulstruc[2])$ is called a \intro*\kl{bag}. 
Thus, $\ynmulstruc[2](\bag)$ is the finite "structure" labeled at $g$, for any "bag" $g$ of $\ynmulstruc[2]$.
For a \kl{structure} $\ynstruc[1]$,
a \intro*\kl{tree decomposition} of $\ynstruc[1]$ is an $\allfinstruc$-labeled "binary" \kl{tree}\footnote{\label{footnote: binary tree}%
    $\allfinstruc$-labeled \emph{"binary"} \kl{trees} suffice to enumerate countable \kl{structures} of finite \kl{treewidth}, "cf", "eg", \cite[\S 4]{benediktStepExpressivenessDecidable2016}.
    } $\ynmulstruc[2]$ such that
\begin{itemize}
    \item $\univ{\ynstruc[1]} = \bigcup_{\bag \in \fdom(\ynmulstruc[2])} \univ{\ynmulstruc[2](\bag)}$,
    \item $\interpatom{\ynstruc[1]}{\rsym} = \bigcup_{\bag \in \fdom(\ynmulstruc[2])} \interpatom{\ynmulstruc[2](\bag)}{\rsym}$ for each $\rsym \in \Rels$,
    \item $\univ{\ynmulstruc[2](\bag)} \cap \univ{\ynmulstruc[2](\bag'')} \subseteq \univ{\ynmulstruc[2](\bag')}$ for all 
    $\bag, \bag', \bag'' \in \fdom(\ynmulstruc[2])$ such that $\bag'$ is on the path between $\bag$ and $\bag''$.
\end{itemize}
\AP
The \intro*\kl{width} 
of $\ynmulstruc[2]$ is defined as 
$\sup_{\bag \in \fdom(\ynmulstruc[2])} (\card \univ{\ynmulstruc[2](\bag)} - 1)$.
The \intro*\kl{treewidth} %
\cite{robertsonGraphMinorsII1986,courcelleGraphStructureMonadic2012} of a \kl{structure} $\ynstruc[1]$ is the minimum \kl{width} among \kl{tree decompositions} of $\ynstruc[1]$, or $\infty$ if there are no \kl{tree decompositions}.

\subsection{\kl{GNTC} and \kl{UNTC}}\label{subsection: GNTC and UNTC}
In this paper, we consider decidable syntactic fragments of (first-order) ""transitive closure logic"" (\reintro*\kl{TC}) \cite{immermanLanguagesThatCapture1987}%
\ifthenelse{\boolean{arXiv}}{ (see \Cref{section: TC} for a formal definition of \kl{TC}).}{.}%
\begin{toappendix}
\begin{scope}\knowledgeimport{TC}
\subsection{Definition of \kl{TC}}\label{section: TC}
\AP
The set of \intro*\kl{formulas} in \reintro*\kl{TC} is given by the following grammar:
\phantomintro{\TC}
\begin{align*}
    \fml[1], \fml[2], \fml[3]
    &~\Coloneqq~ \rsym \mul{x} \mid \fml[1] \land \fml[2] \mid \fml[1] \lor \fml[2] \mid \exists x \fml 
    \mid \lnot \fml
    \mid \reintro*\TC{\fml}{\mul{u} \mul{v}}\, \mul{x} \mul{y}.
\end{align*}
Here, $\rsym \mul{x}$ satisfies $\rsym \in \Rels \dcup \set{\EQ}$ and
$\seqlen{\mul{x}} = \arity(\rsym)$ where $\arity(\EQ) \defeq 2$.
We use 
$\TC{\fml}{\mul{u} \mul{v}} \mul{x} \mul{y}$ that satisfies the following:
\begin{itemize}
    \item $\mul{x}, \mul{y}, \mul{u}, \mul{v}$ are sequences of "variables" having the same \kl(word){length} $k \ge 1$ (called, the \reintro*\kl{arity} of the \kl{TC} \kl{formula});
    \item $\mul{u} \mul{v}$ is pairwise distinct.
\end{itemize}
For $\mul{y} = y_1 \dots y_k$, we write $\exists \mul{y} \fml$ for the \kl(TC){formula} $\exists y_1 \dots \exists y_k \fml$ (in particular, $\fml$ when $k = 0$).

\AP
An \intro*\kl{atom} is a \kl{formula} of the form $\rsym \mul{x}$ where $\rsym \in \Rels \dcup \set{\EQ}$.
An \intro*\kl{existentially quantified atom} is a \kl{formula} of the form $\exists \mul{y} \rsym \mul{x}$ where $\rsym \in \Rels \dcup \set{\EQ}$ and $\mul{y}$ is a sequence of pairwise distinct "variables".

\AP \phantomintro{first-order logic}%
We say that a \kl{formula} $\fml$ in \kl{TC} is an \intro*\kl{FO formula} if $\fml$ does not contain any \kl{subformula} of the form $\TC{\fml[2]}{\mul{u} \mul{v}} \mul{x} \mul{y}$.
We say that an \kl{FO formula} $\fml$ is an \intro*\kl{EFO formula} (\reintro*\kl{existential FO formula}) if $\fml$ does not contain any occurrence of $\lnot$ that lies outside the scope of a quantifier.

The \reintro*\kl(formula){size} $\reintro*\fmllen{\fml}$ of a \kl{formula} $\fml$ is defined as the total number of symbols occurring in $\fml$; more precisely, it is defined as follows:
\begin{align*}
    \fmllen{\rsym \mul{x}} &\defeq 1 + \seqlen{\mul{x}},&
    \fmllen{\fml[1] \land \fml[2]} &\defeq 1 + \fmllen{\fml[1]} + \fmllen{\fml[2]},\\
    \fmllen{\fml[1] \lor \fml[2]} &\defeq 1 + \fmllen{\fml[1]} + \fmllen{\fml[2]},&
    \fmllen{\exists x \fml[1]} &\defeq 1 + 1 + \fmllen{\fml[1]},\\
    \fmllen{\lnot \fml[1]} &\defeq 1 + \fmllen{\fml[1]},&
    \fmllen{\TC{\fml[1]}{\mul{u} \mul{v}}\, \mul{x} \mul{y}} &\defeq 1 + \seqlen{\mul{u}\mul{v}\mul{x}\mul{y}} + \fmllen{\fml[1]}.
\end{align*}

Given a \kl{structure} $\ynstruc$ and an \reintro*\kl{interpretation} (a partial map) $\inter \colon \Vars \pto \univ{\ynstruc}$,
\begin{itemize}
    \item for a sequence $\mul{a} = a_1 \dots a_n \in \univ{\ynstruc}^*$ and a pairwise distinct sequence $\mul{x} = x_1 \dots x_n \in \Vars^*$ of the same \kl(word){length} $n \ge 0$,
    we write $\inter\intro*\intersubst{\mul{x}}{\mul{a}}$ for the \kl{interpretation} $\inter$ in which each $\inter(x_i)$ has been replaced with $a_i$ for each $i \in \1{n}$;
    \item for a sequence $\mul{x} = x_1 \dots x_n \in \Vars^*$,
    we write $\inter(\mul{x})$ for the sequence $\inter(x_1) \dots \inter(x_n)$.
\end{itemize}

For a \kl{structure} $\ynstruc$, an \kl{interpretation} $\inter \colon \Vars \pto \univ{\ynstruc}$, and a \kl{TC} \kl{formula} $\fml$ such that $\FV(\fml) \subseteq \fdom(\inter)$,
\AP
the \intro*\kl{semantics} $\ynstruc \reintro*\modelsass{\inter} \fml$ is defined as follows:
\begin{align*}
    \ynstruc \modelsass{\inter} \rsym \mul{x} &\quad\defiff\quad \inter(\mul{x}) \in \interpatom{\ynstruc}{\rsym} \mbox{ where $\rsym \in \Rels$},\\
    \ynstruc \modelsass{\inter} x \EQ y &\quad\defiff\quad \inter(x) \EQ \inter(y),\\
    \ynstruc \modelsass{\inter} \fml[1] \land \fml[2] &\quad\defiff\quad \ynstruc \modelsass{\inter} \fml[1] \mbox{ and } \ynstruc \modelsass{\inter} \fml[2],\\
    \ynstruc \modelsass{\inter} \fml[1] \lor \fml[2] &\quad\defiff\quad \ynstruc \modelsass{\inter} \fml[1] \mbox{ or } \ynstruc \modelsass{\inter} \fml[2],\\
    \ynstruc \modelsass{\inter} \exists x \fml[1] &\quad\defiff\quad \ynstruc \modelsass{\inter\intersubst{x}{a}} \fml[1] \mbox{ for some $a \in \univ{\ynstruc}$}, \\
    \ynstruc \modelsass{\inter} \lnot \fml[1] &\quad\defiff\quad \mbox{not } (\ynstruc \modelsass{\inter} \fml[1]), \\
    \ynstruc \modelsass{\inter} \TC{\fml[2]}{\mul{u} \mul{v}}\, \mul{x} \mul{y} &\quad\defiff\quad \mbox{there are $n \ge 0, \mul{a}_0, \dots, \mul{a}_n$ s.t.\ } \\
    &\phantom{\quad\defiff\quad} \begin{cases}\begin{aligned}
        & \mul{a}_0 = \inter(\mul{x}), \mul{a}_n = \inter(\mul{y}),\\
        & \ynstruc \modelsass{\inter\intersubst{\mul{u} \mul{v}}{\mul{a}_{i-1}\mul{a}_{i}}} \fml[2] \mbox{ for $i \in \1{n}$}.
    \end{aligned}\end{cases}
\end{align*}

The \kl{semantics} of \kl{UNTC} and \kl{GNTC} are presented as syntactic fragments of \kl{TC}, respectively.

We say that two \kl{structures} $\ynstruc[1]$ and $\ynstruc[2]$ are \intro*\kl{isomorphic} if
there is a bijection $f \colon \univ{\ynstruc[1]} \to \univ{\ynstruc[2]}$ such that, for every $\rsym \in \Rels$ and $\mul{a} \in \univ{\ynstruc[1]}^{\arity(\rsym)}$,
\[\mul{a} \in \interpatom{\ynstruc[1]}{\rsym} \quad\mbox{iff}\quad f(\mul{a}) \in \interpatom{\ynstruc[2]}{\rsym}.\]
Here, $f(\mul{a})$ is defined as $f(a_1) \dots f(a_k)$ for $\mul{a} = a_1 \dots a_k$.

For two classes $\fmlclass$ and $\fmlclass'$ of \kl{formulas},
$\fmlclass$ has at least as much \intro*\kl{expressive power} as $\fmlclass'$
if for every \kl{formula} $\fml[2]$ in $\fmlclass'$,
there is a \kl{formula} $\fml[1]$ in $\fmlclass$
such that, for every \kl{structure} $\ynstruc$ and every \kl{interpretation} $\inter$ with $\FV(\fml[1]) \cup \FV(\fml[2]) \subseteq \fdom(\inter)$, we have:
$\ynstruc \modelsass{\inter} \fml[2]$ "iff" $\ynstruc \modelsass{\inter} \fml[1]$.
We say that $\fmlclass$ and $\fmlclass'$ have the same \reintro*\kl{expressive power}
if each of them has at least as much \kl{expressive power} as the other.
\end{scope}
 \end{toappendix}
\AP
Let $\intro*\Vars$ be an infinite set of ""variables"", disjoint from $\Rels$.
We write $\intro*\FV(\fml)$ for the set of \AP""free variables"" occurring in a "formula@@TC" $\fml$.
For a sequence $\mul{x}$ of "variables", we use $\fml(\mul{x})$ to indicate that $\fml$ is a "formula@@TC"  where $\FV(\fml)$ is the set of "variables" occurring in $\mul{x}$.
We shall sometimes abuse notation and write $\intro*\toSet{\mul{x}}$ also to denote the \emph{set} of "variables" occurring in $\mul{x}$ -- it will be clear from the context when we use it as a set.
For a "structure" $\struc A$, "formula@@TC" $\fml$ and \AP""interpretation"" ("ie", "variable assignment") $\inter : \Vars \pto \dom A$ "st" $\FV(\fml) \subseteq \fdom(\inter)$,
we write $\struc A \intro*\modelsass{\inter} \fml$ to denote that $\struc A$ satisfies $\fml$ under $\inter$.
We write $\struc A \intro*\modelsnonass \fml$ when $\fml$ is a \AP""sentence@@TC"" ("ie", $\fml$ has no "free variables").
\AP
The \intro*\kl(formula){size} $\intro*\fmllen{\fml}$ of a "formula@@TC" $\fml$ is defined as the total number of symbols occurring in $\fml$.
For a sequence $\mul{y} = (y_1, \dotsc, y_n) \in \Vars^*$ and a pairwise distinct sequence $\mul{x} = (x_1, \dotsc, x_n) \in \Vars^*$ of the same \kl(word){length} $n \ge 0$,
we write $\fml\intro*\fmlsubst{\mul{x}}{\mul{y}}$ for the \kl(TC){formula} $\fml$ in which each \kl{variable} $x_i$ has been replaced with $y_i$ for each $i \in \1{n}$.

\AP
For a logic $\+L$ (such as "GNFO", "GNTC", etc.), we denote by ""SAT-""$\+L$ its \reintro{satisfiability problem}, that is, the problem of, given a "sentence@@TC" $\phi \in \+L$, whether there exists some "structure" $\struc A$ such that $\struc A \modelsnonass \phi$, in which case we say it is \AP""satisfiable"" and that $\struc A$ is a \AP""model of"" $\phi$.\footnote{$\struc A$ may be infinite since we shall deal with logics which do not enjoy the finite model property.}

\begin{scope}\knowledgeimport{GNTC}
The guarded negation fragment of "TC" (\AP""GNTC"") is defined as follows.
\AP
A ""guard"" is an "atom" $\rsym \mul{x}$ where $\rsym \in \Rels \dcup \set{\EQ}$.
Here, $\rsym \mul{x}$ satisfies $\seqlen{\mul{x}} = \arity(\rsym)$ where $\arity(\EQ) \defeq 2$.
We use $\afml[1], \afml[2]$ to denote "guards".%
\footnote{\label{footnote: existentially quantified atoms vs atoms}%
We can extend "guards" to "existentially quantified atoms" obtaining a "semantically equivalent" logic via polynomial-time translations%
\ifthenelse{\boolean{arXiv}}{ ("cf" \Cref{section: footnote: existentially quantified atoms vs atoms}).}{.}
}%
\AP
\phantomintro{GNTC sentence}
The ""GNTC formulas"" are generated by the following grammar:%
\labeltext{eq:guarded-neg}{\;}{(G-N)}\labeltext{eq:guarded-tc}{\;}{(G-TC)} %
\begin{align*}
    \fml[1], \fml[2] &~\Coloneqq~ \afml \mid \fml[1] \land \fml[2] \mid \fml[1] \lor \fml[2] \mid \exists x ~ \fml
    \mid \underbrace{\afml \land \lnot \fml[1]}_{\text{\nameref*{eq:guarded-neg}}}
    \mid \underbrace{\TC{\afml[1] \land \afml[2] \land \fml[1]}{\mul{u} \mul{v}} \, \mul{x} \mul{y}}_{\text{\nameref*{eq:guarded-tc}}}.
\end{align*}
Here, in the form of the ""guarded negation formula"" \nameref{eq:guarded-neg}, we assume that $\fml[1]$ is ``guarded'' by the "atom" $\afml$, "ie", $\FV(\fml[1]) \subseteq \FV(\afml)$.
Also, for any formula of form \nameref{eq:guarded-tc} -- which we call \AP""unparameterized guarded transitive closure formula"", \reintro*\kl{transitive closure formula} or simply \reintro*\kl{TC-formula} --, we assume the following:
\begin{itemize}
    \item $\mul{x}$, $\mul{y}$, $\mul{u}$, and $\mul{v}$ have the same \kl(word){length} (called, the \emph{"arity"} of the "TC-formula"), and $\mul{u}\mul{v}$ is pairwise distinct,
    \item $\mul{u}$ and $\mul{v}$ are guarded by the "atoms" $\afml[1]$ and $\afml[2]$, respectively:
    $\toSet{\mul{u}} \subseteq \FV(\afml[1])$ and $\toSet{\mul{v}} \subseteq \FV(\afml[2])$,
    \item there are no TC-parameters: $\FV(\afml[1] \land \afml[2] \land \fml[1]) \subseteq \toSet{\mul{u} \mul{v}}$.
\end{itemize}
\AP\phantomintro{GNFO}
A "GNTC formula" $\fml$ is a ""GNFO formula"" if $\fml$ does not contain "TC-formulas" (given by the grammar without the \nameref{eq:guarded-tc} rule).

\begin{remark}\label{remark: GNF(TC)}
"GNTC" is similar to ``"GNF(TC)"'' of \cite{benediktStepExpressivenessDecidable2016}, but
modified in the following two points:
\begin{itemize}
    \item our "TC-formulas" have both a left "guard" $\afml[1]$ and a right "guard" $\afml[2]$; and
    \item we use classical guards from other logics ("eg", "GNFO" and "GFO"), as opposed to the (exponential-sized) formula ``gdd'' from \cite{benediktStepExpressivenessDecidable2016}.
\end{itemize}
These modifications do not affect the "expressive power"%
\ifthenelse{\boolean{arXiv}}{ (see \Cref{section: GNF(TC) and GNTC} for a proof).}{.}
Further, the form of "TC-formulas" in "GNTC" is crucial for the polynomial-time reduction from "GNTC" to "UNTC" ("cf" \Cref{sec:gntc-untc}): if $\bar{v}$ is not guarded (as in "GNF(TC)"), we have to unfold the TC once, and it may make an exponential blowup if the nesting of TC is unbounded.
The second modification seems also necessary for the polynomial-time reduction from "GNTC" to "UNTC" in the "model checking problem" of \Cref{section: model checking} ("cf" \cite[Remark 5.2]{baranyGuardedNegation2015}%
\ifthenelse{\boolean{arXiv}}{ and \Cref{remark: guards in model checking}}{}).
\end{remark}

\begin{toappendix}
\subsection{Expressive power of GNTC variants}
\subsubsection{Supplement of \Cref{footnote: existentially quantified atoms vs atoms}: On \texorpdfstring{"existentially quantified atoms" in "GNTC"}{existentially quantified atoms in GNTC}}\label{section: footnote: existentially quantified atoms vs atoms}
Even if we extend "guards" in "GNTC" to "existentially quantified atoms",
this extension does not increase the "expressive power" under polynomial-time reductions.
\begin{proposition}\label{proposition: atoms vs existentially quantified atoms}
The variant of "GNTC" where "guards" are "existentially quantified atoms" has the same "expressive power" as "GNTC" under polynomial-time reductions.
\end{proposition}
\begin{proof}
We can remove existential quantifiers in "existentially quantified atoms" based on the following equivalences:
\begin{itemize}
    \item The "TC-formula"
    $\TC{(\exists \mul{w} \afml(\mul{u} \mul{v}' \mul{w})) \land (\exists \mul{w}' \afml[2](\mul{u}' \mul{v} \mul{w}')) \land \fml[1]}{\mul{u} \mul{v}}\mul{x} \mul{y}$, where $\toSet{\mul{u}'} \subseteq \toSet{\mul{u}}$ and $\toSet{\mul{v}'} \subseteq \toSet{\mul{v}}$,
    is "semantically equivalent" to%
\[\exists \mul{w}\mul{w}' (\TC{\afml(\mul{u} \mul{v}'\mul{s}) \land \afml[2](\mul{u}'\mul{v}\mul{s}') \land \fml[1]}{\mul{u} \mul{s}\mul{s}' \mul{v} \mul{t}\mul{t}'}\mul{x} \mul{w}\mul{w}' \mul{y} \mul{w}\mul{w}').\]
    (Observe that $\mul{t}\mul{t}'$ does not occur in the "TC-formula".
    The valuation for $\mul{s}\mul{s}'$ is reset in each iteration.)
    \item The "guarded negation formula" $(\exists \mul{z} \afml(\mul{x} \mul{y} \mul{z})) \land \lnot \fml[1](\mul{x})$ is "semantically equivalent" to
    $\exists \mul{z} (\afml(\mul{x} \mul{y} \mul{z}) \land \lnot \fml[1] (\mul{x}))$.
\end{itemize}
\end{proof}
\end{toappendix}

\begin{toappendix}
\subsubsection{Supplement of \Cref{remark: GNF(TC)}: \texorpdfstring{Expressive power equivalence between "GNTC" and "GNF(TC)"}{Expressive power equivalence between GNTC and GNF(TC)}}\label{section: GNF(TC) and GNTC}
\begin{scope}\knowledgeimport{GNF(TC)}
We recall "GNF(TC)" of \cite{benediktStepExpressivenessDecidable2016} (where notations are adapted to our paper).
For a sequence $\mul{u}$ of pairwise distinct "variables",
we write $\intro*\gdd \mul{u}$ for the formula consisting of the disjunction of
all "existentially quantified atoms" that use a "relation name" in $\Rels \dcup \set{=}$ and involve all the "variables" in $\mul{u}$.
The \AP\intro*\kl(GNF(TC)){guards} in "GNF(TC)" are generated by the following grammar:
\[\afml[1],\afml[2] ~\Coloneqq~ \gdd \mul{u}.\]

\AP\phantomintro{GNF(TC)}%
The \intro*\kl{GNF(TC) formulas} are generated by the following grammar:
\begin{align*}
    \fml[1], \fml[2] &~\Coloneqq~
    \afml \mid \fml[1] \land \fml[2] \mid \fml[1] \lor \fml[2] \mid \exists x \fml
    \mid \underbrace{\afml \land \lnot \fml[1]}_{\text{\nameref{eq:guarded-neg}}} 
    \mid \underbrace{\TC{\afml \land \fml[1]}{\mul{u} \mul{v}}\mul{x} \mul{y}}_{\text{(G-TC')}}.
\end{align*}
Here, in the form of the \reintro*\kl{guarded negation formula} \nameref{eq:guarded-neg}, the "free variables" of $\fml[1]$ are guarded by $\afml$: $\FV(\fml[1]) \subseteq \FV(\afml)$.
Also, in the form of the "unparameterized guarded transitive closure formula" (G-TC'), 
\begin{itemize}
    \item $\mul{x}$, $\mul{y}$, $\mul{u}$, and $\mul{v}$ have the same length,
    and $\mul{u} \mul{v}$ is a pairwise distinct sequence of "variables",
    \item $\mul{u}$ is guarded by $\afml[1]$ and the $\afml[1]$ does not contain "free variables" except $\mul{u}$:
    $\toSet{\mul{u}} = \FV(\afml[1])$
    (here, $\mul{v}$ is not guarded in general),
    \item there are no TC-parameters: $\FV(\afml[1] \land \fml[1]) \subseteq \toSet{\mul{u} \mul{v}}$.
\end{itemize}
\end{scope}

The definition of "GNTC" differs from that of "GNF(TC)" in the following two aspects:
\begin{enumerate}
    \item \label{item: GNTC and GNF(TC) guard syntax} the notion of "guards" is different,
    \item \label{item: GNTC and GNF(TC) left right guard} in "TC-formulas" of "GNTC" both left and right "guards" are required (as in \nameref{eq:guarded-tc}), whereas in "GNF(TC)", only a left "guard" is required (as in (G-TC')).
\end{enumerate}

Nevertheless, we show that "GNTC" and "GNF(TC)" have the same "expressive power".
First, we show that the difference in \kl{guards} (\ref{item: GNTC and GNF(TC) guard syntax}) does not change the expressive power.
\begin{proposition}\label{proposition: GNF(TC) to GNF(TC) atom}
"GNF(TC)" has the same "expressive power" as the variant of "GNF(TC)" where "guards" are "atoms".
\end{proposition}
\begin{proof}
First, observe that any "GNF(TC)" formula in which "guards" are "atoms" is equivalent to a "GNF(TC)" formula, since $\afml(\mul{u}) \land \fml$ is "semantically equivalent" to $\gdd \mul{u} \land (\afml(\mul{u}) \land \fml)$ -- note that $\gdd \mul{u}$ contains every "atom" $\afml(\mul{u})$ as a disjunct.

We now give a translation from "GNF(TC)" into the variant of "GNF(TC)" where "guards" are "atoms".

\proofcasethin{\nameref{eq:guarded-neg}} Consider a formula of the form $\gdd \mul{u} \land \lnot \fml[1]$.
Let $\gdd \mul{u} = \bigvee_i \exists \mul{y}_i \afml[2]_i(\mul{u} \mul{y}_i)$, where each $\afml[2]_i$ is an "atom".
"Wlog", we can assume $\toSet{\mul{y}_i} \cap \FV(\fml[1]) = \emptyset$ for each $i$.
Then the \kl(GNF(TC)){formula} $\gdd \mul{u} \land \lnot \fml[1]$ is "semantically equivalent" to the following \kl(TC){formula}:
\[\bigvee_i \exists \mul{y}_i (~ \afml[2]_i(\mul{u} \mul{y}_i) \land \lnot \fml[1] ~).\]

\proofcasethin{(G-TC')}
Consider the "TC-formula" $\TC{\gdd \mul{u} \land \fml[1]}{\mul{u} \mul{v}}\mul{x} \mul{y}$.
Let $\gdd \mul{u} = \bigvee_i \exists \mul{y}_i \afml[2]_i(\mul{u} \mul{y}_i)$, where each $\afml[2]_i$ is an "atom".
We can then eliminate $\lor$ and $\exists$, respectively as follows.

\subproofcasethin{Eliminating $\lor$}
Let us consider the "TC-formula"
$\TC{(\afml[1](\mul{u}) \lor \afml[2](\mul{u})) \land \fml[1]}{\mul{u} \mul{v}}\mul{x} \mul{y}$,
where $\afml[1]$ and $\afml[2]$ are disjunction of "existentially quantified atoms" involving all the "variables" in $\mul{u}$.
This \kl(GNF(TC)){formula} is "semantically equivalent" to the following \kl(GNF(TC)){formula}:%
\footnote{"Cf", $(E+F)^* = E^*(FE^*)^*$ in regular expressions.}
\begin{align*}
    &\exists \mul{z}(~\TC{\afml[1](\mul{u}) \land \fml[1]}{\mul{u} \mul{v}} \mul{x} \mul{z}
    \land \TC{~\afml[2](\mul{u}) \land \exists \mul{w} (\fml[1]\fmlsubst{\mul{v}}{\mul{w}} \land 
    \TC{\afml[1](\mul{u}) \land \fml[1]}{\mul{u} \mul{v}} \mul{w} \mul{v})~
    }{\mul{u} \mul{v}} \mul{z} \mul{y} ~).
\end{align*}

\subproofcasethin{Eliminating $\exists$}
Let us consider the "TC-formula"
$\TC{(\exists \mul{w} \afml(\mul{u} \mul{w})) \land \fml[1]}{\mul{u} \mul{v}}\mul{x} \mul{y}$,
where $\afml$ is an "atom".
This \kl(GNF(TC)){formula} is "semantically equivalent" to the following \kl(GNF(TC)){formula}:%
\[\exists \mul{w} (\TC{\afml(\mul{u} \mul{s}) \land \fml[1]}{\mul{u} \mul{s} \mul{v} \mul{t}}\mul{x} \mul{w} \mul{y} \mul{w}).\]
(Observe that $\mul{t}$ does not occur in the "TC-formula".
The valuation for $\bar{s}$ is reset in each iteration.)

Hence, we can translate "GNF(TC)" into the variant of "GNF(TC)" where "guards" are "atoms".
\end{proof}

We next absorb the difference of the form of \kl{guards} (\ref{item: GNTC and GNF(TC) left right guard}).
We transform "GNF(TC) formulas" into "GNTC formulas", by putting the right "guards" with unfolding "TC-formulas".
\begin{proposition}\label{proposition: GNTC and GNF(TC)}
    "GNTC" has the same "expressive power" as "GNF(TC)".
\end{proposition}
\begin{proof}
    \proofcasethin{"GNF(TC)" $\leadsto$ "GNTC"}
    By \Cref{proposition: GNF(TC) to GNF(TC) atom},
    we can transform
    each "GNF(TC)" "formula@@GNF(TC)" into
    a "semantically equivalent" "GNF(TC)" "formula@@TC" where "guards" are "atoms".
    The "TC-formula"
    $\TC{\afml(\mul{u}) \land \fml[1]}{\mul{u} \mul{v}}\mul{x} \mul{y}$, where $\afml$ is an "atom",
    is "semantically equivalent" to the following "formula@@GNTC":
    \[
        (\bigwedge_{i} \mul{x}\ith{i} = \mul{y}\ith{i}) ~\lor~
        \exists \mul{z}\left(~ \TC{\afml(\mul{u}) \land \afml(\mul{v}) \land \fml[1](\mul{u} \mul{v})}{\mul{u} \mul{v}}\mul{x} \mul{z}
        \land \afml(\mul{z}) \land \fml[1](\mul{z} \mul{y}) ~\right).
    \]
    Hence, we can translate "GNF(TC)" into "GNTC".

    \proofcasethin{"GNTC" $\leadsto$ "GNF(TC)"}
    The converse direction is almost trivial, as
    the "TC-formula" $\TC{\afml[1](\mul{u}\mul{v}') \land \afml[2](\mul{u}'\mul{v}) \land \fml[1]}{\mul{u} \mul{v}}\mul{x} \mul{y}$, where $\toSet{\mul{u}'} \subseteq \toSet{\mul{u}}$
    and $\toSet{\mul{v}'} \subseteq \toSet{\mul{v}}$,
    is "semantically equivalent" to
    the "TC-formula"
    $\TC{\gdd \mul{u} \land (\afml[1](\mul{u}\mul{v}') \land \afml[2](\mul{u}'\mul{v}) \land \fml[1])}{\mul{u} \mul{v}}\mul{x} \mul{y}$.
    Hence, we can translate "GNTC" into "GNF(TC)".
\end{proof}

As a corollary of the translation above,
we also have the following:
\begin{corollary}\label{corollary: GNTC robustness}
The following $8 = 2 \times 4$ variants have the same "expressive power":
"GNTC" and "GNF(TC)" where the "guards" are given by
\begin{itemize}
    \item all "atoms", 
    \item all "existentially quantified atoms",
    \item all disjunctions of 
    "existentially quantified atoms"
    involving all the "variables" in $\mul{u}$, or
    \item $\gdd \mul{u}$.
\end{itemize}
(Here, $\mul{u}$ ranges over sequences of pairwise distinct "variables".)
\end{corollary}

\begin{remark}\label{remark: guards in model checking}
The translations of \Cref{proposition: GNTC and GNF(TC)} make an \emph{exponential} blow-up.
Due to this, we consider "GNTC" as an alternative to "GNF(TC)".
In "GNTC", we adopt "(existentially quantified) atoms@atoms" instead of the special \kl(TC){formula} $\gdd$ ("cf", \cite{benediktStepExpressivenessDecidable2016}) as "guards",
because the \kl(TC){formula} \kl(formula){size} of $\gdd \mul{x}$ is \emph{exponential} in the maximum "arity" of $\Rels$;
this %
possibly causes the "model checking problem" potentially hard, %
"cf" \cite[Remark 5.2]{baranyGuardedNegation2015}.
\end{remark}

\end{toappendix}

\end{scope}

\begin{scope}\knowledgeimport{UNTC}
\AP
\phantomintro{UNTC formula}\phantomintro{UNTC sentence}%
The \reintro{unary negation} fragment of "GNTC" (""UNTC"") is given by the following grammar:\footnote{%
Ignoring minor syntactic differences,
"UNTC" 
can be equivalently defined as the syntactic fragment of "GNTC" where "guards" are restricted to formulas ``$x = x$'' ($x \in \Vars$).
}
\begin{align*}
    \fml[1], \fml[2] &~\Coloneqq~ \afml \mid \fml[1] \land \fml[2] \mid \fml[1] \lor \fml[2] \mid \exists x \, \fml
    \mid \lnot \fml[1] \mid \TC{\fml[1]}{u v} \, x y.
\end{align*}
Here, $\afml$ is an "atom" $\rsym \mul{x}$ where $\rsym \in \Rels \dcup \set{=}$, and in the form of the ""negation formula"" $\lnot \fml[1]$, the number of "free variables" is at most one: $\# \FV(\fml[1]) \le 1$.\footnote{\label{footnote: binary negation}%
Allowing binary negation renders "UNTC" undecidable, as is already the case for "UNFO" with atomic negations or with inequality \cite[\S 7.3]{segoufinUnaryNegation2013}.
Even the query equivalence problem for negation-free "UNTC" (""EPUNTC"")
is undecidable with atomic negations (by \cite[Theorem 50]{nakamuraExistentialCalculiRelations2023}) or with inequality (by \cite[Theorem 4.9]{nakamuraUndecidabilityPositiveCalculus2024}\cite{nakamuraUndecidabilityEmptinessProblem2025}).}
In the form of "TC-formula" $\TC{\fml[1]}{u v}\, x y$, 
there are no TC-parameters: $\FV(\fml[1]) \subseteq \set{u,v}$.
\AP \phantomintro{UNFO}\phantomintro{UNFO sentence}%
We say that a "UNTC formula" $\fml$ is a ""UNFO formula"" if $\fml$ does not contain "TC-formulas".
    In terms of expressive power, "GNTC" strictly contains "UNTC" and is incomparable with "GNFP". For example, the formula $\exists y (\theta(x,y)\land \TC{\theta(u,v)}{u v}(y,x))$ where $\theta(x,y) = R(x,y)\land \lnot R(y,x)$ -- saying that $x$ belongs to a positive-length cycle made only of strictly forward $R$-edges -- is not expressible either in "GNFP" or "UNTC" (see 
        \ifthenelse{\boolean{arXiv}}{\Cref{lem:GNTC-vs-UNTC-GNFP} of the Appendix}{the "full version"}
    for a proof).
 
 \begin{toappendix}
    \subsubsection{Proof that "GNTC" strictly contains "UNTC" and is incomparable with "GNFP"}
    \label{section: GNTC-vs-UNTC-GNFP}

    \begin{lemma}\label{lem:GNTC-vs-UNTC-GNFP}
        There 
        is a "GNTC formula" which is
        not expressible in "UNTC" nor "GNFP".
    \end{lemma}
    \begin{proof}
        A simple example witnessing that "GNTC" adds expressive power over both "UNTC" and "GNFP" is the following. Let $\theta(x,y)$ abbreviate $R(x,y)\land \lnot R(y,x)$, and consider the unary "GNTC" formula $\mathsf{SFCycle}_R(x) := \exists y (\theta(x,y)\land \TC{\theta(u,v)}{u v}(y,x) )$, saying that $x$ belongs to a positive-length cycle made only of strictly forward $R$-edges. 

        This query is not definable in "UNTC": extending the usual ``UN-bisimulation'' invariance argument for "UNFO" \cite{segoufinUnaryNegation2013},  one obtains that unary "UNTC" queries are invariant under full UN-bisimulation; the only additional case is transitive closure, where one follows the finite witnessing path step by step using the back-and-forth clauses and the induction hypothesis for the TC matrix. However, a directed $3$-cycle with self-loops at all nodes and the one-point structure with a self-loop are fully UN-bisimilar, while they are separated by  $\mathsf{SFCycle}_R$.
    
        It is not definable in "GNFP" either by adapting the separation of \cite[Appendix~A.1]{benediktStepExpressivenessDecidable2016full}.
    \end{proof}
 \end{toappendix}

In this paper, for \kl{satisfiability problems}, we will use the following model theoretic property.
A similar property holds even for "GNFP-UP" \cite[Proposition 6]{benediktStepExpressivenessDecidable2016}.
\AP
\begin{proposition}[{\intro*\kl{bounded treewidth model property} \cite[Remark 7.2 and Theorem 8.5]{figueiraCommonAncestorPDL2025}}]\label{proposition: tw}
    Every \kl{satisfiable} \kl{UNTC formula} $\fml$ is 
    \kl{satisfiable} in a \kl{structure} of \kl{treewidth} at most $\fmllen{\fml} - 1$.
\end{proposition}
\end{scope}

\subsection{\kl{2APTA}}\label{subsection: 2APTA}
\AP
For \kl{satisfiability problems},
we will give reductions to \intro*\kl{2-way alternating parity tree automata} (\kl{2APTAs}) \cite{vardiReasoningTwowayAutomata1998}.
Here, we use only \kl{binary} \kl{trees} as inputs ("cf" \Cref{footnote: binary tree})
and we employ labeled backward transitions as in \cite{gollerPDLIntersectionConverse2009}\ifthenelse{\boolean{arXiv}}{ (see \Cref{section: 2APTA} for a formal definition).}{.}
The ""non-emptiness problem"" for "2APTA" is the following problem:  given a "2APTA", is there a "tree" in the "language" of the "2APTA"?
We will rely on the following complexity result for "2APTAs".
\begin{proposition}[{\cite{vardiReasoningTwowayAutomata1998}}]\label{proposition: 2APTA complexity}
    The \kl{non-emptiness problem} for \kl{2APTAs} is in "ExpTime".
\end{proposition}

\begin{remark}\label{remark: weak}
    In this paper, we use only \emph{weak} \cite{mullerAlternatingAutomataWeak1986} \kl{2APTAs} %
    (a succinct form of B{\"u}chi automata \cite{kupfermanWeakAlternatingAutomata2001})%
    ,
    as we consider only \kl{TC}, and hence \kl{alternation-free} fixpoint; see also, "eg", \cite{mullerWeakAlternatingAutomata1988,carreiroPDL$mu$calculusSyntactic2014}.
\end{remark}

\begin{toappendix}
\begin{scope}\knowledgeimport{2APTA}
\subsection{Definition of \kl{2APTA}}\label{section: 2APTA}
The set $\PBallfml(X)$ of \reintro*\kl{positive Boolean formulas} over a set $X$ of \reintro*\kl{Boolean variables} is generated by the following grammar:
\begin{align*}
    \doublefml[1], \doublefml[2] &\;\Coloneqq\; p \mid \doublefalse \mid \doubletrue \mid \doublefml[1] \doublevee \doublefml[2] \mid \doublefml[1] \doublewedge \doublefml[2], \tag{$p \in X$}
\end{align*}
where the dot notation is used to distinguish from \kl(GNTC){formulas} in "UNTC" / "GNTC".

For a (non-empty) finite set $X$,
a \reintro*\kl{2-way alternating parity tree automaton} (\kl{2APTA}) \cite{vardiReasoningTwowayAutomata1998}\footnote{We introduce labeled backward transitions, based on \cite{gollerPDLIntersectionConverse2009}.
This slight extension does not affect the complexity of the "non-emptiness problem" (\Cref{proposition: 2APTA complexity}).}
over $X$-labeled \kl{binary}\footnote{In this paper, it suffices to consider \kl{binary} \kl{trees} as inputs ("cf", \Cref{footnote: binary tree}).} \kl{trees} is a tuple $\automaton = \tup{Q, \delta, \Pri, q_0}$ where
\begin{itemize}
    \item $Q$ is a finite set of \AP\intro*\kl{states},
    \item $\delta \colon Q \times X \to \PBallfml(Q \times \range{\back 2}{2})$ is a \AP\intro*\kl{transition function},
    \item $\Pri \colon Q \to \nat$ is a \kl{priority} function,
    \item $q_0 \in Q$ is an initial \kl{state}.
\end{itemize}
Given an $X$-labeled \kl{binary} \kl{tree} $T$ and an $S \in Q \times \fdom(T)$,
a \AP\intro*\kl{run} of $\automaton$ on $T$ starting from $S$ is a $(Q \times \fdom(T))$-labeled \kl{tree} $\trace$ 
of $\trace(\varepsilon) = S$ such that,
for every $\word \in \fdom(\trace)$ with $\trace(\word) = \tup{q, \bag}$,
the \kl{positive Boolean formula} obtained from replacing in $\delta(q, T(\bag))$
each element $\tup{q', \dir'} \in Q \times \range{\back 2}{2}$
such that $\tup{q', \bag \series \dir'}$ appears in a child "node" of $\word$  ("ie", $\trace(\word \dir)$ for some $\dir$)
with ``$\doubletrue$'',
is "semantically equivalent" to $\doubletrue$.
\AP
A \kl{run} $\trace$ is \intro*\kl{accepting} if, for every infinite path $a_1 a_2 \dots$ in $\trace$,
the \AP\intro*\kl{priority} $\intro*\pathPri{\trace}(a_1 a_2 \dots)$ given by\footnote{$\pathPri{\trace}(a_1 a_2 \dots)$ is equivalent to $\min\set*{\Pri(q) \;\middle|\;
    q \in Q \mbox{ and } \exists \bag, \trace(a_1 \dots a_n) = \tup{q, \bag} \mbox{ for infinitely many $n$}
    }$,
    by the pigeonhole principle "wrt" $Q$.}
\[\min\set*{p \in \nat \;\middle|\; \exists q, \exists \bag, \trace(a_1 \dots a_n) = \tup{q, \bag} \mbox{ and } \Pri(q) = p \mbox{ holds for infinitely many $n$} }\]
is even.
The \intro*\kl{language} $\intro*\automatonlang(\automaton)$ is a subset of $X$-labeled \kl{binary} \kl{trees}, defined by:
\begin{align*}
    \automatonlang(\automaton) &\defeq \set*{T \;\middle|\; 
    \mbox{there is an \kl{accepting} \kl{run} of $\automaton$ on $T$ starting from $\tup{q_0, \varepsilon}$}
    }.
\end{align*}
The \AP\intro*\kl{size} $\intro*\automatonsize{\delta}$ of its \kl{transition function} $\delta$ is defined as $\automatonsize{\delta} \defeq \sum_{\tup{q, x} \in Q \times X} \bfmllen{\delta(q, x)}$ where $\AP\intro*\bfmllen{\doublefml}$ is the number of symbols in $\doublefml$.
The \reintro*\kl{size} $\reintro*\automatonsize{\automaton}$ of a \kl{2APTA} $\automaton$
is defined as $\automatonsize{\delta}$ (which dominates $\card Q$).\footnote{
This definition is for bounding \kl(formula){sizes} of \kl{positive Boolean formulas} in $\delta$.
}

\end{scope}
 \end{toappendix}

\section{Reduction to UNTC}
\label{sec:gntc-untc}

In this section we show that there is a polynomial-time reduction from "SAT-""GNTC" to "SAT-""UNTC". 
However, the reduction does not imply a reduction from "SAT-""GNFO" to "SAT-""UNFO" since, as we shall see, it requires the expressive power of transitive closure even when restricted to "GNFO" sentences.

\begin{theorem}\label{thm:GNTC-UNTC}
    There is a polynomial-time translation from "GNTC" sentences to "UNTC" sentences 
    preserving both "satisfiability@satisfiable" and "finite-satisfiability@finitely-satisfiable".%
    \footnote{
    \AP A sentence $\fml$ is ""finitely-satisfiable"" if there exists a \emph{finite} structure that models the sentence.
    The ""finite-satisfiability problem"" of $\+L$ is that given a "sentence@@TC" $\phi \in \+L$, whether $\phi$ is "finitely-satisfiable".
    }
    Further, the "UNTC" formulas in the reduction use only binary relations.
\end{theorem}
\begin{proof}
    The idea is to use fresh "domain elements" to encode tuples, and some binary relations of the form $R^{[i]}(t,c)$ to denote that the tuple $t$ contains the element $c$ in its $i$-th component. 
        For example, in
        \Cref{fig:gnfo-unfo-idea}, the tuple $(c_1,c_2,c_3)$ of $\struc A$ is represented in $\struc B$ by the element $t_2$, together with the pairs $(t_2,c_1)$, $(t_2,c_2)$, and $(t_2,c_3)$ in the interpretations of $R^{[1]}$, $R^{[2]}$, and $R^{[3]}$. Hence, some "domain elements" of $\struc B$ represent tuples of $\struc A$ (which we call ``\kl{tuple elements}'') while others represent the elements contained in the components of tuples of $\struc A$ (``\kl{component elements}'').
\begin{figure}
    \centering
    \includegraphics[width=.7\textwidth]{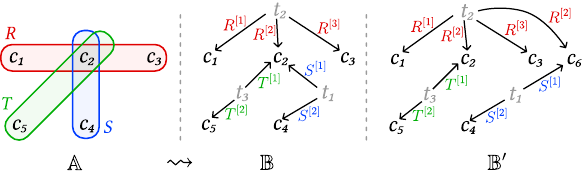}
    \caption{Idea of reduction of \kl{SAT-}\kl{GNTC} to \kl{SAT-}\kl{UNTC}. 
    The structure $\struc A$ is depicted by a set of \kl{domain elements} and some colored hyperedges representing tuples in the ternary relation $R$, the binary relation $T$, and the binary relation $S$. In the structures $\struc B$ and $\struc B'$, all relations are binary, and hence they are depicted as edge-labeled graphs.
    Observe that $\struc B$ and $\struc B'$ both encode $\struc A$, and that $c_2$ and $c_6$ in $\struc B'$ should be regarded as representing the same \kl{element} %
    since they are both the 2nd component in the $R$-tuple represented by $t_2$.%
    }
    \label{fig:gnfo-unfo-idea}
\end{figure}
    However, we cannot enforce these relations to be functional, and thus there could be several $c$ such that $R^{[i]}(t,c)$. 
    In particular, we may have two different but ``semantically equivalent'' encodings.
    Indeed, the same tuple $R(c_1,c_2,c_3)$ of $\struc A$ may also be represented in a structure like $\struc B'$, where $R^{[2]}$ contains both $(t_2,c_2)$ and $(t_2,c_6)$, with $c_6$ coming from another relation $S$. This illustrates that the encoding need not be injective: different 
    "component elements"
    may play the role of the same %
    "domain element"
    of $\struc A$.
    This is not a problem in the presence of transitive closure, since we can encode the ``equality'', consisting of the transitive closure of the relation ``the elements $c$ and $c'$ both appear in the $i$-th position of the same tuple''. We can then do an encoding ensuring that we always work modulo this quotient of equality -- for this we must take advantage of working in the presence of a transitive closure operator.

Now more concretely, let $\phi$ be a \kl{GNTC} sentence, where we assume no repetition in the names of the bound variables.
Let $\AP\intro*\EQrel$ be a fresh binary relation and let $\AP\intro*\phiAtom$ be the formula obtained from $\phi$ by replacing each appearance of $x=y$ with $\EQrel(x,y)$.
Henceforth we shall use $\Rels$ to denote the (finite) set of relation names occurring in $\exists x \ \EQrel(x,x)\land\phiAtom$ (in particular, $\EQrel \in \Rels$),
and we shall use $\AP\intro*\RelsAr$
to denote the (finite) set of pairs $(R,i)$ where $R\in \Rels$ and $i \in \set{1, \dotsc, \arity(R)}$ (in particular, $(\EQrel,1),(\EQrel,2)  \in \RelsAr$).

For each $(R,i) \in \RelsAr$, we use a binary relation $R^{[i]}$. Each pair $(c,c') \in \interpatom{\struc A}{(R^{[i]})}$ denotes that the tuple represented by $c$ contains, in its $i$-th component, the element $c'$. We divide the elements representing tuples from those representing the elements within those tuples:
\AP
\phantomintro{\Phisep}
\phantomintro{\isTup}
\phantomintro{\isConst}
\phantomintro{\PhiEQ}
\begin{align*}
    \reintro*\Phisep &\defeq \lnot \exists x,y,z ~\bigvee_{(R,i),(S,j) \in \RelsAr} R^{[i]}(y,x)
    \land  S^{[j]}(x,z) &
    \reintro*\isTup(t) &\defeq \exists x ~ \bigvee_{(R,i) \in \RelsAr} R^{[i]}(t,x) \\
    \reintro*\PhiEQ &\defeq
            \forall x ~ (\isConst(x) \to  \exists t ~ \bigwedge_{i \in [2]} \EQrel^{[i]}(t, x)) &
            \reintro*\isConst(x) &\defeq \lnot \isTup(x)
\end{align*}
Note that $\PhiEQ$ can be equivalently written in the syntax of "UNFO".
\AP
We shall call "domain elements" satisfying $\isTup$ as ""tuple elements"" and the remaining ones ""component elements"" (denoted by $\isConst$). Note that $\Phisep$ ensures that "component elements" do not have outgoing relations, and hence that relations %
$R^{[i]}$
always link "tuple elements" to "component elements".

We then let the following formula denote that two elements are supposed to represent the same element within a tuple.
\AP
\phantomintro{\eqstar}
\begin{align*}
        (x \reintro*\eqstar y) \defeq \TC[2]{
                \exists t ~ 
                    \bigvee_{i,j \in [2]} (\EQrel^{[i]}(t,u) \land \EQrel^{[j]}(t,v)) 
                    \lor 
                    \bigvee_{(R,i) \in \RelsAr} (R^{[i]}(t,u) \land R^{[i]}(t,v))
                }{u v} x y
\end{align*}

In the reduction we shall use partial functions $\mu : \Vars \pto \Vars \times \Rels \times \Nat$, for which we adopt the notation $\mu(x) = (t,R[i])$ to denote $\mu(x) = (t,R,i)$.
We will inductively define, for subformulas $\psi$ of $\phi$, "UNTC" formulas of the form ``$\Phi_\mu\tup{\psi}$'', where $\fdom(\mu) = \FV(\psi)$, having, as free variables, 
\AP$\intro*\freemu \defeq \set{t \in \Vars : \exists x,i,R \text{ s.t. } \mu(x) = (t,R[i])}$.
The final formula is $\Phi_\emptyset\tup{\phi}$, where $\emptyset$ is the empty partial function.

\smallskip

\proofcase{Definition of $\Phi_\mu\tup{\psi}$}
\emph{Conjunction and disjunction:}
If $*\in\{\land,\lor\}$ then $\Phi_{\mu} \tup{\psi'*\psi''}\defeq\Phi_{\mu'} \tup{\psi'}*\Phi_{\mu''}\tup{\psi''}$, where $\mu', \mu''$ are the restrictions of $\mu$ such that 
$\FV(\psi')=\fdom(\mu')$ and $\FV(\psi'')=\fdom(\mu'')$.
\emph{Atoms:} The case of atoms (or "guards") $R \bar x$ %
is defined by guessing a tuple $t$ which contains $\bar x$ %
when seen as an $R$-tuple.
However, each variable of $\bar x$ is not explicitly available but rather needs to be accessed via some $S^{[j]}(t', \cdot)$ which is stored in $\mu(\bar x \ith{i}) = (t',S[j])$.%

\AP\phantomintro{\PhiTup}
\begin{align*}
    &\Phi_\mu\tup{R(\bar x)} \defeq  {} 
        \exists t ~ \PhiTup{t}{\mu}\tup{R(\bar x)}\text{, where}\\
    &\AP\reintro*\PhiTup{t}{\mu}\tup{ R(x_1,\dotsc,x_n)}\defeq \isTup(t) \quad \land \quad \exists \hat{x}_1,\dotsc,\hat{x}_n ~ \exists \check{x}_1,\dotsc, \check{x}_n \\
    &
            \hspace{3cm}\bigwedge_{i \in \1{n}} R^{[i]}(t,\hat{x}_i) \;\;\land\; 
            \!\!\!\bigwedge_{\substack{i \in \1{n}\\\mu(x_i)=(t_i,S[j_i])}}\!\!\!\! S^{[j_i]}(t_i,\check{x}_i)  \;\;\land\;       \bigwedge_{i\in \1{n}}  \hat{x}_{i} \eqstar \check{x}_i.
\end{align*}

\emph{Quantification:} The existential quantification consists of guessing a "component element" $c$. In our encoding, these elements are endpoints of $R^{[i]}$-relations.
Hence, for doing this, we guess some tuple $t$, a relation $R$ and an index $i$ such that $c$ is the $i$-th component of the $R$-tuple $t$.
\[
        \Phi_\mu\tup{\exists x ~ \phi}  \defeq  \exists t ~ \isTup(t) \land \bigvee_{(R,i) \in \RelsAr} \left((\exists c ~ R^{[i]}(t,c)) \land \Phi_{\mu \dcup \set{x \mapsto (t,R[i])}}\tup{\phi}\right),
    \]
    \emph{Negation:} Guarded negation is transformed into unary negation by guessing a tuple $t$ as before.
    \begin{align*}
        &\Phi_\mu\tup{R(x_1,\dotsc,x_n) \land \lnot \phi} \defeq   %
        \exists t  ~   
        \PhiTup{t}{\mu}\tup{R(x_1,\dotsc,x_n)} \land 
        \lnot \Phi_{\set{x_i \mapsto (t,R[i]) : x_i \in \FV(\phi)} }\tup{\phi}
    \end{align*}
    Observe that $\Phi_{\set{x_i \mapsto (t,R[i]) : x_i \in \FV(\phi)}}\tup{\phi}$ is unary, since it has $t$ as sole free variable, thus fulfilling the requirement of "UNTC".

    \emph{Transitive closure:} Finally, the guarded transitive closure is converted into a unary transitive closure between two tuples $t,t'$:
    \begin{align}
        &\Phi_\mu \tup{\TC{\underbrace{R(\bar u \bar v') \land S(\bar u' \bar v) \land \psi(\bar u \bar v)}_\eta}{\bar u \bar v} \, \bar x \bar y} \defeq \Phi_\mu \tup{\bigwedge_i \EQrel(\bar x\ith{i},\bar y\ith{i})} \label{eq:TC-zerosteps} \quad \lor  \\
        &\hspace{1cm}\Big(\underbrace{\exists t_R, t_S ~ \PhiTup{t_R}{\mu} \tup{R(\bar x \bar y')} \land 
        \PhiTup{t_S}{\mu} \tup{S(\bar x' \bar y)}}_{\text{(a)}} \land 
        \underbrace{\TC{\Phi_{\set{\bar u[i] \mapsto (t,R[i]), \bar v[i] \mapsto (t',S[i])}_i} \tup{\eta}}{t t'} t_R t_S}_{\text{(b)}} \Big)
        \label{eq:TC-recursiveTC-eq:TC-tRtS}
    \end{align}
    where the variables of $\bar u'$ and $\bar v'$ are contained in those of $\bar u$ and $\bar v$, respectively; and $\bar x'$ is such that $\bar x'\ith{i} = \bar x\ith{j}$ "iff" $\bar u'\ith{i} = \bar u\ith{j}$, and similarly for $\bar y'$.
    The formula can be read as follows: ``either the transitive closure is satisfied in 0 steps (line \ref{eq:TC-zerosteps}),
    or we verify $R(\bar x \bar y')$ and $S(\bar x' \bar y)$,
    whence there exist tuples $t_S, t_R$ witnessing these facts (line \ref{eq:TC-recursiveTC-eq:TC-tRtS}-a), and this pair $(t_R,t_S)$ of tuples must be in the transitive closure of $\Phi\tup{\eta}$ under the function mapping  $\bar u[i]$ to the $R[i]$-component of the first tuple $t$, and similarly for $\bar v[i]$ and $S[i]$ with $t'$ (line \ref{eq:TC-recursiveTC-eq:TC-tRtS}-b)''.
    Note that $\Phi_{\set{\bar u[i] \mapsto (t,R[i]), \bar v[i] \mapsto (t',S[i])}_i} \tup{\eta}$ (line \ref{eq:TC-recursiveTC-eq:TC-tRtS}-b) has $t$, $t'$ as sole free variables, meeting the requirements of "UNTC".

\begin{apxclaimrep}\label{claim:gofo-to-unfo}
$\phi$ is ("finitely-@finitely-satisfiable") "satisfiable" "iff" $\Phi_\emptyset(\phiAtom) \land \Phisep  \land \PhiEQ$ is ("finitely-@finitely-satisfiable") "satisfiable".
\end{apxclaimrep}
\begin{claimproof}[Proof sketch]
    \proofcase{$\Rightarrow$}
    For the left-to-right direction, given a "model of" $\phi$, consider adding a new  
    "tuple element"
    for each tuple and the relations $R^{[i]}$ as explained before. It follows that the resulting "structure" verifies $\Phi_\emptyset(\phiAtom) \land \Phisep \land \PhiEQ$.

    \smallskip

    \proofcase{$\Leftarrow$}
    For the right-to-left direction, suppose we have a "model@model of" $\struc B$ of $\Phi_\emptyset(\phiAtom) \land \Phisep \land \PhiEQ$. Consider the structure $\struc A$ resulting from applying the following operations in the following order:

    \begin{enumerate}[(i)]
        \item\label{strucA:quotient} quotienting by $\eqstar$, "ie", replacing each element $c$ with its $\eqstar$-equivalence class $\equiveq c$, where  \AP$\intro*\equiveq c \defeq \set{c' \in \dom B : \struc B \modelsnonass (c \eqstar c')}$;
        \item\label{strucA:rels} for each "arity" $n$ relation $R$ of $\Rels$, adding a tuple $(c_1, \dotsc, c_n)$ to $R$ if for some $c$ we have that $(c,c_i)$ is in $R^{[i]}$ for every $i$; and
        \item\label{strucA:remove} removing the "tuple elements", "ie", removing, for each possible $R,i$, all elements $c$ appearing as first component of a $R^{[i]}$-pair.
    \end{enumerate}
    Observe that applying these operations to either $\struc B$ or $\struc B'$ of \Cref{fig:gnfo-unfo-idea} we obtain a "structure" isomorphic to the original structure $\struc A$.
    
    It remains to show that  $\struc A \modelsnonass \phiAtom$.
    The proof goes via structural induction on $\phiAtom$, using the following inductive hypothesis:
    For every 
         subformula $\psi$ of $\phiAtom$;
        $\mu : \Vars \pto \Vars \times \Rels \times \Nat$ such that $\fdom(\mu) = \FV(\psi)$;
        assignment $\xi : \freemu \to \dom B$ such that the image $Im(\xi)$ ranges over the "tuple elements" of $\struc B$,
        assignment $\xi' : \FV(\psi) \to \dom A$ mapping each variable $x \in \FV(\psi)$ to $\equiveq c$, where $c$ is a 
        $R^{[i]}$-successor of $\xi(t)$ in $\struc B$ such that $\mu(x) = (t, R[i])$,
    we have
        $\struc B\modelsass\xi \Phi_\mu\tup{\psi} $ 
            "iff"  
        $\struc A\modelsass{\xi'} \psi $.
\end{claimproof}
\begin{claimproof}[Proof of Claim \ref{claim:gofo-to-unfo}]
\proofcase{$\Rightarrow$}
    For the left-to-right direction, given a "model of" $\phi$, consider adding a new "domain element" for each tuple and the relations $R^{[i]}$ as explained before. It follows that the resulting "structure" verifies $\Phi_\emptyset(\phiAtom) \land \Phisep \land \PhiEQ$.

    \smallskip

    \proofcase{$\Leftarrow$}
    For the right-to-left direction, suppose we have a "model@model of" $\struc B$ of $\Phi_\emptyset(\phiAtom) \land \Phisep \land \PhiEQ$. Consider the structure $\struc A$ resulting from applying the following operations in the following order:

    \begin{enumerate}[(i)]
        \item\label{strucA:quotient full} quotienting by $\eqstar$, that is, replacing every element $c$ with its $\eqstar$-equivalence class $\equiveq c$, defined as  \AP$\reintro*\equiveq c \defeq \set{c' \in \dom B : \struc B \modelsnonass (c \eqstar c')}$; then
        \item\label{strucA:rels full} for each "arity" $n$ relation $R$ of $\Rels$, adding a tuple $(c_1, \dotsc, c_n)$ to $R$ if for some $c$ we have that $(c,c_i)$ is in $R^{[i]}$ for every $i$; and finally
        \item\label{strucA:remove full} removing the "tuple elements", that is, removing, for each possible $R,i$, all elements $c$ appearing as first component of a $R^{[i]}$-pair.
    \end{enumerate}
    Observe that applying these operations to either $\struc B$ or $\struc B'$ of \Cref{fig:gnfo-unfo-idea} we obtain a "structure" isomorphic to the original structure $\struc A$.
    
    It remains to show that  $\struc A \modelsnonass \phiAtom$.
    The proof goes via structural induction on $\phiAtom$, using the following inductive hypothesis:
    For every 
    \begin{itemize}
        \item subformula $\psi$ of $\phiAtom$ 
        \item $\mu : \Vars \pto \Vars \times \Rels \times \Nat$ such that $\fdom(\mu) = \FV(\psi)$
        \item assignment $\xi : \freemu \to \dom B$ such that the image $Im(\xi)$ ranges over the "tuple elements" of $\struc B$,
        \item assignment $\xi' : \FV(\psi) \to \dom A$ mapping each variable $x \in \FV(\psi)$ to $\equiveq c$, where $c$ is a 
        $R^{[i]}$-successor of $\xi(t)$ in $\struc B$ such that $\mu(x) = (t, R[i])$
    \end{itemize}
    we have
\begin{equation}\label{equisat:iff}
\struc B\modelsass\xi \Phi_\mu\tup{\psi}
            \mbox{\quad"iff"\quad}  
        \struc A\modelsass{\xi'} \psi.
\end{equation}
We proceed by induction on the structure of $\psi$.

\smallskip

\proofcase{$\psi= R(x_1,\dotsc,x_n)$}
Let $\struc B'$ be the model resulting after quotienting $\struc B$ by 
$\eqstar$ (step (\ref{strucA:quotient full}) of the construction). By the definition of $\Phi_\mu\tup{\psi}$ we have that $\struc B\modelsass\xi \Phi_\mu\tup{\psi}$ "iff" 
there are $b_0,\dots,b_n,c_1,\dots, c_n\in\dom {B}$ such that
\begin{enumerate}
\item\label{guard:1} $b_i$ is a "tuple element" of $\struc B$ for each  $i=0,\dots,n$, 
\item\label{guard:2} $\equiveq {c_i}$ is a $R^{[i]}$-successor of $\equiveq {b_0}$ in $\struc B'$ for each $i\in\1{%
n}$, and
\item\label{guard:3} $\equiveq {c_i}$ is also a $S_i^{[j_i]}$-successor of $\equiveq {b_i}$ in $\struc B'$ for each $i\in\1{n}$,
\end{enumerate} 
where for each $i\in\1{n}$ we have
$\mu(x_i) = (t_i, S_i[j_i])$,  
$\xi'(x_i)=\equiveq {c_i}$ and
$\xi(t_i)=b_i$. We obtain $\struc A$ by applying steps (\ref{strucA:rels full}) and (\ref{strucA:remove full}) of the construction to $\struc B'$.
One can see that %
$(\xi'(x_1),\dots,\xi'(x_n)%
)\in \interpatom{\struc A}{R}$ "iff" $\struc B\modelsass\xi \Phi_\mu\tup{\psi}$.

\smallskip

\proofcase{$\psi=\exists x ~ \phi$}
By definition, we have that $\struc B\modelsass\xi \Phi_\mu\tup{\exists x ~ \phi}$ "iff"
\[
\struc B\modelsass\xi  \exists t ~ \isTup(t) \land \bigvee_{(R,i) \in \RelsAr} \left(\exists c ~ R^{[i]}(t,c) \land \Phi_{\mu \dcup \set{x \mapsto (t,R[i])}}\tup{\phi}\right),
\]
"iff" there exists a "tuple element" $b\in\dom B$, a $R^{[i]}$-successor $c$ of $b$, and $(R,i)\in\RelsAr$ such that
$\struc B \modelsass{\tilde\xi} \Phi_{\tilde\mu}\tup{\phi}$, where $\tilde\xi=\xi\dcup\{t\mapsto b\}$ and $\tilde\mu=\mu\dcup\{x\mapsto (t,R[i])\}$. By inductive hypothesis, this happens "iff"
\begin{align}\label{claim:gofo-to-unfo:proof}
\parbox{0.8\linewidth}{
there exists a "tuple element" $b\in\dom B$, a $R^{[i]}$-successor $c$ of $b$, and $(R,i)\in\RelsAr$ such that $\struc A\modelsass{\tilde\xi'} \phi$,
where $\tilde\xi'=\xi'\dcup\{x\mapsto \equiveq c\}$, and $c$ is a $R^{[i]}$-successor of $\tilde\xi(t)(=b)$ in $\struc B$
}
\end{align}

For the left-to-right implication of \eqref{equisat:iff}, suppose that $\struc B\modelsass\xi \Phi_\mu\tup{\exists x ~ \phi}$. Hence (\ref{claim:gofo-to-unfo:proof}) is true and then $\struc A\modelsass{\tilde\xi'} \phi$, where $\tilde\xi'$ extends $\xi'$. This implies that $\struc A\modelsass{\xi'} \exists x ~ \phi$.

For the right-to-left implication, suppose $\struc A\modelsass{\xi'} \exists x ~ \phi$. 
Then there exists $\equiveq c\in\dom A$ such that $\struc A\modelsass{\tilde \xi'} \phi$, where ${\tilde \xi'=\xi'\dcup\{x\mapsto\equiveq c\}}$ for some "component element" $c\in\dom B$. Since $\struc B\models\PhiEQ$, there is a "tuple element" $b\in\dom B$ which has $c$ as a $\EQrel^{[1]}$-successor. Then (\ref{claim:gofo-to-unfo:proof}) holds for $(R,i)=(\EQrel,1)$ and therefore $\struc B\modelsass\xi \Phi_\mu\tup{\exists x ~ \phi}$.

\smallskip

\proofcase{$\psi= R(x_1,\dotsc,x_n) \land \lnot \phi$}

Without loss of generality, we can assume that $\FV(\phi) = \set{x_1,\dots,x_n}$.
As in the case of the guards, we have that $\struc B\modelsass\xi \Phi_\mu\tup{\psi}$ "iff" there are $b_0,\dots,b_n,c_1,\dots, c_n\in\dom {B}$ such that items \ref{guard:1}), \ref{guard:2}) and \ref{guard:3}) hold, and also
\begin{equation}\label{eq:gofo-to-unfo:neg}
\struc B \not\modelsass{\{t\mapsto b_0\}} \Phi_{\set{x_i \mapsto (t,R[i])}_{i\in\1{n}}}\tup{\phi}.
\end{equation}
By inductive hypothesis \eqref{eq:gofo-to-unfo:neg} is equivalent to 
\begin{equation}\label{eq:gofo-to-unfo:neg2}
\struc A\not\modelsass{\xi'} \phi,
\end{equation}
where $\xi'=\{x_i\mapsto\equiveq{c_i}:i\in\1{n}\}$.
By item \ref{guard:2}) and the definition of $\struc A$ in (\ref{strucA:rels full}) we have that
$(\equiveq {c_1},\dots,\equiveq {c_n}) \in \interpatom{\struc A}{R}$. Hence \eqref{eq:gofo-to-unfo:neg2} is equivalent to $\struc A\modelsass{\{x_i\mapsto\equiveq{c_i}\}_{i\in\1{n}}} \psi$, and this concludes the proof of this case.

\smallskip

\proofcase{$
\psi=\TC{
    R(\bar u \bar v') \land 
    S(\bar u' \bar v) \land 
    \phi
    }{\bar u \bar v} \bar x \bar y 
$
}
For simplicity, let us assume that $\bar v'=\bar u'=\emptyset$. The general case is similar. Hence we may assume that $\psi$ is of this form:
$$
\psi=
\TC{
    \underbrace{
    \alpha(\bar u) \land \beta(\bar v) \land    
    \phi}_{\eta}
    }{\bar u \bar v} \bar x \bar y
$$
where
\begin{align*}
\alpha(\bar u) &= R(u_1,\dots,u_n),
\\
\beta(\bar v)  &= S(v_1,\dots,v_n),
\end{align*}
and $\seqlen{\bar u}=\seqlen{\bar v}=\seqlen{\bar x}=\seqlen{\bar y}=n$ (we notate $u_i$ for $u\ith{i}$ and the same for the other tuples of length $n$). As we reasoned in the case of the guards, let $\struc B'$ be the model resulting after quotienting $\struc B$ by 
$\eqstar$ (step (\ref{strucA:quotient full}) of the construction of $\struc A$). 

By the definition of $\Phi_\mu\tup{\psi}$ we have that $\struc B\modelsass\xi \Phi_\mu\tup{\psi}$ "iff" one of the following is true:
\begin{enumerate}
    \item\label{guardTC:a} $\struc B\modelsass\xi \Phi_\mu\tup{\EQrel(x_i,y_i)}$ for each $i\in\1{n}$;
    \item\label{guardTC:b} there are $N\geq 2$ and elements
    $(b^{(j)})_{j\in\1{N}}$, 
    $(c_i^{(j)})_{i\in\1{%
    n},j\in\1{N-1}}$, 
    $(d_i^{(j)})_{i\in\1{%
    n},j\in\1{N-1}}$ of $\dom B$, 
such that 
    \begin{enumerate}
    \item\label{guardTC:1} $b^{(j)}$ is a "tuple element" of $\struc B$ for each  $j\in\1{N}$,

    \item \label{guardTC:1b} $b^{(1)}$ and $b^{(N)}$ satisfy
    $\struc B \modelsass{\xi \sqcup \set{t_R \mapsto b^{(1)}}} \PhiTup{t_R}{\mu} \tup{R(\bar x \bar y')}$
    and
    $\struc B \modelsass{\xi \sqcup \set{t_S \mapsto b^{(N)}}} \PhiTup{t_S}{\mu} \tup{S(\bar x' \bar y)}$,
    
    \item\label{guardTC:2} for all $j\in\1{N-1}$, $\equiveq {c_i^{(j)}}$ is a $R^{[i]}$-successor of $\equiveq {b^{(j)}}$ in $\struc B'$ for each $i\in\1{%
    n}$, and $\equiveq {d_i^{(j)}}$ is a $S^{[i]}$-successor of $\equiveq {b^{(j+1)}}$ in $\struc B'$ for each $i\in\1{%
    n}$,

    \item\label{guardTC:4} 
    for all $j\in\1{N-1}$, we have $\struc B\modelsass{\xi_j}\Phi_{\tilde\mu} \tup{\eta}$, where 
    $\tilde\mu={\set{u_i \mapsto (t,R[i]), v_i \mapsto (t',S[i])}_{i\in\1{n}}}$ and
    $\xi_j=\{t\mapsto b^{(j)}, t'\mapsto b^{(j+1)}\}$.
\end{enumerate} 
The case \ref{guardTC:a} is straightforward, since it is just the case of the guards, so let us focus on case \ref{guardTC:b}. By \ref{guardTC:1} and \ref{guardTC:2}, the side conditions of the inductive hypothesis are satisfied; applying it to \ref{guardTC:4}, we have that \ref{guardTC:4} holds "iff" $\struc A\modelsass{\xi'_j} \eta$ for all $j\in\1{N-1}$
    where 
    $$
    \xi'_j=\{u_i\mapsto\equiveq{c_i^{(j)}}, v_i\mapsto\equiveq{d_i^{(j)}}\}_{i\in\1{n}}.
    $$
    Together with \ref{guardTC:1b} and \ref{guardTC:2}, this is true "iff" $\struc A\modelsass{\xi'} \psi$, where $\xi'={\{x_i\mapsto\equiveq{c_i^{(1)}}, y_i\mapsto\equiveq{d_i^{(N-1)}}\}_{i\in\1{n}}}$.
\end{enumerate}

\smallskip

\proofcase{$\psi=\phi_1*\phi_2$ for $*\in\set{\lor,\land}$} Straightforward.
\end{claimproof}
This concludes the reduction.
\end{proof}
\section{Satisfiability of UNTC: Proof Approach}
\label{sec:untc-proof-approach}
"UNTC" enjoys, as all decidable logics of this family, a ``tree-like'' model property. In other words, a sentence $\phi$ is satisfiable if, and only if, it is satisfiable in a "structure" of "treewidth" $\leq \fmllen{\phi}$. 
It is then sufficient to produce an algorithm which answers whether there exists such a bounded "treewidth" "model of" $\phi$.
We shall define a tree automaton $\automaton[1]^{\phi}$ that decides such a property, based on local properties of the bags, in such a way that $\automaton[1]^{\phi}$ has a non-empty language "iff" $\phi$ is satisfiable.
Intuitively, such an automaton $\automaton[1]^{\phi}$ will guess a "tree decomposition" of a (possibly infinite) "structure" and verify that it is a "model of" $\phi$. 

As is usual in similar approaches, we see "tree decompositions" as binary trees labeled by "structures" (``"bags"'') of size $\leq \fmllen{\phi}$ and, in order to make the alphabet finite, we abstract the information contained in "bags" in a standard way (using a fixed set of elements of size $\leq \fmllen{\phi}$). 
However, the crucial and challenging part of the construction is to define the states and transitions so that the automaton $\automaton[1]^{\phi}$ is only of single exponential size, as otherwise the satisfiability procedure would suffer from an exponential blowup.

The way we approach this issue is by defining some ``local'' rules on abstract semantics that need to be verified on the "tree decomposition" to ensure that the denoted "structure" verifies $\phi$.
Such rules are given via a form of ``"local checker@@UNTC"''. 
In the rest of the section we build some more precise intuition on the nature of the abstract semantics and the  tableau-like rules of the  "local checker@@UNTC".
This serves as a guide to the more technical terminology and definitions presented in \Cref{section: abstract semantics on tree decompositions,section: UNFO,section: UNTC}.

\smallskip

A "tree decomposition" $\ynmulstruc$ is just a binary tree, labeled by finite "structures", often referred to as \reintro{bags}.
The "structure" $\glue\ynmulstruc$ that such a decomposition denotes is the result of making the disjoint union of all the finite "structures", and ``gluing'' any two equal elements from adjacent "structures", as will be shown in \Cref{figure: gluing}.
Such "structure" $\glue\ynmulstruc$ is unique up to isomorphisms.
We say that $\ynmulstruc$ has "width" $k-1$ if the maximum "cardinality" among all the labeled "structures" is $k$ ("eg", the one in \Cref{figure: gluing} has "width" 2). %
It is plain that for every "tree decomposition" of "width" less than $k$
there is another one denoting the same "structure" and using only some $k$ fixed ``named'' elements\footnote{This is a standard construction, by inserting an additional "bag" between each pair of adjacent "bags", "cf", "eg", \cite[p.~27:23]{nakamuraDerivativesGraphsPositive2025}.} $\set{\fixedvertex_1, \dots, \fixedvertex_{k}}$, and for this reason (combined with the bounded "treewidth" property of "UNTC", "cf" \Cref{proposition: tw}) we restrict to such "tree decompositions".

Note that, standing at a "bag" $\bag$ of a "tree decomposition" $\ynmulstruc$, we can distinguish the following possible ``\emph{locations}'' where an element of $\glue\ynmulstruc$ can be found:
\begin{itemize}
    \item ``\textsl{the $\fixedvertex_i$ element of this bag}'', to refer to the element of $\glue\ynmulstruc$ which is ``locally'' represented as $\fixedvertex_i$ in $\bag$, "ie", the element of $\glue\ynmulstruc$ obtained as the result of gluing $\fixedvertex_i$ of the "structure" in $\bag$ with some other elements, 
    \item  ``\textsl{in the "direction" $d=1$}'' ("resp" $d=2$) to refer to any (non-local) element of $\glue\ynmulstruc$ that can be found in the left subtree ("resp" right subtree) of $\bag$, and
    \item ``\textsl{in the "direction" $d=\back 1$}'' ("resp" $d=\back 2$) to denote an element which is elsewhere in the tree ("ie", in a node outside the subtree rooted at $\bag$), assuming $\bag$ is a left-child  in $\ynmulstruc$ ("resp" assuming it is a right-child).
\end{itemize}
Hence, each "bag" $\bag$ has a set (or ``domain'') of available locations
$\ADB{\ynmulstruc}{\bag} \subseteq \set{\fixedvertex_1,\dotsc,\fixedvertex_k} \dcup \set{1,2,\back 1,\back 2}$, which depends on the domain of the labeled "structure", and whether $\bag$ is a left/right child and has left/right children.
We can then \emph{abstract} an element $w$ from $\glue\ynmulstruc$ by its location $\absfunc{\ynmulstruc}{\bag}(\word) \in \ADB{\ynmulstruc}{\bag}$ -- "ie", either by a ``local name'' $\fixedvertex_i$ or by a "direction" $d \in \set{1,2,\back 1,\back 2}$ where it can be found from $\bag$.
Conversely, the \emph{"concretization@concretization function"} of an element $\fixedvertex_i$ in the context of $(\ynmulstruc,\bag)$ is the corresponding element of $\glue\ynmulstruc$, and the "concretization@concretization function" of a "direction" $d$ is the set of all elements of $\glue\ynmulstruc$ that can be found in the "direction" $d$ from $\bag$;
denoted by $\conc{\ynmulstruc}{\bag}(\fixedvertex_i)$ and $\conc{\ynmulstruc}{\bag}(d)$ respectively. In this way each "bag" $\bag$ induces the %
\emph{finite}
partition $\set{\conc{\ynmulstruc}{\bag}(z)}_{z \in \ADB{\ynmulstruc}{\bag}}$  on the
(possibly infinite)
domain of $\glue\ynmulstruc$.

In light of this, for reasoning about the satisfiability of "UNTC", instead of working with classical interpretations ("ie", mappings from variables to elements of $\glue\ynmulstruc$) we shall work with \emph{"abstract interpretations"}, mapping each variable $x$ to a set\footnote{%
For convenience, we map to sets, not elements.
This definition simplifies the rule \eqref{rule: mov}, because the "concretization" of a "direction" $d$ corresponds to a set in an adjacent "bag" ("cf" "moving set" $\MD{\ynmulstruc}{\bag}{\dir}$).} $\ainter(x) \subseteq \ADB{\ynmulstruc}{\bag}$ of locations relative to a "bag" $\bag$. A \emph{"concretization"} of such an "abstract interpretation" $\ainter$ is then just a (classical) interpretation with the same domain as $\ainter$, but mapping each variable $x$ to an element of $\conc{\ynmulstruc}{\bag}(z)$ for some $z \in \ainter(x)$ -- "ie", to an element of the "concretization" of some $\ainter(x)$ member.

By means of such "abstract interpretations" we define a \emph{"local checker@@UNTC"} ("LC@@UNTC") over "tree decompositions".
Concretely, an \emph{"instance@@LMC"} %
of the "LC@@UNTC" (called ``\kl(LMC){state}'' in later sections) consists of a "tree decomposition" $\ynmulstruc$, a node $\bag$ thereof, a set $\Gamma$ of "UNTC" formulas, and an "abstract interpretation" $\ainter$ of the free variables of these formulas.
We denote such "instances@@LMC" by $(\Gamma)_{\ainter,\bag}^{\ynmulstruc}$.
The "LC@@UNTC" has the task of checking that there exists some "concretization" of the "abstract interpretation" $\ainter$ (in the context of ($\ynmulstruc,\bag$)), satisfying all formulas of $\Gamma$ on $\glue\ynmulstruc$.
In particular, for a "sentence@@UNTC" $\phi$, the "LC@@UNTC" on the   "instance@@LMC" 
$(\fml)_{\emptyset,\eps}^{\ynmulstruc}$ checks if $\glue\ynmulstruc$ is a "model of" $\fml$.%
\footnote{
        Although 
        the fact of working with abstractions
        may compromise completeness due to the loss of information outside the bag at the given location,
    the \kl(UNTC){local checker} %
    remains complete under 
    the evaluation strategy we will exhibit.
}%

The "LC@@UNTC" will solve this problem by making calls to a Boolean combination of other "LC@@UNTC" "instances@@LMC" -- in some sense mimicking the Boolean operations of the formula it checks.
For this reason it is useful to have two (dual) ``modes'' of "instances@@LMC": the ``positive mode'' (denoted with a superscript 1) defined just as before, and the ``negative mode'' (with superscript 0) testing the opposite -- namely, that no "concretization" verifies all formulas of $\Gamma$. We will hence use positive Boolean formulas on these signed model checking "instances@@LMC" (we shall use $\doublewedge$ and $\doublevee$ as Boolean connectors to avoid confusion with the "UNFO" syntax).
For example, we will denote by 
$\mkstate{(\fmlset[1])}{1}{\ynmulstruc}{\ainter}{\bag} \doublevee \mkstate{(\fmlset[2])}{0}{\ynmulstruc}{\ainter}{\bag'}$
the truth value $p \lor \lnot q$ where $p$ is the truth value of the model checking problem on $\mkstate{(\fmlset[1])}{1}{\ynmulstruc}{\ainter}{\bag}$ and $q$ is the one on $\mkstate{(\fmlset[2])}{1}{\ynmulstruc}{\ainter}{\bag'}$.

With this intuition, the "LC@@UNTC" for "UNFO" is defined via the intuitive derivation rules of \Cref{figure: rules UNFO} -- somewhat akin to sequent calculus rules or tableau rules -- showing how "instances@@LMC" are derived from other "instances@@LMC".\footnote{In these rules, we use ``$\Gamma, \Delta$'' and ``$\Gamma, \fml$'' to denote $\Gamma \cup \Delta$ and $\Gamma \cup \set\fml$ respectively, and the shorthand $\doublewedge^0 \defeq \doublevee$, $\doublewedge^1 \defeq \doublewedge$, $\doublevee^0 \defeq \doublewedge$, $\doublevee^1 \defeq \doublevee$.}
A rule ``$A \leadstoUNFO B$'' should be read as ``\textsl{for checking $A$, it suffices to check $B$}'' if $A$ is in positive mode ("ie", $p=1$) or as ``\textsl{for checking $A$, it is necessary to check $B$}'' if $A$ is in negative mode ($p=0$).\footnote{For understanding rule \eqref{rule: mov}, $\bag \series \dir$ denotes the "bag" adjacent to $\bag$ in the "direction" $\dir$, and $\ainter_{\dir}$ is the expected way of updating the "abstract interpretation" $\ainter$ with the information that we have moved towards $\dir$.}
Notably, to apply the rule \eqref{rule: UN}, all the free "variables" of $\fml$ should be "concretized"; this condition is required for the soundness of the "LC@@UNTC".%
\footnote{While rule \eqref{rule: UN} is complete for unary negation cases, it would be incomplete for "GNFO", "eg", when two variables indicate the vertices not belonging to the same "bag" at the same time.}

The rules for the transitive closure (TC), given in \Cref{figure: rules UNTC}, are slightly more involved and necessitate further definitions. However, the reader can already see that they correspond to the different ways of shrinking a TC unravelling: either by consuming an element from one side if its endpoint is local to the current "bag" -- \eqref{rule: tc-l} and \eqref{rule: tc-r} --, or by splitting it into two TCs when an endpoint is not local -- \eqref{rule: tc-split-l} and \eqref{rule: tc-split-r}.

Using these rules, a %
\emph{"derivation tree@@LMC"} (called `"run@@LMC"' in later sections) for an "instance@@LMC" $\mkstate{(\fmlset[1])}{p}{\ynmulstruc}{\ainter}{\bag}$ is an "instance@@LMC"-labeled (possibly infinite)  tree such that each positive-"instance@@LMC" ("resp" negative-"instance@@LMC") is ``consistent'' with some ("resp" every) applicable rule.\footnote{More precisely, a rule $A \leadstoUNFO \doublefml[1]$ is ``consistent'' with a node if its left-hand side $A$ is the node's label and by replacing all the children labels with ``true''  makes the "positive Boolean formula" $\doublefml[1]$ true.}
For example, if a node in the tree is labeled $A$ and has %
two children with labels $B$ and $C$, %
and if the applicable rules for $A$ are $A \leadstoUNFO B \doublewedge C$ and $A \leadstoUNFO D \doublevee A$, then it will verify the "derivation tree@@LMC" condition at the node if $A$ is positive (it is consistent with the first rule), but not if $A$ is negative (it is not consistent with the second rule).

While these rules can be easily seen to be sound, they still do not yield an algorithm, since: 
    (i) the input "tree decomposition" is an infinite object and 
    (ii) the "derivation trees@@LMC" may be infinite. 
However, the "LC@@UNTC" rules can be straightforwardly interpreted as the "transition function" of a ("2APTA") tree automaton, where "derivation trees@@LMC" %
are seen as ``\kl(LMC){runs}'' and each "LC@@UNTC" "instance@@LMC" $\mkstate{(\fmlset[1])}{p}{\ynmulstruc}{\ainter}{\bag}$ is abstracted as a ``"state@@LMC"'' ${\fmlset[1]}_{\ainter}$ with "priority" $p$.\footnote{As everywhere else in this section, we disregard minor technicalities; for instance, the "2APTA" would actually require some auxiliary "states@@2APTA"%
\ifthenelse{\boolean{arXiv}}{ ("cf" \Cref{definition: 2APTA UNFO,definition: 2APTA UNTC}).}{.}
}
One can then define \emph{"accepting@@LMC"} "derivation trees@@LMC" in the same way as "accepting" "runs@@LMC" for "2APTA", "ie", requiring that every infinite branch has infinitely many negative "instances@@LMC".
With these semantics, %
the "local (model) checker@local checker@UNTC"
is indeed both sound and complete (corollaries of \Cref{theorem: completeness closure UNFO,theorem: completeness closure UNTC}): an "instance@@LMC" is a yes-instance if, and only if, there exists an "accepting@@LMC" "derivation tree@@LMC" that contains it at the "root".

However, the "derivation trees@@LMC" starting with an "instance@@LMC" $\mkstate{(\set{\fml})}{0}{\ynmulstruc}{\ainter}{\eps}$ may in principle yield an automaton $\automaton[1]^{\fml}$ with an infinite number of "states@@2APTA" of the form ${\fmlset[1]}_{\ainter}$. 
To solve this, we show that starting with the "sentence@@UNTC" $\fml$, the number of "states@@LMC" needed in any "accepting" "run@@LMC" is not too many, and can be drawn from a set $\clUNFO(\fml)$ of "states@@LMC" which is of (single) exponential size (\Cref{theorem: completeness closure UNFO,theorem: completeness closure UNTC} and \Cref{proposition: closure UNFO 2,proposition: closure size UNTC 2}).
Crucially, this result relies on devising a \emph{derivation strategy} for the "LC@@UNTC" rules, using only "instances@@LMC" with formulas  from $\clUNFO(\fml)$ which always finds an "accepting@@LMC" "derivation tree@@LMC" whenever there is one -- the intuition behind the strategy will be given in \Cref{section: evaluation strategy UNFO}.
In this way, we obtain an exponential-time reduction from the satisfiability of "UNTC" to the "non-emptiness problem" for "2APTA", which in turn is in "ExpTime" (\Cref{proposition: 2APTA complexity}), thus obtaining a "2ExpTime" procedure (\Cref{corollary: 2-EXPTIME UNTC}).

\subparagraph{Organization of "SAT-"{}"UNTC" proof.}
The next three sections are organized as follows.
In \Cref{section: abstract semantics on tree decompositions} we fix the notations for our "LC@@UNFO".
We then give the "LC@@UNFO" for "UNFO" in \Cref{section: UNFO} and extend it to "UNTC" in \Cref{section: UNTC}.
\section{Abstract Semantics on Tree Decompositions}\label{section: abstract semantics on tree decompositions}
In this section
we introduce an ``abstract'' semantics on \kl{tree decompositions}, an alternative to the standard semantics.
For simplicity, we consider only \kl{structures} of binary relations ("ie" $\arity(R) = 2$ for each $R \in \Rels$), relying on the polynomial-time reduction of \Cref{thm:GNTC-UNTC}, and we consider only the following "structures" whose ``names'' are drawn from the %
\intro*\kl(LMC){fixed names}
$\AP \set{\intro*\fixedvertex_1, \dots, \fixedvertex_{k}}$ %
for \kl{tree decompositions} of \kl{width} at most $k-1$:
\AP\phantomintro{\allatstruc}
\[\reintro*\allatstruc_{k} \;\defeq\; \set{\ynstruc \in \allcountablestruc : \univ{\ynstruc} \subseteq \set{\fixedvertex_1, \dots, \fixedvertex_k}}.\]
Hence, we consider $\allatstruc_{k}$-labeled "binary" \kl(LMC){trees} (henceforth just \intro*\kl(LMC){trees}) to express \kl{tree decompositions} of \kl{width} less than $k$.
For a \kl(LMC){tree} $\ynmulstruc$,
the \AP\intro*\kl{glued structure} %
$\reintro*\glue \ynmulstruc$ is defined as the \kl{structure} obtained from
the \kl(structure){disjoint union} of all \kl{structures} labeled in $\ynmulstruc$
by gluing {elements} having the same name coming from adjacent "bags"\ifthenelse{\boolean{arXiv}}{ (see \Cref{section: definition of gluing} for a formal definition).}{.}
Each {element} of $\glue \ynmulstruc$ is expressed
as a quotient class $\AP\reintro*\quo{\tup{\bag, \vertex}}_{\quorel{\ynmulstruc}}$
"wrt" the ``gluing'' \kl{equivalence relation} {$\reintro*\quorel{\ynmulstruc}$},
where $\tup{\bag, \vertex}$ is a pair of a \kl{bag} $\bag \in \fdom(\ynmulstruc)$ and its {element} $\vertex \in \univ{\ynmulstruc(\bag)}$.
By definition, for every \kl{tree decomposition} $\ynmulstruc[2]$ of any \kl{structure} $\ynstruc$,
the \kl{glued structure} $\glue \ynmulstruc[2]$ is \kl{isomorphic} to $\ynstruc$.

\begin{toappendix}
\subsection{Formal Definition of Gluing}\label{section: definition of gluing}
Formally for an indexed family $\set{\ynstruc_i}_{i \in I}$ of \kl{structures},
the \AP\intro*\kl(structure){disjoint union} $\intro*\strucbigdcup{i \in I} \ynstruc_i$ is the \kl{structure} defined as follows:
\begin{align*}
    \univ{\strucbigdcup{i \in I} \ynstruc_i} &\;\defeq\; \bigcup_{i \in I} \set{\tup{i, \vertex} \mid \vertex \in \univ{\ynstruc_i}},&
    \interpatom{\strucbigdcup{i \in I} \ynstruc_i}{\rsym} &\;\defeq\; \bigcup_{i \in I} \set{\tup{\tup{i, \vertex}, \tup{i, \vertex'}} \mid \tup{\vertex, \vertex'} \in \interpatom{\ynstruc_i}{\rsym}}.
\end{align*}
\AP For a \kl{structure} $\ynstruc$ and an \kl{equivalence relation} $\sim$ on the set $\univ{\ynstruc}$,
the \intro*\kl{quotient structure} $\ynstruc/{\sim}$ is defined as follows:
\begin{align*}
    \univ{\ynstruc/{\sim}} &\;\defeq\; \mbox{``the set of all \kl{equivalence classes} in $\univ{\ynstruc}$ "wrt" $\sim$''},\\
    \interpatom{\ynstruc/{\sim}}{\rsym} &\;\defeq\; \set{\tup{C, C'} \mid \exists \vertex \in C, \exists \vertex' \in C', \tup{\vertex, \vertex'} \in \interpatom{\ynstruc}{\rsym}}.
\end{align*}
Then for a \kl(LMC){tree} $\ynmulstruc$,
the \AP\reintro*\kl{glued structure} $\intro*\glue \ynmulstruc$ is defined as follows:
\[\glue \ynmulstruc \;\defeq\;  (\strucbigdcup{\bag \in \fdom(\ynmulstruc)} \ynmulstruc(\bag)) \; /\; {\quorel{\ynmulstruc}}\]
\AP where $\intro*\quorel{\ynmulstruc}$ is the minimal \kl{equivalence relation} closed under the following rule:
for all adjacent $\bag[1], \bag[2] \in \fdom(\ynmulstruc)$ and all $\vertex \in \univ{\ynmulstruc(\bag)} \cap \univ{\ynmulstruc(\bag[2])}$, $\tup{\bag[1], \vertex} \quorel{\ynmulstruc} \tup{\bag[2], \vertex}$.
We write $\AP\intro*\quo{\tup{\bag, \vertex}}_{\quorel{\ynmulstruc}}$ for the \kl{equivalence class} of $\tup{\bag, \vertex}$ "wrt" $\quorel{\ynmulstruc}$.
\end{toappendix}

For instance, consider the \kl(LMC){tree} $\ynmulstruc$ with a single binary \kl{relation name}, in \Cref{figure: gluing}.
\kl{Structures} are depicted as directed \kl{graphs}, where "node" labels indicate the corresponding (named) \kl{domain elements}.
Then, the \kl{structure} $\glue \ynmulstruc$ has the shape shown in \Cref{figure: gluing},
by taking the quotient "wrt" the \kl{equivalence relation} ${\quorel{\ynmulstruc}}$ expressed by dashed lines.

\begin{figure}[t]
    \centering
\begin{align*}
    \hspace{5em} \ynmulstruc &\;=\; 
        \begin{tikzpicture}[baseline = -.5cm]
            \tikzstyle{vert}=[draw = black, circle, fill = gray!30, inner sep = 2pt, minimum size = 1pt, font = \scriptsize]
            \tikzstyle{edge}=[draw = gray!30]
            \tikzstyle{bag}=[draw=gray, circle]
            \foreach \i/\v/\c in {0/$\fixedvertex_1$/{},1/$\fixedvertex_3$/{},2/$\fixedvertex_2$/{}} {
                \node[vert, \c, xshift = -2.5cm, yshift = -.5cm](b0_\i) at ({270 + 360/3 * (\i - 1)}:.5) {\v};
            }
            \node[bag, ellipse, fit=(b0_0)(b0_1)(b0_2), inner sep = 0pt](b0) {}; 
            \node[above, color = gray, font = \scriptsize] at ($(b0.north)+(0,-.05)$) {\kl{bag} $\eps$};
            \foreach \i/\v/\c in {2/$\fixedvertex_2$/{}} {
                \node[vert, \c, xshift = -1.cm, yshift = .5cm](b1_\i) at ({60 - 360/3 * (\i - 1)}:.4) {\v};
            }
            \node[bag, ellipse, fit=(b1_2), inner sep = 0pt](b1) {}; 
            \node[above, color = gray, font = \scriptsize] at ($(b1.north)+(0,-.05)$) {\kl{bag} $1$};
            \foreach \i/\v/\c in {0/$\fixedvertex_2$/{},1/$\fixedvertex_1$/{}} {
                \node[vert, \c, xshift = 1.cm, yshift = .2cm](b11_\i) at ({0 + 360/2 * (\i - 1)}:.4) {\v};
            }
            \node[bag, ellipse, fit=(b11_0)(b11_1), inner sep = 0.5pt](b11) {}; 
            \node[above, color = gray, font = \scriptsize] at ($(b11.north)+(0,-.05)$) {\kl{bag} $11$};
            \foreach \i/\v/\c in {0/$\fixedvertex_1$/{}} {
                \node[vert, \c, xshift = 3.5cm, yshift = .2cm](b111_\i) at ({0 + 360/2 * (\i - 1)}:.4) {\v};
            }
            \node[bag, ellipse, fit=(b111_0), inner sep = 0.5pt](b111) {}; 
            \node[above, color = gray, font = \scriptsize] at ($(b111.north)+(0,-.05)$) {\kl{bag} $111$};
            \foreach \i/\v/\c in {0/$\fixedvertex_3$/{},1/$\fixedvertex_2$/} {
                \node[vert, \c, xshift = 0cm, yshift = -1.2cm](b2_\i) at ({0 + 360/2 * (\i - 1)}:.4) {\v};
            }
            \node[bag, ellipse, fit=(b2_0)(b2_1), inner sep = 0pt](b2) {};
            \node[above, color = gray, font = \scriptsize] at ($(b2.north)+(0,-.05)$) {\kl{bag} $2$};
            \foreach \i/\v/\c in {0/$\fixedvertex_1$/{},1/$\fixedvertex_3$/{}} {
                \node[vert, \c, xshift = 3.cm, yshift = -1.2cm](b112_\i) at ({0 + 360/2 * (\i - 1)}:.4) {\v};
            }
            \node[bag, ellipse, fit=(b112_0)(b112_1), inner sep = 0pt](b112) {};
            \node[above, color = gray, font = \scriptsize] at ($(b112.north)+(0,-.05)$) {\kl{bag} $112$};
            \graph[use existing nodes, edges={color=black, pos = .5, earrow}, edge quotes={fill=white, inner sep=1pt,font= \scriptsize}]{
                b0 ->["$1$"{auto, font = \scriptsize}, gray, pos = .5] b1;
                b1 ->["$1$"{auto, font = \scriptsize}, gray] b11;
                b11 ->["$1$"{auto, font = \scriptsize}, gray] b111;
                b0 ->["$2$"{auto, font = \scriptsize}, gray] b2;
                b11 ->["$2$"{auto, font = \scriptsize}, gray] b112;
            };
            \graph[use existing nodes, edges={color=black, pos = .5, earrow}, edge quotes={fill=white, inner sep=1pt,font= \scriptsize}]{
                b0_1 -> b0_0; b0_0 -> b0_2;
                b2_1 -> b2_0;
                b11_0 -> b11_1;
                b112_0 -> b112_1;
            };
            \graph[use existing nodes, edges={color=olive, pos = .5, earrow, densely dashed, line width = 1pt,}, edge quotes={draw = gray, inner sep=1pt,font= \scriptsize}]{
                b0_2 --[bend left = 50] b1_2 --[bend left = 50] b11_0; b0_1 --[out = -45, in = -135, looseness = .5] b2_0;
                b0_2 --[out = -15, in = 45, looseness = 1.3] b2_1;
                b11_1 --[bend left = 60] b111_0;
                b11_1 --[bend right = 45] b112_0;
            };
        \end{tikzpicture}\\[-5ex]
        \rotatebox{-90}{$\leadsto$} & \\[-1ex]
        \glue \mul{\ynstruc} &\;=\; \hspace{.5em} \begin{tikzpicture}[baseline = -5.5ex]
            \tikzstyle{vert}=[draw = black, circle, fill = gray!30, inner sep = 1.5pt, minimum size = 1pt, font = \scriptsize]
            \graph[grow right = .6cm, branch down = 3.5ex]{
            {/,2/{$b_1$}[vert, xshift = .5em]} -!-
            {/,/, 1/{$b_3$}[vert]} -!-
            {/, 3/{$b_2$}[vert, xshift = -.5em]} -!-
            {/, 0/{$b_4$}[vert]} -!-
            {/,/,0'/{$b_5$}[vert]}
            };
            \graph[use existing nodes, edges={color=black, pos = .5, earrow}, edge quotes={fill=white, inner sep=1pt,font= \scriptsize}]{
                3 -> 0 -> 0';
                1 -> 2 -> 3-> 1;
            };
            \end{tikzpicture}
             \mbox{\scriptsize \hspace{1.5em} \colorbox{gray!10}{$\begin{aligned}
                b_1 &= \quo{\tup{\eps, \fixedvertex_1}}_{\quorel{\ynmulstruc}} &
                b_4 &= \quo{\tup{11, \fixedvertex_1}}_{\quorel{\ynmulstruc}}\\
                b_2 &= \quo{\tup{\eps, \fixedvertex_2}}_{\quorel{\ynmulstruc}} &
                b_5 &= \quo{\tup{112, \fixedvertex_3}}_{\quorel{\ynmulstruc}}\\
                b_3 &= \quo{\tup{\eps, \fixedvertex_3}}_{\quorel{\ynmulstruc}}
            \end{aligned}$}}
    \end{align*}     \caption{Illustrative example of the \kl{gluing operator} $\glue$.}
    \label{figure: gluing}
\end{figure}

\subsection{The Abstract Semantics Setting}\label{subsection: abstract semantics on tree decompositions}
Inspired by the abstract interpretation from program semantics \cite{cousotAbstractInterpretationUnified1977}, we consider abstracting "nodes" based on \kl{directions} (recall \Cref{sec:prelim}).
\AP The \intro*\kl{abstract domain} $\intro*\AD$ is defined as the following set:  %
\begin{align*}
    \AD &\;\defeq\; \set{\fixedvertex_i \mid i \in \pnat} \dcup \set{\back 2, \back 1, 1, 2}.
\end{align*}
We use $\elem$ to denote an element in $\AD$,
$\vertex$ to denote a \emph{\kl(LMC){fixed name}} in $\set{\fixedvertex_i \mid i \in \pnat}$, and
$\dir$ to denote a \emph{\kl{direction}} in $\set{\back 2, \back 1, 1, 2}$, respectively.

Let $\ynmulstruc$ be a \kl(LMC){tree}.
\AP For each \kl{bag} $\bag \in \fdom(\ynmulstruc)$,
the \intro*\kl{abstraction function} $\intro*\absfunc{\ynmulstruc}{\bag} \colon \univ{{\glue} \ynmulstruc} \to \AD$ is defined as follows:
\begin{align*}
\absfunc{\ynmulstruc}{\bag}(\word)
&\defeq \begin{cases}
    \vertex &\hspace{-0.em} \text{if $\word = \quo{\tup{\bag, \vertex}}_{\quorel{\ynmulstruc}}$},\\
    \dir &\hspace{-0.em} \begin{aligned}[t]
        & \text{else if $\word = \quo{\tup{\bag', \vertex'}}_{\quorel{\ynmulstruc}}$ for some $\bag'$ in the}\\
        & \text{\; \kl{direction} $\dir$ on $\bag$ and some $\vertex' \in \univ{\ynmulstruc(\bag')}$.}
    \end{aligned}
\end{cases}
\end{align*}
For instance, when $\ynmulstruc$ is the \kl(LMC){tree} in \Cref{figure: gluing},
the "elements"
$\quo{\tup{\eps, \fixedvertex_1}}_{\quorel{\ynmulstruc}}$,
$\quo{\tup{\eps, \fixedvertex_2}}_{\quorel{\ynmulstruc}}$,
$\quo{\tup{\eps, \fixedvertex_3}}_{\quorel{\ynmulstruc}}$,
$\quo{\tup{11, \fixedvertex_1}}_{\quorel{\ynmulstruc}}$,
$\quo{\tup{112, \fixedvertex_3}}_{\quorel{\ynmulstruc}}$ are
mapped to $\fixedvertex_1, \fixedvertex_2, \fixedvertex_3, 1, 1$ by $\absfunc{\ynmulstruc}{\eps}$,
mapped to $\back 1, \fixedvertex_2, \back 1, 1, 1$ by $\absfunc{\ynmulstruc}{1}$, and 
mapped to $\back 1, \fixedvertex_2, \back 1, \fixedvertex_1, 2$ by $\absfunc{\ynmulstruc}{11}$.
\Cref{figure: abstraction of vertices} is an illustration of $\absfunc{\ynmulstruc}{11}$.
\begin{figure}[t]
    \centering
\begin{tikzpicture}[baseline = -.5cm]
    \tikzstyle{vert}=[draw = black, circle, fill = gray!30, inner sep = 2pt, minimum size = 1pt, font = \scriptsize]
    \tikzstyle{local}=[draw = blue, line width = 1.pt]
    \tikzstyle{edge}=[draw = gray!30]
    \tikzstyle{zone}=[rounded corners, fill = blue!30, draw = blue, fill opacity=.3, line width = 1.pt]
    \tikzstyle{bag}=[draw=gray, circle]
    \foreach \i/\v/\c in {0//{},1//{},2/$\fixedvertex_2$/{local}} {
        \node[vert, \c, xshift = -2.5cm, yshift = -.5cm](b0_\i) at ({270 + 360/3 * (\i - 1)}:.5) {\v};
    }
    \node[bag, ellipse, fit=(b0_0)(b0_1)(b0_2), inner sep = 0pt](b0) {}; 
    \node[above, color = gray, font = \scriptsize] at (b0.north) {};
    \foreach \i/\v/\c in {2/$\fixedvertex_2$/{local}} {
        \node[vert, \c, xshift = -1.cm, yshift = .5cm](b1_\i) at ({60 - 360/3 * (\i - 1)}:.4) {\v};
    }
    \node[bag, ellipse, fit=(b1_2), inner sep = 0pt](b1) {}; 
    \node[above, color = gray, font = \scriptsize] at (b1.north) {};
    \foreach \i/\v/\c in {0/$\fixedvertex_2$/{local},1/$\fixedvertex_1$/{local}} {
        \node[vert, \c, xshift = 1.cm, yshift = .2cm](b11_\i) at ({0 + 360/2 * (\i - 1)}:.4) {\v};
    }
    \node[bag, ellipse, fit=(b11_0)(b11_1), inner sep = 0.5pt](b11) {}; 
    \node[above, color = gray, font = \scriptsize] at ($(b11.north)+(0,-.05)$) {\textbf{\kl{bag}} $11$};
    \foreach \i/\v/\c in {0/$\fixedvertex_1$/{local}} {
        \node[vert, \c, xshift = 3.5cm, yshift = .2cm](b111_\i) at ({0 + 360/2 * (\i - 1)}:.4) {\v};
    }
    \node[bag, ellipse, fit=(b111_0), inner sep = 0.5pt](b111) {}; 
    \node[above, color = gray, font = \scriptsize] at (b111.north) {};
    \foreach \i/\v/\c in {0//{},1/$\fixedvertex_2$/{local}} {
        \node[vert, \c, xshift = 0cm, yshift = -1.2cm](b2_\i) at ({0 + 360/2 * (\i - 1)}:.4) {\v};
    }
    \node[bag, ellipse, fit=(b2_0)(b2_1), inner sep = 0pt](b2) {};
    \node[above, color = gray, font = \scriptsize] at (b2.north) {};
    \foreach \i/\v/\c in {0/$\fixedvertex_1$/{local}, 1//} {
        \node[vert, \c, xshift = 3.cm, yshift = -1.2cm](b112_\i) at ({0 + 360/2 * (\i - 1)}:.4) {\v};
    }
    \node[bag, ellipse, fit=(b112_0)(b112_1), inner sep = 0pt](b112) {};
    \node[above, color = gray, font = \scriptsize] at (b112.north) {};
    \graph[use existing nodes, edges={color=black, pos = .5, earrow}, edge quotes={fill=white, inner sep=1pt,font= \scriptsize}]{
        b0 ->[{auto, font = \scriptsize}, gray, pos = .5] b1;
        b1 ->["1"{auto, font = \scriptsize}, gray] b11;
        b11 ->["1"{auto, font = \scriptsize}, gray] b111;
        b0 ->[""{auto, font = \scriptsize}, gray] b2;
        b11 ->["2"{auto, font = \scriptsize}, gray] b112;
    };
    \graph[use existing nodes, edges={color=black, pos = .5, earrow}, edge quotes={fill=white, inner sep=1pt,font= \scriptsize}]{
    };
    \graph[use existing nodes, edges={color=olive, pos = .5, earrow, densely dashed, line width = 1pt,}, edge quotes={draw = gray, inner sep=1pt,font= \scriptsize}]{
        b0_2 --[bend left = 50] b1_2 --[bend left = 50] b11_0; b0_1 --[out = -45, in = -135, looseness = .5] b2_0;
        b0_2 --[out = -15, in = 45, looseness = 1.3] b2_1;
        b11_1 --[bend left = 65] b111_0;
        b11_1 --[bend right = 45] b112_0;
    };
    \draw[zone]
    ($(b0_0) + (-.4,.2)$) --
        ($(b0_0) + (-.4,-1.45)$) --
            ($(b2_0) + (.3,-.5)$) -- ($(b2_0) + (.3,.25)$) --
        ($(b0_1) + (1.4,-.25)$) -- ($(b0_1) + (.5,-.1)$)  --
    ($(b0_0) + (.3,.2)$) -- cycle;
    \node[font= \footnotesize] at ($(b0_0) + (0.4, -1.25)$){\kl{direction} $\back 1$};
    \draw[zone]
    ($(b112_1) + (-.3,.2)$) --
        ($(b112_1) + (-.3,-.2)$) -- ($(b112_1) + (1.9,-.5)$) --
    ($(b112_1) + (1.9,.5)$) -- cycle;
    \node[font= \footnotesize] at ($(b112_1) + (1.15, -0.15)$){\kl{direction} $2$};
\end{tikzpicture}
     \caption{Illustration of $\absfunc{\ynmulstruc}{11}$, with $\ynmulstruc$ taken from \Cref{figure: gluing}.}
    \label{figure: abstraction of vertices}
\end{figure}

\AP For each \kl{bag} $\bag \in \fdom(\ynmulstruc)$,
the \intro*\kl{concretization function} $\intro*\conc{\ynmulstruc}{\bag} \colon \AD \to \pset{\univ{\glue \ynmulstruc}}$ is defined as the inverse image function of $\absfunc{\ynmulstruc}{\bag}$.
By definition, for each "element" $\vertex$,
the set $\conc{\ynmulstruc}{\bag}(\vertex)$ is the singleton $\set{[\tup{\bag, \vertex}]_{\quorel{\ynmulstruc}}}$ if
$\vertex$ is in the "direction" $0$ ("ie", in the "bag" $\bag$),
and the empty set $\emptyset$ otherwise;
for each "direction" $\dir$,
the set $\conc{\ynmulstruc}{\bag}(\dir)$ consists of "elements" in the "direction" $\dir$ but not in the "direction" $0$.
Note that $\conc{\ynmulstruc}{\bag}(\dir)$ is possibly empty even if the "bag" $\bag \series \dir$ exists in $\ynmulstruc$.
Additionally, we lift the \kl(function){domain} of $\conc{\ynmulstruc}{\bag}$ from $\AD$ to its powerset $\pset{\AD}$ by $\conc{\ynmulstruc}{\bag}(X) \defeq \bigcup_{\elem \in X} \conc{\ynmulstruc}{\bag}(\elem)$.

\AP
The "abstract domain" $\intro*\ADB{\ynmulstruc}{\bag}$ of $\ynmulstruc$ on $\bag$ is defined as $\ADB{\ynmulstruc}{\bag} \defeq \set{\elem \in \AD \mid \conc{\ynmulstruc}{\bag}(\elem) \neq \emptyset}$.
The family $\set{\conc{\ynmulstruc}{\bag}(\elem)}_{\elem \in \ADB{\ynmulstruc}{\bag}}$ is a partition of $\univ{\glue \ynmulstruc}$.

\AP A (powerset-lifted) \intro*\kl{abstract interpretation} on a \kl{bag} $\bag$ is a partial function $\ainter \colon \Vars \pto \pset{\ADB{\ynmulstruc}{\bag}}$.
For an \kl{interpretation} $\inter \colon \Vars \pto \univ{\glue{\ynmulstruc}}$ and an \kl{abstract interpretation} $\ainter \colon \Vars \pto \pset{\ADB{\ynmulstruc}{\bag}}$ on a \kl{bag} $\bag$ with $\fdom(\ainter) = \fdom(\inter)$,
we say that $\ainter$ is an \AP\intro*\kl{abstraction} of $\inter$ (or, $\inter$ is a \AP\intro*\kl{concretization} of $\ainter$) on $\bag$ if
$\inter(x) \in \conc{\ynmulstruc}{\bag}(\ainter(x))$ for all $x \in \fdom(\inter)$.
Additionally, we will use functions $\ainter$ by lifting the \kl(function){domain} from $\Vars$ to its powerset $\pset{\Vars}$ by $\ainter(X) \defeq \bigcup_{x \in X} \ainter(x)$.

\begin{definition}\label{definition: states}
For a class $\fmlsetclass$ of finite sets of \kl(UNTC){formulas},
the \intro*\kl(LMC){state} set $\intro*\allstates_{\fmlsetclass}$
is defined as the set of $\mkstate{\fmlset}{\pri}{\ynmulstruc}{\ainter}{\bag}$, where
\begin{itemize}
    \item $\fmlset \in \fmlsetclass$,
    \item $\pri \in \set{\prizero, \prione}$ is a \intro*\kl(LMC){priority},
    \item $\ynmulstruc$ is a \kl(LMC){tree},
    \item $\bag \in \fdom(\ynmulstruc)$ is a \kl{bag}, and
    \item $\ainter$ is an \kl{abstract interpretation} on $\bag$ "st" $\FV(\fmlset) \subseteq \fdom(\ainter)$.\lipicsEnd
\end{itemize}
\end{definition}
We use $\mystate[1],\mystate[2],\mystate[3],\dotsc$ to denote \kl(LMC){states}.
\AP The \reintro*\kl(LMC){priority} $\intro*\statePri(\mystate)$ of a \kl(LMC){state} $\mystate = \mkstate{\fmlset}{\pri}{\ynmulstruc}{\ainter}{\bag}$ is defined by $\statePri(\mystate) \defeq \pri$.

We write $\allstates_{\mathrm{\kl{UNFO}}}$ ("resp" $\allstates_{\mathrm{\kl{UNTC}}}$) for the set $\allstates_{\fmlsetclass}$ where $\fmlsetclass$ is the class of all finite sets of \kl{UNFO} ("resp" \kl{UNTC}) \kl(TC){formulas}.
For $k \in \pnat$, let \[\allstates_{\fmlsetclass}^{(k)} \defeq \set{\mkstate{\fmlset}{\pri}{\ynmulstruc}{\ainter}{\bag} \in \allstates_{\fmlsetclass}
\mid \text{$\ynmulstruc$ is an $\allatstruc_{k}$-labeled \kl(LMC){tree}}}.\]

We now define the abstract semantics on \kl{tree decompositions}.
\begin{definition}\label{definition: semantics on tree decomposition}
\AP For $\mkstate{\fmlset[1]}{\pri}{\ynmulstruc}{\ainter}{\bag} \in \allstates_{\fmlsetclass}$,
    we define $\intro*\absmodels \mkstate{\fmlset[1]}{\pri}{\ynmulstruc}{\ainter}{\bag}$ as follows:
    \begin{align*}
        \absmodels \mkstate{\fmlset[1]}{\prione}{\ynmulstruc}{\ainter}{\bag} &\;\defiff\; \text{$\glue \ynmulstruc \modelsass{\inter} \bigwedge \fmlset[1]$ for a \kl{concretization} $\inter$ of $\ainter$ on $\bag$}, &
        \absmodels \mkstate{\fmlset[1]}{\prizero}{\ynmulstruc}{\ainter}{\bag} &\;\defiff\; \not\absmodels \mkstate{\fmlset[1]}{1}{\ynmulstruc}{\ainter}{\bag}.
    \end{align*}
    \lipicsEnd
\end{definition}
In particular, when $\fmlset[1]$ is the singleton of a \kl(UNTC){sentence} $\fml$ and $\pri = \prione$, it coincides with the standard semantics by definition:
\begin{proposition}\label{proposition: standard semantics and semantics on tree decompositions}
    For every \emph{\kl(UNTC){sentence}} $\fml$, we have:
    \[\glue \ynmulstruc \modelsnonass \fml \quad\Longleftrightarrow\quad {} \absmodels \set{\fml}^{1,\ynmulstruc}_{\emptyset, \bag}.\]
    (Hence, this equivalence holds particularly when $\bag = \eps$.)
\end{proposition}
\begin{example}\label{example: abstract semantics}
    Recall the \kl(LMC){tree} $\ynmulstruc$  in \Cref{figure: gluing}.
    We use $\intro*\Erel$ for the unique binary \kl{relation name}.
    Let $\fml \defeq \lnot \exists y \Erel x y$.
    Then,
    $\glue \ynmulstruc \modelsass{\inter} \fml$ "iff"
    $\inter(x)$ has no successor "iff"
    $\inter(x) = b_5$.
    In the abstract semantics, for instance, we have $\absmodels \mkstate{\set{\fml}}{1}{\ynmulstruc}{x \mapsto \set{2}}{11}$ by $b_5 \in \conc{\ynmulstruc}{11}(2)$
    and we have $\not\absmodels \mkstate{\set{\fml}}{1}{\ynmulstruc}{x \mapsto \set{\fixedvertex_1, \fixedvertex_2, \back 1}}{11}$ by $b_5 \not\in \conc{\ynmulstruc}{11}(\set{\fixedvertex_1, \fixedvertex_2, \back 1})$.
\end{example}
In the sequel, we use this abstract semantics as an alternative to the standard semantics.

\AP
Additionally, in \Cref{definition: semantics on tree decomposition},
we canonically lift $\allstates_{\fmlsetclass}$ to the set $\PBallfml(\allstates_{\fmlsetclass})$ of "(positive) Boolean formulas@positive Boolean formulas",
where, for example, ${} \absmodels \doublefml[1] \doublewedge \doublefml[2]$ "iff" ${} \absmodels \doublefml[1] \text{ and } {} \absmodels \doublefml[2]$. 
\AP
Here, the set $\intro*\PBallfml(X)$ of \intro*\kl{positive Boolean formulas} over a set $X$ of \intro*\kl{Boolean variables}
is generated by the 
grammar:
\begin{align*}
    \doublefml[1], \doublefml[2] &\;\Coloneqq\; p \mid \doublefalse \mid \doubletrue \mid \doublefml[1] \doublevee \doublefml[2] \mid \doublefml[1] \doublewedge \doublefml[2], \tag{$p \in X$}
\end{align*}
where the dot notation is used to distinguish from \kl(TC){formulas} in "UNTC" / "GNTC".
\AP
\phantomintro{semantical equivalence relation}
\phantomintro{semantical entailment relation}
We denote by ${\intro*\lequiv}$ ("resp" ${\intro*\lleqq}$\phantomintro\lgeqq) the \kl{semantical equivalence relation} ("resp" "entailment relation@semantical entailment relation").
For $\mathit{op} \in \set{\doublefalse, \doubletrue, \doublevee, \doublewedge}$,
let $\mathit{op}^{\pri}$ be $\mathit{op}$ itself if $\pri = \prione$ and be the ``""dual""'' of $\mathit{op}$ if $\pri = \prizero$.

\subsection{Moving on Bags in the Abstract Semantics}\label{section: mov}

For each \kl{direction} $\dir \in \range{\back 2}{2}$,
the ""neighbouring set"" $\intro*\ND{\ynmulstruc}{\bag}{\dir} \subseteq \ADB{\ynmulstruc}{\bag}$ is defined so that its "concretization" is the set of all "elements" in "bags" in the direction $\dir$ on $\bag$:
\[\conc{\ynmulstruc}{\bag}(\ND{\ynmulstruc}{\bag}{\dir}) \;=\; \bigcup_{\bag' \in \yndirdom{\bag}{\dir}(\ynmulstruc)} \conc{\ynmulstruc}{\bag'}(\univ{\ynmulstruc(\bag')}).\]
More concretely,
$\ND{\ynmulstruc}{\bag}{\dir} \;\defeq\; (\univ{\ynmulstruc(\bag \series \dir)} \cup \set{\dir}) \cap \ADB{\ynmulstruc}{\bag}$ if $\bag \series \dir$ is defined
and $\ND{\ynmulstruc}{\bag}{\dir} \defeq \emptyset$ otherwise.
In particular, $\ND{\ynmulstruc}{\bag}{0} = \univ{\ynmulstruc(\bag)}$.
For instance,
when $\ynmulstruc$ is the \kl(LMC){tree} in \Cref{figure: gluing} and $\bag = 11$,
we have
$\ND{\ynmulstruc}{\bag}{0} = \set{\fixedvertex_1, \fixedvertex_2}$, 
$\ND{\ynmulstruc}{\bag}{1} = \set{\fixedvertex_1}$, 
$\ND{\ynmulstruc}{\bag}{2} = \set{\fixedvertex_1, 2}$, 
$\ND{\ynmulstruc}{\bag}{\back 1} = \set{\fixedvertex_2, \back 1}$, and
$\ND{\ynmulstruc}{\bag}{\back 2} = \emptyset$.%

When $\ainter(\FV(\fmlset[1])) \subseteq \ND{\ynmulstruc}{\bag}{\dir}$ and $\dir \neq 0$,
we can move from $\bag$ to the adjacent \kl{bag} $\bag \series \dir$.
Before defining the "abstract interpretation" obtained from moving towards $\dir$ (which we shall define below as $\ainter_{\dir}$), we need to define the "abstract domain" on which it lives.

For each \kl{direction} $\dir \in \set{\back 2, \back 1, 1, 2}$,
the ""moving set"" $\intro*\MD{\ynmulstruc}{\bag}{\dir} \subseteq \ADB{\ynmulstruc}{\bag \series \dir}$ is defined so that 
its "concretization" on $\bag \series \dir$ coincides with "that@concretization" of $\dir$ on $\bag$:
\[\conc{\ynmulstruc}{\bag \series \dir}(\MD{\ynmulstruc}{\bag}{\dir}) \;=\; \conc{\ynmulstruc}{\bag}(\dir).\]
More concretely, $\MD{\ynmulstruc}{\bag}{\dir} \;\defeq\; \AD^{\ynmulstruc}_{\bag \series \dir} \setminus (\univ{\ynmulstruc(\bag)} \cup \set{\back \dir})$.
For instance, for the \kl(LMC){tree} $\ynmulstruc$ in \Cref{figure: gluing,figure: abstraction of vertices},
$\MD{\ynmulstruc}{1}{1} = \set{\fixedvertex_1, 2}$,
$\MD{\ynmulstruc}{11}{2} = \set{\fixedvertex_3}$, 
and $\MD{\ynmulstruc}{112}{\back 2} = \set{\fixedvertex_2, \back 1}$.
\Cref{figure: mov} depicts $\MD{\ynmulstruc}{1}{1}$.

\begin{figure}[t]
    \centering
\begin{tikzpicture}
    \tikzstyle{vert}=[draw =  black, circle, fill = gray!30, inner sep = 2pt, minimum size = 1pt, font = \scriptsize]
    \tikzstyle{local}=[draw = blue, line width = 1.pt]
    \tikzstyle{edge}=[draw = gray!30]
    \tikzstyle{zone}=[rounded corners, fill = blue!30, draw = blue, fill opacity=.3, line width = 1.pt]
    \tikzstyle{bag}=[draw=gray, circle]
    \foreach \i/\v/\style in {2/$\fixedvertex_2$/{local, line width = .3pt},1//{white}} {
        \node[vert, xshift = -2.cm, yshift = .3cm, apply style/.expand once = \style](b-1_\i) at ({60 - 360/3 * (\i - 1)}:.4) {\v};
    }
    \node[bag, ellipse, fit=(b-1_2), inner sep = 0pt](b-1) {}; 
    \node[above, color = gray, font = \scriptsize] at (b-1.north) {\kl{bag} $1$};
    \foreach \i/\v/\style in {0/$\fixedvertex_2$/{local, line width = .3pt},1/$\fixedvertex_1$/local} {
        \node[vert, xshift = 0cm, yshift = 0cm, apply style/.expand once = \style](b_\i) at ({0 - 360/2 * (\i - 1)}:.35) {\v};
    }
    \node[bag, ellipse, fit=(b_0)(b_1), inner sep = 0pt](b) {}; 
    \node[above, color = gray, font = \scriptsize] at (b.north) {\kl{bag} $1 1$};
    \foreach \i/\v/\style in {0/$\fixedvertex_1$/local} {
        \node[vert, xshift = 3.85cm, yshift = .37cm, apply style/.expand once = \style](b1_\i) at ({-60 + 360/3 * (\i - 1)}:.35) {\v};
    }
    \node[bag, ellipse, fit=(b1_0), inner sep = 0pt](b1) {}; 
    \foreach \i/\v/\style in {0/$\fixedvertex_1$/local,1//{}} {
        \node[vert, xshift = 2.9cm, yshift = -.3cm, apply style/.expand once = \style](b2_\i) at ({0 + 360/2 * (\i - 1)}:.35) {\v};
    }
    \node[bag, ellipse, fit=(b2_0)(b2_1), inner sep = 0pt](b2) {};
    \draw[zone] ($(b2_0) + (3,.5)$) -- ($(b2_0) + (0.38,.15)$) -- ($(b2_0) + (0.38,-.15)$) -- ($(b2_0) + (3,-.5)$) -- cycle;
    \node[font= \footnotesize, align = right] at ($(b2_0) + (2.2, -.01)$){\kl{direction} $2$\\ on \kl{bag}  $1 1$};
    \graph[use existing nodes, edges={color=black, pos = .5, earrow}, edge quotes={fill=white, inner sep=1pt,font= \scriptsize}]{
        (b-1) ->["$1$"{auto, font = \scriptsize}, gray, pos = .5] b;
        b ->["$1$"{auto, font = \scriptsize}, gray] b1;
        b ->["$2$"{auto, font = \scriptsize}, gray, pos = .4] b2;
    };
    \draw[zone, fill opacity = 0, line width = .3pt] ($(b_0) + (.35, .4)$) -- ($(b_0) + (.35,-.4)$) -- ($(b_0) + (6.2,-1.)$) -- ($(b_0) + (6.2,1.)$)  -- cycle;
    \node[font= \footnotesize, align = right] at ($(b1_0) + (1.6, .19)$){\kl{direction} $1$\\on \kl{bag} $1$};
    \draw[color = red, bend right= 40, line width = 1pt, ->] ($(b-1.south)+(-.2,-0.1)$) to["\eqref{rule: mov}"{font=\footnotesize}] ($(b.south)+(-.2,-0.1)$);
\end{tikzpicture}     \caption{Illustration of $\MD{\ynmulstruc}{1}{1}$, with $\ynmulstruc$ taken from \Cref{figure: gluing}.
    The \kl{direction} $1$ on \kl{bag} $1$ has the same \kl{concretization} as $\set{\const{\vertex_1}, 2}$ on \kl{bag} $11$.}
    \label{figure: mov}
\end{figure}

When $\ainter(\FV(\fmlset[1])) \subseteq \ND{\ynmulstruc}{\bag}{\dir}$ and $\dir \in \set{\back 2,\back 1, 1, 2}$,
the \kl{abstract interpretation} $\ainter_{\dir} \colon \Vars \pto \pset{\AD^{\ynmulstruc}_{\bag \series \dir}}$ on $\bag \series \dir$ is defined as follows:
\[\ainter_{\dir}(x) \;\defeq\; \begin{cases}
    (\ainter(x) \setminus \set{\dir}) \cup \MD{\ynmulstruc}{\bag}{\dir} & \text{if $\dir \in \ainter(x)$,}\\
    \ainter(x) & \text{otherwise}.
\end{cases}\]
By construction, $\conc{\ynmulstruc}{\bag \series \dir}(\ainter_{\dir}(x)) = \conc{\ynmulstruc}{\bag}(\ainter(x))$.
Thus, we have:
\[{} \absmodels \mkstate{\fmlset[1]}{\pri}{\ynmulstruc}{\ainter}{\bag} \quad\iff\quad {} \absmodels \mkstate{\fmlset[1]}{\pri}{\ynmulstruc}{\ainter_{\dir}}{\bag \series \dir}.\]

Such an $\ainter_{\dir}$ can be taken thanks to $\ainter(\FV(\fmlset[1])) \subseteq \ND{\ynmulstruc}{\bag}{\dir}$;
in general, there is no $\ainter'$ such that $\conc{\ynmulstruc}{\bag \series \dir}(\ainter'(x)) = \conc{\ynmulstruc}{\bag}(\ainter(x))$.
Nevertheless, by an appropriate evaluation (\Cref{section: evaluation strategy UNFO}),
we can always reach a \kl(LMC){state} $\mkstate{\fmlset[1]}{\pri}{\ynmulstruc}{\ainter}{\bag}$ satisfying $\ainter(\FV(\fmlset[1])) \subseteq \ND{\ynmulstruc}{\bag}{\dir}$ for some $\dir \in \set{\back 2, \back 1, 1, 2}$.

\begin{scope}\knowledgeimport{UNFO}\knowledgeimport{LMC}

\section{Local Checker for \texorpdfstring{"SAT-""UNFO"}{SAT-UNFO}}\label{section: UNFO}

The \kl{satisfiability problem} for \kl{UNFO} is known to be in "2ExpTime" \cite[Theorem 4.5]{segoufinUnaryNegation2013}.
In \cite{segoufinUnaryNegation2013}, this is shown via 
an exponential reduction to the \kl{satisfiability problem} for two-way modal $\mu$-calculus. 
In this section, we present a different "2ExpTime" algorithm by a direct \kl{2APTA} construction based on a "local checker" on \kl{tree decompositions}.
We will extend this algorithm to \kl{UNTC} in the next section.

\subsection{A \texorpdfstring{"Local Checker"}{Local Checker}}\label{section: local model checker UNFO}
\AP %
We next define the ""local checker"" (\reintro*\kl{LC}) for \kl{UNFO}, based on \kl{2APTAs}.
For convenience in writing examples, given
$\fmlset[1]$ in ``$\mkstate{\fmlset[1]}{\pri}{\ynmulstruc}{\ainter}{\bag}$'',
we may put $\ainter(x)$ in the superscript of some occurrences $x$,
may just write $\dir$ for $\set{\dir}$ in the superscript,
and may use the set notations as in the sequent calculus. For instance,
$\mkstate{(\rsym[1] x^{\set{\back 1, 1}} y, \exists x \rsym[2] x y^{2})}{\pri}{\ynmulstruc}{\ainter}{\bag}$
denotes the \kl{state}
$\mkstate{\set{\rsym[1] x y, \exists x \rsym[2] x y}}{\pri}{\ynmulstruc}{\ainter}{\bag}$,
with $\ainter(x) = \set{\back 1, 1}$ and $\ainter(y) = \set{2}$.

\AP For two sequences $\mul{x} = x_1 \dots x_k$ and $\mul{\vertex} = \vertex_1 \dots \vertex_k$ where the elements of $\mul{x}$ are pairwise distinct,
we write $\ainter\intro*\aintersubst{\mul{x}}{\mul{\vertex}}$ for the $\ainter$ in which each $\ainter(x_i)$ has been substituted with $\vertex_i$ for each $i \in \1{k}$.

\AP The binary relation $(\intro*\leadstoUNFO) \subseteq \allstates_{\mathrm{\kl{UNFO}}} \times \PBallfml(\allstates_{\mathrm{\kl{UNFO}}})$
is defined as the smallest relation closed under the ""rules"" in \Cref{figure: rules UNFO}.
\begin{figure}[t]%
    \centering
\begin{tcolorbox}[standard jigsaw, opacityback = 0.8, colframe=black!80, boxrule=.3mm, left=.2mm, right=.2mm]
    \vspace{-4ex}
    \begin{align*}
    \mkstate{()}{\pri}{\ynmulstruc}{\ainter}{\bag} &\;\leadstoUNFO\; \doubletrue^{\pri} \tag{emp}\label{rule: emp}\\
    \mkstate{(\fmlset[1], \fml[1] \land \fml[2])}{\pri}{\ynmulstruc}{\ainter}{\bag} &\;\leadstoUNFO\;  \mkstate{(\fmlset[1], \fml[1], \fml[2])}{\pri}{\ynmulstruc}{\ainter}{\bag} \tag{$\land$} \label{rule: and}\\
    \mkstate{(\fmlset[1], \fml[1] \lor \fml[2])}{\pri}{\ynmulstruc}{\ainter}{\bag} &\;\leadstoUNFO\;  \mkstate{(\fmlset[1], \fml[1])}{\pri}{\ynmulstruc}{\ainter}{\bag} \doublevee^{\pri} \mkstate{(\fmlset[1], \fml[2])}{\pri}{\ynmulstruc}{\ainter}{\bag} \tag{$\lor$}\label{rule: or}\\
    \mkstate{(\lnot \fml)}{\pri}{\ynmulstruc}{\ainter}{\bag} &\;\leadstoUNFO\;  \mkstate{(\fml)}{1 - \pri}{\ynmulstruc}{\ainter}{\bag}
    \quad \text{ if 
    $\begin{cases}
    \text{$\ainter(x) \subseteq \univ{\ynmulstruc(\bag)}$
    and 
    $\card {\ainter(x)} = 1$
    }\\
    \text{for each $x \in \FV(\fml)$}    
    \end{cases}$} \tag{$\lnot$} \label{rule: UN}\\
    \mkstate{\fmlset[1]}{\pri}{\ynmulstruc}{\ainter}{\bag} &\;\leadstoUNFO\;  \bigdoublevee_{\elem \in \ainter(x)}^{\pri} \mkstate{\fmlset[1]}{\pri}{\ynmulstruc}{\ainter\aintersubst{x}{\set{\elem}}}{\bag} \tag{conc}\label{rule: concrete}\\
    \mkstate{(\afml)}{\pri}{\ynmulstruc}{\ainter}{\bag} &\;\leadstoUNFO\; \doubletrue^{\pri} \hspace{.5em} \text{ if 
    $\begin{cases}
    \text{for each $x \in \FV(\afml)$, there is a $\vertex \in \univ{\ynmulstruc(\bag)}$ \kl{st}}\\
    \text{$\ainter(x) = \set{\vertex}$, and $\ynmulstruc(\bag) \modelsass{\set{x \mapsto \vertex \mid x \mapsto \set{\vertex} \in \ainter}} \afml$}    
    \end{cases}$} \tag{$\afml$}\label{rule: a1}\\
    \mkstate{(\fmlset[1], \exists x \fml)}{\pri}{\ynmulstruc}{\ainter}{\bag} &\;\leadstoUNFO\;  \mkstate{(\fmlset[1], \fml[1][z/x])}{\pri}{\ynmulstruc}{\ainter\aintersubst{z}{\AD^{\ynmulstruc}_{\bag}}}{\bag} \tag{$\exists$} \label{rule: exists}\\
    \mkstate{(\fmlset[1], \fmlset[2])}{\pri}{\ynmulstruc}{\ainter}{\bag} &\;\leadstoUNFO\;  \mkstate{\fmlset[1]}{\pri}{\ynmulstruc}{\ainter}{\bag} \doublewedge^{\pri} \mkstate{\fmlset[2]}{\pri}{\ynmulstruc}{\ainter}{\bag} \quad 
\text{ if 
    $\begin{cases}
    \text{$\ainter(x) \subseteq \univ{\ynmulstruc(\bag)}$
    and 
    $\card {\ainter(x)} = 1$
    }\\
    \text{for each $x \in \FV(\fmlset[1]) \cap \FV(\fmlset[2])$}    
    \end{cases}$} \tag{split} \label{rule: split}\\
    \mkstate{\fmlset[1]}{\pri}{\ynmulstruc}{\ainter}{\bag} &\;\leadstoUNFO\; \mkstate{\fmlset[1]}{\pri}{\ynmulstruc}{\ainter_{\dir}}{\bag \series \dir}
    \quad \text{ if $\ainter(\FV(\fmlset[1])) \subseteq \ND{\ynmulstruc}{\bag}{\dir}$} \tag{move}\label{rule: mov}
    \end{align*}
\end{tcolorbox}
    \caption{\reintro*\kl{Rules} of the \kl{LC} for \kl{UNFO}. Here, $z$ is used as a \kl{fresh} \kl{variable}.}
    \label{figure: rules UNFO}
\end{figure}
We lift this relation to $(\reintro*\leadstoUNFO) \subseteq \PBallfml(\allstates_{\mathrm{\kl{UNFO}}}) \times \PBallfml(\allstates_{\mathrm{\kl{UNFO}}})$,
as the smallest relation that contains the original $(\leadstoUNFO)$ and is single-step compatible with the operators in \kl{positive Boolean formulas} ("eg", if $\mkstate{\fmlset[1]}{\pri}{\ynmulstruc}{\ainter}{\bag} \leadstoUNFO \doublefml[2]$, then $\mkstate{\fmlset[1]}{\pri}{\ynmulstruc}{\ainter}{\bag} \doublevee \doublefml[3] \leadstoUNFO \doublefml[2] \doublevee \doublefml[3]$).

\AP A \intro*\kl(LMC){run} starting from $\mystate$ is a $\allstates_{\mathrm{\kl{UNFO}}}$-labeled \kl{tree} $\trace$ where $\trace(\varepsilon) = \mystate$ and such that for every $\word \in \fdom(\trace)$ we have that for\\
$\left\{
\begin{aligned}
&\text{some $\trace(\word) \leadstoUNFO \doublefml[2]$} && \text{if $\statePri(\trace(\word)) = \prione$} \\
&\text{every $\trace(\word) \leadstoUNFO \doublefml[2]$} && \text{if $\statePri(\trace(\word)) = \prizero$}
\end{aligned}
\right\}$
the \kl{positive Boolean formula} resulting from replacing in $\doublefml[2]$ each child "state"  of $\word$ ("ie", each element $\trace(\word a)$) with $\doubletrue$ is "semantically equivalent" to $\doubletrue$.

\AP A \kl(LMC){run} $\trace$ is \intro*\kl(LMC){accepting} if for every infinite path $a_1 a_2 \dots$ in $\trace$, there are infinitely many $n$ with $\statePri(\trace(a_1 a_2 \dots a_n)) = \prizero$ (in which case we say that the path has priority $\prizero$).
\AP We write $\intro*\vdashUNFO \mystate \phantomintro\nvdashUNFO$ if there is an \kl(LMC){accepting} \kl(LMC){run} starting from $\mystate$.
The definition of "runs@@LMC" is almost the same as that of "runs@@2APTA" for alternating automata.
For example, consider $\trace(\word)$ with odd "priority" $\statePri(\trace(\word)) = \prione$ and let $\mystate[1], \mystate[2], \mystate[3] \in \allstates_{\mathrm{\kl{UNFO}}}$.
If there is a "rule" $\trace(\word) \leadstoUNFO \mystate[1] \doublewedge \mystate[2]$,
the condition is passed at the position $\word$ when both $\mystate[1]$ and $\mystate[2]$ occur in the "children" of $\word$.
Similarly, if there is a "rule" $\trace(\word) \leadstoUNFO \mystate[1] \doublevee \mystate[2]$,
the condition is passed at $\word$ when either $\mystate[1]$ or $\mystate[2]$ occurs.
If multiple "rules" exist for $\trace(\word)$, "eg", $\trace(\word) \leadstoUNFO \mystate[1] \doublewedge \mystate[2]$ and $\trace(\word) \leadstoUNFO \mystate[3]$, %
then they can be combined into the single "rule" $\trace(\word) \leadstoUNFO (\mystate[1] \doublewedge \mystate[2]) \doublevee \mystate[3]$.
When $\trace(\word) = \mkstate{\fmlset}{\pri}{\ynmulstruc}{\ainter}{\bag}$ has "priority" $\pri = \prizero$,
the definition is given so that the merged %
"formula@positive Boolean formula" for $\mkstate{\fmlset}{\prizero}{\ynmulstruc}{\ainter}{\bag}$ is the ``"dual"'' of "that@positive Boolean formula" for $\mkstate{\fmlset}{1}{\ynmulstruc}{\ainter}{\bag}$.
Each "rule" $\mystate[1] \leadstoUNFO \doublefml[2]$ in \Cref{figure: rules UNFO} is defined so that
${} \absmodels \mystate[1] \quad\iff\quad {} \absmodels \doublefml[2]$.

The "LC" is sound and complete "wrt" the semantics on \kl{tree decompositions} (\Cref{theorem: completeness closure UNFO}).
Below, we give a toy example%
\ifthenelse{\boolean{arXiv}}{  (see \Cref{section: additional example UNFO} for an additional example for even priority).}{.}%
\begin{example}\label{example: model checker UNFO 1}
    Let $\ynmulstruc$ be the \kl{tree} 
    in which $\ynmulstruc(1)$, $\ynmulstruc(\eps)$, and $\ynmulstruc(2)$ are
    given as follows, and $\ynmulstruc(\bag)$ is undefined for all other $\bag$:
    $\left(\begin{tikzpicture}[baseline = -.5ex]
    \graph[grow right = 1.cm, branch down = 2.5ex]{
    {s1/{$\fixedvertex_1$}[vert]} -!- {t1/{$\fixedvertex_2$}[vert]}
    };
    \graph[use existing nodes, edges={color=black, pos = .5, earrow}, edge quotes={fill=white, inner sep=1pt,font= \scriptsize}]{
        s1 ->["$\rsym[1]$", pos = .4] t1;
    };
\end{tikzpicture} \right)$,
$\left(\begin{tikzpicture}[baseline = -.5ex]
    \graph[grow right = 1.cm, branch down = 2.5ex]{
    {s1/{$\fixedvertex_2$}[vert]}
    };
    \graph[use existing nodes, edges={color=black, pos = .5, earrow}, edge quotes={fill=white, inner sep=1pt,font= \scriptsize}]{
    };
\end{tikzpicture} \right)$,
$\left(\begin{tikzpicture}[baseline = -.5ex]
    \graph[grow right = 1.cm, branch down = 2.5ex]{
    {s1/{$\fixedvertex_2$}[vert]} -!- {t1/{$\fixedvertex_3$}[vert]}
    };
    \graph[use existing nodes, edges={color=black, pos = .5, earrow}, edge quotes={fill=white, inner sep=1pt,font= \scriptsize}]{
        s1 ->["$\rsym[2]$", pos = .4] t1;
    };
\end{tikzpicture} \right)$.
    Then,
    $\glue \ynmulstruc = \left(\begin{tikzpicture}[baseline = -.5ex]
        \graph[grow right = .8cm, branch down = 4.ex]{
        {1/{}[vert]} -!- {2/{}[vert]} -!- {3/{}[vert]}
        };
        \graph[use existing nodes, edges={color=black, pos = .5, earrow}, edge quotes={fill=white, inner sep=1pt,font= \scriptsize}]{
            1 ->["$\rsym[1]$"] 2 ->["$\rsym[2]$"] 3;
        };
    \end{tikzpicture}\right)$.
    Let $\fml$ be %
    $\exists x \exists y \exists z (\rsym[1] x z \land \rsym[2] z y)$.
    Then, $\absmodels \mkstate{(\fml)}{\prione}{\ynmulstruc}{\emptyset}{\eps}$ holds.
    In our "LC", we can also construct an \kl(LMC){accepting} \kl(LMC){run}, as follows,
    where we abbreviate \kl{abstract interpretations} $\emptyset[\dots]$ to $-$:
    \begin{align*}
        \mkstate{\exists x \exists y \exists z (\rsym[1] x z \land \rsym[2] z y)}{\prione}{\ynmulstruc}{\emptyset}{\eps}
        &~\leadstoUNFO_{\eqref{rule: exists}\eqref{rule: concrete}}^{*}~ \bigdoublevee_{\elem_1 \in \set{\fixedvertex_2, 1, 2}} \mkstate{(\exists y \exists z  (\rsym[1] x^{\elem_1} z \land \rsym[2] z y))}{\prione}{\ynmulstruc}{\bl}{\eps} \\
        &~\leadstoUNFO_{\eqref{rule: exists}\eqref{rule: concrete}\eqref{rule: and}}^{*}~ \bigdoublevee_{\elem_1, \elem_2, \elem_3 \in \set{\fixedvertex_2, 1, 2}} \mkstate{(\rsym[1] x^{\elem_1} z^{\elem_3}, \rsym[2] z^{\elem_3} y^{\elem_2})}{\prione}{\ynmulstruc}{\bl}{\eps}\\
        &~\lgeqq~ \mkstate{(\rsym[1] x^{1} z^{\fixedvertex_2}, \rsym[2] z^{\fixedvertex_2} y^{2})}{\prione}{\ynmulstruc}{\bl}{\eps}
        ~~\leadstoUNFO_{\eqref{rule: split}}~~ \mkstate{(\rsym[1] x^{1} z^{\fixedvertex_2})}{\prione}{\ynmulstruc}{\bl}{\eps} \doublewedge \mkstate{(\rsym[2] z^{\fixedvertex_2} y^{2})}{\prione}{\ynmulstruc}{\bl}{\eps}\\
        &~\leadstoUNFO_{\eqref{rule: mov}}^{*}~ \mkstate{(\rsym[1] x^{\fixedvertex_1} z^{\fixedvertex_2})}{\prione}{\ynmulstruc}{\bl}{1} \doublewedge \mkstate{(\rsym[2] z^{\fixedvertex_2} y^{\fixedvertex_3})}{\prione}{\ynmulstruc}{\bl}{2}
        ~~\leadstoUNFO_{\eqref{rule: a1}}^{*}\lgeqq~~ \doubletrue.
    \end{align*}
    Thus, there is a (finite) \kl(LMC){accepting} \kl(LMC){run}, and hence $\vdashUNFO \mkstate{(\fml)}{\prione}{\ynmulstruc}{\emptyset}{\eps}$.
    \lipicsEnd
\end{example}

\begin{toappendix}
\begin{scope}\knowledgeimport{UNFO}\knowledgeimport{LMC}
\subsection{Additional Example}\label{section: additional example UNFO}
\begin{example}[On infinite tree]\label{example: model checker UNFO 2}
    Let $\ynmulstruc$ be the \kl{tree} given by
    $\ynmulstruc(1^{2m}) = \left(\begin{tikzpicture}[baseline = -.5ex]
        \graph[grow right = 1.cm, branch down = 2.5ex]{
        {s1/{$\fixedvertex_1$}[vert]}
        };
        \path (s1) edge [color=black, pos = .5, earrow, out = 150, in = 210, looseness = 4, ->] node[fill=white, inner sep=1pt,font= \scriptsize] {$\rsym[1]$} (s1);
    \end{tikzpicture} \right)$ and
    $\ynmulstruc(1^{2m+1}) = \left(\begin{tikzpicture}[baseline = -.5ex]
        \graph[grow right = 1.cm, branch down = 2.5ex]{
        {s1/{$\fixedvertex_1$}[vert]} -!- {t1/{$\fixedvertex_2$}[vert]}
        };
        \path (s1) edge [color=black, pos = .4, earrow, ->] node[fill=white, inner sep=1pt,font= \scriptsize] {$\rsym[1]$} (t1);
    \end{tikzpicture} \right)$ where $m \ge 0$, %
    and $\ynmulstruc(\bag)$ is undefined otherwise;
    then,
    \[\glue \ynmulstruc = \left(\begin{tikzpicture}[baseline = -2.5ex]
        \graph[grow right = 1.cm, branch down = 4.ex]{
        {/,1/{}[vert]} -!- {0/{}[vert],2/{}[vert, xshift = -.2cm]} -!- {/,3/{}[vert, xshift = -.2cm]} -!- {/[xshift = -.6cm], /{$\dots$}[xshift = -.6cm]}
        };
        \path (0) edge [color=black, pos = .4, earrow, ->] node[fill=white, inner sep=1pt,font= \scriptsize] {$\rsym[1]$} (1);
        \path (0) edge [color=black, pos = .4, earrow, ->] node[fill=white, inner sep=1pt,font= \scriptsize] {$\rsym[1]$} (2);
        \path (0) edge [color=black, pos = .4, earrow, ->] node[fill=white, inner sep=1pt,font= \scriptsize] {$\rsym[1]$} (3);
        \path (0) edge [color=black, pos = .4, earrow, out = 120, in = 180, looseness = 15, ->] node[fill=white, inner sep=1pt,font= \scriptsize] {$\rsym[1]$} (0);
    \end{tikzpicture}\right).\]
    Let $\fml$ be the \kl{UNFO sentence} $\lnot \exists x \lnot \exists y \rsym[1] y x$,
    which is "semantically equivalent" to ``$\forall x \exists y \rsym[1] y x$''.
    Then $\absmodels \mkstate{(\fml)}{\prione}{\ynmulstruc}{\emptyset}{\eps}$ holds.
    In the "LC", we construct an \kl(LMC){accepting} \kl(LMC){run}, as follows.
    First, we reach to an even "priority" "state" as follows:
    \begin{align*}
        \mkstate{(\lnot \exists x \lnot \exists y \rsym[1] y x)}{\prione}{\ynmulstruc}{\emptyset}{\eps}
        &\;\leadstoUNFO_{\eqref{rule: UN}}\; \mkstate{(\exists x \lnot \exists y \rsym[1] y x)}{\prizero}{\ynmulstruc}{\emptyset}{\eps}.
    \end{align*}
    For even "priority", we consider all applications of rules (namely, we consider a demonic choice).
    As the applicable rules are \eqref{rule: exists}, \eqref{rule: split}, and \eqref{rule: mov}, by letting $\fml[2] \defeq \exists x \lnot \exists y \rsym[1] y x$ (for short), we consider each conjunct of the "positive Boolean formula" below:
    \[\underbrace{\mkstate{(\lnot \exists y \rsym[1] y x^{\set{\fixedvertex_1, 1}})}{\prizero}{\ynmulstruc}{\bl}{\eps}}_{\eqref{rule: exists}}
    \;\doublewedge\;\;
    \underbrace{\mkstate{\fml[2]}{\prizero}{\ynmulstruc}{\emptyset}{1}}_{\eqref{rule: mov}}
    \;\doublewedge\;
    \underbrace{(\mkstate{()}{\prizero}{\ynmulstruc}{\emptyset}{\eps} \doublevee \mkstate{\fml[2]}{\prizero}{\ynmulstruc}{\emptyset}{\eps})}_{\eqref{rule: split}}.\]
    The only crucial part is for \eqref{rule: exists}.
    Depending on whether the "abstract interpretation" of $x$ is "concretized" to $\set{\fixedvertex_1}$ or $\set{1}$, the following two conjuncts are crucial:
    \[\underbrace{\mkstate{(\exists y \rsym[1] y x^{\fixedvertex_1})}{\prione}{\ynmulstruc}{\bl}{\eps}}_{\subalign{&\eqref{rule: concrete}\eqref{rule: UN}\\ &x \mapsto \set{\fixedvertex_1}}}
    \;\doublewedge\;
    \underbrace{\mkstate{(\lnot \exists y \rsym[1] y x^{\set{\fixedvertex_2, 1}})}{\prizero}{\ynmulstruc}{\bl}{1}}_{\subalign{&\eqref{rule: concrete}\eqref{rule: mov}\\ &x \mapsto \set{1}}}
    \;\doublewedge\; \dots\]
    \proofcasethin{Case $x \mapsto \set{\fixedvertex_1}$} We then have:
    \begin{align*}
        \mkstate{(\exists y \rsym[1] y x^{\fixedvertex_1})}{\prione}{\ynmulstruc}{\bl}{\eps}
        &\leadstoUNFO_{\eqref{rule: exists}\eqref{rule: concrete}}^{*}\lgeqq \mkstate{(\rsym[1] y^{\fixedvertex_1} x^{\fixedvertex_1})}{\prione}{\ynmulstruc}{\bl}{\eps}
        \leadstoUNFO_{\eqref{rule: a1}} \doubletrue,
    \end{align*}
    Thus, we obtain a finite "accepting" "run" for $\mkstate{(\exists y \rsym[1] y x^{\fixedvertex_1})}{\prione}{\ynmulstruc}{\bl}{\eps}$.
    As each path is finite, this case does not affect the acceptance.

    \proofcasethin{Case $x \mapsto \set{1}$}
    We consider 
    $\mkstate{(\lnot \exists y \rsym[1] y x^{\set{\fixedvertex_2, 1}})}{\prizero}{\ynmulstruc}{\bl}{2m+1}$ where $m \ge 0$
    by generalizing $\mkstate{(\lnot \exists y \rsym[1] y x^{\set{\fixedvertex_2, 1}})}{\prizero}{\ynmulstruc}{\bl}{1}$.
    For each $m$, similar to the above, we distinguish the following two cases.

    \subproofcasethin{Sub-Case $x \mapsto \set{\fixedvertex_2}$} 
    Similar to Case $x \mapsto \set{\fixedvertex_1}$,
    after applying \eqref{rule: UN},%
    we can give a finite "accepting" "run" for $\mkstate{(\exists y \rsym[1] y x^{\fixedvertex_2})}{\prione}{\ynmulstruc}{\bl}{1}$.

    \subproofcasethin{Sub-Case $x \mapsto \set{1}$}
    Similar to Case $x \mapsto \set{1}$, by \eqref{rule: concrete}\eqref{rule: mov},
    we have a conjunct $\mkstate{(\lnot \exists y \rsym[1] y x^{\set{\fixedvertex_2, 1}})}{\prizero}{\ynmulstruc}{\bl}{1^{2m+3}}$.
    We then go back to the above case analysis.

    Then in each infinite path,
    $\mkstate{(\bl)}{\prizero}{\ynmulstruc}{\bl}{\bl}$ appears infinitely often.
    Hence, the \kl(LMC){run} obtained from the above construction has \kl(LMC){priority} $\prizero$.
    Thus, the \kl(LMC){run} is \kl(LMC){accepting}, and hence $\vdashUNFO \mkstate{(\fml)}{\prione}{\ynmulstruc}{\emptyset}{\eps}$.
    \lipicsEnd
\end{example}
\end{scope}
\end{toappendix}

\subsection{Evaluation Strategy} \label{section: evaluation strategy UNFO}
We now present a strategy to evaluate $\absmodels \mkstate{\fmlset[1]}{\pri}{\ynmulstruc}{\ainter}{\bag}$ in our "LC".

\proofcasethin{Step 1}\labeltext{evaluation strategy UNFO 1}{}{1}
We eliminate the outermost $\land$ using the "rule" \eqref{rule: and}, and the outermost $\lor$ and $\exists$ by nondeterministically selecting one disjunct via the "rules" \eqref{rule: or}\eqref{rule: exists}\eqref{rule: concrete}, as much as possible.
We can then assume:
\begin{enumerate}
    \item[a)] \labeltext{cond: evaluation strategy UNFO 1}{}{a} $\card \ainter(x) = 1$ for each $x \in \FV(\fmlset[1])$, and
    \item[b)] \labeltext{cond: evaluation strategy UNFO 2}{}{b} $\fmlset[1]$ is of the form $(\fml[2]_1, \dots, \fml[2]_n)$
where each $\fml[2]_i$ is one of the following forms:
an "atom" $\afml$ %
or a "negation formula" $\lnot \fml[3]$,
\end{enumerate}
from which it follows that
\begin{enumerate}
    \item[c)] \labeltext{cond: evaluation strategy UNFO 3}{}{c} for each $\fml[2]_i$,
    for some $\dir \in \range{\back 2}{2}$,
    $\ainter(\FV(\fml[2]_i)) \subseteq \ND{\ynmulstruc}{\bag}{\dir}$.
\end{enumerate}
Here when $\fml[2]_i$ is of the form $\afml$,
this is derived from $\absmodels \mkstate{\fmlset[1]}{\pri}{\ynmulstruc}{\ainter}{\bag}$.
When $\fml[2]_i$ is of the form $\lnot \fml[3]$,
it immediately follows from the condition of unary negation: $\card\FV(\fml[3]) \le 1$.

\proofcasethin{Step 2}\labeltext{evaluation strategy UNFO 2}{}{2}
If $\fml[2]_i$ satisfies the condition \nameref{cond: evaluation strategy UNFO 3}) for $\dir = 0$,
we can eliminate $\fml[2]_i$ by applying \eqref{rule: split}.
We can then assume that
\begin{enumerate}
    \item[c')] \labeltext{cond: evaluation strategy UNFO 3'}{}{c'} for each $\fml[2]_i$, for a $\dir \in \set{\back 2, \back 1, 1, 2}$, $\dir \in \ainter(\FV(\fml[2]_i)) \subseteq \ND{\ynmulstruc}{\bag}{\dir}$.
\end{enumerate}
Here for the eliminated \kl{state} $\mkstate{(\fml[2]_i)}{\pri}{\ynmulstruc}{\ainter}{\bag}$, 
we apply \eqref{rule: a1}%
\footnote{\label{footnote: a1 and mov}Here, we sometimes need also \eqref{rule: mov}; "eg", in \Cref{figure: gluing}, $\mkstate{(\Erel x^{\fixedvertex_2} y^{\fixedvertex_3})}{\prione}{\ynmulstruc}{\bl}{\eps} \leadstoUNFO_{\eqref{rule: mov}}
\mkstate{(\Erel x^{\fixedvertex_2} y^{\fixedvertex_3})}{\prione}{\ynmulstruc}{\bl}{2} \leadstoUNFO_{\eqref{rule: a1}} \doubletrue$.}
if $\fml[2]_i$ is $\afml$
and apply \eqref{rule: UN} if $\fml[2]_i$ is $\lnot \fml[3]$,
and then we return to Step \nameref{evaluation strategy UNFO 1}.

\proofcasethin{Step 3}\labeltext{evaluation strategy UNFO 3}{}{3}
From the condition \nameref{cond: evaluation strategy UNFO 3'}),
by applying the "rule" \eqref{rule: split},
we can assume that 
\begin{enumerate}
    \item[c'')] \labeltext{cond: evaluation strategy UNFO 3''}{}{c''} for some $\dir \in \set{\back 2, \back 1, 1, 2}$, $\dir \in \ainter(\FV(\fmlset[1])) \subseteq \ND{\ynmulstruc}{\bag}{\dir}$.
\end{enumerate}

\proofcasethin{Step 4}\labeltext{evaluation strategy UNFO 4}{}{4}
By the condition \nameref{cond: evaluation strategy UNFO 3''}),
we apply the "rule" \eqref{rule: mov} based on the argument of \Cref{section: mov};
then, we move from $\bag$ to the adjacent \kl{bag} $\bag \series \dir$.
Finally, we go back to Step \nameref{evaluation strategy UNFO 1}.

\subsection{Closure Property}\label{section: closure property UNFO}
Based on the evaluation strategy above,
we can obtain a stronger completeness with a closure property, as we shall see in \Cref{theorem: completeness closure UNFO}.
\begin{definition}\label{definition: closure UNFO}
    For a \kl{UNFO} \kl{formula} $\fml$,
    the \intro*\kl(UNFO){closure} $\intro*\clUNFO(\fml)$ is the \kl{set} of \kl{UNFO} \kl{formula} \kl{sets} defined as follows:
    \begin{align*}
        \clUNFO(\afml) &\;\defeq\; \set{(\afml), ()},&
        \clUNFO(\fml[1] \lor \fml[2]) &\;\defeq\; \set{(\fml[1] \lor \fml[2])} \cup \clUNFO(\fml[1]) \cup \clUNFO(\fml[2]),\\
        \clUNFO(\fml[1] \land \fml[2]) &\;\defeq\; \set{(\fml[1] \land \fml[2])} \cup \set{ (\fmlset[1], \fmlset[2]) \mid \fmlset[1] \in \clUNFO(\fml[1]), \fmlset[2] \in \clUNFO(\fml[2]), \FV(\fmlset[1]) \cap \FV(\fmlset[2]) \subseteq \FV(\fml[1]) \cap \FV(\fml[2]) }, \span\span\\
        \clUNFO(\exists x \fml[1]) &\;\defeq\; \set{(\exists x \fml[1])} \cup \bigcup_{\text{$z$ is "fresh"}} \clUNFO(\fml[1][z/x]),&
        \clUNFO(\lnot \fml[1]) &\;\defeq\; \set{(\lnot \fml[1])} \cup \clUNFO(\fml[1]).
        \tag*{\lipicsEnd}
    \end{align*}
\end{definition}
By straightforward induction on $\fml$, we can show the following monotonicity: $(\fmlset[1], \fmlset[2]) \in \clUNFO(\fml)$ implies $\fmlset[1] \in \clUNFO(\fml)$.

For a class $\fmlsetclass$ of finite \kl{formula} sets,
$\fmlsetclass$-""runs@clruns@LMC"" are defined in the same way as "runs",
where for each "rule" $\mystate \leadstoUNFO \doublefml[2]$ in the "LC" (\Cref{figure: rules UNFO}),
each "state" $\mystate[2]$ occurring in $\doublefml[2]$ has been
replaced with $\left\{\begin{aligned}
    & \doublefalse && \text{if $\pri = \prione$}\\
    & \doubletrue && \text{if $\pri = \prizero$}
\end{aligned}\right\}$ if $\mystate[2] \not\in \allstates_{\fmlsetclass}$.
For a "state" $\mkstate{\fmlset}{\pri}{\ynmulstruc}{\ainter}{\bag} \in \allstates_{\fmlsetclass}$,
we write $\intro*\vdashUNFOsub_{\fmlsetclass} \mystate \phantomintro\nvdashUNFOsub$ 
if there is an "accepting@@LMC" $\fmlsetclass$-"run@clrun@@LMC" starting from $\mystate$.

We then have the following soundness and completeness theorem with a closure property, based on the strategy of \Cref{section: evaluation strategy UNFO}.
\ifthenelse{\boolean{arXiv}}{\begin{theorem}[\Cref{section: theorem: completeness UNFO}]}{\begin{theorem}}\AP\label{theorem: completeness closure UNFO}
    \gdef\completenessclosureUNFO{%
    Let $\fml$ be a \kl{UNFO} \kl{formula}.
    For every $\mystate \in \allstates_{\clUNFO(\fml)}$,
    we have:
    \[\absmodels \mystate \quad\Longleftrightarrow\quad {} \vdashUNFOsub_{\clUNFO(\fml)} \mystate.\]
    }%
    \completenessclosureUNFO
\end{theorem}
On the size of the \kl(UNFO){closure}, we have the following.
\begin{proposition}\label{proposition: closure size UNFO}
    For all \kl{UNFO} \kl{formulas} $\fml$,
    the \kl{cardinality} of $\clUNFO(\fml)$, up to renaming \kl{free variables}, is at most $(2 \fmllen{\fml})^{2 \fmllen{\fml}}$.
\end{proposition}
\begin{proof}
    By easy induction on $\fml$%
    \ifthenelse{\boolean{arXiv}}{  (\Cref{section: proposition: closure size UNFO}).}{.}
\end{proof}
Hence, we can give an exponential bound on the number of "states", up to an appropriate equivalence, as follows.
\begin{proposition}\label{proposition: closure UNFO 2}
    For \kl{UNFO} \kl{formulas} $\fml$,
    the number of $\mkstate{\fmlset[1]}{\pri}{\ynmulstruc}{\ainter}{\bag} \in \allstates_{\clUNFO(\fml)}^{(k)}$
    is $2^{\mathcal{O}(\fmllen{\fml} (\log \fmllen{\fml} + k))}$,
    up to forgetting $\ainter(x)$ for $x \not\in \FV(\fmlset[1])$, $\ynmulstruc$, and $\bag$ ("ie", ``$\fmlset[1]^{\pri}_{\ainter \restriction \FV(\fmlset[1])}$'')
    and up to renaming \kl{free variables}.
\end{proposition}
\begin{proof}
    By $(2 \fmllen{\fml})^{2 \fmllen{\fml}}$ (the number of $\fmlset[1] \in \clUNFO(\fml)$; \Cref{proposition: closure size UNFO})
    $\times$ $2$ (the number of $\pri$)
    $\times$ $\mathcal{O}((2^{(k+4)})^{\fmllen{\fml}})$ (the number of $\ainter$; note that $\card \FV(\fmlset[1]) \le \fmllen{\fml}$ and $\card \AD^{\ynmulstruc} \le k + 4$)
    $\le$ $2^{\mathcal{O}(\fmllen{\fml} (\log \fmllen{\fml} + k))}$.
\end{proof}

\subsection{Reduction to \kl{2APTAs}}\label{section: reducing to 2APTAs UNFO}
Using the "LC",
we can naturally reduce the \kl{satisfiability problem} over \kl{structures} of \kl{treewidth} at most $k-1$ to the \kl{non-emptiness problem} for \kl{2APTAs}.\footnote{A minor difference of the "LC" from \kl{2APTAs} is that some auxiliary predicates ("eg" $\ADB{\ynmulstruc}{\bag}$, $\ND{\ynmulstruc}{\bag}{\dir}$, and $\MD{\ynmulstruc}{\bag}{\dir}$) are used in the "rules" of the "LC" (\Cref{figure: rules UNFO}), but we can easily encode them inside \kl{2APTAs}%
 \ifthenelse{\boolean{arXiv}}{  (see \Cref{section: 2APTA construction UNFO} for a precise construction of \kl{2APTAs}).}{.}}
Using \Cref{proposition: closure UNFO 2}, we can show that the \kl(2APTA){size} of the \kl{2APTA} is $2^{\poly(\|\fml\|,k)}$.
Because "UNFO" has the \kl{bounded treewidth model property} \cite[Theorem 3.4]{segoufinUnaryNegation2013} (\Cref{proposition: tw}),
we can give an exponential-time reduction from the \kl{satisfiability problem} for \kl{UNFO} to the \kl{non-emptiness problem} for \kl{2APTAs}.
Because the \kl{non-emptiness problem} for \kl{2APTAs} is in "ExpTime" \cite{vardiReasoningTwowayAutomata1998} (\Cref{proposition: 2APTA complexity}),
we have the following known result.
\begin{corollary}[\textnormal{also shown in {\cite[Theorem 4.5]{segoufinUnaryNegation2013}}, even for "UNFP"}]\label{corollary: 2-EXPTIME UNFO}
     "SAT-""UNFO" is in "2ExpTime".
\end{corollary}

\end{scope}

\begin{toappendix}
\begin{scope}\knowledgeimport{UNFO}\knowledgeimport{LMC}
\subsection{Proof of {\Cref{theorem: completeness closure UNFO}}:
Soundness and Completeness}\label{section: theorem: completeness UNFO}
We recall the \kl(UNFO){LC} given in \Cref{section: local model checker UNFO}.
In \Cref{section: preservation property UNFO,section: duality UNFO,section: consistency UNFO}, we prepare some properties of the \kl(UNFO){LC}.
In \Cref{section: theorem: completeness UNFO proof}, we prove \Cref{theorem: completeness closure UNFO}.
In \Cref{section: lemma: parameter UNFO even,section: lemma: parameter UNFO odd}, we prove the two lemmas used in \Cref{section: theorem: completeness UNFO proof} for proving \Cref{theorem: completeness closure UNFO}.

\subsubsection{Preservation property}\label{section: preservation property UNFO}
\AP
For each "state" $\mystate$,
the \intro*\kl(LMC){transition formula} $\intro*\transitionUNFO(\mystate)$ is defined as the
$\left\{\begin{aligned}
    &\text{disjunction} && \text{if $\statePri(\mystate) = \prione$}\\
    &\text{conjunction} && \text{if $\statePri(\mystate) = \prizero$}
\end{aligned}\right\}$ of all $\doublefml[2]$ "st" $\mystate[1] \leadstoUNFO \doublefml[2]$.
For a class $\fmlsetclass$ of finite "formula" sets,
we define $\reintro*\transitionUNFO_{\fmlsetclass}(\mystate)$ in the same way as above,
where the "rules" have been changed to \kl[rules]{those} for $\fmlsetclass$-\kl[clruns]{runs}.
Note that $\transitionUNFO(\mystate)$ (resp.\ $\transitionUNFO_{\fmlsetclass}(\mystate)$) has infinite disjunction/conjunction, because the "rule" \eqref{rule: exists} has infinite patterns according to the name of $z$.
Below for the notational convenience, we canonically extend "positive Boolean formulas" with infinite disjunction and infinite conjunction and extend the notions "wrt" "positive Boolean formulas" ("eg", ${\absmodels}$ and $\leadstoUNFO$).

We first note that each "rule" preserves the truth.
\begin{proposition}[preservation property]\label{proposition: sem equiv UNFO}
    For each "rule" $\mystate[1] \leadstoUNFO \doublefml[2]$ in the "local checker" for "UNFO" (\Cref{figure: rules UNFO}),
    we have     that
    ${} \absmodels \mystate[1]$ "iff" ${} \absmodels \doublefml[2]$.
    Hence, we have:
    \[{} \absmodels \mystate[1] \quad\iff\quad {} \absmodels \transitionUNFO(\mystate[1]).\]
\end{proposition}
\begin{proof}
    By a routine verification.
\end{proof}

\subsubsection{Duality}\label{section: duality UNFO}
\AP
For $\mystate[1] = \mkstate{\fmlset[1]}{\pri[1]}{\ynmulstruc}{\ainter}{\bag[1]} \in \allstatesUNFO$,
we write $\intro*\addPri{\mystate[1]}$ for $\mkstate{\fmlset[1]}{(1 - \pri[1])}{\ynmulstruc}{\ainter}{\bag[1]}$.
Moreover, for each "positive Boolean formula" $\doublefml$ over $\allstatesUNFO$,
we write $\reintro*\addPri{\doublefml}$ for the "dual" of $\doublefml$ in which
each $\mystate[2] \in \allstatesUNFO$ has been replaced with $\addPri{\mystate[2]}$.
We then have the following.
\begin{proposition}[Duality]\label{proposition: dual UNFO}
    For all $\mystate[1] \in \allstatesUNFO$,
    $\transitionUNFO(\addPri{\mystate[1]}) = \addPri{\transitionUNFO(\mystate[1])}$.
\end{proposition}
\begin{proof}
By the form of \Cref{figure: rules UNFO},
for each "rule" $\mystate[1] \leadstoUNFO \doublefml[2]$,
we have that $\addPri{\mystate[1]} \leadstoUNFO \addPri{\doublefml[2]}$ is also a "rule".
\end{proof}

\subsubsection{Consistency}\label{section: consistency UNFO}
The ""negation depth"" of a "formula" $\fml$ is defined as the maximum nesting depth of negations $\lnot$ in the syntax tree of $\fml$.
\begin{proposition}[Weak alternating]\label{proposition: weak alternating UNFO}
    Let $\fml$ be a "UNFO" \kl{formula}.
    Let $\mystate = \mkstate{\fmlset}{\pri}{\ynmulstruc}{\ainter}{\bag}$ be a "state".
    For every infinite \kl{path} $a_1 a_2 \dots$ of $\clUNFO(\fml)$-\kl[clruns]{runs} $\trace$ starting from $\mystate$,
    the number of alternations of "priority" is finite (more precisely, at most the "negation depth" of $\fmlset$).
    Hence, for sufficiently large $n_0 \ge 0$, either one of the following holds:
\begin{itemize}
    \item for all $n \ge n_0$, $\statePri(\trace(a_1 \dots a_n)) = \prizero$;
    \item for all $n \ge n_0$, $\statePri(\trace(a_1 \dots a_n)) = \prione$.
\end{itemize}
\end{proposition}
\begin{proof}
    Observe that the "rule" \eqref{rule: UN} only changes the "priority".
    After we apply \eqref{rule: UN},
    the nesting depth of $\lnot$ decreases by one.
\end{proof}

\begin{proposition}[Consistency]\label{proposition: vdashUNFO con}
    Let $\fml$ be a "UNFO" \kl{formula}.
    For every $\mystate \in \allstates_{\clUNFO(\fml)}$, we have:
    \[\vdashUNFOsub_{\clUNFO(\fml)} \mystate \quad\Longrightarrow\quad \nvdashUNFOsub_{\clUNFO(\fml)} \addPri{\mystate}.\]
\end{proposition}
\begin{proof}
    Towards a contradiction, assume that $\vdashUNFOsub_{\clUNFO(\fml)} \addPri{\mystate}$.
    Let $\trace[1]$ and $\trace[2]$ be "accepting" $\clUNFO(\fml)$-\kl[clrun]{runs} starting from $\mystate$
    and from $\addPri{\mystate}$, respectively.
    \begin{claim}
        There are
        an infinite "path" $a_1 a_2 \dots$ on $\trace[1]$ and
        an infinite "path" $b_1 b_2 \dots$ on $\trace[2]$ such that for all $n \in \nat$,
        \[\trace[1](a_1 \dots a_n) = \addPri{\trace[2](b_1 \dots b_n)}.\]
    \end{claim}
    \begin{proof}
        Suppose that $\word \in \fdom(\trace[1])$ and $\word' \in \fdom(\trace[2])$ such that
        \[\trace[1](\word) = \addPri{\trace[2](\word')}.\]
        Let $\transitionUNFO_{\clUNFO(\fml)}(\trace[1](\word)) \lequiv \bigdoublevee_{l} \bigdoublewedge_{k} {\mystate[2]_{l, k}}$.
        As $\trace[1]$ is a $\clUNFO(\fml)$-\kl[clrun]{run}, we have:
        \begin{itemize}
            \item for some $l$, for all $k$, $\mystate[2]_{l, k}$ occurs on a \kl{child} of $\trace(\word)$.
        \end{itemize}
        Also, 
        we have
        $\transitionUNFO_{\clUNFO(\fml)}(\addPri{\trace[2](\word')}) = \addPri{\transitionUNFO_{\clUNFO(\fml)}(\trace[2](\word'))} \lequiv 
        \bigdoublewedge_{l} \bigdoublevee_{k} \addPri{\mystate[2]_{l, k}}$ by \Cref{proposition: dual UNFO}.
        As $\trace[2]$ is a $\clUNFO(\fml)$-\kl[clrun]{run}, we have:
        \begin{itemize}
            \item for all $l$, for some $k$, $\addPri{\mystate[2]_{l, k}}$ occurs on a \kl{child} of $\trace[2](\word')$.
        \end{itemize}
        Thus, by choosing $l$ and $k$ appropriately,
        we have that there is some $\mystate[2]$ such that 
        $\mystate[2]$ occurs on a \kl{child} of $\trace(\word)$ and 
        $\addPri{\mystate[2]}$ occurs on a \kl{child} of $\trace(\word')$.
        Hence, this completes the proof.
    \end{proof}
    Thus, we have $\Pri_{\trace}(a_1 a_2 \dots) \neq \Pri_{\trace[2]}(b_1 b_2 \dots)$ (note that one is $\dots 0^{\omega}$ and the other is $\dots 1^{\omega}$ for "priority" by \Cref{proposition: weak alternating UNFO}),
    which contradicts that both $\trace[1]$ and $\trace[2]$ are "accepting".
\end{proof}

\subsubsection{Proof of \Cref{theorem: completeness closure UNFO}}\label{section: theorem: completeness UNFO proof}
\AP
A \intro*\kl(LMC){subrun} starting from $\mystate$ is a $\allstates_{\mathrm{\kl{UNFO}}}$-labeled 
\kl{tree} $\trace$ with $\trace(\varepsilon) = \mystate$ such that
for each $\word \in \fdom(\trace)$ "st" $\word = \eps$ or $\word$ is not a "leaf" of $\trace$,
we have that for $\left\{\begin{aligned}
    &\text{some $\doublefml[2]$} && \text{if $\statePri(\trace(\word)) = \prione$} \\
    &\text{every $\doublefml[2]$} && \text{if $\statePri(\trace(\word)) = \prizero$}
\end{aligned}\right\}$ "st" $\trace(\word) \leadstoUNFO \doublefml[2]$,
the \kl{positive Boolean formula} resulting from
replacing in $\doublefml[2]$
each "child" "state" of $\word$ ("ie", each element $\trace(\word a)$) with $\doubletrue$ is "semantically equivalent" to $\doubletrue$.
That is, \kl(LMC){subruns} are \kl(LMC){runs} where the condition for "leaves" (except the "root") is disregarded.
\AP
Similarly, for a class $\fmlsetclass$ of finite "formula" sets, $\fmlsetclass$-\intro*\kl[subclruns]{subruns} are defined in the same way.

To prove \Cref{theorem: completeness closure UNFO},
it suffices to prove the following two lemmas.
Here, $\paramUNFO(\mystate)$ is a well-founded parameter such that if $\paramUNFO(\mystate[1]) \ge \paramUNFO(\mystate[2])$ and $\absmodels \mystate[1]$, then $\absmodels \mystate[2]$, which will be defined (\Cref{label: paramater UNFO}).
\begin{lemma}[Even case, \Cref{section: lemma: parameter UNFO even}]\label{lemma: parameter UNFO even}
    Let $\fml$ be a "UNFO" \kl{formula}.
    For every $\mystate[1] \in \allstates_{\clUNFO(\fml)}$ "st" $\statePri(\mystate[1]) = \prizero$,
    if $\absmodels \mystate[1]$,
    then there is a $\clUNFO(\fml)$-\kl[subclruns]{subrun} $\trace$ of finite "height" starting from $\mystate[1]$ such that 
    $\absmodels \trace(\word)$ for each "leaf" $\word$ of $\trace$.
\end{lemma}
\begin{lemma}[Odd case, \Cref{section: lemma: parameter UNFO odd}]\label{lemma: parameter UNFO odd}
    Let $\fml$ be a "UNFO" \kl{formula}.
    For every $\mystate[1] \in \allstates_{\clUNFO(\fml)}$ "st" $\statePri(\mystate[1]) = \prione$,
    if $\absmodels \mystate[1]$,
    then there is a finite $\clUNFO(\fml)$-\kl[subclruns]{subrun} $\trace$ starting from $\mystate[1]$ such that 
    $\paramUNFO(\mystate[1]) > \paramUNFO(\trace(\word))$ (so, $\absmodels \trace(\word)$)
    for each "leaf" $\word$ of $\trace$.
\end{lemma}
Using them, we can prove \Cref{theorem: completeness closure UNFO} as follows.
\begin{theorem*}[Restatement of \Cref{theorem: completeness closure UNFO}]
    \completenessclosureUNFO
\end{theorem*}
\begin{proof}
    \proofcasethin{($\Longrightarrow$)}
    Let $\trace$ be the $\clUNFO(\fml)$-\kl[clrun]{run},
    obtained from the singleton \kl{tree} with $\trace(\eps) = \mystate$
    by extending each \kl{leaf} with the $\clUNFO(\fml)$-\kl[subclruns]{subrun} of \Cref{lemma: parameter UNFO odd,lemma: parameter UNFO even}, iteratively.
    We then have:
    \begin{claim}
        For all infinite \kl{paths} $a_1 a_2 \dots$ on $\trace$,
        $\Pri_{\trace}(a_1 a_2 \dots) = \prizero$.
    \end{claim}
    \begin{proof}
        By construction (of \Cref{lemma: parameter UNFO odd}),
        for every $n \ge 0$,
        if $\statePri(\trace(a_1 \dots a_n)) = \prione$,
        then we have $\paramUNFO(\trace(a_1 \dots a_n)) > \paramUNFO(\trace(a_1 \dots a_m))$ for some $m > n$.
        Because the ordering is well-founded,
        the "priority" $\statePri(\trace(a_1 \dots a_m))$ is eventually changed to the even "priority" $\prizero$ for some $m > n$.
        Hence, $\statePri_{\trace}(a_1 a_2 \dots) = \prizero$.
    \end{proof}
    Hence, $\trace$ is an \kl{accepting} $\clUNFO(\fml)$-\kl[clrun]{run}, whereby $\vdashUNFOsub_{\clUNFO(\fml)} \mystate$.

    \proofcasethin{($\Longleftarrow$)}
    We have:
    \begin{align*}
        {} \vdashUNFOsub_{\clUNFO(\fml)} \mystate
        &\quad\Longrightarrow\quad {} \nvdashUNFOsub_{\clUNFO(\fml)} \addPri{\mystate} \tag*{(By \Cref{proposition: vdashUNFO con})}\\
        &\quad\Longrightarrow\quad {} \not\absmodels \addPri{\mystate} \tag*{(By the direction ($\Longrightarrow$))}\\
        &\quad\Longrightarrow\quad {} \absmodels \mystate. \tag*{(By \Cref{definition: states})}
    \end{align*}
    Hence, this completes the proof.
\end{proof}
Below, we prove the remaining parts (\Cref{lemma: parameter UNFO odd,lemma: parameter UNFO even}).
\subsubsection{Proof of \Cref{lemma: parameter UNFO even} (even case)}\label{section: lemma: parameter UNFO even}
\begin{proof}
    By \Cref{proposition: sem equiv UNFO}, we have $\absmodels \transitionUNFO(\mystate[1])$.
    For each \kl{state} $\mystate[2]$ with odd "priority" in $\transitionUNFO(\mystate[1])$,
    the rule \eqref{rule: UN} is always applied from $\mystate[1]$, since the "rule" \eqref{rule: UN} only changes the "priority".
    Thus $\mystate[2] \in \allstates_{\clUNFO(\fml)}$.
    For each \kl{state} $\mystate[2]$ with even "priority" in $\transitionUNFO(\mystate[1])$, it is not replaced or replaced with $\doubletrue$.
    Thus, we have $\transitionUNFO(\mystate[1]) \lleqq \transitionUNFO_{\clUNFO(\fml)}(\mystate[1])$,
    and hence $\absmodels \transitionUNFO_{\clUNFO(\fml)}(\mystate[1])$.
    Let $\transitionUNFO_{\clUNFO(\fml)}(\trace[1](\word)) \lequiv \bigdoublevee_{l} \bigdoublewedge_{k} {\mystate[2]_{l, k}}$.
    Then $\absmodels \bigdoublewedge_{k} {\mystate[2]_{l, k}}$ for some $l$.
    Hence, by taking the $\clUNFO(\fml)$-\kl[subclruns]{subrun} (of "height" 1) having the "leaves" of $\mystate[2]_{l, k}$, this completes the proof.
\end{proof}

\subsubsection{Proof of \Cref{lemma: parameter UNFO odd} (odd case)}\label{section: lemma: parameter UNFO odd}

\paragraph{A well-founded parameter}\label{label: paramater UNFO}
\AP
For $\mkstate{\fmlset[1]}{\pri}{\ynmulstruc}{\ainter}{\bag[1]} \in \allstatesUNFO$ with odd "priority" $\pri$,
the ""concretization set"" $\intro*\paramconcUNFO(\mkstate{\fmlset[1]}{\pri}{\ynmulstruc}{\ainter}{\bag[1]})$ is defined by:
\begin{align*}
    \paramconcUNFO(\mkstate{\fmlset[1]}{\pri}{\ynmulstruc}{\ainter}{\bag[1]}) &
    \;\defeq\;
    \set*{\concstate{\fmlset[1]}{\glue \ynmulstruc}{\inter} \;\middle|\; \mbox{$\inter$ is a \kl{concretization} of $\ainter$ on $\bag[1]$, and $\glue\ynmulstruc \modelsass{\inter} \bigdoublewedge \fmlset[1]$}}.
\end{align*}
Note that $\concstate{\fmlset[1]}{\ynstruc}{\inter}$ expresses the $3$-tuple of a "formula" set $\fmlset[1]$, a "structure" $\ynstruc$, and an "interpretation" $\inter$ on $\ynstruc$.
By definition (\Cref{definition: semantics on tree decomposition}), for $\mystate[1] \in \allstatesUNFO$ with odd "priority", we have:
\[\absmodels \mystate[1] \quad\Longleftrightarrow\quad \paramconcUNFO(\mystate[1]) \neq \emptyset.\]

For two "bags" $\bag[1], \bag[1]' \in \fdom(\ynmulstruc)$,
we write $\treedistance(\bag[1], \bag[1]')$ for the distance between $\bag[1]$ and $\bag[1]'$ on the \kl{tree} $\ynmulstruc$.
\AP
For a "bag" $\bag \in \fdom(\ynmulstruc)$ and a finite set $C \subseteq \univ{\glue \ynmulstruc}$ of \kl(LMC){fixed names},
the ""distance"" $\intro*\treedistance(\bag, C)$ is defined as follows:
\[\treedistance(\bag[1], C) ~\defeq~
\max_{\vertex \in C} \min_{\substack{\bag[1]' \in \fdom(\ynmulstruc) \text{ "st"}\\
\vertex \in \conc{\ynmulstruc}{\bag[1]'}(\univ{\ynmulstruc(\bag[1]')})
}} \treedistance(\bag[1], \bag[1]').\]
(Here, $\conc{\ynmulstruc}{\bag[1]'}(\univ{\ynmulstruc(\bag[1]')})$ expresses the set of \kl(LMC){fixed names} in $\glue \ynmulstruc$
indicated by \kl(LMC)[fixed names]{those} of $\univ{\ynmulstruc(\bag[1]')}$.)

\AP For a "state" $\mystate[1] = \mkstate{\fmlset[1]}{\pri}{\ynmulstruc}{\ainter}{\bag[1]} \in \allstatesUNFO$,
the ""parameter"" $\intro*\paramUNFO(\mystate[1]) \in \nat^2 \dcup \set{\infty}$ is defined as follows:
\begin{align*}
    \paramUNFO(\mystate[1])
    &~\defeq~ 
    \begin{cases}
    {\displaystyle \min_{\fmlset[2]^{\glue \ynmulstruc}_{\inter} \in \paramconcUNFO(\mystate[1])}
    \tup{\fmllen{\fmlset[2]}, \treedistance(\bag, \inter(\FV(\fmlset[1])))}
    } & \mbox{if $\pri = \prione$ and $\absmodels \mystate[1]$,}\\
    \tup{0,0} & \mbox{if $\pri = \prizero$ and $\absmodels \mystate[1]$,}\\
    \infty & \mbox{otherwise,}
    \end{cases} \span \span \span
\end{align*}
On this "parameter",
we use the lexicographical ordering on $\nat^2$ extended with the maximum element $\infty$,
which is clearly a well-founded ordering.
By definition, we have:
${} \absmodels \mystate[1]$ "iff"
$\paramUNFO(\mystate[1]) \neq \infty$.

\paragraph{Proof of \Cref{lemma: parameter UNFO odd}}
Based on the evaluation strategy presented in \Cref{section: closure property UNFO},
we give a $\clUNFO(\fml)$-\kl[subclrun](LMC){subrun} that decreases the "parameter" $\paramUNFO$.
\begin{proof}
    By induction on $\paramUNFO(\mystate[1])$.
    Let $\mystate[1] = \mkstate{\fmlset[1]}{\prione}{\ynmulstruc}{\ainter}{\bag}$.
    We distinguish the following cases.

    \proofcasethin{Step 1}
    \labeltext{lemma: parameter UNFO concrete}{}{1}
    If $\fmlset[1]$ contains a \kl{free variable} $x$ with $\card\ainter(x) \ge 2$,
    then we apply the "rule" \eqref{rule: concrete}.
    In this application, while the "parameter" may not strictly decrease,
    it suffices to assume the following in the subsequent cases, without loss of generality:
    \begin{enumerate}
        \item[a)] \labeltext{cond: evaluation strategy UNFO proof 1}{}{a} $\card \ainter(x) = 1$ for each $x \in \FV(\fmlset[1])$.
    \end{enumerate}

    \proofcasethin{Step 1'}
    \labeltext{lemma: parameter UNFO or and exists}{}{1'}
    If $\fmlset[1]$ contains a \kl{formula} of the form $\fml[2] \lor \fml[3]$, $\fml[2] \land \fml[3]$, or $\exists x \fml[2]$:
    In this case, by applying the corresponding "rule" \eqref{rule: or}, \eqref{rule: and}, or \eqref{rule: exists} with IH, this case is shown.

    Here, the transformed "state" is indeed in $\allstates_{\clUNFO(\fml)}$ ("ie", the "formula" set is in $\clUNFO(\fml)$).
    For instance, we have that 
    $(\fmlset[2], \fml[2]_1 \lor \fml[2]_2) \in \clUNFO(\fml)$ implies $(\fmlset[2], \fml[2]_i) \in \clUNFO(\fml)$.
    This is shown by straightforwardly transforming the derivation tree of $\clUNFO(\fml)$, as follows:
    \begin{center}
        \scalebox{.9}{\begin{prooftree}[separation = .5em]
            \hypo{\dots}
            \hypo{\mathstrut}
            \infer1{(\fml[2]_1 \lor \fml[2]_2) \in \clUNFO(\fml[2]_1 \lor \fml[2]_2)}
            \hypo{\dots}
            \infer[double]3{(\fmlset[2], \fml[2]_1 \lor \fml[2]_2) \in \clUNFO(\fml)}
        \end{prooftree}}
        \hspace{1em}to\hspace{1em}
        \scalebox{.9}{\begin{prooftree}[separation = .5em]
            \hypo{\dots}
            \hypo{\mathstrut}
            \infer1{(\fml[2]_i) \in \clUNFO(\fml[2]_i)}
            \infer1{(\fml[2]_i) \in \clUNFO(\fml[2]_1 \lor \fml[2]_2)}
            \hypo{\dots}
            \infer[double]3{(\fmlset[2], \fml[2]_i) \in \clUNFO(\fml)}
        \end{prooftree}}.\footnote{We use double lines to indicate that multiple rules are applied.}
    \end{center}
    Here, by using the monotonicity ("ie", $(\fmlset[1], \fmlset[2]) \in \clUNFO(\fml)$ implies $\fmlset[2] \in \clUNFO(\fml)$),
    we can assume that $\fml[2]_1 \lor \fml[2]_2$ only occurs in the path from the \kl{root} to $(\fml[2]_1 \lor \fml[2]_2) \in \clUNFO(\fml[2]_1 \lor \fml[2]_2)$, in the derivation tree in the left-hand side.
    Then by replacing each occurrence $(\bl, \fml[2]_1 \lor \fml[2]_2)$ with $(\bl, \fml[2]_i)$,
    we can obtain the derivation tree in the right-hand side.
    (Similarly for the other cases.)

    This step allows us to assume the following in the subsequent cases.
    \begin{enumerate}
        \item[b)] \labeltext{cond: evaluation strategy UNFO proof 2}{}{b} $\fmlset[1]$ is of the form $(\fml[2]_1, \dots, \fml[2]_n)$
        where each $\fml[2]_i$ is one of the following forms:
        an "atom" $\afml$ or $\lnot \fml[3]$.
    \end{enumerate}

    Moreover, we can assume the following in the subsequent cases.
    \begin{enumerate}
        \item[c)] \labeltext{cond: evaluation strategy UNFO proof 3}{}{c} for each $\fml[2]_i$,
        for some $\dir \in \range{\back 2}{2}$,
        $\ainter(\FV(\fml[2]_i)) \subseteq \ND{\ynmulstruc}{\bag}{\dir}$.
    \end{enumerate}
    If $\fml[2]_i$ is of the form $\afml$,
    then \nameref{cond: evaluation strategy UNFO proof 3}) holds by $\absmodels \mkstate{(\fml[2]_i)}{\prione}{\ynmulstruc}{\ainter}{\bag}$.
    If $\fml[2]_i$ is of the form $\lnot \fml[3]$,
    then \nameref{cond: evaluation strategy UNFO proof 3}) holds by the condition of unary negation: $\card\FV(\fml[3]) \le 1$.

    \proofcasethin{Step 2}\labeltext{lemma: parameter UNFO split}{}{2}
    Else if the following does not hold:
    \begin{enumerate}
        \item[c')] \labeltext{cond: evaluation strategy UNFO proof 3'}{}{c'} for each $\fml[2]_i$, for a $\dir \in \set{\back 2, \back 1, 1, 2}$, $\dir \in \ainter(\FV(\fml[2]_i)) \subseteq \ND{\ynmulstruc}{\bag}{\dir}$.
    \end{enumerate}
    Let $\fml[2]_i$ be a "formula" "st" $\ainter(\FV(\fml[2]_i)) \subseteq \ND{\ynmulstruc}{\bag}{0}$.
    We then distinguish the following cases.

    \subproofcasethin{Case $\card \fmlset[1] \ge 2$}
    Then by applying \eqref{rule: split}, we can eliminate $\fml[2]_i$ from $\fmlset[1]$.
    Then by applying IH, this case is shown.
    
    \subproofcasethin{Case $\card \fmlset[1] = 1$}
    By $\absmodels \mkstate{\set{\fml[2]_i}}{\prione}{\ynmulstruc}{\ainter}{\bag}$ and $\ainter(\FV(\fml[2]_i)) \subseteq \ND{\ynmulstruc}{\bag}{0}$,
    we have that $\ynmulstruc(\bag) \modelsass{\set{x \mapsto \vertex\mid x \mapsto \set{\vertex} \in \ainter}} \fml[2]_i$.
    Then by applying \eqref{rule: a1} or \eqref{rule: UN}, this case is shown.
    Here, in \eqref{rule: a1}, we may also apply \eqref{rule: mov} multiple times; "cf" Footnote~\ref{footnote: a1 and mov}.

    This step allows us to assume \nameref{cond: evaluation strategy UNFO proof 3'}) in the subsequent cases.
    
    \proofcasethin{Step 3}\labeltext{evaluation strategy UNFO fmlset split}{}{3}
    Else if the following condition does not hold:
    \begin{enumerate}
        \item[c'')] \labeltext{cond: evaluation strategy UNFO proof 3''}{}{c''} for some $\dir \in \set{\back 2, \back 1, 1, 2}$, $\dir \in \ainter(\FV(\fmlset[1])) \subseteq \ND{\ynmulstruc}{\bag}{\dir}$.
    \end{enumerate}
    First, if $\card \fmlset[1] = 0$, then by applying the "rule" \eqref{rule: emp}, this case is shown.
    Otherwise ("ie", $\card \fmlset[1] \ge 2$ by \nameref{cond: evaluation strategy UNFO proof 3'})), by applying the "rule" \eqref{rule: split} with IH, this case is shown.

    This step allows us to assume \nameref{cond: evaluation strategy UNFO proof 3''}) in the subsequent cases.

    \proofcasethin{Step 4}\labeltext{evaluation strategy UNFO mov}{}{4}
    Otherwise, by applying the rule \eqref{rule: mov} with IH, this case is shown.
    Note that the "parameter" strictly decreases, because the "distance" strictly decreases by $d \in \ainter(\FV(\fmlset[1]))$.

    Hence, this completes the proof.
\end{proof}

\subsection{Proof of {\Cref{proposition: closure size UNFO}}}\label{section: proposition: closure size UNFO}
\AP
We write $\fmlset[1] \intro*\renameeq \fmlset[2]$ if $\fmlset[1]$ and $\fmlset[2]$ are the same \kl{formula} set up to renaming of \kl{free variables}.
We denote by $\clUNFO(\fml)/{\renameeq}$ the set of \kl{equivalence classes} of $\clUNFO(\fml)$ w.r.t.\ $\renameeq$. 
\begin{proof}
    By easy induction on the \kl(formula){size} $\fmllen{\fml}$.

    \proofcasethin{Case $\fml[1] = \afml$}
    By $\card \clUNFO(\afml)/{\renameeq} = 2$.

    \proofcasethin{Case $\fml[1] = \fml[2] \lor \fml[3]$}
    We have
    \begin{align*}
        \card \clUNFO(\fml[1])/{\renameeq}
        \le 1 + \card \clUNFO(\fml[2])/{\renameeq} + \card \clUNFO(\fml[3])/{\renameeq}
        &\le 1 + (2 \fmllen{\fml[2]})^{2\fmllen{\fml[2]}} + (2 \fmllen{\fml[3]})^{2\fmllen{\fml[3]}} \tag{IH}\\
        &\le (2 (1 + \fmllen{\fml[2]} + \fmllen{\fml[3]}))^{2\max(\fmllen{\fml[2]}, \fmllen{\fml[3]})}\\
        &\le (2 \fmllen{\fml[1]})^{2\fmllen{\fml}}.
    \end{align*}

    \proofcasethin{Case $\fml[1] = \fml[2] \land \fml[3]$}
    We have
    \begin{align*}
        \card \clUNFO(\fml[1])/{\renameeq}
        \le 1 + \card \clUNFO(\fml[2])/{\renameeq} \times \card \clUNFO(\fml[3])/{\renameeq}
        &\le 1 + (2 \fmllen{\fml[2]})^{2\fmllen{\fml[2]}} \times (2 \fmllen{\fml[3]})^{2\fmllen{\fml[3]}} \tag{IH}\\
        &\le (2 (1 + \fmllen{\fml[2]} + \fmllen{\fml[3]}))^{2(\fmllen{\fml[2]} + \fmllen{\fml[3]})}\\
        &\le (2 \fmllen{\fml})^{2\fmllen{\fml}}.
    \end{align*}
    (The first inequality is thanks to the condition ``$\FV(\fmlset[1]) \cap \FV(\fmlset[2]) \subseteq \FV(\fml[1]) \cap \FV(\fml[2])$'' in \Cref{definition: closure UNFO}.)

    \proofcasethin{Case $\fml[1] = \exists x \fml[2]$}
    We have
    \begin{align*}
        \card \clUNFO(\fml[1])/{\renameeq}
        &\le 1 + \card \clUNFO(\fml[2][z/x])/{\renameeq}  \tag{By $\clUNFO(\fml[2][z/x])/{\renameeq} = \clUNFO(\fml[2][z'/x])/{\renameeq}$}\\
        &\le 1 + (2 \fmllen{\fml[2]})^{2\fmllen{\fml[2]}} \tag{IH}\\
        &\le (2 (1 + \fmllen{\fml[2]}))^{2\fmllen{\fml[2]}}
        \le (2 \fmllen{\fml[1]})^{2\fmllen{\fml[1]}}.
    \end{align*}

    \proofcasethin{Case $\fml[1] = \afml \land \lnot \fml[2]$}
    We have
    \begin{align*}
        \card \clUNFO(\fml[1])/{\renameeq}
        \le 1 + \card \clUNFO(\fml[2])/{\renameeq}
        &\le 1 + (2 \fmllen{\fml[2]})^{2 \fmllen{\fml[2]}} \tag{IH}\\
        &\le (2 (1 + \fmllen{\fml[2]}))^{2 \fmllen{\fml[2]}}
        = (2 \fmllen{\fml[1]})^{2 \fmllen{\fml[1]}}. \qedhere
    \end{align*}
\end{proof}

\subsection{Proof of {\Cref{corollary: 2-EXPTIME UNFO}}: 2APTA construction for UNFO}\label{section: 2APTA construction UNFO}

In this section, from the "local checker" for "UNFO", we construct \kl{2APTAs}.

Let \[\AD^{(k)} \defeq \set{\fixedvertex_1, \dots, \fixedvertex_k} \dcup \set{\back 2, \back 1, 1, 2}.\]
As in \Cref{proposition: standard semantics and semantics on tree decompositions},
we use $\elem$ to denote an element in $\AD^{(k)}$,
we use $\fixedvertex_1, \dots, \fixedvertex_k$ for indicating "nodes", and
we use each $\dir \in \set{\back 2, \back 1, 1, 2}$ for indicating a \kl{direction} on \kl{tree decompositions}.

\AP
Let \[\intro*\allsymbols \defeq \set{\ADBsymbol} \dcup \set{\NDsymbol_{\dir}, \MDsymbol_{\dir} \mid \dir \in \set{\back 2, \back 1, 1, 2}}\]
be the set of symbols.
Intuitively, we use them as follows, respectively:
\begin{itemize}
    \item the symbol $\AP\intro*\ADBsymbol$ indicates the set $\ADB{\ynmulstruc}{\bag}$ (recall \Cref{section: abstract semantics on tree decompositions}),
    \item the symbol $\AP\intro*\NDsymbol_{\dir}$ indicates the set $\ND{\ynmulstruc}{\bag}{\dir}$, and
    \item the symbol $\AP\intro*\MDsymbol_{\dir}$ indicates the set $\MD{\ynmulstruc}{\bag}{\dir}$.
\end{itemize}

\begin{definition}\label{definition: 2APTA UNFO}
    For $k \in \pnat$ and a "UNFO" \kl{sentence} $\fml$,
    the \kl{2APTA} $\automaton[1]_{k}^{\fml} = \tup{Q, \transitionUNFO, \Pri, q_0}$ over $\allatstruc_{k}$ is defined as follows:
    \begin{itemize}
        \item $q^{\pri} \in Q$ consists of the following:

        \begin{itemize}
            \item ${\fmlset[1]}^{\pri}_{\ainter}$ modulo $\renameeq$ where
            \begin{itemize}
                \item $\fmlset[1] \in \clUNFO(\fml)$,
                \item $\pri \in \set{\prizero, \prione}$, and
                \item $\ainter \colon \FV(\fmlset[1]) \to \pset{\AD^{(k)}}$,
            \end{itemize}
            (We write
            ${\fmlset[1]}^{\pri}_{\ainter} \renameeq {\fmlset[2]}^{\pri'}_{\ainter'}$ if they coincide up to renaming of \kl{variables} in $\fdom(\ainter)$.)

            \item ${\fmlset[1]}^{\pri}_{\ainter, \chi}$ modulo $\renameeq$ where
                \begin{itemize}
                    \item $\fmlset[1] \in \clUNFO(\fml)$,
                    \item $\pri \in \set{\prizero, \prione}$, 
                    \item $\ainter \colon \FV(\fmlset[1]) \to \pset{\AD^{(k)}}$, and
                    \item $\chi \colon \allsymbols \to \pset{\AD^{(k)}}$.
                \end{itemize}
            (We write
            ${\fmlset[1]}^{\pri}_{\ainter, \chi} \renameeq {\fmlset[2]}^{\pri'}_{\ainter', \chi'}$ if they coincide up to renaming of \kl{variables} in $\fdom(\ainter)$.)
            \item $(\bullet = \elemset)^{\prione}$ where
            \begin{itemize}
                \item $\bullet \in \allsymbols$, and
                \item $\elemset \subseteq \AD^{(k)}$,
            \end{itemize}
 
            \item $\elem^{\pri}$ where
            \begin{itemize}
                \item $\elem \in \AD^{(k)}$, and
                \item $\pri \in \set{\prizero, \prione}$.
            \end{itemize}
        \end{itemize}
        
    \item The relation $(\leadsto) \subseteq (Q \times \allatstruc_{k}) \times \PBallfml(Q \times \range{\back 2}{2})$
        is defined as the minimal binary relation closed under the rules in \Cref{figure: rules UNFO'}.
        Here, 
        we use the Iverson bracket notation \cite{knuthTwoNotesNotation1992}, given by
        \[\intro*\iverson{P} \defeq \begin{cases}
            1 & \mbox{if $P$ holds,}\\
            0 & \mbox{otherwise.}
        \end{cases}\]
        We then define $\transitionUNFO(q^{\pri}, \ynstruc) \defeq \bigdoublevee^{\pri} \set{ \doublefml[2] \mid \tup{q^{\pri}, \ynstruc} \leadsto \doublefml[2]}$,
        \item $\Pri(q^{\pri}) = \pri$, and
        \item $q_0 = (\fml)^{\prione}_{\emptyset}$. \lipicsEnd
    \end{itemize}
\end{definition}

\begin{figure*}[t]
    \centering
\begin{tcolorbox}[standard jigsaw, opacityback = 0.8, colframe=black!80, boxrule=.3mm, left=.2mm, right=.2mm]
    \vspace{-4ex}
    \begin{align*}
    \tup{()^{\pri}_{\ainter,\chi}, \ynstruc} &\ \leadstoUNFO\ \doubletrue^{\pri}, \tag{emp}\\
    \tup{(\afml)^{\pri}_{\ainter,\chi}, \ynstruc} &\ \leadstoUNFO\ 
    ()^{\pri}_{\ainter,\chi} \text{ if 
    $\begin{cases}
    \text{for each $x \in \FV(\afml)$, there is}\\
    \text{\quad some $\vertex \in \univ{\ynstruc}$ \kl{st} $\ainter(x) = \set{\vertex}$,}\\
    \text{and $\ynstruc \modelsass{\set{x \mapsto \vertex \mid x \mapsto \set{\vertex} \in \ainter}} \afml$,}    
    \end{cases}$} \tag{$\afml$}\\
    \tup{(\fmlset[1], \fml[1] \lor \fml[2])^{\pri}_{\ainter,\chi}, \ynstruc} &\ \leadstoUNFO\  \tup{(\fmlset[1], \fml[1])^{\pri}_{\ainter,\chi}, 0} \doublevee^{\pri} \tup{(\fmlset[1], \fml[2])^{\pri}_{\ainter,\chi}, 0}, \tag*{($\lor$)}\\
    \tup{(\fmlset[1], \fml[1] \land \fml[2])^{\pri}_{\ainter,\chi}, \ynstruc} &\ \leadstoUNFO\  \tup{(\fmlset[1], \fml[1], \fml[2])^{\pri}_{\ainter,\chi}, 0}, \tag{$\land$} \\
    \tup{(\fmlset[1], \exists x \fml)^{\pri}_{\ainter,\chi}, \ynstruc} &\ \leadstoUNFO\
    \tup{(\fmlset[1], \fml[1]\fmlsubst{x}{z})^{\pri}_{\ainter[\chi(\ADBsymbol)/z], \chi}, 0} \quad\text{ if $z$ is \kl{fresh}}, \tag{$\exists$} \\
    \tup{(\fmlset[1])^{\pri}_{\ainter, \chi}, \ynstruc} &\ \leadstoUNFO\  \bigdoublevee_{\elem \in \ainter(x)}^{\pri} \tup{(\fmlset[1])^{\pri}_{\ainter\aintersubst{x}{\elem},\chi}, 0} \quad\text{ if $x \in \FV(\fmlset[1])$,} \tag{conc}\\
    \tup{(\lnot \fml)^{\pri}_{\ainter,\chi}, \ynstruc} &\ \leadstoUNFO\ 
    \tup{(\fml)^{(1 - \pri)}_{\ainter,\chi}, 0} \quad\text{ if 
    $\begin{cases}
    \text{$\ainter(x) \subseteq \univ{\ynstruc}$
    and 
    $\card {\ainter(x)} = 1$
    }\\
    \text{for each $x \in \FV(\fml)$,}    
    \end{cases}$} \quad\tag{$\lnot$}\\
    \tup{(\fmlset[1], \fmlset[2])^{\pri}_{\ainter,\chi}, \ynstruc} &\ \leadstoUNFO\  \tup{(\fmlset[1])^{\pri}_{\ainter,\chi},0} \doublewedge^{\pri} \tup{\fmlset[2]^{\pri}_{\ainter,\chi},0} \text{ if 
    $\begin{cases}
    \text{$\ainter(x) \subseteq \univ{\ynstruc}$
    and 
    $\card {\ainter(x)} = 1$
    }\\
    \text{for each $x \in \FV(\fmlset[1]) \cap \FV(\fmlset[2])$,}    
    \end{cases}$}
    \tag{split}\\
    \tup{{\fmlset[1]}^{\pri}_{\ainter,\chi}, \ynstruc} &\ \leadstoUNFO\
    \tup{{\fmlset[1]}^{\pri}_{\ainter_{\dir}}, \dir} \quad\text{ if $\ainter(\FV(\fmlset[1])) \subseteq \chi(\NDsymbol_{\dir})$,}\\
    & \hspace{0em}\text{where $\ainter_{\dir}(x) \defeq \begin{cases}
        (\ainter(x) \setminus \set{\dir}) \cup \chi(\MDsymbol_{\dir}) & \text{if $\dir \in \ainter(x)$}\\
        \ainter(x) & \text{otherwise},
    \end{cases}$}
    \tag{move} \label{rule: UNFO' mov}\\
    \tikz \draw[dotted] (0,0)--(5,0); \text{\footnotesize {} ($\downarrow$ rules for encoding ``$\ADB{\ynmulstruc}{\bag}$'', ``$\ND{\ynmulstruc}{\bag}{\dir}$'', ``$\MD{\ynmulstruc}{\bag}{\dir}$'') } \tikz \draw[dotted] (0,0)--(1,0); \span \\
    \tup{{\fmlset[1]}^{\pri}_{\ainter}, \ynstruc} &\ \leadstoUNFO\ 
    \bigdoublevee_{\chi \colon \allsymbols \to \pset{\AD^{(k)}}}
    (\bigdoublewedge_{\bullet \in \allsymbols} 
    \tup{(\bullet = \chi(\bullet))^{\prione}, 0} \doublewedge \tup{{\fmlset[1]}^{\pri}_{\ainter, \chi}, 0}),
    \tag{move'} \label{rule: UNFO' mov'}\\[1.ex]
    \tup{(\ADBsymbol = \elemset)^{\prione}, \ynstruc} & \ \leadstoUNFO\ 
    \bigdoublewedge_{\elem \in \AD^{(k)}} \tup{\elem^{\iverson{\elem \in \elemset}}, 0},\\
    \tup{(\NDsymbol_{\dir} = \elemset)^{\prione}, \ynstruc} & \ \leadstoUNFO\ 
    \bigdoublewedge_{\elem \in \univ{\ynstruc}} \tup{\elem^{\iverson{\elem \in \elemset}}, \dir} \doublewedge \tup{\dir^{[\dir \in \elemset]}, 0}
    \text{ if $\elemset \subseteq \univ{\ynstruc} \dcup \set{\dir}$,}\\
    \tup{(\MDsymbol_{\dir} = \elemset)^{\prione}, \ynstruc} & \ \leadstoUNFO\ 
    \bigdoublewedge_{\elem \in \AD^{(k)} \setminus (\univ{\ynstruc} \dcup \set{\back \dir})} \tup{\elem^{\iverson{\elem \in \elemset}}, \dir} \text{ if $\elemset \subseteq \AD^{(k)} \setminus (\univ{\ynstruc} \dcup \set{\back \dir})$,}\\
    \tup{\vertex_i^{\pri}, \ynstruc} & \ \leadstoUNFO\ \doubletrue^{\pri} \text{ if $\vertex_i \in \univ{\ynstruc}$,}\\
    \tup{\vertex_i^{\pri}, \ynstruc} & \ \leadstoUNFO\ \doublefalse^{\pri} \text{ if $\vertex_i \not\in \univ{\ynstruc}$,}\\
    \tup{\dir^{\pri}, \ynstruc} & \ \leadstoUNFO\ 
    \bigdoublevee_{\elem \in \AD^{(k)} \setminus \set{\back \dir}}^{\pri} \tup{\elem^{\pri}, \dir} \text{ for $\dir \in \set{\back 2, \back 1, 1, 2}$}.
    \end{align*}
\end{tcolorbox}     \caption{\kl{2APTA} transition rules from the "local checker" for "UNFO" (\Cref{figure: rules UNFO}).
    Here, we just write $\fmlset[2]^{\pri}_{\ainter}$ for $\fmlset[2]^{\pri}_{\ainter\restriction \FV(\fmlset[2])}$, for short.}
    \label{figure: rules UNFO'}
\end{figure*}

Each rule above the dotted line has almost the same shape as the corresponding rule in \Cref{figure: rules UNFO}.
Here, the map $\chi \colon \allsymbols \to \pset{\AD^{(k)}}$ is introduced for expressing the unary predicates ``$\ADB{\ynmulstruc}{\bag}$'', ``$\ND{\ynmulstruc}{\bag}{\dir}$'', ``$\MD{\ynmulstruc}{\bag}{\dir}$'' in the \kl{2APTA} construction.
After we move "nodes" by applying the rule \eqref{rule: UNFO' mov}, this map is reset.
Then, we can only apply the rule \eqref{rule: UNFO' mov'}.
Using the rule \eqref{rule: UNFO' mov'}, we newly set the map $\chi$ so that $\chi$ expresses these unary predicates correctly.

Each rule below the dotted line is to define the map $\chi$ (in the rule \eqref{rule: UNFO' mov'}).
The states $(\bullet = \elemset)$ are used to assert that $\elemset = \bullet^{\ynmulstruc}_{\bag}$.
The states $\elem^{\pri}$ are used to assert that $\conc{\ynmulstruc}{\bag}(\elem) \neq \emptyset$ when $\pri = \prione$
and that $\conc{\ynmulstruc}{\bag}(\elem) = \emptyset$ when $\pri = \prizero$.
The rules for $\fixedvertex_i$ are trivially given.
In the rule for states $\dir \in \set{\back 2, \back 1, 1, 2}$,
we search whether there exists some "node" in the direction $\dir$ from the current position, by moving "nodes" nondeterministically.
The rules for states $(\ADBsymbol = \elemset)$, $(\NDsymbol_{\dir} = \elemset)$, and $(\MDsymbol_{\dir} = \elemset)$,
are induced from the definition of $\ADB{\ynmulstruc}{\bag}$, $\ND{\ynmulstruc}{\bag}{\dir}$, and $\MD{\ynmulstruc}{\bag}{\dir}$ (given in \Cref{section: abstract semantics on tree decompositions}), respectively.

From this, the \kl{2APTA} $\automaton[1]_{k}^{\fml}$ satisfies the following.
\begin{proposition}\label{proposition: 2APTA UNFO}
    Let $k \in \pnat$ and $\fml$ be a "UNFO" \kl{sentence}.
    For every $\allatstruc_{k}$-labeled \kl{binary} \kl{tree} $\ynmulstruc$,
    we have:
    \[\vdashUNFOsub_{\clUNFO(\fml)} \mkstate{(\fml)}{\prione}{\ynmulstruc}{\emptyset}{\eps} \quad\Longleftrightarrow\quad \ynmulstruc \in \automatonlang(\automaton[1]_{k}^{\fml}).\]
\end{proposition}
\begin{proof}[Proof Sketch]
    Let $\automatonlang(\automaton, \tup{q^{\pri}, \bag})$ be the set of $\allatstruc_{k}$-labeled \kl{binary} \kl{trees} $\ynmulstruc$
    such that 
    there is an \kl{accepting} \kl{run} of $\automaton$ starting from $\tup{q^{\pri}, \bag}$.
    By construction, we have:
    \begin{itemize}
        \item $\ynmulstruc \in \automatonlang(\automaton[1]_{k}^{\fml}, \tup{\elem^{\pri}, \bag})$
        iff 
        $\begin{cases}
        \elem \in \ADB{\ynmulstruc}{\bag} & (\mbox{if $\pri = \prione$}),\\   
        \elem \not\in \ADB{\ynmulstruc}{\bag} & (\mbox{if $\pri = \prizero$}).
        \end{cases}$

        \item $\ynmulstruc \in \automatonlang(\automaton[1]_{k}^{\fml}, \tup{(\ADBsymbol = \elemset)^{\prione}, \bag})$
        iff 
        $\ADB{\ynmulstruc}{\bag} = \elemset$.

        \item $\ynmulstruc \in \automatonlang(\automaton[1]_{k}^{\fml}, \tup{(\NDsymbol_{\dir} = \elemset)^{\prione}, \bag})$
        iff 
        $\ND{\ynmulstruc}{\bag}{\dir} = \elemset$.

        \item $\ynmulstruc \in \automatonlang(\automaton[1]_{k}^{\fml}, \tup{(\MDsymbol_{\dir} = \elemset)^{\prione}, \bag})$
        iff 
        $\MD{\ynmulstruc}{\bag}{\dir} = \elemset$.
    \end{itemize}
    From them with that the rules for ${\fmlset[1]}^{\pri}_{\ainter}$ in \Cref{figure: rules UNFO'} are the same as those in \Cref{figure: rules UNFO},
    we can convert an \kl{accepting} \kl{run} of $\vdashUNFOsub_{\clUNFO(\fml)} \mkstate{(\fml)}{\prione}{\ynmulstruc}{\emptyset}{\eps}$
    into an \kl{accepting} \kl{run} witnessing $\ynmulstruc \in \automatonlang(\automaton[1]_{k}^{\fml})$, and vice versa.
\end{proof}

\begin{proof}[Proof of \Cref{corollary: 2-EXPTIME UNFO}]
    For all "UNFO" \kl{sentences} $\fml$,
    we have:
    \begin{align*}
    &\text{$\fml$ is \kl{satisfiable}}\\
    &~\Longleftrightarrow~
    \text{$\fml$ is \kl{satisfiable} in a \kl{structure} of \kl{treewidth} at most $\fmllen{\fml}-1$} \tag{By \Cref{proposition: tw}}\\
    &~\Longleftrightarrow~
    \text{$\glue \ynmulstruc \modelsnonass \fml$ for some $\allatstruc_{\fmllen{\fml}}$-labeled \kl{binary} \kl{tree} $\ynmulstruc$} \tag{By \Cref{section: abstract semantics on tree decompositions}}\\
    &~\Longleftrightarrow~
    \text{$\absmodels \mkstate{(\fml)}{\prione}{\ynmulstruc}{\emptyset}{\eps}$ for some $\allatstruc_{\fmllen{\fml}}$-labeled \kl{binary} \kl{tree} $\ynmulstruc$} \tag{By \Cref{proposition: standard semantics and semantics on tree decompositions}}\\
    &~\Longleftrightarrow~
    \text{$\vdashUNFOsub_{\clUNFO(\fml)} \mkstate{(\fml)}{\prione}{\ynmulstruc}{\emptyset}{\eps}$  for some $\allatstruc_{\fmllen{\fml}}$-labeled \kl{binary} \kl{tree} $\ynmulstruc$}
    \tag{By \Cref{theorem: completeness closure UNFO}}\\
    &~\Longleftrightarrow~
    \text{$\automatonlang(\automaton[1]^{\fml}_{\fmllen{\fml}}) \neq \emptyset$}.
    \tag{By \Cref{proposition: 2APTA UNFO}}
    \end{align*}
    On the \kl(2APTA){size}, by restricting the set of \kl{relation names} $\Rels$ to those occurring in $\fml$,
    we have:
    \begin{itemize}
        \item the alphabet size $\card \allatstruc_{\fmllen{\fml}}$ is $2^{\mathcal{O}(\fmllen{\fml} \log \fmllen{\fml})}$,

        \item the number of \kl{states} is $2^{\poly(\fmllen{\fml})}$, and
        
        \item (from the two above) the transition size is $2^{\poly(\fmllen{\fml})}$.
    \end{itemize}
    Hence by \Cref{proposition: 2APTA complexity}, the \kl{satisfiability problem} for "UNFO" is in "2EXPTIME".
\end{proof}

\end{scope} %
\end{toappendix}
\begin{scope}\knowledgeimport{UNTC}\knowledgeimport{LMC}

\section{Local Checker for \texorpdfstring{"SAT-""UNTC"}{SAT-UNTC}}\label{section: UNTC}

In this section, we extend the "local checker@@UNFO" from "UNFO" to "UNTC".

\subsection{A \texorpdfstring{"Local Checker"}{Local Checker}}\label{section: local model checker UNTC}
\AP \phantomintro*"local checker"
Recall the \kl{state} set $\allstates_{\mathrm{\kl{UNTC}}}$ (\Cref{definition: semantics on tree decomposition}).
\AP
We now define the relation $(\intro*\leadstoUNTC) \subseteq \allstates_{\mathrm{\kl{UNTC}}} \times \PBallfml(\allstates_{\mathrm{\kl{UNTC}}})$
as the minimal binary relation closed under the ""rules"" in \Cref{figure: rules UNTC}.
\begin{figure}[t]
    \centering
\begin{tcolorbox}[standard jigsaw, opacityback = 0.8, colframe=black!80, boxrule=.3mm, left=.2mm, right=.2mm]
    \vspace{-4ex}
    \begin{align*}
    \dots&\;\leadstoUNTC\; \dots  \tag{all the \kl{rules} for \kl{UNFO} given in \Cref{figure: rules UNFO}}\\
    \mkstate{(\fmlset[1], \TC{\fml[2]}{u v} {x} {y})}{\pri}{\ynmulstruc}{\ainter}{\bag}
    &\;\leadstoUNTC\;
    \mkstate{(\fmlset[1], {x} \EQ {y})}{\pri}{\ynmulstruc}{\ainter}{\bag}
    \doublevee^{\pri}
    \mkstate{(\fmlset[1], \fml[2]\fmlsubst{u v}{{z} {z'}}, \TC{\fml[2]}{u v} {z'} {y})}{\pri}{\ynmulstruc}{\ainter\aintersubst{z}{\set{\vertex}}\aintersubst{z'}{\ADB{\ynmulstruc}{\bag}}}{\bag}\\
    & \hspace{.5em}\mbox{\small if $\ainter({x}) = \set{\vertex} \subseteq \univ{\ynmulstruc(\bag)}$,}  \label{rule: tc-l}\tag{TC-l}\\
    \mkstate{(\fmlset[1], \TC{\fml[2]}{u v} {x} {y})}{\pri}{\ynmulstruc}{\ainter}{\bag}
    &\;\leadstoUNTC\;
    \mkstate{(\fmlset[1], {x} \EQ {y})}{\pri}{\ynmulstruc}{\ainter}{\bag}
    \doublevee^{\pri} \mkstate{(\fmlset[1], \TC{\fml[2]}{u v} {x} {z}, \fml[2]\fmlsubst{u v}{z z'})}{\pri}{\ynmulstruc}{\ainter\aintersubst{z}{\ADB{\ynmulstruc}{\bag}}\aintersubst{z'}{\set{\vertex}}}{\bag} \\
    & \hspace{.5em}\mbox{\small if $\ainter({y}) = \set{\vertex} \subseteq \univ{\ynmulstruc(\bag)}$,} \label{rule: tc-r}\tag{TC-r}\\
    \mkstate{(\fmlset[1], \TC{\fml[2]}{u v} {x} {y})}{\pri}{\ynmulstruc}{\ainter}{\bag}
    &\;\leadstoUNTC\;
    \smashoperator{\sbigdoublevee_{\hspace{2em}\dir' \in \range{\back 2}{2} \setminus \set{\dir}}^{\pri}}
    \mkstate{(\fmlset[1], \TC{\fml[2]}{u v} {x} {z}, \fml[2]\fmlsubst{u v}{z z'}, \TC{\fml[2]}{u v} {z}' {y})}{\pri}{\ynmulstruc}{\ainter\aintersubst{z}{\ND{\ynmulstruc}{\bag}{\dir}}\aintersubst{z'}{\ND{\ynmulstruc}{\bag}{\dir'}}}{\bag}\\
    & \hspace{.5em}\mbox{\small if $\ainter({x}) = \set{\dir} \subseteq \set{\back 2, \back 1, 1, 2}$,
    $\ainter({y}) = \set{\elem}\not\subseteq \ND{\ynmulstruc}{\bag}{\dir}$,} \label{rule: tc-split-l}\tag{TC-sp-l}\\
    \mkstate{(\fmlset[1], \TC{\fml[2]}{u v} {x} {y})}{\pri}{\ynmulstruc}{\ainter}{\bag}
    &\;\leadstoUNTC\;
    \smashoperator{\sbigdoublevee_{\hspace{2em}\dir' \in \range{\back 2}{2} \setminus \set{\dir}}^{\pri}}
    \mkstate{(\fmlset[1], \TC{\fml[2]}{u v} {x} {z}, \fml[2]\fmlsubst{u v}{z z'}, \TC{\fml[2]}{u v} {z}' {y})}{\pri}{\ynmulstruc}{\ainter\aintersubst{z}{\ND{\ynmulstruc}{\bag}{\dir'}}\aintersubst{z'}{\ND{\ynmulstruc}{\bag}{\dir}}}{\bag}\\
    & \hspace{.5em}\mbox{\small if $\ainter({y}) = \set{\dir} \subseteq \set{\back 2, \back 1, 1, 2}$, 
    $\ainter({x}) = \set{\elem} \not\subseteq \ND{\ynmulstruc}{\bag}{\dir}$.} \label{rule: tc-split-r}\tag{TC-sp-r}
    \end{align*}
\end{tcolorbox}
     \caption{\reintro*\kl{Rules} of the \kl(UNTC){LC} for \kl{UNTC}. Here, $z$ and $z'$ are used as \kl{fresh} \kl{variables}.}
    \label{figure: rules UNTC}
\end{figure}
We lift this relation to $(\leadstoUNTC) \subseteq \PBallfml(\allstates_{\mathrm{\kl{UNTC}}}) \times \PBallfml(\allstates_{\mathrm{\kl{UNTC}}})$ following \Cref{section: local model checker UNFO}.
All notions of \kl(LMC){runs}, \kl(LMC){accepting} \kl(LMC){runs},
$\intro*\vdashUNTC \mkstate{\fmlset[1]}{\pri}{\ynmulstruc}{\ainter}{\bag} \phantomintro\nvdashUNTC$,
and $\intro*\vdashUNTCsub_{\fmlsetclass} \mkstate{\fmlset[1]}{\pri}{\ynmulstruc}{\ainter}{\bag} \phantomintro\nvdashUNTCsub$
are defined exactly as before, except that there are some extra "rules" for "UNTC", which we discuss next.
The four "TC" "rules" above are introduced for obtaining a closure property (\Cref{example: completeness UNTC,section: closure property UNTC}).

\proofcasethin{"Rules" \eqref{rule: tc-l} and \eqref{rule: tc-r}}
These "rules" are obtained by left-unfolding and right-unfolding a "TC-formula" in a standard way.

\proofcasethin{"Rule" \eqref{rule: tc-split-l}}
This "rule" is also obtained by unfolding. 
We confirm the property ``$\absmodels \mystate[1]$ "iff" $\absmodels \doublefml[2]$'' (where $\mystate[1] \leadstoUNTC \doublefml[2]$)
for odd $\pri$ (the even case is similar).
\proofcasethin{($\Leftarrow$)} Trivial.
\proofcasethin{($\Rightarrow$)}
There are some $n \ge 1$ and $\dir_1, \dots, \dir_{n-1} \in \range{\back 2}{2}$ "st" $\absmodels \mkstate{(\fmlset[1], \bigwedge_{i = 1}^{n}\fml[2]\fmlsubst{u v}{z_{i-1} z_{i}})}{\pri}{\ynmulstruc}{\ainter\aintersubst{z_1 \dots z_{n-1}}{\ND{\ynmulstruc}{\bag}{\dir_{1}}, \dots, \ND{\ynmulstruc}{\bag}{\dir_{n-1}}}}{\bag}$, where $z_0 = x$, $z_n = y$, and $z_1, \dots, z_{n-1}$ are "fresh".
Let $\dir_0 = \dir$ and let $\dir_n$ be such that $\elem \in \ND{\ynmulstruc}{\bag}{\dir_n}$.
By $\dir_0 \neq \dir_n$, there is some $m \in \1{n}$ such that $\dir_{m-1} = \dir$ and $\dir_m \neq \dir$.
By letting $\dir' = \dir_m$, we can show it.

\proofcasethin{"Rule" \eqref{rule: tc-split-r}}
Similar to the above.

The "LC@@UNTC" is sound and complete "wrt" the semantics on \kl{tree decompositions} (\Cref{theorem: completeness closure UNTC}).
Below, we give a toy example.
\begin{example}\label{example: completeness UNTC}
    Let $\ynmulstruc$ be the \kl(LMC){tree} in which
    $\ynmulstruc(1)$, $\ynmulstruc(\eps)$, and $\ynmulstruc(2)$ are given as follows, and $\ynmulstruc(\bag)$ is undefined for all other $\bag$:
    \[\left(\begin{tikzpicture}[baseline = -5.ex]
        \graph[grow right = 1.cm, branch down = 4.5ex]{
        {/, s5/{$\fixedvertex_5$}[vert], s4/{$\fixedvertex_4$}[vert]} -!-
        {s0/{$\fixedvertex_0$}[vert], s6/{$\fixedvertex_6$}[vert], s3/{$\fixedvertex_3$}[vert]}
        };
        \graph[use existing nodes, edges={color=black, pos = .5, earrow}, edge quotes={fill=white, inner sep=1pt,font= \scriptsize}]{
            s5 ->["$\rsym[2]$", pos = .4] s6;
            s5 ->["$\rsym[1]$", out = 150, in = -150, looseness = 4] s5;
            s4 ->["$\rsym[2]$", pos = .35, bend right = 60] s5;
            s4 ->["$\rsym[1]$", out = 150, in = -150, looseness = 4] s4;
            s3 ->["$\rsym[2]$", pos = .4] s4;
        };
    \end{tikzpicture} \right),
    \left(\;\begin{tikzpicture}[baseline = -5.ex]
        \graph[grow right = 1.cm, branch down = 4.5ex]{
        {s0/{$\fixedvertex_0$}[vert], s6/{$\fixedvertex_6$}[vert], s3/{$\fixedvertex_3$}[vert]}
        };
        \graph[use existing nodes, edges={color=black, pos = .5, earrow}, edge quotes={fill=white, inner sep=1pt,font= \scriptsize}]{
        };
    \end{tikzpicture}\; \right), \mbox{ and }
    \left(\begin{tikzpicture}[baseline = -5.ex]
        \graph[grow right = 1.cm, branch down = 4.5ex]{
        {s0/{$\fixedvertex_0$}[vert], s6/{$\fixedvertex_6$}[vert], s3/{$\fixedvertex_3$}[vert]} -!-
        {s1/{$\fixedvertex_1$}[vert], /, s2/{$\fixedvertex_2$}[vert]}
        };
        \graph[use existing nodes, edges={color=black, pos = .5, earrow}, edge quotes={fill=white, inner sep=1pt,font= \scriptsize}]{
            s0 ->["$\rsym[1]$", pos = .4] s1;
            s1 ->["$\rsym[2]$", out = 30, in = -30, looseness = 4] s1;
            s1 ->["$\rsym[1]$", pos = .4] s2;
            s2 ->["$\rsym[2]$", out = 30, in = -30, looseness = 4] s2;
            s2 ->["$\rsym[1]$", pos = .4] s3;
        };
    \end{tikzpicture} \right),\]
    Let $\fml$ be the \kl{UNTC} \kl{formula} $\TC{\exists z(\rsym[1] x z \land \rsym[2] z y)}{x y} x y$ and let $\fml[2] \defeq \exists z(\rsym[1] x z \land \rsym[2] z y)$.
    Intuitively, $\fml$ means that there is a $(\rsym[1] \rsym[2])^*$-path from $x$ to $y$.
    Let $\ainter$ be such that $\ainter(x) = \set{\fixedvertex_0}$ and $\ainter(y) = \set{\fixedvertex_6}$.
    Then $\absmodels \mkstate{(\fml)}{\prione}{\ynmulstruc}{\ainter}{\eps}$ holds, by the form of $\glue \ynmulstruc$.
    In the "LC@@UNTC",
    we construct an \kl(LMC){accepting} \kl(LMC){run}, as follows.
    First, we have:
    \begin{align*}
    \mkstate{(\fml)}{\prione}{\ynmulstruc}{\ainter}{\eps}
    & \leadstoUNTC_{\eqref{rule: tc-l}\eqref{rule: concrete}}^{*}\lgeqq
        \mkstate{(\fml[2]\fmlsubst{xy}{z_0^{\fixedvertex_0}z_1^{2}}, \TC{\fml[2]}{x y} z_1^{2} y^{\fixedvertex_6})}{1}{ \ynmulstruc}{\bl}{\eps}\\
    & \leadstoUNTC_{\eqref{rule: tc-r}\eqref{rule: concrete}}^{*}\lgeqq
    \mkstate{\fml[2]\fmlsubst{xy}{z_0^{\fixedvertex_0}z_1^{2}}, \TC{\fml[2]}{x y} z_1^{2} z_5^{1}, \fml[2]\fmlsubst{xy}{z_5^{1} z_6^{\fixedvertex_6}}}{\prione}{\ynmulstruc}{\bl}{\eps}.
    \end{align*}
    We then split the "formula" $\TC{\fml[2]}{x y} z_1 z_5$ as follows:
    \begin{tcolorbox}[ams align*, standard jigsaw, opacityback = 0, boxrule=0pt, boxsep=0pt, top=0pt, bottom=0pt, left=.5mm, right=0pt, frame hidden, enhanced, sharpish corners, borderline west = {.4mm}{0mm}{black}, left skip=1.7em,
        overlay={
            \node[xshift = -1em] at (frame.west) {\labeltext{example: completeness UNTC split}{$\bigstar$}{$\bigstar$}};}
        ]
        &\mkstate{(\fml[2]\fmlsubst{xy}{z_0^{\fixedvertex_0}z_1^{2}},\  \uline{\TC{\fml[2]}{x y} z_1^{2} z_5^{1}},\  \fml[2]\fmlsubst{xy}{z_5^{1} z_6^{\fixedvertex_6}})}{\prione}{\ynmulstruc}{\bl}{\eps} \span \\[0.3ex]
        &\leadstoUNTC_{\eqref{rule: tc-split-l}\eqref{rule: concrete}}^{*} \lgeqq
        (\fml[2]\fmlsubst{xy}{z_0^{\fixedvertex_0}z_1^{2}}, \TC{\fml[2]}{x y} z_1^{2} z_2^{2},\  \uline{\fml[2]\fmlsubst{xy}{z_2^{2} z_4^{1}}}, \  \TC{\fml[2]}{x y} z_4^{1} z_5^{1}, \fml[2]\fmlsubst{xy}{z_5^{1} z_6^{\fixedvertex_6}})^{1,\ynmulstruc}_{-,\eps}\\[0.3ex]
        &\leadstoUNTC_{\eqref{rule: and}\eqref{rule: exists}\eqref{rule: concrete}}^{*}\lgeqq
        (\fml[2]\fmlsubst{xy}{z_0^{\fixedvertex_0}z_1^{2}}, \TC{\fml[2]}{x y} z_1^{2} z_2^{2},\  \rsym[1] z_2^{2} z_3^{\fixedvertex_3},\  \rsym[2] z_3^{\fixedvertex_3} z_4^{1}, \  \TC{\fml[2]}{x y} z_4^{1} z_5^{1}, \fml[2]\fmlsubst{xy}{z_5^{1} z_6^{\fixedvertex_6}})^{1, \ynmulstruc}_{-, \eps}\\[0.3ex]
        &\leadstoUNTC_{\eqref{rule: split}}
        \mkstate{(\fml[2]\fmlsubst{xy}{z_0^{\fixedvertex_0}z_1^{2}}, \TC{\fml[2]}{x y} z_1^{2} z_2^{2}, \rsym[1] z_2^{2} z_3^{\fixedvertex_3})}{\prione}{\ynmulstruc}{\bl}{\eps} ~\doublewedge~ \mkstate{(\rsym[2] z_3^{\fixedvertex_3} z_4^{1}, \TC{\fml[2]}{x y} z_4^{1} z_5^{1}, \fml[2]\fmlsubst{xy}{z_5^{1} z_6^{\fixedvertex_6}})}{\prione}{\ynmulstruc}{\bl}{\eps}. \tag*{}
    \end{tcolorbox}
    Here, in the step using \eqref{rule: tc-split-l}, the assignment $z_4 \mapsto \set{1}$ is guessed,
    which indicates a "fixed name" that escapes from "direction" $2$ during the iterations of "TC".
    
    \proofcasethin{Left-hand side}
    We have:
    \begin{align*}
        \mkstate{(\fml[2]\fmlsubst{xy}{z_0^{\fixedvertex_0}z_1^{2}}, \TC{\fml[2]}{x y} z_1^{2} z_2^{2}, \rsym[1] z_2^{2} z_3^{\fixedvertex_3})}{\prione}{\ynmulstruc}{\bl}{\eps}
        &\leadstoUNTC_{\eqref{rule: mov}\eqref{rule: concrete}} \lgeqq
        \mkstate{(\fml[2]\fmlsubst{xy}{z_0^{\fixedvertex_0}z_1^{\fixedvertex_1}}, \TC{\fml[2]}{x y} z_1^{\fixedvertex_1} z_2^{\fixedvertex_2}, \rsym[1] z_2^{\fixedvertex_2} z_3^{\fixedvertex_3})}{\prione}{\ynmulstruc}{\bl}{2} \\
        &\leadstoUNTC_{\eqref{rule: split}\eqref{rule: exists}\eqref{rule: concrete}\eqref{rule: and}\eqref{rule: a1}}^{*}\lgeqq
        \mkstate{(\TC{\fml[2]}{x y} z_1^{\fixedvertex_1} z_2^{\fixedvertex_2})}{\prione}{\ynmulstruc}{\bl}{2} \\
        &\leadstoUNTC_{\eqref{rule: tc-l}\eqref{rule: tc-l}\eqref{rule: concrete}}^{*}\lgeqq
        \mkstate{(\fml[2][z_1^{\fixedvertex_1} z^{\fixedvertex_2}/x y], z^{\fixedvertex_2} \EQ z_2^{\fixedvertex_2})}{1} {\ynmulstruc}{\bl}{2}\\
        &\leadstoUNTC^{*}\lgeqq \const{true}.
    \end{align*}

    \proofcasethin{Right-hand side}
    Similarly, we have:
    \begin{align*}
        \mkstate{(\rsym[2] z_3^{\fixedvertex_3} z_4^{1}, \TC{\fml[2]}{x y} z_4^{1} z_5^{1}, \fml[2]\fmlsubst{xy}{z_5^{1}z_6^{\fixedvertex_6}})}{\prione}{\ynmulstruc}{\bl}{\eps}
        &\leadstoUNTC_{\eqref{rule: mov}\eqref{rule: concrete}}^{*} \mkstate{(\rsym[2] z_3^{\fixedvertex_3} z_4^{\fixedvertex_4}, \TC{\fml[2]}{x y} z_4^{\fixedvertex_4} z_5^{\fixedvertex_5}, \fml[2]\fmlsubst{xy}{z_5^{\fixedvertex_5}z_6^{\fixedvertex_6}})}{\prione}{\ynmulstruc}{\bl}{1}\\
        &\leadstoUNTC^{*}\lgeqq \const{true}.
    \end{align*}
    Hence, we have $\vdashUNTC \mkstate{(\fml)}{\prione}{\ynmulstruc}{\ainter}{\eps}$.
    \lipicsEnd
\end{example}

\subsection{Evaluation Strategy} \label{section: evaluation strategy UNTC}
We now present a strategy to evaluate $\absmodels \mkstate{\fmlset[1]}{\pri}{\ynmulstruc}{\ainter}{\bag}$ in the "LC@@UNTC".
Step \nameref{evaluation strategy UNTC 2} is crucial for the "2EXPTIME" upper bound.

\proofcasethin{Step 1}\labeltext{evaluation strategy UNTC 1}{}{1}
By the same strategy as Steps \nameref{evaluation strategy UNFO 1} and \nameref{evaluation strategy UNFO 2} for "UNFO" (\Cref{section: evaluation strategy UNFO}),
we can assume that
\begin{enumerate}
    \item[a)] \labeltext{cond: evaluation strategy UNTC 1}{}{a} $\card \ainter(x) = 1$ for each $x \in \FV(\fmlset[1])$.

    \item[b)] \labeltext{cond: evaluation strategy UNTC 2}{}{b} $\fmlset[1]$ is of the form $(\fml[2]_1, \dots, \fml[2]_n)$,
where each $\fml[2]_i$ is one of the following forms:
    an "atom" $\afml$, $\lnot \fml[3]$, or $\TC{\fml[3]}{u v} {x} {y}$.

    \item[c)] \labeltext{cond: evaluation strategy UNTC 3}{}{c} for each $\fml[2]_i$, for a $\dir \in \set{\back 2, \back 1, 1, 2}$, 
    $\dir \in \ainter(\FV(\fml[2]_i)) \subseteq \ND{\ynmulstruc}{\bag}{\dir}$,
        except for the case when $\fml[2]_i$ is a \kl{TC-formula}.
\end{enumerate}

\proofcasethin{Step 2}\labeltext{evaluation strategy UNTC 2}{}{2}
    When $\fml[2]_i$ is a \kl{TC-formula} $\TC{\fml[3]}{u v} {x} {y}$,
    we divide $\fmlset[1] = (\fmlset[1]', \fmlset[1]'')$ into the following four cases,
    according to whether $\fml[2]_i$ is left- and/or right-unfolded:
    \begin{enumerate}[i)]
        \item \labeltext{state: biclosure UNTC}{}{i} $\fmlset[1]'' = (\phantom{\fmlset[2], }\, \TC{\fml[3]}{u v} {x} {y})$
        (observe that $\FV(\fmlset[1]') \cap \FV(\fmlset[1]'') \subseteq \set{x, y}$),
        \item \labeltext{state: closure UNTC l}{}{ii} $\fmlset[1]'' = (\fmlset[2], \TC{\fml[3]}{u v} {x} {y})$
        where $\FV(\fmlset[1]') \cap \FV(\fmlset[1]'') \subseteq \set{y}$,
        \item \labeltext{state: closure UNTC r}{}{iii} $\fmlset[1]'' = (\phantom{\fmlset[2], }\, \TC{\fml[3]}{u v} {x} {y}, \fmlset[3])$
        where $\FV(\fmlset[1]') \cap \FV(\fmlset[1]'') \subseteq \set{x}$,
        \item \labeltext{state: closure UNTC}{}{iv}  $\fmlset[1]'' = (\fmlset[2], \TC{\fml[3]}{u v} {x} {y}, \fmlset[3])$
        where $\FV(\fmlset[1]') \cap \FV(\fmlset[1]'') = \emptyset$,
    \end{enumerate}
    where $\fmlset[2]$ and $\fmlset[3]$ are \kl{formula} sets obtained by left-unfolding and right-unfolding some $\TC{\fml[3]}{u v}{\bl}{\bl}$, respectively.
    Note that $\FV(\fmlset[2]) \cap \FV(\TC{\fml[3]}{u v} {x} {y}) \subseteq \set{x}$ and $\FV(\TC{\fml[3]}{u v} {x} {y}) \cap \FV(\fmlset[3]) \subseteq \set{y}$ also hold, as 
    free \kl{variables} are always replaced with \kl{fresh} \kl{variables} in unfolded \kl{formulas} in the "rules" for "TC".

    If $\ainter(x) \subseteq \ND{\ynmulstruc}{\bag}{0}$,
    we apply the "rules" \eqref{rule: split} and \eqref{rule: tc-l}.
    For instance, for the case \nameref{state: closure UNTC l}), we have:
    \begin{align*}
    \mkstate{(\underbrace{\fmlset[1]',\  \fmlset[2], \TC{\fml[3]}{u v} {x} {y}}_{\text{Case \nameref{state: closure UNTC l})}})}{\prione}{\ynmulstruc}{\ainter}{\bag}
    &~\leadstoUNTC_{\eqref{rule: split}}~ \mkstate{(\fmlset[2])}{\prione}{\ynmulstruc}{\bl}{\bag} \doublewedge
    \mkstate{(\fmlset[1]',\  \TC{\fml[3]}{u v} {x} {y})}{\prione}{\ynmulstruc}{\bl}{\bag} \\
    &~\leadstoUNTC_{\eqref{rule: tc-l}}~ \mkstate{(\fmlset[2])}{\prione}{\ynmulstruc}{\bl}{\bag} \doublewedge
    \mkstate{(\underbrace{\fmlset[1]',\ \fml[3]\fmlsubst{u v}{xz^{\ADB{\ynmulstruc}{\bag}}}, \TC{\fml[3]}{u v} {z} {y}}_{\text{Case \nameref{state: closure UNTC l})}})}{\prione}{ \ynmulstruc}{\bl}{\bag}.
    \end{align*}
    The transformed \kl{formula} set is still in the case \nameref{state: closure UNTC l}).

    If $\ainter(y) \subseteq \ND{\ynmulstruc}{\bag}{0}$,
    we apply the "rules" \eqref{rule: split} and \eqref{rule: tc-r}, similarly.
    After that, we can assume that 
    \begin{itemize}
        \item for some $\dir, \dir' \in \set{\back 2, \back 1, 1, 2}$, $\ainter({x}) = \set{\dir}$ and $\ainter({y}) = \set{\dir'}$.
    \end{itemize}

    When $\dir \neq \dir'$, we apply the "rule" \eqref{rule: tc-split-l} or \eqref{rule: tc-split-r}.
    For instance, for the case \nameref{state: closure UNTC l}), we have:
     \begin{tcolorbox}[ams align*, standard jigsaw, opacityback = 0, boxrule=0pt, boxsep=0pt, top=0pt, bottom=0pt, left=0pt, right=0pt, frame hidden, sharp corners, enhanced, borderline west = {.4mm}{0mm}{black}]
    & \ \mkstate{(\underbrace{\fmlset[1]',\; \fmlset[2], \TC{\fml[3]}{u v} {x} {y}}_{\text{Case \nameref{state: closure UNTC l})}})}{\prione}{\ynmulstruc}{\ainter}{\bag}
    ~\leadstoUNTC_{\eqref{rule: tc-split-l}}\lgeqq~ \mkstate{(\fmlset[1]',\; \fmlset[2], \TC{\fml[3]}{u v} {x} {z},\, \uline{\fml[3]\fmlsubst{u v}{zz'}} ,\, \TC{\fml[3]}{u v} {z'} {y})}{\prione}{\ynmulstruc}{\ainter}{\bag}.
    \end{tcolorbox}
    After applying this argument recursively to the \kl{formula} $\fml[3]\fmlsubst{u v}{zz'}$,
    we eventually reach two sets $\fmlset[2]'$ and $\fmlset[3]'$ as follows:
    \begin{tcolorbox}[ams align*, standard jigsaw, opacityback = 0, boxrule=0pt, boxsep=0pt, top=0pt, bottom=0pt, left=0pt, right=0pt, frame hidden, sharp corners, enhanced, borderline west = {.4mm}{0mm}{black}]
    \dots & ~\leadstoUNTC^{*}\lgeqq~ \mkstate{(\fmlset[1]',\ \fmlset[2], \TC{\fml[3]}{u v} {x} {z},\fmlset[2]', \; \fmlset[3]', \TC{\fml[3]}{u v} {z'} {y})}{\prione}{\ynmulstruc}{\ainter}{\bag},
    \end{tcolorbox}
    \noindent where $\fmlset[2]'$ is obtained by left-unfolding some $\TC{\fml[3]}{u v} {x} {\bl}$ and
    $\fmlset[3]'$ is obtained by right-unfolding some $\TC{\fml[3]}{u v} {\bl} {y}$, "cf", the lines (\nameref{example: completeness UNTC split}) in \Cref{example: completeness UNTC};
    the "formula" $\TC{\fml[2]}{x y} z_1^{2} z_5^{1}$
    is split to $\fmlset[2]' = (\TC{\fml[2]}{x y} z_1^{2} z_2^{2}, \rsym[1] z_2^{2} z_3^{\fixedvertex_3})$
    and $\fmlset[3]' = (\rsym[2] z_3^{\fixedvertex_3} z_4^{1}, \TC{\fml[2]}{x y} z_4^{1} z_5^{1})$.
    We can then apply \eqref{rule: split}, as follows:
     \begin{tcolorbox}[ams align*, standard jigsaw, opacityback = 0, boxrule=0pt, boxsep=0pt, top=0pt, bottom=0pt, left=0pt, right=0pt, frame hidden, sharp corners, enhanced, borderline west = {.4mm}{0mm}{black}]
    \dots & ~\leadstoUNTC_{\eqref{rule: split}}~
    \mkstate{(\underbrace{\fmlset[2], \TC{\fml[3]}{u v} {x} {z}, \fmlset[2]'}_{\text{Case \nameref{state: closure UNTC})}})}{\prione}{\ynmulstruc}{\ainter}{\bag} \doublewedge
    \mkstate{(\underbrace{\fmlset[1]', \fmlset[3]', \TC{\fml[3]}{u v} {z'} {y}}_{\text{Case \nameref{state: closure UNTC l})}})}{\prione}{\ynmulstruc}{\ainter}{\bag}.
    \end{tcolorbox}
    Both the transformed \kl{formula} sets are still in the cases \nameref{state: closure UNTC}) and \nameref{state: closure UNTC l}), respectively.
    After the step, we can assume that
    \begin{enumerate}
        \item[c')] \labeltext{cond: evaluation strategy UNTC 3'}{}{c'} for each $\fml[2]_i$, for a $\dir \in \set{\back 2, \back 1, 1, 2}$, $\dir \in \ainter(\FV(\fml[2]_i)) \subseteq \ND{\ynmulstruc}{\bag}{\dir}$.
    \end{enumerate}

\proofcasethin{Step 3}\labeltext{evaluation strategy UNTC 3}{}{3}
Finally, by the condition \nameref{cond: evaluation strategy UNTC 3'}),
we apply \eqref{rule: split}\eqref{rule: mov} in the same way as Steps \nameref{evaluation strategy UNFO 3} and \nameref{evaluation strategy UNFO 4} for "UNFO", and then we go back to Step \nameref{evaluation strategy UNTC 1}.

\subsection{Closure Property}\label{section: closure property UNTC}
Based on the evaluation strategy above,
we can obtain a completeness with a closure property (\Cref{theorem: completeness closure UNTC}).
The following \kl{closure} is inspired by \kl[closure]{that} for \intro*\kl{derivatives} in \intro*\kl{regular expressions} \cite{brzozowskiDerivativesRegularExpressions1964,antimirovPartialDerivativesRegular1996} and
in the \kl{positive calculus of relations with transitive closure} \cite{nakamuraPartialDerivativesGraphs2017,nakamuraDerivativesGraphsPositive2025}
and by \intro*\kl{Fischer--Ladner closure} in \kl{PDL} \cite{fischerPropositionalDynamicLogic1979}.
The main difference from them is that our \kl{closure} has \emph{two} unfoldings:
left-unfolding and right-unfolding, based on the "rules" for "TC".
\begin{definition}\label{definition: closure UNTC}
    \AP For a \kl{UNTC} \kl{formula} $\fml$,
    the \intro*\kl{uni-closure} $\intro*\preclUNTC(\fml)$ is the \kl{set} of \kl{UNTC} \kl{formula} \kl{sets} defined as follows:
    \begin{align*}
        &\preclUNTC(\TC{\fml}{u v} {x} {y}) ~\defeq~ 
        \smashoperator{\bigcup_{\substack{\text{$z_1$ is $x$ or "fresh"}\\\text{$z_2$ is $y$ or "fresh"}}}} \preclUNTC(z_1 \EQ z_2)
        \cup  \smashoperator{\bigcup_{\substack{\text{$z_1, z_2$ are "fresh"}}}} \preclUNTC(\fml[1]\fmlsubst{u v}{z_1z_2}) \quad \cup \quad \set{(\TC{\fml}{u v} {x} {y})} \tag*{i)} \label{definition: closure UNTC not unfolded}\\
        &\cup \set*{
        \begin{aligned}
        &(\fmlset[2], \TC{\fml[1]}{u v}  z_2 z_3) \mid\; \text{$\fmlset[2] \in \preclUNTC(\fml[1]\fmlsubst{u v}{z_1z_2})$, where $z_1, z_2$ are "fresh",}\\
        & \text{$z_3$ is $y$ or "fresh", and $\FV(\fmlset[2]) \cap \set{z_2, z_3} \subseteq \set{z_2}$}
        \end{aligned}
        } \tag*{ii)}\label{definition: closure UNTC l}\\
        &\cup \set*{
        \begin{aligned}
        &(\TC{\fml[1]}{u v} z_1 z_2, \fmlset[3]) \mid\; \text{$\fmlset[3] \in \preclUNTC(\fml[1]\fmlsubst{u v}{z_2 z_3})$, where $z_2, z_3$ are "fresh",}\\
        & \text{ $z_1$ is $x$ or "fresh", and $\set{z_1,z_2} \cap \FV(\fmlset[3]) \subseteq \set{z_2}$}
        \end{aligned}
        } \tag*{iii)}\label{definition: closure UNTC r}
    \end{align*}
    The other definitions are given based on the $\clUNFO$ of \Cref{definition: closure UNFO}. \lipicsEnd
\end{definition}
\begin{definition}\label{definition: biclosure UNTC}
    \AP For a \kl{UNTC} \kl{formula} $\fml$,
    the \intro*\kl{bi-closure} $\intro*\prebiclUNTC(\fml)$ is defined as follows:
    \begin{align*}
        &\prebiclUNTC(\fml) \defeq \preclUNTC(\fml) \cup {}\\
        &\hspace{-.7em}\set*{
        \begin{aligned}
        & (\fmlset[2], \TC{\fml[2]}{u v} z_2 z_3, \fmlset[3])\mid\; \text{$\TC{\fml[2]}{u v} x y$ is a subformula of $\fml$, $\fmlset[2] \in \preclUNTC(\fml[2]\fmlsubst{u v}{z_1 z_2})$, and}\\
        &\text{$\fmlset[3] \in \preclUNTC(\fml[2]\fmlsubst{u v}{z_3 z_4})$, where $z_1, z_2, z_3, z_4$ are "fresh", $\FV(\fmlset[2]) \cap \set{z_2, z_3} \subseteq \set{z_2}$,}\\
        &\text{$\set{z_2, z_3} \cap \FV(\fmlset[3]) \subseteq \set{z_3}$, and $\FV(\fmlset[2]) \cap \FV(\fmlset[3]) = \emptyset$}
        \end{aligned}} \tag*{iv)} \label{definition: closure UNTC bi}
    \end{align*}
\end{definition}

We recall the case \nameref{state: closure UNTC}) in Step \nameref{evaluation strategy UNTC 2} of the evaluation strategy presented in \Cref{section: local model checker UNTC}.
In this case, by applying \eqref{rule: split}, we can split to $\fmlset[1]'$ and $(\fmlset[2], \TC{\fml[3]}{u v} {x} {y}, \fmlset[3])$,
and thus we can assume $\fmlset[1]' = \emptyset$.
Namely, it suffices to consider bi-unfolded \kl{formula} sets at the \emph{outermost}.
The definition of $\prebiclUNTC$ (\Cref{definition: biclosure UNTC}) is based on this fact.
We observe that this extension causes only a quadratic, not an exponential blowup,
which is crucial to obtain the "2EXPTIME" upper bound.\footnote{%
It causes an exponential blowup if $\preclUNTC(\TC{\fml}{u v} {x} {y})$ is defined so that it contains all bi-unfolded \kl{formulas} $(\fmlset[2], \TC{\fml[1]}{u v} {z}{z'}, \fmlset[3])$ (and if the nesting of \kl{TC-formulas} is unbounded).}
Also, note that \emph{two} \kl{formula} sets in $\prebiclUNTC(\fml)$ may simultaneously appear between \eqref{rule: tc-split-l} and \eqref{rule: split} ("cf", the lines of \Cref{example: completeness UNTC} (\nameref{example: completeness UNTC split})).
Based on this, we define the \intro*\kl{closure} set $\clUNTC(\fml)$, as follows:
\[\intro*\clUNTC(\fml) ~\defeq~ \set{\fmlset[1] \cup \fmlset[2] \mid \fmlset[1], \fmlset[2] \in \prebiclUNTC(\fml)}.\]
By the strategy presented in \Cref{section: local model checker UNTC},
we have that the "LC@@UNTC" is sound and complete "wrt" the semantics on \kl{tree decompositions} under a \kl{closure} property, as follows.
\ifthenelse{\boolean{arXiv}}{\begin{theorem}[{\Cref{section: theorem: completeness UNTC}}]}{\begin{theorem}}\label{theorem: completeness closure UNTC}
    \gdef\completenessclosureUNTC{%
    Let $\fml$ be a \kl{UNTC} \kl{formula}.
    For all $\mystate \in \allstates_{\prebiclUNTC(\fml)}$,
    we have:
    \[{} \absmodels \mystate \quad\Longleftrightarrow\quad \vdashUNTCsub_{\clUNTC(\fml)} \mystate.\]
    }%
    \completenessclosureUNTC
\end{theorem}

For the size of the \kl{closure} set, we have the following.
\begin{proposition}\label{proposition: closure size UNTC}
    For all \kl{UNTC} \kl{formulas} $\fml$,
    the \kl{cardinality} of $\preclUNTC(\fml)$, up to renaming \kl{free variables}, is at most $(2 \fmllen{\fml})^{2 \fmllen{\fml}}$.
\end{proposition}
\begin{proof}
    By easy induction on $\fml$%
    \ifthenelse{\boolean{arXiv}}{ (\Cref{section: proposition: closure size UNTC}).}{.}
\end{proof}

\begin{proposition}\label{proposition: closure size UNTC 2}
    For \kl{UNTC} \kl{formulas} $\fml$,
    the number of $\mkstate{\fmlset[1]}{\pri}{\ynmulstruc}{\ainter}{\bag} \in \allstates_{\clUNTC(\fml)}^{(k)}$
    is $2^{\mathcal{O}(\fmllen{\fml} (\log \fmllen{\fml} + k))}$,
    up to forgetting $\ainter(x)$ for $x \not\in \FV(\fmlset[1])$, $\ynmulstruc$, and $\bag$ ("ie", ``$\fmlset[1]^{\pri}_{\ainter \restriction \FV(\fmlset[1])}$'')
    and up to renaming \kl{free variables}.
\end{proposition}
\begin{proof}
    Similar to \Cref{proposition: closure UNFO 2}, using \Cref{proposition: closure size UNTC}.
\end{proof}

\subsection{Reduction to \kl{2APTAs}}\label{section: reducing to 2APTAs UNTC}
Similar to \Cref{section: reducing to 2APTAs UNFO}, using the "LC@@UNTC",
we can give an exponential-time reduction from "SAT-""UNTC" to the \kl{non-emptiness problem} for \kl{2APTAs} using the "LC@@UNTC" for \kl{UNTC}%
\ifthenelse{\boolean{arXiv}}{  (see \Cref{section: 2APTA construction UNTC} for a precise construction).}{.}
Using \Cref{proposition: closure size UNTC 2}, we can show that the \kl(2APTA){size} of the \kl{2APTA} is $2^{\poly(\|\fml\|, k)}$.
Thus, in the same vein as \Cref{section: reducing to 2APTAs UNFO}, we have the following result.
\begin{corollary}[%
\textnormal{"cf" \cite[Theorem 8.5]{figueiraCommonAncestorPDL2025}}]\label{corollary: 2-EXPTIME UNTC}
    "SAT-""UNTC" is in "2EXPTIME".
\end{corollary}
\noindent Combining this with \Cref{thm:GNTC-UNTC} yields the following main result.
\begin{corollary}\label{corollary: 2-EXPTIME GNTC}
    "SAT-""GNTC" is in "2EXPTIME".
\end{corollary}

\end{scope}

\begin{toappendix}
\begin{scope}\knowledgeimport{UNTC}\knowledgeimport{LMC}

\subsection{Proof of {\Cref{theorem: completeness closure UNTC}}:
Soundness and Completeness}\label{section: theorem: completeness UNTC}
We recall the \kl(UNTC){LC} for \kl{UNTC} (\Cref{section: local model checker UNTC}).
The proof is closely based on \Cref{section: theorem: completeness UNFO}, but some details are extended.
We define the \reintro*\kl(LMC){transition formula} $\intro*\transitionUNTC(\mystate)$ and $\reintro*\transitionUNTC_{\fmlsetclass}(\mystate)$ in the same way as in \Cref{section: theorem: completeness UNFO}, where the \kl(UNTC){rules} have been changed for \kl{UNTC}.
We first observe that the following properties hold also for the \kl(UNTC){LC} for \kl{UNTC}.
\begin{proposition}[Preservation property]\label{proposition: sem equiv UNTC}
    For each "rule" $\mystate[1] \leadstoUNTC \doublefml[2]$ in the "LC" for "UNTC" (\Cref{figure: rules UNTC}),
    ${} \absmodels \mystate[1]$ "iff" ${} \absmodels \doublefml[2]$.
    Hence, we have:
    \[{} \absmodels \mystate[1] \quad\iff\quad {} \absmodels \transitionUNTC(\mystate[1]).\]
\end{proposition}
\begin{proposition}[Duality]\label{proposition: dual UNTC}
    For all $\mystate[1] \in \allstatesUNTC$,
    $\transitionUNTC(\addPri{\mystate[1]}) = \addPri{\transitionUNTC(\mystate[1])}$.
\end{proposition}

\begin{proposition}[Consistency]\label{proposition: vdashUNTC con}
    Let $\fml$ be a "UNTC" \kl{formula}.
    For every $\mystate \in \allstates_{\clUNTC(\fml)}$, we have:
    \[\vdashUNTCsub_{\clUNTC(\fml)} \mystate \quad\Longrightarrow\quad \nvdashUNTCsub_{\clUNTC(\fml)} \addPri{\mystate}.\]
\end{proposition}

\subsubsection{Proof of \Cref{theorem: completeness closure UNTC}}\label{section: proof theorem: completeness closure UNTC}
Below, we use the notion of \emph{\kl[subclruns]{subruns}} (\Cref{section: theorem: completeness UNFO proof}) also for "UNTC".
To prove \Cref{theorem: completeness closure UNTC},
we use the following two lemmas.
Here, $\paramUNTC(\mystate)$ is a well-founded \kl{parameter} such that if $\paramUNTC(\mystate[1]) \ge \paramUNTC(\mystate[2])$ and $\absmodels \mystate[1]$, then $\absmodels \mystate[2]$,
which is modified from the \kl{parameter} given in \Cref{section: lemma: parameter UNFO odd} and
will be defined in \Cref{section: proof lemma: parameter UNTC odd}.
Notably, we only consider $\allstates_{\prebiclUNTC(\fml)}$ for odd priority (\Cref{lemma: parameter UNTC even}) but this is sufficient,
because when we return to odd priority from even priority by applying the rule \eqref{rule: UN} in $\allstates_{\clUNTC(\fml)}$, we can always return to a \kl(LMC){state} in $\allstates_{\prebiclUNTC(\fml)}$ (\Cref{proposition: negation cl to ucl}).
\begin{lemma}[Even case, \Cref{section: proof lemma: parameter UNTC even}]\label{lemma: parameter UNTC even}
    Let $\fml$ be a "UNTC" \kl{formula}.
    For every $\mystate[1] \in \allstates_{\clUNTC(\fml)}$
    "st" $\statePri(\mystate[1]) = \prizero$,
    if $\absmodels \mystate[1]$,
    then there is some depth-finite $\clUNTC(\fml)$-\kl[subclrun]{subrun} $\trace$ starting from $\mystate[1]$ such that for each "leaf" $\word$ of $\trace$ with $\trace(\word) = \mkstate{\fmlset[2]}{\pri}{\ynmulstruc}{\ainter'}{\bag[1]'}$, we have:
    \begin{itemize}
        \item $\absmodels \trace(\word)$, and
        \item $\fmlset[2] \in \prebiclUNTC(\fml)$ if $\pri = \prione$.
    \end{itemize}
\end{lemma}
\begin{lemma}[Odd case, \Cref{section: proof lemma: parameter UNTC odd}]\label{lemma: parameter UNTC odd}
    Let $\fml$ be a \kl{UNTC} \kl{formula}.
    For every $\mystate[1] \in \allstates_{\prebiclUNTC(\fml)}$
    "st" $\statePri(\mystate[1]) = \prione$,
    if $\absmodels \mystate[1]$,
    then there is some finite $\clUNTC(\fml)$-\kl[subclrun]{subrun} $\trace$ starting from $\mystate[1]$ such that for each "leaf" $\word$ of $\trace$ with $\trace(\word) = \mkstate{\fmlset[2]}{\pri}{\ynmulstruc}{\ainter'}{\bag[1]'}$, we have:
    \begin{itemize}
        \item $\paramUNTC(\mystate[1]) > \paramUNTC(\trace(\word))$
        (so, $\absmodels \trace(\word)$), and
        \item $\fmlset[2] \in \prebiclUNTC(\fml)$.
    \end{itemize}
\end{lemma}
Using them, we can prove \Cref{theorem: completeness closure UNTC}, as follows.
\begin{theorem*}[Restatement of \Cref{theorem: completeness closure UNTC}]
    \completenessclosureUNTC
\end{theorem*}
The proof is almost the same as \Cref{theorem: completeness closure UNFO}.
The remarkable difference is that we show only for "states" $\mystate$ in $\allstates_{\prebiclUNTC(\fml)}$,
whereas "states" in $\allstates_{\clUNTC(\fml)}$ and not in $\allstates_{\prebiclUNTC(\fml)}$ may appear in the "LC" for "UNTC".
\begin{proof}
    \proofcasethin{($\Longrightarrow$)}
    Let $\trace$ be the $\clUNTC(\fml)$-\kl[clrun]{run},
    obtained from the singleton \kl{tree} with $\trace(\eps) = \mystate$
    by extending each \kl{leaf} in $\allstates_{\prebiclUNTC(\fml)}$ with the $\clUNTC(\fml)$-\kl[subclrun]{subrun} of \Cref{lemma: parameter UNTC odd,lemma: parameter UNTC even}, iteratively.
    Here, because each "leaf" with odd "priority" $1$ in the \kl[subclruns]{subruns} above is always in $\allstates_{\prebiclUNTC(\fml)}$ (by the conditions of \Cref{lemma: parameter UNTC odd,lemma: parameter UNTC even}), we can indeed construct such a $\trace$.
    For the resulting $\trace$,
    we can show that $\trace$ is an \kl{accepting} $\clUNTC(\fml)$-\kl[clrun]{run} in the same way as the proof of \Cref{theorem: completeness closure UNFO}, whereby $\vdashUNTCsub_{\clUNTC(\fml)} \mystate$.

    \proofcasethin{($\Longleftarrow$)}
    Similar to the direction ($\Longleftarrow$) of \Cref{theorem: completeness closure UNFO},
    this part is shown by the direction ($\Longrightarrow$) with \Cref{proposition: vdashUNTC con}.

    Hence, this completes the proof.
\end{proof}

Below, we prove the remaining parts (\Cref{lemma: parameter UNTC even,lemma: parameter UNTC odd}).

\subsubsection{Proof of \Cref{lemma: parameter UNTC even} (even case)}\label{section: proof lemma: parameter UNTC even}
For negated "formulas",
we first observe the following property.
\begin{proposition}\label{proposition: negation cl to ucl}
    Let $\fml$ be a "UNTC" \kl{formula}.
    For every "UNTC" \kl{formula} of the form $\lnot \fml[2]$,
    if $\lnot \fml[2] \in \clUNTC(\fml)$,
    then $\lnot \fml[2] \in \preclUNTC(\fml)$.
\end{proposition}
\begin{proof}
    Clear, by the definitions of $\clUNTC$ and $\prebiclUNTC$.
\end{proof}
Using this, we can prove \Cref{lemma: parameter UNTC even}, as follows.
\begin{proof} (Similar to \Cref{lemma: parameter UNFO even}.)
    By \Cref{proposition: sem equiv UNTC}, we have $\absmodels \transitionUNTC(\mystate[1])$.
    For each \kl{state} $\mystate[2]$ with odd "priority" in $\transitionUNTC(\mystate[1])$,
    the rule \eqref{rule: UN} is always applied from $\mystate[1]$, since the "rule" \eqref{rule: UN} only changes the "priority".
    Hence $\mystate[2] \in \allstates_{\preclUNTC(\fml)}$ (by \Cref{proposition: negation cl to ucl}).
    Also, for each \kl{state} $\mystate[2]$ with even "priority" in $\transitionUNTC(\mystate[1])$, it is not replaced or replaced with $\doubletrue$ (if $\mystate[2] \not\in \allstates_{\clUNTC(\fml)}$).
    Thus, we have $\transitionUNTC(\mystate[1]) \lleqq \transitionUNTC_{\clUNTC(\fml)}(\mystate[1])$,
    and hence $\absmodels \transitionUNTC_{\clUNTC(\fml)}(\mystate[1])$.
    Let $\transitionUNTC_{\clUNTC(\fml)}(\trace[1](\word)) \lequiv \bigdoublevee_{l} \bigdoublewedge_{k} {\mystate[2]_{l, k}}$.
    Then $\absmodels \bigdoublewedge_{k} {\mystate[2]_{l, k}}$ for some $l$.
    Hence, by taking the $\clUNTC(\fml)$-\kl[subclruns]{subrun} having the "leaves" of $\mystate[2]_{l, k}$, this completes the proof.
\end{proof}

\subsubsection{Proof of \Cref{lemma: parameter UNTC odd} (odd case)}\label{section: proof lemma: parameter UNTC odd}
For $n \in \nat$,
the "UNTC" "formula" $\intro*\kthTC{\fml[2]}{n}{u v} {x}{y}$ is defined as follows:
\begin{align*}
    \kthTC{\fml[2]}{0}{u v} {x}{y} &\;\defeq\; {x} \EQ {y},&
    \kthTC{\fml[2]}{n+1}{u v} {x}{y} &\;\defeq\; \exists {z} (\fml[2]\fmlsubst{u v}{x z} \land \kthTC{\fml[2]}{n}{u v}{z}{y}) \tag*{where $z$ is \kl{fresh}.}
\end{align*}
For a \kl{UNTC} \kl{formula} $\fml$, we write $\intro*\UNF(\fml)$ for the set of \kl{UNTC} \kl{formulas}
obtained from $\fml$ by unfolding each \emph{non-$\lnot$-scoped occurrence} of $\TC{\fml[2]}{u v}{x}{y}$ into a \kl{formula} $\kthTC{\fml[2]}{n}{u v}{x}{y}$ for some $n \ge 0$.
More precisely, $\UNF(\fml)$ is inductively defined as follows:
\begin{align*}
    \UNF(\afml) &\defeq \set{\afml},&
    \UNF(\fml[2] \land \fml[3]) &\defeq \set{\fml[2]' \land \fml[3]' \mid \fml[2]' \in \UNF(\fml[2]) \mbox{ and } \fml[3]' \in \UNF(\fml[3]) },\\
    \UNF(\lnot \fml[2]) &\defeq \set{\lnot \fml[2]},&
    \UNF(\fml[2] \lor \fml[3]) &\defeq \set{\fml[2]' \lor \fml[3]' \mid \fml[2]' \in \UNF(\fml[2]) \mbox{ and } \fml[3]' \in \UNF(\fml[3]) },\\
    \UNF(\exists x \fml[2]) &\defeq \set{\exists x \fml[2]' \mid \fml[2]' \in \UNF(\fml[2])},&
    \UNF(\TC{\fml[2]}{u v}{x}{y}) &\defeq \bigcup_{n \ge 0} \UNF(\kthTC{\fml[2]}{n}{u v}{x}{y}).
\end{align*}
For a \kl{UNTC} \kl{formula} \kl{set} $\fmlset[1]$,
we write $\UNF(\fmlset[1])$ for the \kl{set} of \kl{UNTC} \kl{formula} \kl{sets}
obtained from $\fmlset[1]$ by replacing each $\fml \in \fmlset[1]$ with some $\fml' \in \UNF(\fml)$.
We write $\UNF(\mkstate{\fmlset[1]}{\pri}{\ynmulstruc}{\ainter}{\bag})$ for the \kl{set} $\set{\mkstate{\fmlset[2]}{\pri}{\ynmulstruc}{\ainter}{\bag} \mid \fmlset[2] \in \UNF(\fmlset[1])}$.

\paragraph{A well-founded parameter}
For $\mystate[1] = \mkstate{\fmlset[1]}{\pri}{\ynmulstruc}{\ainter}{\bag} \in \allstates_{\mathrm{\kl{UNTC}}}$ with odd priority $\pri$,
the \intro*\kl(UNTC){concretization set} $\intro*\paramconcUNTC(\mystate[1])$ is defined by:
\begin{align*}
    \paramconcUNTC(\mystate[1])
    &~\defeq~
    \set*{\fmlset[2]^{\glue \ynmulstruc}_{\inter} \;\middle|\;
    \mbox{$\inter$ is a \kl{concretization} of $\ainter$ on $\bag$,
    $\fmlset[2] \in \UNF(\fmlset[1])$, and $\glue\ynmulstruc \modelsass{\inter} \bigdoublewedge \fmlset[2]$}}.
\end{align*}
By definition, we have:
\[\absmodels \dot{\fmlset[1]} \quad\Longleftrightarrow\quad \paramconcUNTC(\dot{\fmlset[1]}) \neq \emptyset.\]

For a "state" $\mystate[1] = \mkstate{\fmlset[1]}{\pri}{\ynmulstruc}{\ainter}{\bag[1]} \in \allstatesUNTC$,
the \intro*\kl(UNTC){parameter} $\intro*\paramUNTC(\mystate[1]) \in \nat^2 \dcup \set{\infty}$ is defined as follows:
\begin{align*}
    \paramUNTC(\mystate[1])
    &~\defeq~
    \begin{cases}
    {\displaystyle \min_{\fmlset[2]^{\glue \ynmulstruc}_{\inter} \in \paramconcUNTC(\mystate[1])}
    \tup{\fmllen{\fmlset[2]}, \treedistance(\bag, \inter(\FV(\fmlset[1])))}
    } & \mbox{if $\pri = \prione$ and $\absmodels \mystate[1]$,}\\
    \tup{0,0} & \mbox{if $\pri = \prizero$ and $\absmodels \mystate[1]$,}\\
    \infty & \mbox{otherwise,}
    \end{cases} 
\end{align*}
(This definition is the same as \Cref{section: lemma: parameter UNFO odd} except the definition of $\paramconcUNTC$.)

\paragraph{Splitting lemma}
The completeness is shown based on the evaluation strategy presented in \Cref{section: evaluation strategy UNTC}.
We first show the following lemma,
which corresponds to the argument given in the lines of (\nameref{example: completeness UNTC split}).

\begin{lemma}[Splitting lemma]\label{lemma: split UNTC}
    Let $\fml$ be a "UNTC" "formula".
    Let $\dir \in \set{\back 2, \back 1, 1, 2}$.
    For every $\mystate = \mkstate{(\fmlset[1]',\ \fml')}{\prione}{\ynmulstruc}{\ainter}{\bag} \in \allstates_{\mathrm{\prebiclUNTC(\fml)}}$
    where $\fmlset[1]'$ is a \kl{UNTC} \kl{formula} \kl{set} and $\fml'$ is a \kl{UNTC} \kl{formula},
    if $\absmodels \mystate$,
    then there is some finite $\clUNTC(\fml)$-\kl[subclrun]{subrun} $\trace$
    with one leaf $\word$ starting from $\mystate$ such that
    \begin{itemize}
        \item $\paramUNTC(\mystate) \ge \paramUNTC(\trace(\word))$;

        \item $\trace(\word) = \mkstate{(\fmlset[1]',\  \fmlset[2], \fmlset[3])}{\prione}{\ynmulstruc}{\ainter'}{\bag}$ where $\fmlset[2], \fmlset[3] \in \preclUNTC(\fml)$ and $\FV(\fmlset[1]') \cap \FV((\fmlset[2], \fmlset[3])) \subseteq \FV(\fml)$;

        \item $\ainter'(\FV(\fmlset[2])) \subseteq \ND{\ynmulstruc}{\bag}{\dir}$ and $\dir \not\in \ainter'(\FV(\fmlset[3]))$ (so, $\ainter'(\FV(\fmlset[2]) \cap \FV(\fmlset[3])) \subseteq \ND{\ynmulstruc}{\bag}{0}$).
    \end{itemize}
\end{lemma}
\begin{proof}
    By induction on the \kl(formula){size} of $\fml'$.
    When $\ainter(\FV(\fml')) \subseteq \ND{\ynmulstruc}{\bag}{\dir}$,
    by letting
    $\tup{\fmlset[2], \fmlset[3]} = \tup{(\fml'), ()}$,
    this case has been proved.
    When $\dir \not\in \ainter(\FV(\fml'))$,
    by letting $\tup{\fmlset[2], \fmlset[3]} = \tup{(), (\fml')}$,
    this case has been proved.
    Otherwise, $\dir \in \ainter(\FV(\fml')) \not\subseteq \ND{\ynmulstruc}{\bag}{\dir}$.
    Also, by applying \eqref{rule: concrete}, "wlog", it suffices to show when $\card\ainter(x) = 1$ for every $x \in \FV(\fml')$.
    We distinguish the following cases:

    \proofcasethin{Case $\fml'$ is $\afml$}
    By $\dir \in \ainter(\FV(\afml)) \not\subseteq \ND{\ynmulstruc}{\bag}{\dir}$,
    it contradicts $\models \mkstate{(\fmlset[1]',\ \afml)}{\prione}{\ynmulstruc}{\ainter}{\bag}$.
    Hence, this case has already been passed.
   
    \proofcasethin{Case $\fml'$ is $\lnot \fml[2]$}
    By $\dir \in \ainter(\FV(\lnot \fml[2])) \not\subseteq \ND{\ynmulstruc}{\bag}{\dir}$,
    it contradicts the condition of "unary negation": $\card(\FV(\fml[2])) \le 1$.
    Hence, this case has already been passed.
    
    \proofcasethin{Case $\fml'$ is $\fml[2] \lor \fml[3]$}
    We have:
    \begin{align*}
        \mkstate{(\fmlset[1]',\  \fml[2] \lor \fml[3])}{\prione}{\ynmulstruc}{\ainter}{\bag}
        &\;\leadstoUNTC_{\eqref{rule: or}}\;
        \mkstate{(\fmlset[1]',\  \fml[2])}{\prione}{\ynmulstruc}{\ainter}{\bag}
        \doublevee
        \mkstate{(\fmlset[1]',\  \fml[3])}{\prione}{\ynmulstruc}{\ainter}{\bag}.
    \end{align*}
    Thus by IH, this case has been shown.

    \proofcasethin{Case $\fml'$ is $\exists x \fml[2]$}
    We have:
    \begin{align*}
        \mkstate{(\fmlset[1]',\ \exists x \fml[2])}{\prione}{\ynmulstruc}{\ainter}{\bag}
        &\;\leadstoUNTC_{\eqref{rule: exists}}\; \mkstate{(\fmlset[1]',\  \fml[2]\fmlsubst{x}{z})}{\prione}{\ynmulstruc}{\ainter\aintersubst{z}{\ADB{\ynmulstruc}{\bag}}}{\bag},
    \end{align*}
    where $z$ is \kl{fresh}.
    Thus by IH, this case has been shown.
 
    \proofcasethin{Case $\fml'$ is $\fml[2] \land \fml[3]$}
    We have:
    \begin{align*}
        &\mkstate{(\fmlset[1]',\ \fml[2] \land \fml[3])}{\prione}{\ynmulstruc}{\ainter}{\bag}
        \;\leadstoUNTC_{\eqref{rule: and}}\;
        \mkstate{(\fmlset[1]',\ \fml[2], \fml[3])}{\prione}{\ynmulstruc}{\ainter}{\bag}
        \;=\;
        \mkstate{(\fmlset[1]' \cup \set{\fml[3]},\  \fml[2])}{\prione}{\ynmulstruc}{\bl}{\bag}\\
        &\;\leadstoUNTC^*\lgeqq\;
        \mkstate{(\fmlset[1]' \cup \set{\fml[3]},\  \fmlset[2], \fmlset[3])}{\prione}{\ynmulstruc}{\bl}{\bag}
        \;=\;
        \mkstate{(\fmlset[1]' \cup \fmlset[2] \cup \fmlset[3],\  \fml[3])}{\prione}{\ynmulstruc}{\bl}{\bag}
        \tag{By IH "wrt" $\fml[2]$}\\
        &\;\leadstoUNTC^*\lgeqq\;
        \mkstate{(\fmlset[1]' \cup \fmlset[2] \cup \fmlset[3],\  \fmlset[2]', \fmlset[3]')}{\prione}{\ynmulstruc}{\bl}{\bag}
        \;=\;
        \mkstate{(\fmlset[1]',\ (\fmlset[2] \cup \fmlset[2]'), (\fmlset[3] \cup \fmlset[3]'))}{\prione}{\ynmulstruc}{\bl}{\bag}.
        \tag{By IH "wrt" $\fml[3]$}
    \end{align*}
    Hence, we have obtained the desired $\clUNTC(\fml)$-\kl[subclrun]{subrun}.

    \proofcasethin{Case $\fml'$ is $\TC{\fml[2]}{u v} {x} {y}$}
    By $\dir \in \ainter(\FV(\TC{\fml[2]}{u v} {x} {y}))$,
    either $\ainter(x) = \set{\dir}$ or $\ainter(y) = \set{\dir}$.
    We distinguish the following cases:

    \subproofcasethin{Case $\ainter(x) = \set{\dir}$}
    By $\ainter(y) \not\subseteq \ND{\ynmulstruc}{\bag}{\dir}$,
    we have:
    \begin{align*}
        &\mkstate{(\fmlset[1]',\  \TC{\fml[2]}{u v} {x} {y})}{\prione}{\ynmulstruc}{\ainter}{\bag}
        \leadstoUNTC_{\eqref{rule: tc-split-l}\eqref{rule: concrete}}^*\lgeqq\\
        &\mkstate{(\fmlset[1]',\  \TC{\fml[2]}{u v} {x^{\dir}} {z_1^{\ND{\ynmulstruc}{\bag}{\dir}}}, \fml[2]\fmlsubst{u v}{z_1 z_2}, \TC{\fml[2]}{u v} z_2^{\ND{\ynmulstruc}{\bag}{\dir'}} y)}{\prione}{\ynmulstruc}{\bl}{\bag}
        \tag*{(where $\dir' \in \range{\back 2}{2} \setminus \set{\dir}$ is chosen so as to preserve $\absmodels$ and $z_1$ and $z_2$ are "fresh")}\\
        &=
        \mkstate{(\fmlset[1]' \cup \set{\TC{\fml[2]}{u v} {x} {z_1}, \TC{\fml[2]}{u v} {z_2} {y}}, \fml[2]\fmlsubst{u v}{z_1 z_2})}{\prione}{\ynmulstruc}{\bl}{\bag}\\
        &\leadstoUNTC^*\lgeqq
        \mkstate{(\fmlset[1]' \cup \set{\TC{\fml[2]}{u v} {x} {z_1}, \TC{\fml[2]}{u v} {z_2} {y}}, \fmlset[2], \fmlset[3])}{\prione}{\ynmulstruc}{\bl}{\bag} \tag{$\fmlset[2]$ and $\fmlset[3]$ are given by IH "wrt" $\fml[2]\fmlsubst{u v}{z_1 z_2}$}\\
        &=
        \mkstate{(\fmlset[1]', \set{\TC{\fml[2]}{u v} {x} {z_1}} \cup \fmlset[2], \fmlset[3] \cup \set{\TC{\fml[2]}{u v} z_2 {y}})}{\prione}{\ynmulstruc}{\bl}{\bag}.
    \end{align*}
    Hence, we have obtained the desired $\clUNTC(\fml)$-\kl[subclrun]{subrun}.
    
    \subproofcasethin{Case $\ainter(y) = \set{\dir}$}
    Similarly,
    by $\ainter(x) \not\subseteq \ND{\ynmulstruc}{\bag}{\dir}$,
    we have:
    \begin{align*}
        &\mkstate{(\fmlset[1]',\  \TC{\fml[2]}{u v} {x} {y})}{\prione}{\ynmulstruc}{\ainter}{\bag}
        \leadstoUNTC_{\eqref{rule: tc-split-r}\eqref{rule: concrete}}^*\lgeqq\\
        &\mkstate{(\fmlset[1]',\ \TC{\fml[2]}{u v} {x} {z_1^{\ND{\ynmulstruc}{\bag}{\dir'}}}, \fml[2]\fmlsubst{u v}{z_1 z_2}, \TC{\fml[2]}{u v} z_2^{\ND{\ynmulstruc}{\bag}{\dir}} y^{\dir})}{\prione}{\ynmulstruc}{\bl}{\bag}
        \tag*{(where $\dir' \in \range{\back 2}{2} \setminus \set{\dir}$ is chosen so as to preserve $\absmodels$
        and $z_1$ and $z_2$ are "fresh")}\\
        &=
        \mkstate{(\fmlset[1]' \cup \set{\TC{\fml[2]}{u v} {x} {z_1}, \TC{\fml[2]}{u v} z_2 {y}}, \fml[2]\fmlsubst{u v}{z_1 z_2})}{\prione}{\ynmulstruc}{\bl}{\bag}\\
        &\leadstoUNTC^*\lgeqq
        \mkstate{(\fmlset[1]' \cup \set{\TC{\fml[2]}{u v} {x} {z_1}, \TC{\fml[2]}{u v} z_2 {y}}, \fmlset[2], \fmlset[3])}{\prione}{\ynmulstruc}{\bl}{\bag} \tag{$\fmlset[2]$ and $\fmlset[3]$ are given by IH "wrt" $\fml[2]\fmlsubst{u v}{z_1 z_2}$}\\
        &=
        \mkstate{(\fmlset[1]', \set{\TC{\fml[2]}{u v} {x} {z_1}} \cup \fmlset[3], \fmlset[2] \cup \set{\TC{\fml[2]}{u v} z_2 {y}})}{\prione}{\ynmulstruc}{\bl}{\bag}.
        \qedhere
    \end{align*}
\end{proof}

\paragraph{Proof of \Cref{lemma: parameter UNTC odd}}
We now construct a finite $\clUNTC(\fml)$-\kl[subclrun]{subrun}, as follows.

\begin{proof}
    By induction on the \kl{parameter} $\paramUNTC(\mystate)$.
    Let $\mystate = \mkstate{\fmlset[1]}{\prione}{\ynmulstruc}{\ainter}{\bag[1]}$.
    We distinguish the following cases.

    \proofcasethin{Steps 1 and 1'}
    \labeltext{lemma: parameter UNTC concrete}{}{1}\labeltext{lemma: parameter UNTC or and exists}{}{1'}
    We apply the same arguments as Steps \nameref{lemma: parameter UNFO concrete} and \nameref{lemma: parameter UNFO or and exists} in \Cref{lemma: parameter UNFO odd}.
    They allow us to assume the following in the subsequent cases.
    \begin{enumerate}
        \item[a)] \labeltext{cond: evaluation strategy UNTC proof 1}{}{a} $\card \ainter(x) = 1$ for each $x \in \FV(\fmlset[1])$.
        \item[b)] \labeltext{cond: evaluation strategy UNTC proof 2}{}{b} $\fmlset[1]$ is of the form $(\fml[2]_1, \dots, \fml[2]_n)$
        where each $\fml[2]_i$ is one of the following forms:
        an "atom" $\afml$, $\lnot \fml[3]$, or $\TC{\fml[3]}{u v} x y$.
    \end{enumerate}

    \proofcasethin{Step 2}\labeltext{lemma: parameter UNTC split}{}{2}
    Else if the following condition does not hold:
    \begin{enumerate}
        \item[c')] \labeltext{cond: evaluation strategy UNTC proof 3}{}{c'}
        For each $\fml[2]_i$,
        for a $\dir \in \set{\back 2, \back 1, 1, 2}$,
        $\ainter(\FV(\fml[2]_i)) \subseteq \ND{\ynmulstruc}{\bag}{\dir}$.
    \end{enumerate}
    Then, except the case when $\fml[2]_i$ is a "TC-formula",
    we can show \nameref{cond: evaluation strategy UNTC proof 3}) in the same way as in \Cref{lemma: parameter UNFO odd}.
    Let $\fml[2]_i$ be a "TC-formula".
    Suppose $\ainter(\FV(\fml[2]_i)) \not\subseteq \ND{\ynmulstruc}{\bag}{\dir}$ for all $\dir \in \set{\back 2, \back 1, 1, 2}$.
    By the four cases \ref{definition: closure UNTC not unfolded}\ref{definition: closure UNTC l}\ref{definition: closure UNTC r}\ref{definition: closure UNTC bi} in the definition of $\prebiclUNTC$ (\Cref{definition: biclosure UNTC}) with $\fmlset[1] \in \prebiclUNTC(\fml)$,
    there is a "subformula" $\TC{\fml[3]}{u v} x y$ of $\fml$ such that 
    $\fmlset[1]$ satisfies either one of the following conditions:

    \begin{enumerate}[i)]
        \item \labeltext{state: tc self appendix}{}{i} $\fmlset[1] = (\fmlset[1]',\ \phantom{\fmlset[2], }\, \TC{\fml[3]}{u v} x y)$, where
        \begin{itemize}
        \item $\fml[2]_i = \TC{\fml[3]}{u v} x y$, and
        \item $(\TC{\fml[3]}{u v} x y) \in \preclUNTC(\TC{\fml[3]}{u v} x y)$ occurs
        in the derivation tree of $\fmlset[1] \in \prebiclUNTC(\fml)$.
        \end{itemize}

        \item \labeltext{state: tc l appendix}{}{ii} $\fmlset[1] = (\fmlset[1]',\ \fmlset[2], \TC{\fml[3]}{u v} z_2 z_3)$, where
        \begin{itemize}
            \item $\fml[2]_i = \TC{\fml[3]}{u v} z_2 z_3$,
            \item $\FV(\fmlset[1]') \cap \FV(\fmlset[2]) = \emptyset$,
            $\FV(\fmlset[1]') \cap \set{z_2, z_3} \subseteq \set{z_3}$,
            $\FV(\fmlset[2]) \cap \set{z_2, z_3} \subseteq \set{z_2}$,
            $z_1, z_2, z_3$ are pairwise distinct, and 
            \item 
            in the derivation tree of $\fmlset[1] \in \prebiclUNTC(\fml)$,
            both $(\fmlset[2], \TC{\fml[3]}{u v} z_2 z_3) \in \preclUNTC(\TC{\fml[3]}{u v} x y)$ and $\fmlset[2] \in \preclUNTC(\fml[3]\fmlsubst{u v}{z_1 z_2})$ occur.
        \end{itemize}

        \item \labeltext{state: tc r appendix}{}{iii} $\fmlset[1] = (\fmlset[1]',\ \phantom{\fmlset[2], }\, \TC{\fml[3]}{u v} z_1 z_2, \fmlset[3])$, where
        \begin{itemize}
            \item $\fml[2]_i = \TC{\fml[3]}{u v} z_1 z_2$,
            \item $\FV(\fmlset[1]') \cap \set{z_1, z_2} \subseteq \set{z_1}$,
            $\FV(\fmlset[1]') \cap \FV(\fmlset[3]) = \emptyset$, 
            $\set{z_1, z_2} \cap \FV(\fmlset[3]) \subseteq \set{z_2}$,
            $z_1, z_2, z_3$ are pairwise distinct, and 
            \item
            in the derivation tree of $\fmlset[1] \in \prebiclUNTC(\fml)$,
            both $(\TC{\fml[3]}{u v} z_1 z_2, \fmlset[3]) \in \preclUNTC(\TC{\fml[3]}{u v} x y)$ and $\fmlset[3] \in \preclUNTC(\fml[3]\fmlsubst{u v}{z_2 z_3})$ occur.
        \end{itemize}

        \item \labeltext{state: tc bi appendix}{}{iv}  $\fmlset[1] = (\phantom{\fmlset[1]',\ }\, \fmlset[2], \TC{\fml[3]}{u v} z_2 z_3, \fmlset[3])$, where
        \begin{itemize}
            \item $\fml[2]_i = \TC{\fml[3]}{u v} z_2 z_3$,
            \item $\FV(\fmlset[2]) \cap \set{z_2, z_3} \subseteq \set{z_2}$, $\set{z_2, z_3} \cap \FV(\fmlset[3]) \subseteq \set{z_3}$,
            $\FV(\fmlset[3]) \cap \FV(\fmlset[2]) = \emptyset$,
            $z_1,z_2,z_3,z_4$ are pairwise distinct, and
            \item
            in the derivation tree of $\fmlset[1] \in \prebiclUNTC(\fml)$,
            all $(\fmlset[2], \TC{\fml[3]}{u v} z_2 z_3, \fmlset[3]) \in \prebiclUNTC(\TC{\fml[3]}{u v} x y)$, $\fmlset[2] \in \preclUNTC(\fml[3]\fmlsubst{u v}{z_1 z_2})$, $\fmlset[3] \in \preclUNTC(\fml[3]\fmlsubst{u v}{z_3 z_4})$ occur.
        \end{itemize}
    \end{enumerate}
    For convenience, we use $\fmlset[1]_0$, $x_0$ and $y_0$ such that
    $\fmlset[1] = (\fmlset[1]_0, \fml[2]_i)$ and
    $\fml[2]_i = \TC{\fml[3]}{u v} x_0 y_0$.
    We then distinguish the following cases.
 
    \proofcasethin{Case: $\ainter(x_0) \subseteq \ND{\ynmulstruc}{\bag}{0}$}
    If $\absmodels \mkstate{(\fmlset[1]_0, \kthTC{\fml[3]}{0}{u v} x_0 y_0)}{\prione}{\ynmulstruc}{\ainter}{\bag}$,
    we have:
    \begin{align*}
    & \mkstate{(\fmlset[1]_0, \ \TC{\fml[3]}{u v} x_0 y_0)}{\prione}{\ynmulstruc}{\ainter}{\bag}
    \;\leadstoUNTC_{\eqref{rule: tc-l}}\lgeqq\;
    \mkstate{(\fmlset[1]_0,\ {x_0 \EQ y_0})}{\prione}{\ynmulstruc}{\ainter}{\bag}
    \;\leadstoUNTC_{\eqref{rule: a1}}\;
    \mkstate{(\fmlset[1]_0)}{\prione}{\ynmulstruc}{\ainter}{\bag}.
    \end{align*}
    Thus by IH, this case has been shown.
    Otherwise, we distinguish the following four cases according to the above.

    \subproofcasethin{Case \nameref{state: tc self appendix})}
    We have:
    \begin{align*}
    & \mkstate{(\fmlset[1]',\ \TC{\fml[3]}{u v} {x} {y})}{\prione}{\ynmulstruc}{\ainter}{\bag} 
    \;\leadstoUNTC_{\eqref{rule: tc-l}}\lgeqq\;
    \underbrace{\mkstate{(\fmlset[1]',\ \fml[3]\fmlsubst{u v}{x z}, \TC{\fml[3]}{u v} {z} {y})}{\prione}{\ynmulstruc}{\bl}{\bag}}_{\text{Case \nameref{state: tc l appendix})}}.
    \end{align*}

    \subproofcasethin{Case \nameref{state: tc l appendix})}
    We have:
    \begin{align*}
    \mkstate{(\fmlset[1]',\ \fmlset[2], \TC{\fml[3]}{u v} {z_2} {z_3})}{\prione}{\ynmulstruc}{\ainter}{\bag}
    &\leadstoUNTC_{\eqref{rule: split}\eqref{rule: tc-l}}\lgeqq\; \mkstate{\fmlset[2]}{\prione}{\ynmulstruc}{\ainter}{\bag}
    \doublewedge^{}
    \underbrace{\mkstate{(\fmlset[1]',\ \fml[3]\fmlsubst{u v}{z_2 z}, \TC{\fml[3]}{u v} {z} {z_3})}{\prione}{\ynmulstruc}{\bl}{\bag}}_{\text{Case \nameref{state: tc l appendix})}}.
    \end{align*}

    \subproofcasethin{Case \nameref{state: tc r appendix})}
    We have:
    \begin{align*}
    \mkstate{(\fmlset[1]',\ \TC{\fml[3]}{u v} {z_1} {z_2}, \fmlset[3])}{\prione}{\ynmulstruc}{\ainter}{\bag}
    & \;\leadstoUNTC_{\eqref{rule: split}\eqref{rule: tc-l}}\lgeqq\;
    \mkstate{(\fmlset[1]')}{\prione}{\ynmulstruc}{\ainter}{\bag} \doublewedge^{} \underbrace{\mkstate{(\fml[3]\fmlsubst{u v}{z_1 z}, \TC{\fml[3]}{u v} {z} {z_2}, \fmlset[3])}{\prione}{\ynmulstruc}{\bl}{\bag}}_{\text{Case \nameref{state: tc bi appendix})}}.
    \end{align*}

    \subproofcasethin{Case \nameref{state: tc bi appendix})}
    We have:
    \begin{align*}
    \mkstate{(\fmlset[2], \TC{\fml[3]}{u v} {z_2} {z_3}, \fmlset[3])}{\prione}{\ynmulstruc}{\ainter}{\bag}
    & \;\leadstoUNTC_{\eqref{rule: split}\eqref{rule: tc-l}}\lgeqq\;
    \mkstate{(\fmlset[2])}{\prione}{\ynmulstruc}{\ainter}{\bag}
    \doublewedge
    \underbrace{\mkstate{(\fml[3]\fmlsubst{u v}{z_2 z}, \TC{\fml[3]}{u v}{z} {z_3}, \fmlset[3])}{\prione}{\ynmulstruc}{\bl}{\bag}}_{\text{Case \nameref{state: tc bi appendix})}}.
    \end{align*}

    In each case, the \kl(UNTC){parameter} strictly decreases, as the "TC-formula" is unfolded.
    Thus by IH, this case has been shown.
    
    \proofcasethin{Case: $\ainter(y_0) \subseteq \ND{\ynmulstruc}{\bag}{0}$}
    Similar to Case: $\ainter(x_0) \subseteq \ND{\ynmulstruc}{\bag}{0}$.

    \proofcasethin{Otherwise}
    Then there exist distinct $\dir, \dir' \in \set{\back 2, \back 1, 1, 2}$ such that
    $\ainter(x_0) = \set{\dir}$ and $\ainter(y_0) = \set{\dir'}$.
    We distinguish the following three cases except Case \nameref{state: tc self appendix}).
    Case \nameref{state: tc self appendix}) is considered in the last.

    \subproofcasethin{Case \nameref{state: tc l appendix})}
    We have:
    \begin{align*}
        \mkstate{(\fmlset[1]',\  \fmlset[2], \TC{\fml[3]}{u v} z_2 z_3)}{\prione}{\ynmulstruc}{\ainter}{\bag}
        \leadstoUNTC_{\eqref{rule: tc-split-l}}
        &~\mkstate{(\fmlset[1]',\ \underbrace{\fmlset[2], \TC{\fml[3]}{u v} z_2 z_2', \fml[3]\fmlsubst{u v}{z_2' z_3'}}_{\text{$\in \prebiclUNTC(\TC{\fml[3]}{u v} x y)$}},
        \underbrace{\TC{\fml[3]}{u v} z_3' z_3}_{\text{$\in \preclUNTC(\TC{\fml[3]}{u v} x y)$}})}{\prione}{\ynmulstruc}{\bl}{\bag}\\
        \leadstoUNTC^*
        &~\mkstate{(\fmlset[1]',\ \underbrace{\fmlset[2], \TC{\fml[3]}{u v} z_2 z_2', \fmlset[2]'}_{\text{$\in \prebiclUNTC(\TC{\fml[3]}{u v} x y)$}},\  \underbrace{\fmlset[3]', \TC{\fml[3]}{u v} z_3' z_3}_{\text{$\in \preclUNTC(\TC{\fml[3]}{u v} x y)$}})}{\prione}{\ynmulstruc}{\bl}{\bag} \tag{\Cref{lemma: split UNTC}}\\
        \leadstoUNTC_{\eqref{rule: split}}
        &~\underbrace{\mkstate{(\fmlset[2], \TC{\fml[3]}{u v} z_2 z_2', \fmlset[2]')}{\prione}{\ynmulstruc}{\bl}{\bag}}_{\text{Case \nameref{state: tc bi appendix})}}
        \;\doublewedge^{}\;
        \underbrace{\mkstate{(\fmlset[1]',\  \fmlset[3]', \TC{\fml[3]}{u v} z_3' z_3)}{\prione}{\ynmulstruc}{\bl}{\bag}}_{\text{Case \nameref{state: tc l appendix})}}
    \end{align*}
    where $z_2'$ and $z_3'$ are "fresh" and $\fmlset[2]', \fmlset[3]' \in \preclUNTC(\fml[3]\fmlsubst{u v}{z_2' z_3'})$.
    Both the resulting \kl{formula} \kl{sets} above are in the set $\prebiclUNTC(\fml)$,
    by the definition of "bi-closure" (\Cref{definition: biclosure UNTC}).
    Thus, by IH, this case has been shown.
    Observe that in the steps above (from \eqref{rule: tc-split-l} to \eqref{rule: split}),
    it is necessary to temporarily use "states" in $\clUNTC(\fml)$ (not in $\prebiclUNTC(\fml)$).

    For Cases \nameref{state: tc r appendix}) and \nameref{state: tc bi appendix}),
    we can show in the same way as this case.
    Below, we only write the application of the "rules".

    \subproofcasethin{Case \nameref{state: tc r appendix})}
    We have:
    \begin{align*}
        \mkstate{(\fmlset[1]',\  \TC{\fml[3]}{u v} z_1 z_2, \fmlset[3])}{\prione}{\ynmulstruc}{\ainter}{\bag}
        \leadstoUNTC_{\eqref{rule: tc-split-l}}
        &~\mkstate{(\fmlset[1]',\ \underbrace{\TC{\fml[3]}{u v} z_1 z_1'}_{\text{$\in \preclUNTC(\TC{\fml[3]}{u v} x y)$}}, 
        \underbrace{\fml[3]\fmlsubst{u v}{z_1' z_2'}, \TC{\fml[3]}{u v} z_2' z_2, \fmlset[3]}_{\text{$\in \prebiclUNTC(\TC{\fml[3]}{u v} x y)$}})}{\prione}{\ynmulstruc}{\bl}{\bag}\\
        \leadstoUNTC^*
        &~\mkstate{(\fmlset[1]',\ \underbrace{\TC{\fml[3]}{u v} z_1 z_1', \fmlset[2]'}_{\text{$\in \preclUNTC(\TC{\fml[3]}{u v} x y)$}},\  \underbrace{\fmlset[3]', \TC{\fml[3]}{u v} z_2' z_2, \fmlset[3]}_{\text{$\in \prebiclUNTC(\TC{\fml[3]}{u v} x y)$}})}{\prione}{\ynmulstruc}{\bl}{\bag} \tag{\Cref{lemma: split UNTC}}\\
        \leadstoUNTC_{\eqref{rule: split}}
        &~\underbrace{\mkstate{(\fmlset[1]',\  \TC{\fml[3]}{u v} z_1 z_1', \fmlset[2]')}{\prione}{\ynmulstruc}{\bl}{\bag}}_{\text{Case \nameref{state: tc r appendix})}}
        \;\doublewedge^{}\;
        \underbrace{\mkstate{(\fmlset[3]', \TC{\fml[3]}{u v} z_2' z_2, \fmlset[3])}{\prione}{\ynmulstruc}{\bl}{\bag}}_{\text{Case \nameref{state: tc bi appendix})}}
    \end{align*}

    \subproofcasethin{Case \nameref{state: tc bi appendix})}
    We have:
    \begin{align*}
        \mkstate{(\fmlset[2], \TC{\fml[3]}{u v} z_2 z_3, \fmlset[3])}{\prione}{\ynmulstruc}{\ainter}{\bag}
        \leadstoUNTC_{\eqref{rule: tc-split-l}}
        &~\mkstate{(\underbrace{\fmlset[2], \TC{\fml[3]}{u v} z_2 z_2', \fml[3]\fmlsubst{u v}{z_2' z_3'}}_{\text{$\in \prebiclUNTC(\TC{\fml[3]}{u v} x y)$}},
        \underbrace{\TC{\fml[3]}{u v} z_3' z_3, \fmlset[3]}_{\text{$\in \preclUNTC(\TC{\fml[3]}{u v} x y)$}})}{\prione}{\ynmulstruc}{\bl}{\bag}\\
        \leadstoUNTC^*
        &~\mkstate{(\underbrace{\fmlset[2], \TC{\fml[3]}{u v} z_2 z_2', \fmlset[2]'}_{\text{$\in \prebiclUNTC(\TC{\fml[3]}{u v} x y)$}},\  \underbrace{\fmlset[3]', \TC{\fml[3]}{u v} z_3' z_3, \fmlset[3]}_{\text{$\in \prebiclUNTC(\TC{\fml[3]}{u v} x y)$}})}{\prione}{\ynmulstruc}{\bl}{\bag} \tag{\Cref{lemma: split UNTC}}\\
        \leadstoUNTC_{\eqref{rule: split}}
        &~\underbrace{\mkstate{(\fmlset[2], \TC{\fml[3]}{u v} z_2 z_2', \fmlset[2]')}{\prione}{\ynmulstruc}{\bl}{\bag}}_{\text{Case \nameref{state: tc bi appendix})}}
        \;\doublewedge^{}\;
        \underbrace{\mkstate{(\fmlset[3]', \TC{\fml[3]}{u v} z_3' z_3, \fmlset[3])}{\prione}{\ynmulstruc}{\bl}{\bag}}_{\text{Case \nameref{state: tc bi appendix})}}
    \end{align*}

    \subproofcasethin{Case \nameref{state: tc self appendix})}
    Assuming that the above cases are passed for each "TC-formula",
    we can let 
    \[\fmlset[1] =
    (\set{\fmlset[1]_{\dir}'}_{\dir \in \set{\back 2, \back 1, 1, 2}},\     \set{\TC{\fml[3]_i}{u_i v_i} x_i y_i}_{i}),\]
    where for each $\dir \in \set{\back 2, \back 1, 1, 2}$, $\ainter(\FV(\fmlset[1]_{\dir}')) \subseteq \ND{\ynmulstruc}{\bag}{\dir}$, and for each $i$, we have:
    \begin{itemize}
        \item $\ainter(x_i) = \set{\dir}$ and $\ainter(y_i) = \set{\dir'}$ for some pairwise distinct $\dir, \dir' \in \set{\back 2, \back 1, 1, 2}$,
        \item $\TC{\fml[3]_i}{u_i v_i} x_i y_i \in \preclUNTC(\TC{\fml[3]_i}{u_i v_i} x_i y_i)$ occurs
        in the derivation tree of $\fmlset[1] \in \prebiclUNTC(\fml)$.
    \end{itemize}

    We then have:
    \begin{align*}
        &\mkstate{(\set{\fmlset[1]_{\dir}'}_{\dir \in \set{\back 2, \back 1, 1, 2}},\
        \set{\TC{\fml[3]_i}{u_i v_i} x_i y_i}_{i})}{\prione}{\ynmulstruc}{\ainter}{\bag}
        \\
        \leadstoUNTC_{\eqref{rule: tc-split-l}}^*
        &~(\set{\fmlset[1]_{\dir}'}_{\dir \in \set{\back 2, \back 1, 1, 2}}, \set{\underbrace{\TC{\fml[3]_i}{u_i v_i} x_i z_{i,1}, \fml[3]_i\fmlsubst{u_i v_i}{z_{i,1} z_{i,2}}}_{\text{$\in \preclUNTC(\TC{\fml[3]_i}{u_i v_i} x_i y_i)$}},
        \underbrace{\TC{\fml[3]_i}{u_i v_i} z_{i,2} y_i}_{\text{$\in \preclUNTC(\TC{\fml[3]_i}{u_i v_i} x_i y_i)$}}}_{i})\mkstate{}{\prione}{\ynmulstruc}{\bl}{\bag}\\
        \leadstoUNTC^*
        &~(\set{\fmlset[1]_{\dir}'}_{\dir \in \set{\back 2, \back 1, 1, 2}}, \set{\underbrace{\TC{\fml[3]_i}{u_i v_i} x_i z_{i,1}, \fmlset[2]_{i}'}_{\text{$\in \preclUNTC(\TC{\fml[3]_i}{u_i v_i} x_i y_i)$}},\ 
        \underbrace{\fmlset[3]_{i}', \TC{\fml[3]_i}{u_i v_i} z_{i,2} y_i}_{\text{$\in \preclUNTC(\TC{\fml[3]_i}{u_i v_i} x_i y_i)$}}}_{i})\mkstate{}{\prione}{\ynmulstruc}{\bl}{\bag} \tag{\Cref{lemma: split UNTC}}\\
        & \tag*{where $z_{i,1}$ and $z_{i,2}$ are "fresh" and $\fmlset[2]_i', \fmlset[3]_i' \in \preclUNTC(\fml[3]_i\fmlsubst{u_i v_i}{z_{i,1} z_{i,2}})$}\\
        \leadstoUNTC_{\eqref{rule: split}}^*
        &~\bigdoublewedge_{\dir \in \set{\back 2, \back 1, 1, 2}} \mkstate{(\fmlset[1]_{\dir}', \set{\fmlset[2]_{i, \dir}'', \fmlset[3]_{i, \dir}''}_{i})}{\prione}{\ynmulstruc}{\bl}{\bag}.
    \end{align*}
    Here,
    \begin{align*}
    \fmlset[2]_{i, \dir}''
    &\defeq \begin{cases}
        \set{\TC{\fml[3]_i}{u_i v_i} x_i z_{i,1}, \fmlset[2]_{i}'} & \text{if $\ainter(x_i) = \set{\dir}$}\\
        \emptyset & \text{otherwise},
        \end{cases}\\
    \fmlset[3]_{i, \dir}''
    &\defeq \begin{cases}
        \set{\fmlset[3]_{i}', \TC{\fml[3]_i}{u_i v_i} z_{i,2} y_i} & \text{if $\ainter(y_i) = \set{\dir}$}\\
        \emptyset & \text{otherwise}.
        \end{cases}
    \end{align*}
    Because $\ainter(x_i) \neq \ainter(y_i)$,
    either $\fmlset[2]_{i, \dir}''$ or $\fmlset[3]_{i, \dir}''$ is the empty set for every $i$ and $\dir$.
    Thus, $(\fmlset[1]_{\dir}', \set{\fmlset[2]_{i, \dir}'', \fmlset[3]_{i, \dir}''}_{i}) \in \prebiclUNTC(\fml)$.
    Hence, by IH, this case has been shown.
    After that, we can assume \nameref{cond: evaluation strategy UNTC proof 3}).

    \proofcasethin{Step 3}
    Finally, by the condition \nameref{cond: evaluation strategy UNTC proof 3}),
    we can show the remaining case in the same way as Steps \nameref{evaluation strategy UNFO fmlset split} and \nameref{evaluation strategy UNFO mov} in \Cref{lemma: parameter UNFO odd}.
    
    Hence, this completes the proof.
\end{proof}

\begin{figure*}[t]
    \centering
\begin{tcolorbox}[standard jigsaw, opacityback = 0.8, colframe=black!80, boxrule=.3mm, left=.2mm, right=.2mm]
    \vspace{-4ex}
    \begin{align*}
    \dots&\leadstoUNTC \dots \tag*{(all the rules for \kl{UNFO}, given in \Cref{figure: rules UNFO'})}\\
    \tup{(\fmlset[1], \TC{\fml[2]}{u v} {x} {y})^{\pri}_{\ainter, \chi}, \ynstruc}
    &\leadstoUNTC
    (\tup{(\fmlset[1], {x} \EQ {y})^{\pri}_{\ainter, \chi}, 0} \doublevee^{\pri}\\
    &\hspace{-7em} \tup{(\fmlset[1], \fml[2]\fmlsubst{u v}{z z'}, \TC{\fml[2]}{u v}{z' y})^{\pri}_{\ainter\aintersubst{z}{\set{\vertex}}\aintersubst{z'}{\chi(\ADBsymbol)}, \chi}, 0})
    \text{ if $\ainter(x) = \set{\vertex} \subseteq \univ{\ynstruc}$},
    \tag*{(TC-l)}\\
    \tup{(\fmlset[1], \TC{\fml[2]}{u v} {x} {y})^{\pri}_{\ainter, \chi}, \ynstruc}
    &\leadstoUNTC
    (\tup{(\fmlset[1], {x} \EQ {y})^{\pri}_{\ainter, \chi}, 0} \doublevee^{\pri}\\
    &\hspace{-7em} \tup{(\fmlset[1], \TC{\fml[2]}{u v} {x} {z}, \fml[2]\fmlsubst{u v}{z z'})^{\pri}_{\ainter\aintersubst{z}{\chi(\ADBsymbol)}\aintersubst{z'}{\set{\vertex}}, \chi}, 0})
    \text{ if $\ainter(y) = \set{\vertex} \subseteq \univ{\ynstruc}$},
    \tag*{(TC-r)}\\
    \tup{(\fmlset[1], \TC{\fml[2]}{u v} {x} {y})^{\pri}_{\ainter, \chi}, \ynstruc}
    &\leadstoUNTC\\
    &\hspace{-7em} \bigdoublevee_{\dir' \in \range{\back 2}{2} \setminus \set{\dir}}^{\pri}
    \tup{(\fmlset[1], \TC{\fml[2]}{u v} {x} {z}, \fml[2]\fmlsubst{u v}{z z'}, \TC{\fml[2]}{u v} {z}' {y})^{\pri}_{\ainter\aintersubst{z}{\chi(\NDsymbol_{\dir})}\aintersubst{z'}{\chi(\NDsymbol_{\dir'})}, \chi}, 0} \\
    & \hspace{-4em} \text{if $\ainter({x}) = \set{\dir} \subseteq \set{\back 2, \back 1, 1, 2}$ and
    $\ainter(y) = \set{\elem} \not\subseteq \chi(\NDsymbol_{\dir})$,}\tag*{(TC-split-l)}\\
    \tup{(\fmlset[1], \TC{\fml[2]}{u v} {x} {y})^{\pri}_{\ainter, \chi}, \ynstruc}
    &\leadstoUNTC\\
    &\hspace{-7em} \bigdoublevee_{\dir' \in \range{\back 2}{2} \setminus \set{\dir}}^{\pri}
    \tup{(\fmlset[1], \TC{\fml[2]}{u v} {x} {z}, \fml[2]\fmlsubst{u v}{z z'}, \TC{\fml[2]}{u v} {z}' {y})^{\pri}_{\ainter\aintersubst{z}{\chi(\NDsymbol_{\dir'})}\aintersubst{z'}{\chi(\NDsymbol_{\dir})}, \chi}, 0} \\
    & \hspace{-4em} \text{if $\ainter({y}) = \set{\dir} \subseteq \set{\back 2, \back 1, 1, 2}$ and
    $\ainter(x) = \set{\elem} \not\subseteq \chi(\NDsymbol_{\dir})$.}\tag*{(TC-split-r)}
    \end{align*}
\end{tcolorbox}     \caption{\kl{2APTA} transition rules for the "local checker" for \kl{UNTC} (\Cref{figure: rules UNTC}). Here, $z$ and $z'$ are "fresh" "variables".}
    \label{figure: rules UNTC'}
\end{figure*}

\subsection{Proof of {\Cref{proposition: closure size UNTC}}}\label{section: proposition: closure size UNTC}
\begin{proof}
    By easy induction on the \kl(formula){size} $\fmllen{\fml}$, similar to \Cref{proposition: closure size UNFO}.
    We write only the case of \kl{transitive closure formulas}.
    
    \proofcasethin{Case $\fml = [\fml[2]]_{u v}^* {x} {y}$}
    We have
    \begin{align*}
        \card\preclUNTC(\fml)/{\renameeq}
        &\le 1 + \card \preclUNTC({x} \EQ {y})/{\renameeq} + \card \preclUNTC({x} \EQ {z})/{\renameeq}
        + \card \preclUNTC({z} \EQ {y})/{\renameeq} +  \card \preclUNTC({z} \EQ {z}')/{\renameeq} \\
        &\hspace{2em} + 5 \card\preclUNTC(\fml[2][{z}_1{z}_2/u v])/{\renameeq}\\
        &\le 1 + 4 (2 \fmllen{{x} = {y}})^{2 \fmllen{{x} = {y}}} + 5 (2 \fmllen{\fml[2]})^{2 \fmllen{\fml[2]}} \tag{IH}\\ 
        &\le 1 + (2 \fmllen{{x} = {y}})^{2 \fmllen{\fml[1]}} + (2 \fmllen{\fml[2]})^{2 \fmllen{\fml[1]}}\\
        &\le (2 (1 + \fmllen{{x} = {y}} + \fmllen{\fml[2]}))^{2 \fmllen{\fml[1]}}
        \le (2 \fmllen{\fml})^{2 \fmllen{\fml}}. \qedhere
    \end{align*}
\end{proof}

\subsection{Proof of {\Cref{corollary: 2-EXPTIME UNTC}}: 2APTA construction for UNTC}\label{section: 2APTA construction UNTC}

In this section, from the "local checker" for \kl{UNTC}, we construct \kl{2APTAs}.
\begin{definition}\label{definition: 2APTA UNTC}
    For $k \in \pnat$ and a \kl{UNTC} \kl{sentence} $\fml$,
    the \kl{2APTA} $\automaton[1]_{k}^{\fml} = \tup{Q, \delta, \Pri, q_0}$ over $\allatstruc_{k}$
    is defined in the same way as the definition of \Cref{definition: 2APTA UNFO} where
    the relation $(\leadsto) \subseteq (Q \times \allatstruc_{k}) \times \PBallfml(Q \times \range{\back 2}{2})$
    is defined as the minimal binary relation closed under the rules in \Cref{figure: rules UNTC'}.
\end{definition}
This \kl{2APTA} satisfies the following.
\begin{proposition}\label{proposition: 2APTA UNTC}
    Let $k \in \pnat$ and $\fml$ be a \kl{UNTC} \kl{sentence}.
    For every $\allatstruc_{k}$-labeled \kl{binary} \kl{tree} $\ynmulstruc$,
    we have:
    \[\vdashUNTCsub_{\clUNTC(\fml)} (\fml)^{1, \ynmulstruc}_{\emptyset, \eps} \quad\Longleftrightarrow\quad \ynmulstruc \in \automatonlang(\automaton[1]_{k}^{\fml}).\]
\end{proposition}
\begin{proof}[Proof Sketch]
    Similar to \Cref{section: 2APTA construction UNFO}.
    The rules in \Cref{figure: rules UNTC'} are the same as those in \Cref{figure: rules UNTC}.
    Thus, we can convert an \kl{accepting} \kl{run} of $\vdashUNTCsub_{\clUNTC(\fml)} (\fml)^{1, \ynmulstruc}_{\emptyset, \eps}$
    into an \kl{accepting} \kl{run} witnessing $\ynmulstruc \in \automatonlang(\automaton[1]_{k}^{\fml})$, and vice versa.
\end{proof}
\begin{proof}[Proof of \Cref{corollary: 2-EXPTIME UNTC}]
    For all \kl{UNTC} \kl{sentences} $\fml$, we have:
    \begin{align*}
    &\text{$\fml$ is \kl{satisfiable}}\\
    &~\Longleftrightarrow~
    \text{$\fml$ is \kl{satisfiable} in a \kl{structure} of \kl{treewidth} at most $\fmllen{\fml}-1$} \tag{By \Cref{proposition: tw}}\\
    &~\Longleftrightarrow~
    \text{$\glue \ynmulstruc \modelsnonass \fml$ for some $\allatstruc_{\fmllen{\fml}}$-labeled \kl{binary} \kl{tree} $\ynmulstruc$} \tag{By \Cref{section: abstract semantics on tree decompositions}}\\
    &~\Longleftrightarrow~ \text{$\absmodels \mkstate{(\fml)}{\prione}{\ynmulstruc}{\emptyset}{\eps}$ for some $\allatstruc_{\fmllen{\fml}}$-labeled \kl{binary} \kl{tree} $\ynmulstruc$} \tag{By \Cref{proposition: standard semantics and semantics on tree decompositions}}\\
    &~\Longleftrightarrow~ \text{$\vdashUNTCsub_{\clUNTC(\fml)} \mkstate{(\fml)}{\prione}{\ynmulstruc}{\emptyset}{\eps}$ for some $\allatstruc_{\fmllen{\fml}}$-labeled \kl{binary} \kl{tree} $\ynmulstruc$} \tag{By \Cref{theorem: completeness closure UNTC}}\\
    &~\Longleftrightarrow~ \text{$\automatonlang(\automaton[1]^{\fml}_{\fmllen{\fml}}) \neq \emptyset$} \tag{By \Cref{proposition: 2APTA UNTC}}.
    \end{align*}

    On the \kl(2APTA){size}, by restricting $\Rels$ to those occurring in $\fml$,
    \begin{itemize}
        \item the alphabet size $\card \allatstruc_{\fmllen{\fml}}$ is $2^{\mathcal{O}(\fmllen{\fml} \log \fmllen{\fml})}$,

        \item the number of \kl{states} is $2^{\poly(\fmllen{\fml})}$,
        
        \item (from the two above) the transition size is $2^{\poly(\fmllen{\fml})}$.
    \end{itemize}
    Hence, by \Cref{proposition: 2APTA complexity}, the \kl{satisfiability problem} for \kl{UNTC} is in "2ExpTime".
\end{proof}

\end{scope}
\end{toappendix}
\section{Model Checking}\label{section: model checking}
In this section, we study the combined complexity of the ""model checking problem"":
given a "sentence@@TC" $\fml$ and a "structure" $\ynstruc$, to decide whether $\ynstruc \modelsnonass \fml$.
Here, the size of a "structure" $\ynstruc$ is defined as the sum of the number of "domain elements" and the number of "tuples".
For "UNFO" and "GNFO", the following complexity results are known.
Here, ""PNPlog2"" is the class of decision problems solvable by a $\mathsf{P}$ machine with $\mathcal{O}(\log^2 n)$ accesses to an $\mathsf{NP}$ oracle, where $n$ is the input length.
\begin{proposition}[\cite{segoufinUnaryNegation2013,baranyGuardedNegation2015}]\label{proposition: model checking UNFO GNFO}
    The "model checking problem" is "PNPlog2"-complete
    for both "UNFO" \cite{segoufinUnaryNegation2013} and "GNFO" \cite{baranyGuardedNegation2015}.
\end{proposition}
However, the precise complexity was left open for "UNTC" (and thus also for "GNTC") \cite[Open Problem 2 and Footnote 16 of version v1]{figueiraCommonAncestorPDL2025}.
Nevertheless, the above complexity bounds continue to hold even if we extend the logics with guarded "TC-formulas".%
\footnote{"Cf", "PNP"-complete for "LFP" and even the \intro*\kl{alternation-free} fragment of \intro*\kl{UNFP} \cite[Theorem 5.5]{segoufinUnaryNegation2013}.}

The basic idea is the reduction \cite{segoufinUnaryNegation2013} from "UNFO" to ``Tree Block Satisfaction'' \intro*\TBSAT{1}.
The problem \TBSAT{q} is in "PNPlog2" for every fixed $q \ge 1$ \cite[Corollary 3.5]{schnoebelenOracleCircuitsBranchingTime2003}.
"Transitive closure@transitive closure" can still be handled by using \TBSAT{2} instead of \TBSAT{1}.
Building on this idea, we can extend the reduction to "UNTC".
For "GNTC", we can reduce the "model checking problem" to "that@model checking problem" of "UNTC" by adding a new "domain element" for each tuple, "cf", \cite[Theorem 5.1]{baranyGuardedNegation2015} for "GNFO" and "GNFP".\footnote{%
Recall that the size of the "structure" is defined as the number of "domain elements" and the number of "tuples".
Note that the "domain@@structure" of the resulting "structure" is exponentially larger than "that@domain@structure" of the original "structure" when the maximal "arity" of "atoms" is unbounded.
In contrast, for "UNTC" (and also $\WHERE{q, k}$), the "model checking problem" is "PNPlog2"-complete even when the size is defined as the number of "domain elements".
}
\ifthenelse{\boolean{arXiv}}{\begin{theorem}[{\Cref{section: theorem: model checking}}]}{\begin{theorem}}\label{theorem: model checking}
\gdef\modelchecking{
The "model checking problem" is
"PNPlog2"-complete for both "UNTC" and "GNTC".
}
\modelchecking
\end{theorem}
As a corollary,
the "model checking problem" for "UNFOreg" is also "PNPlog2"-complete,
which provides a revised proof of \cite[Theorem 18]{jungQueryingUnaryNegation2018}; see also \cite[Footnote 16 of version v1]{figueiraCommonAncestorPDL2025} for a report on the error of the proof presented in \cite[Theorem 18]{jungQueryingUnaryNegation2018}.

In our proof, we introduce a useful logic $\WHERE{q, k}$.
The syntax of $\intro*\WHERE{q, k}$ \intro*\kl(where){formulas} is given by the following grammar:
\begin{align*}
    \hat{\fml} &~\Coloneqq~ \fml[2] ~\textsf{where}~ (X_1 x_{1, 1} \dots x_{1,\breve{k}_1} \leftarrow  \hat{\fml}_1,\;  \dots,\; X_m x_{m, 1} \dots x_{m,\breve{k}_m} \leftarrow  \hat{\fml}_m),
\end{align*}
where
\begin{itemize}
    \item $\fml[2]$ is an \reintro*\kl{EFO formula} (\reintro*\kl{existential FO formula}) over $\Rels \dcup \set{X_1, \dots, X_m}$ such that
    \begin{itemize}
        \item \emph{$X_j$ occurs at most $q$ times} for each $j \in \1{m}$,
    \end{itemize}
    \item $m \ge 0$, $X_1, \dots, X_m$ are pairwise distinct, and for each $j \in \1{m}$,
    \begin{itemize}
        \item $\breve{k}_j \le k$, and
        \item %
        $\FV(\hat{\fml}_j) \subseteq \set{x_{j,1}, \dots, x_{j, \breve{k}_{j}}}$ 
        and the "variables" $x_{j,1}, \dots, x_{j, \breve{k}_j}$ are pairwise distinct,
    where $\hat{\fml}_j$ is a \WHERE{q, k} \kl(where){formula}.
    \end{itemize}
\end{itemize}
The "arity" of $X_i$ is bounded by $k$,
while the "arity" of $\rsym \in \Rels$ is not bounded.

Each \WHERE{q, k} \kl(where){formula} $\hat{\fml}$ naturally expresses an \reintro*\kl{FO formula} $\WHEREtoFO{\hat{\fml}}$,
more precisely,
\AP
$\intro*\WHEREtoFO{\hat{\fml}}$ is defined as the \kl(where){formula} $\fml[2]$ in which each occurrence $X_j \mul{z}$ has been replaced with $\WHEREtoFO{\hat{\fml[1]}_j}\fmlsubst{\mul{y}_j}{\mul{z}}$ (possibly with renaming bound "variables" in $\WHEREtoFO{\hat{\fml[1]}_j}$ to avoid naming conflicts).

For instance, if $\hat{\fml}$ is the following \WHERE{2, 3} \kl(where){formula}
\begin{align*}
&(\exists z~X x z \land X z y) \land Y w~\textsf{where}~(\\
&\quad X x y \leftarrow x = y \lor E x y ~\textsf{where}~(),\\
&\quad Y x \leftarrow \exists w \lnot Z x x w ~\textsf{where}~(Z x y z \leftarrow E y z)\\
&),
\end{align*}
then $\WHEREtoFO{\hat{\fml}}$ is $(\exists z ((x = z \lor Exz) \land (z = y \lor Ezy))) \land \exists w' \lnot E w w'$.
The satisfaction relation $\ynstruc \modelsass{\inter} \hat{\fml}$ is defined as $\ynstruc \modelsass{\inter} \WHEREtoFO{\hat{\fml}}$, using the standard "FO" satisfaction relation.

The "model checking problem" for $\WHERE{q, k}$ can be naturally encoded by \TBSAT{q},
and hence we can show the "problem@model checking problem" is "PNPlog2"-complete.
\ifthenelse{\boolean{arXiv}}{\begin{lemma}[{\Cref{section: lemma: WHERE}}]}{\begin{lemma}}
\label{lemma: WHERE}
\gdef\modelcheckingWHERE{%
For each fixed integer $q, k \ge 1$,
the "model checking problem" for $\WHERE{q, k}$ is "PNPlog2"-complete.
}
\modelcheckingWHERE
\end{lemma}
The "model checking problem" for "UNTC" can be reduced to the "model checking problem" for $\WHERE{2, 2}$ in polynomial time, and thus is "PNPlog2"-complete.

Our approach also shows that, in \Cref{theorem: model checking},
the condition of "unary negation" can be generalized to $\kappa$-ary negation for any fixed $\kappa$ ("cf", \Cref{footnote: binary negation}).
\AP
\phantomintro{genUNTC formula}\phantomintro{genUNTC sentence}
We say that a "formula@@TC" $\fml$ (in "TC") is $\tup{\kappa, \lambda, \xi}$-""genUNTC"" if
(i)
each subformula of the form $\lnot \fml[2]$ has at most $\kappa$ free variables: $\card\FV(\fml[2]) \le \kappa$, 
(ii)
for each subformula of the form $\TC{\fml[2]}{\mul{u}} \mul{x}$,
the "arity" is at most $\lambda$: $\card \toSet{\mul{u}} \le 2 \lambda$, and
(iii)
each subformula of the form $\TC{\fml[2]}{\mul{u}} \mul{x}$
has at most $\xi$ parameters: $\card(\FV(\fml[2]) \setminus \toSet{\mul{u}}) \le \xi$.

\AP
\phantomintro{genGNTC formula}\phantomintro{genGNTC sentence}
We say that a "formula@@TC" $\fml$ (in "TC") is $\tup{\kappa, \lambda, \xi}$-""genGNTC"" if 
(i) each negation appears in the form of $(\bigwedge_{i=1}^{\kappa} \afml_{i}) \land \lnot \fml[2]$ where $\FV(\fml[2]) \subseteq \FV(\bigwedge_{i=1}^{\kappa} \afml_{i})$,
(ii) each TC appears in the form of $[(\bigwedge_{i=1}^{\lambda} \afml[1]_{i}) \land (\bigwedge_{i=1}^{\lambda} \afml[2]_{i}) \land \fml[2]]^*_{\mul{u}, \mul{v}} \mul{x} \mul{y}$
where $\toSet{\mul{u}} \subseteq \FV(\bigwedge_{i=1}^{\lambda} \afml[1]_{i})$ and $\toSet{\mul{v}} \subseteq \FV(\bigwedge_{i=1}^{\lambda} \afml[2]_{i})$, and
(iii)
each subformula of the form $\TC{\fml[2]}{\mul{u}} \mul{x}$
has at most $\xi$ parameters: $\card(\FV(\fml[2]) \setminus \toSet{\mul{u}}) \le \xi$.
Note that 
$\tup{1, 1, 0}$-"genUNTC" coincides with "UNTC" and
$\tup{1, 1, 0}$-"genGNTC" coincides with "GNTC", respectively.
Additionally, $\tup{\kappa, \lambda, \xi}$-"genGNTC" has at least as much expressive power as $\tup{\kappa, \lambda, \xi}$-"genUNTC" for every $\kappa, \lambda, \xi \ge 0$.

Similar to \Cref{theorem: model checking},
the "model checking problem" for $\tup{\kappa, \lambda, \xi}$-"genUNTC" can be reduced to "that@model checking problem" for $\WHERE{2, \max(\kappa, 2\lambda + \xi)}$ in polynomial time, and hence we can show the problem is "PNPlog2"-complete.
\ifthenelse{\boolean{arXiv}}{\begin{theorem}[\Cref{section: theorem: model checking kFOTC kGNTC}]}{\begin{theorem}}\label{theorem: model checking kFOTC kGNTC}
\gdef\modelcheckingkFOTCkGNTC{%
For every fixed integer $\kappa \ge 1$ and $\lambda, \xi \ge 0$,
the "model checking problem" is "PNPlog2"-complete
for both $\tup{\kappa, \lambda, \xi}$-"genUNTC" and  $\tup{\kappa, \lambda, \xi}$-"genGNTC".
}
\modelcheckingkFOTCkGNTC
\end{theorem}

\begin{toappendix}

\subsection{Preliminaries for \Cref{section: lemma: WHERE}: Tree Block Satisfaction}
In this section, we recall the definition of ``Tree Block Satisfaction'' \TBSAT{q} \cite{schnoebelenOracleCircuitsBranchingTime2003}.

For $w \ge 1$ and $n \in \range{0}{2^{w} - 1}$,
we write $\intro*\mkbin{w}{n}$ to denote the binary encoding of $n$ as a sequence over $\set{\const{0}, \const{1}}$ of length $w$ ("eg", $\mkbin{4}{5} = \const{0}\const{1}\const{0}\const{1}$).

For $q \ge 1$ and $k \ge 1$,
a \intro*\kl{TB-tree} (without inputs) of type $q \times \mathrm{M}$ and \intro*\kl(TB-tree){width} $k$ is a \kl{tree} consisting of \intro*\kl(TB-tree){blocks},
where each \kl(TB-tree){block} is a kind of Boolean circuit having $k$ output gates and having $k$ input gates for each of its children, see \Cref{figure: TB}.

\begin{figure}[h]
\centering
\begin{tikzpicture}[every node/.style={font=\small}, node distance=1cm]
\draw[thick, fill=blue!10] (0,0) -- (8.5,0) -- (8.0,2.5) -- (0.5,2.5) -- cycle;
\node[above] (z1) at (1.9,2.7) {$z_0$};
\node[right=.8cm of z1] (z2) {$z_1$};
\node[right=.8cm of z2] (dots_z) {$\cdots$};
\node[right=.8cm of dots_z] (zk) {$z_{k-1}$};
\node[below=.5cm of z1, draw, dashed, fill=white, minimum width=0.8cm, minimum height=0.6cm] (chi1) {$\chi_0$};
\node[below=.5cm of z2, draw, dashed, fill=white, minimum width=0.8cm, minimum height=0.6cm] (chi2) {$\chi_1$};
\node[below=.5cm of dots_z] (dots_chi) {$\cdots$};
\node[below=.5cm of zk, draw, dashed, fill=white, minimum width=0.8cm, minimum height=0.6cm] (chik) {$\chi_{k-1}$};
\draw[thick] ($(z1.south)$) -- (chi1.north);
\draw[thick] ($(z2.south)$) -- (chi2.north);
\draw[thick] ($(zk.south)$) -- (chik.north);

\node (y1) at (1.1, .5) {$y_0^{(1)} \cdots y_{k-1}^{(1)}$};
\node[right=.6cm of y1] (y2) {$y_0^{(2)} \cdots y_{k-1}^{(2)}$};
\node[right=.6cm of y2] (dots_y) {$\dots$};
\node[right=.6cm of dots_y] (yn) {$y_0^{(m)} \cdots y_{k-1}^{(m)}$};

\draw[fill = blue!10] ($(y1.south) + (-0.5,-.5)$) -- ($(y1.south) + (-0.8,-1.)$) -- ($(y1.south) + (0.8,-1.)$) -- ($(y1.south) + (0.5,-.5)$) -- cycle;
\draw[fill = blue!10] ($(y2.south) + (-0.5,-.5)$) -- ($(y2.south) + (-0.8,-1.)$) -- ($(y2.south) + (0.8,-1.)$) -- ($(y2.south) + (0.5,-.5)$) -- cycle;
\draw[fill = blue!10] ($(yn.south) + (-0.5,-.5)$) -- ($(yn.south) + (-0.8,-1.)$) -- ($(yn.south) + (0.8,-1.)$) -- ($(yn.south) + (0.5,-.5)$) -- cycle;
\draw[thick] ($(y1.south)+ (-.4,0)$) -- ($(y1.south) + (-.4,-.5)$);
\draw[thick] ($(y1.south)+ (-.2,0)$) -- ($(y1.south) + (-.2,-.5)$);
\node[below right = 0.15cm and -.2cm of y1.south] {$\cdots$};
\draw[thick] ($(y1.south)+ (.4,0)$) -- ($(y1.south) + (.4,-.5)$);
\draw[thick] ($(y2.south)+ (-.4,0)$) -- ($(y2.south) + (-.4,-.5)$);
\draw[thick] ($(y2.south)+ (-.2,0)$) -- ($(y2.south) + (-.2,-.5)$);
\node[below right = 0.15cm and -.2cm of y2.south] {$\cdots$};
\draw[thick] ($(y2.south)+ (.4,0)$) -- ($(y2.south) + (.4,-.5)$);
\draw[thick] ($(yn.south)+ (-.4,0)$) -- ($(yn.south) + (-.4,-.5)$);
\draw[thick] ($(yn.south)+ (-.2,0)$) -- ($(yn.south) + (-.2,-.5)$);
\node[below right = 0.15cm and -.2cm of yn.south] {$\cdots$};
\draw[thick] ($(yn.south)+ (.4,0)$) -- ($(yn.south) + (.4,-.5)$);
\node[below=0.6cm of dots_y.south] {$\cdots$};
\end{tikzpicture}
 \caption{A \kl(TB-tree){block} with $m$ children in a \kl{TB-tree} of \kl(TB-tree){width} $k$.}
\label{figure: TB}
\end{figure}

The $i$-th output ($z_i$ in \Cref{figure: TB}) of a \kl(TB-tree){block} is the \intro*\kl{truth value} ($\mathtt{0}$/$\mathtt{false}$ or $\mathtt{1}$/$\mathtt{true}$) defined in terms of the input gates by means of an existentially quantified Boolean formula $\chi_i$ of the form:
\[\exists \mul{b}_1 c_1 \dots \mul{b}_M c_M \mul{d}\left(\bigwedge_{\ell = 1}^{M} c_{\ell} = \const{in}^{(i_{\ell})}(\mul{b}_{\ell}) \land \fml[2]\right)\]
where
\begin{itemize}
\item for each $\ell \in \1{M}$,
\begin{itemize}
    \newcommand{\enc}{\const{enc}}
    \item $i_{\ell} \in \1{m}$,
    \item $\mul{b}_{\ell}$ is a \kl{Boolean variable} sequence of length $\ceil{\log_2 k}$,
    \item $\const{in}^{(i_{\ell})}(\mul{b}_{\ell})$ represents the \kl{truth value} of the $(\enc(\mul{b}_{\ell}))$-th bit of the $i_{\ell}$-th child \kl(TB-tree){block} ($y_{\enc(\mul{b}_{\ell})}^{(i_{\ell})}$ in \Cref{figure: TB})
    where $\enc(\mul{b}_{\ell})$ is the integer such that $\mul{b}_{\ell} = \mkbin{\ceil{\log_2 k}}{\enc(\mul{b}_{\ell})}$,
\end{itemize}
\item $\mul{d}$ is a \kl{Boolean variable} sequence of arbitrary length,
\item $\fml[2]$ is a \intro*\kl{Boolean formula} using any of the existentially quantified \kl{Boolean variables}, and
\item $\chi_i$ only uses $q$ bits from each input vector: $\card \set{\ell \in \1{M} \mid j = i_\ell} \le q$ for each $j \in \1{m}$.
\end{itemize}

\TBSAT{q} is the following problem: given a \kl{TB-tree} of type $q \times \mathrm{M}$ (where $q$ is fixed and the \kl(TB-tree){width} $k$ is arbitrary), does the first ($0$-th) output of the root \kl(TB-tree){block} have a $\mathtt{1}$?

\begin{proposition}[{\cite[Corollary 3.5]{schnoebelenOracleCircuitsBranchingTime2003}}]\label{proposition: TBSAT}
    For each fixed $q \ge 1$,
    \TBSAT{q} is "PNPlog2"-complete.
\end{proposition}

\subsection{Proof of \Cref{lemma: WHERE} %
}\label{section: lemma: WHERE}
In this section, we prove \Cref{lemma: WHERE}.
\begin{theorem*}[Restatement of \Cref{lemma: WHERE}]
\modelcheckingWHERE
\end{theorem*}

\begin{proof}[Proof of \Cref{lemma: WHERE}]
\proofcasethin{(Lower bound)}
By the lower bound for "UNFO" \cite{segoufinUnaryNegation2013}.
(We can reduce the "model checking problem" for "UNFO" to \kl[model checking problem]{that} for $\WHERE{1, 1}$.)

\proofcasethin{(Upper bound)}
We give a polynomial-time reduction from this problem to \TBSAT{q}.

Let $\ynstruc$ be a \kl{structure}.
Let $N \defeq \card\univ{\ynstruc}$ be the \kl{cardinality} of the \kl(structure){domain} of $\ynstruc$.
"Wlog", we can assume $\univ{\ynstruc} = \range{0}{N-1}$, by taking an \kl{isomorphic} \kl{structure}.
"Wlog", we can assume $N = 2^{L}$ for some $L \ge 0$,
by padding $\ynstruc$ with "fresh" \kl{elements},
extending $\ynstruc$ with a \kl{fresh} unary \kl{relation name} $U$ such that $\interpatom{\ynstruc}{U}$ is the set of all original \kl{elements}, and
replacing each "subformula" $\exists x \fml[2]$ with $\exists x (Ux \land \fml[2])$.
This transformation is done in polynomial-time.
We thus also let $L \defeq \log_2 N$.

For each $\WHERE{q, k}$ \kl(where){formula} $\hat{\fml}$ and $k$ distinct \kl{variables} $x_1, \dots, x_{k}$ "st" $\FV(\hat{\fml}) \subseteq \set{x_1, \dots, x_{k}}$,
we construct a \kl{TB-tree} $T^{\ynstruc}_{\hat{\fml}, x_1, \dots, x_{k}}$ of \kl(TB-tree){width} $N^{k}$,
so that for every $i_1, \dots, i_{k} \in \univ{\ynstruc}$,
\begin{align*}
&\ynstruc \modelsass{x_1 \dots x_{k} \mapsto i_1 \dots i_{k}} \hat{\fml} \quad\text{"iff"}\quad\text{the $(\sum_{j=1}^{k} i_j N^{j-1})$-th output of $T^{\ynstruc}_{\hat{\fml}, x_1, \dots, x_{k}}$ is $\mathtt{1}$.}
\end{align*}

By renaming "variables",
"wlog", using a sequence $x_1, \dots, x_k$ of pairwise distinct "variables",
we can assume that $\hat{\fml}$ is of the following form:
\[\fml[2] ~\textsf{where}~(X_1 x_1 \dots x_{k_1} \leftarrow  \hat{\fml}_1,\;  \dots,\; X_m x_{1} \dots x_{k_m} \leftarrow  \hat{\fml}_m).\]

By taking its prenex normal form, "wlog",
we can assume that $\fml[2]$ is of the form: $\exists z_1 z_2 \dots z_{\breve{n}} \fml[2]'$,%
where $\fml[2]'$ is built from \kl{atoms} over $\Rels \dcup \set{X_1, \dots, X_m}$ using $\land, \lor, \lnot$ (without $\exists$)
and $x_1, \dots, x_{k}, z_1, \dots, z_{\breve{n}}$ are pairwise distinct.
Let $M$ be the number of "atoms" over the set $\set{X_1, \dots, X_m}$ occurring in $\fml[2]$
and let $\fml[3]_{\ell} = X_{I_{\ell}} \bl$ be the $\ell$-th "atom" over $\set{X_1, \dots, X_m}$ for each $\ell \in \1{M}$.
By introducing a \kl{fresh} \kl{variable} $z'$ and replacing $\fml[3]_\ell$ with $\fml[3]_\ell\fmlsubst{z}{z'} \land z \EQ z'$,
"wlog", we can assume that
$\FV(\fml[3]_{\ell}) \subseteq \set{z'_{k(\ell-1)+1}, \dots, z'_{k \ell}}$ for each $\ell \in \1{M}$ and 
$z'_1, \dots, z'_{k M}$ are the first $k M$ \kl{variables} from the sequence $z_1 \dots z_{\breve{n}}$.

We construct the "TB-tree" $T^{\ynstruc}_{\hat{\fml}, x_1 \dots, x_{k}}$ %
from
the "TB-trees" $\set{T^{\ynstruc}_{\hat{\fml[1]}_j, x_{1} \dots x_{k}}}_{j = 1}^{m}$. %
When $m = 0$, we are in the base case of the construction.
Otherwise, it is defined by adding a new \kl{root} \kl(TB-tree){block} whose children are the \kl{roots} of $T^{\ynstruc}_{\hat{\fml[1]}_j, x_{1} \dots x_{k}}$,
and whose $(\sum_{j=1}^{k} i_j N^{j-1})$-th output is defined by the \kl(EFO){formula}:%
\footnote{We use $\mul{b} = \mul{b}'$ to denote the "Boolean formula" $\bigwedge_{i} \mul{b}\ith{i} = \mul{b}'\ith{i}$,
where $b = b'$ denotes the "Boolean formula" $(b \land b') \lor (\lnot b \land \lnot b')$ as usual.}
\begin{align*}
    &\chi_{\sum_{j=1}^{k} i_j N^{j-1}}
    \defeq \exists \mul{b}_1' c_1 \dots \mul{b}_{M}' c_{M} \mul{b}_1 \dots \mul{b}_{\breve{n}}
    \Big(
    \bigwedge_{\ell = 1}^{M} c_\ell = \const{in}^{(I_{\ell})}(\mul{b}_\ell') \land \bigwedge_{\ell = 1}^{M} {\mul{b}'_{\ell} = \mul{b}_{k(\ell-1)+1} \dots \mul{b}_{k\ell}} \land \chi_{\ynstruc} \Big), \span
\end{align*}
where
each $\mul{b}_i$ is of length $L$ and the "Boolean formula" $\chi_{\ynstruc}$ is obtained from $\hat{\fml}$ as follows:
\begin{itemize}
    \item each $\fml[3]_\ell$ is replaced with $c_{\ell}$,

    \item each \kl{atom} $\rsym v_{1} \dots v_{\breve{k}}$, where $\rsym \in \Rels$, is
    replaced with a \kl{Boolean formula} enumerating all tuples in $\interpatom{\ynstruc}{\rsym}$:
    \begin{align*}
    & \bigvee_{\tup{\breve{i}_1, \dots, \breve{i}_{\breve{k}}} \in \interpatom{\ynstruc}{\rsym}}
    (v_1^{\star} = \mkbin{L}{\breve{i}_1} \land \dots \land v_{\breve{k}}^{\star} = \mkbin{L}{\breve{i}_{\breve{k}}}),
    \text{where $v^{\star} \defeq \begin{cases}
        \mkbin{L}{i_{j}} & \text{if $v = x_{j}$},\\
        \mul{b}_i & \text{if $v = z_i$}.
    \end{cases}$}
    \end{align*}
\end{itemize}
For each $j$, since each $X_j$ occurs at most $q$ times in $\hat{\fml[1]}$, each ``$\const{in}^{(j)}$'' also occurs at most $q$ times.
Hence, the \kl(EFO){formula} only uses $q$ bits from each input vector.
By \Cref{proposition: TBSAT}, this completes the proof.
\end{proof}

\subsection{Proof of {\Cref{theorem: model checking}}}\label{section: theorem: model checking}
In this section, we prove the following theorem.
\begin{theorem*}[Restatement of \Cref{theorem: model checking}]
\modelchecking
\end{theorem*}

\subsubsection{For \texorpdfstring{"UNSQ"}{UNSQ}: \texorpdfstring{"UNFO"}{UNFO} with \texorpdfstring{"squaring"}{squaring}}
We first show that "UNFO" with "squaring" ("UNSQ") is "PNPlog2"-complete.
\AP
\phantomintro{UNSQ formula}\phantomintro{UNSQ sentence}%
""UNSQ"" is given by the following grammar:
\begin{align*}
    \fml[1], \fml[2] &~\Coloneqq~ \rsym \mul{x} \mid \fml[1] \land \fml[2] \mid \fml[1] \lor \fml[2] \mid \exists x \fml
    \mid \lnot \fml[1] \mid \SQ{\fml[1]}{u v}x y.
\end{align*}
Here, 1) $\rsym \in \Rels \dcup \set{=}$, 2) in the form of the "negation formula" $\lnot \fml[1]$, the number of "free variables" is at most one: $\# \FV(\fml[1]) \le 1$,
and 3) in the form of the ""squaring formula"" $\intro*\SQ{\fml[1]}{u v}x y$, 
there are no parameters: $\FV(\fml[1]) \subseteq \set{u, v}$.

$\SQ{\fml[1]}{u v}$ denotes the "squaring" of $[\fml[1]]_{u v}$.
The semantics is precisely given as follows (it can be generalized for $k$-ary):
\begin{align*}
\ynstruc \modelsass{\inter} \SQ{\fml[2]}{\mul{u} \mul{v}} x_1 \dots x_k y_1 \dots y_k \quad\defiff\quad \mbox{there is some $\mul{a}$ s.t.\ } \\
    \ynstruc \modelsass{\inter\intersubst{\mul{u} \mul{v}}{\inter(x_1) \dots \inter(x_k) \mul{a}}} \fml[2] \text{ and }
    \ynstruc \modelsass{\inter\intersubst{\mul{u} \mul{v}}{\mul{a} \inter(y_1) \dots \inter(y_k)}} \fml[2].
\end{align*}

\begin{lemma}\label{lemma: model checking UNSQ}
    For each fixed $n \ge 1$,
    the "model checking problem" for "UNSQ" is "PNPlog2"-complete.
\end{lemma}
\begin{proof}
\proofcasethin{(Lower bound)}
By the lower bound for "UNFO" \cite{segoufinUnaryNegation2013}.

\proofcasethin{(Upper bound)}
We give a polynomial-time reduction to the "model checking problem" for \WHERE{2, 2}.

We define a truth preserving translation from "UNSQ formulas" of at most two free variables to \WHERE{2, 2} \kl(where){formulas}. 
For an "UNSQ formula" $\fml$ and a length-$2$ sequence of distinct "variables" $\mul{z} = z_1 z_2$ such that $\FV(\fml) \subseteq \set{z_1, z_2}$,
the \WHERE{2, 2} \kl(where){formula} $\intro*\UNTCtoWHERE(\fml, \mul{z})$ is inductively defined as follows where $X^{(\fml, \mul{z})}$ is pairwise distinct for each $\fml$ and $\mul{z}$:

\proofcasethin{Case $\fml = \lnot \fml[2]$}
\begin{align*}
&\UNTCtoWHERE(\lnot \fml[2], \mul{z}) \quad\defeq\quad
\lnot X^{(\fml[2], \mul{z})} \mul{z} \mbox{\hspace{.5em}\textsf{where}\hspace{.5em}} (X^{(\fml[2], \mul{z})} \mul{z} \leftarrow \UNTCtoWHERE(\fml[2],\mul{z})).
\end{align*}

\proofcasethin{Case $\fml = \SQ{\fml[2]}{u v} x y$}
\begin{align*}
&\UNTCtoWHERE(\SQ{\fml[2]}{u v} x y, \mul{z}) \quad\defeq\quad \exists z' X^{(\fml[2], u v)} x z' \land X^{(\fml[2], u v)} z' y
\mbox{\hspace{.5em}\textsf{where}\hspace{.5em}} (X^{(\fml[2], u v)} u v \leftarrow \UNTCtoWHERE(\fml[2],u v)),
\end{align*}
where $z'$ is a "fresh" "variable".

\proofcasethin{Otherwise}
$\fml$ is built from \kl{atoms} and \kl(UNSQ){formulas} having at most two \kl{free variables} using $\land$, $\lor$, and $\exists$.
Let $\fml[2]_1, \dots, \fml[2]_M$ be the maximal strict \kl{subformulas} of $\fml$ having at most two \kl{free variables}.
For each $\ell \in \1{M}$,
let $\mul{z}_{\ell}$ be a length-$2$ sequence of distinct "variables" such that $\FV(\fml[2]_{\ell}) \subseteq \toSet{\mul{z}_{\ell}}$.
Then,
\begin{align*}
& \UNTCtoWHERE(\fml, \mul{z}) \quad\defeq\quad \fml' \mbox{\hspace{.5em}\textsf{where}\hspace{.5em}} (X^{(\fml[2]_1,\mul{z}_1)} \mul{z}_1 \leftarrow \UNTCtoWHERE(\fml[2]_1,\mul{z}_1),\; \dots,\; X^{(\fml[2]_M,\mul{z}_M)} \mul{z}_M \leftarrow \UNTCtoWHERE(\fml[2]_M, \mul{z}_M)),
\end{align*}
where $\fml'$ is the \kl(EFO){formula} $\fml$ in which 
each $\fml[2]_{\ell}$ is replaced with $X^{(\fml[2]_{\ell}, \mul{z}_{\ell})} \mul{z}_{\ell}$.

We then have that $\fml[1]$ and $\UNTCtoWHERE(\fml[1], \mul{z})$ are "semantically equivalent",
by easy induction on $\fml[1]$.
By applying the construction of \proofcase{Otherwise} outermost, we can translate even all "UNSQ formulas" possibly with more than two "free variables" into \WHERE{2, 2} \kl(where){formulas}.
Hence, by \Cref{lemma: WHERE}, this completes the proof.
\end{proof}

\begin{remark}\label{remark: model checking UNTC supplement 1}
The proof is essentially based on the reduction for "UNFO" \cite{segoufinUnaryNegation2013}.
For "UNFO", we can reduce the "model checking problem" to \kl[model checking problem]{that} for \WHERE{1, 1}.
To handle the case of $\SQ{\fml[2]}{u v} x y$, we reduce to \WHERE{2, 2} instead of \WHERE{1, 1}, that is,
\begin{itemize}
    \item We reduce to \TBSAT{2} not \TBSAT{1}.
    This is for calculating the "squaring" (note that ``$X^{(\fml[2], u v)}$'' appears twice and only in the case of the "squaring").
    Nevertheless, this difference is not essential because \TBSAT{2} is easily reduced to \TBSAT{1} \cite[Corollary 3.5]{schnoebelenOracleCircuitsBranchingTime2003}.

    \item We use \kl{TB-trees} of \kl(TB-tree){width} $\mathcal{O}(\card\univ{\ynstruc}^2)$ not $\mathcal{O}(\card\univ{\ynstruc})$.
    This is because the "squaring formula" $\SQ{\fml[2]}{u v} x y$ may have two \kl{free variables}.
    Nevertheless, the reduction is still in polynomial time.
\end{itemize}
\end{remark}

\subsubsection{For \texorpdfstring{"UNTC"}{UNTC}}
For an "UNSQ formula" $\fml[2]$ and an integer $\ell \ge 0$,
the "UNSQ formula" $\intro*\SQgen{\fml[2]}{\le 2^{\ell}}{u v} x y$ is defined as follows:
\[\SQgen{\fml[2]}{\le 2^{\ell}}{u v} x y \defeq \begin{cases}
    \SQ{(\SQgen{\fml[2]}{\le 2^{\ell-1}}{u v}u v)}{u v} x y & \text{if $\ell \ge 1$},\\
    x = y \lor \fml[2]\fmlsubst{u v}{xy} & \text{if $\ell =0$}.
\end{cases}\]

\begin{proof}[Proof of \Cref{theorem: model checking} for \texorpdfstring{"UNTC"}{UNTC}]
\proofcasethin{(Lower bound)}
By the lower bound for "UNFO" \cite{segoufinUnaryNegation2013}.

\proofcasethin{(Upper bound)}
We give a polynomial-time reduction to the "model checking problem" for "UNSQ".
Let $\ynstruc$ be a given "structure" and let $L \defeq \ceil{\log_2 \card \univ{\ynstruc}}$.
Let $\fml$ be a given "UNTC sentence".
We then let $\fml'$ be the "UNSQ sentence" obtained from $\fml$
by replacing each subformula of the form $\TC{\fml[2]}{u v} x y$ with the "UNSQ formula" $\SQgen{\fml[2]}{\le 2^{L}}{u v} x y$.
We then have that $\ynstruc \modelsass{\inter} \TC{\fml[2]}{u v} x y$ "iff" $\ynstruc \modelsass{\inter} \SQgen{\fml[2]}{\le 2^{L}}{u v} x y$,
because $\card \univ{\ynstruc} \le 2^{L}$.
We thus have $\ynstruc \modelsnonass \fml$ "iff" $\ynstruc \modelsnonass \fml'$.
Hence, by \Cref{lemma: model checking UNSQ}, this completes the proof.
\end{proof}

\begin{remark}\label{remark: model checking UNTC supplement 2}
    If the maximum nesting of \kl{transitive closure formulas} is bounded, 
    there is a naive polynomial-time reduction from the \kl{model checking problem} for "UNTC" into \kl[model checking problem]{that} for "UNFO" by unfolding each \kl{transitive closure formula} \cite[Proposition 9.4]{figueiraCommonAncestorPDL2025}.
    Thus, the "PNPlog2" upper bound is obtained from \kl[model checking problem]{that} for "UNFO", in this case.
    (This situation is the same for "UNSQ".)
    However, the general case was left open \cite[Open question 2]{figueiraCommonAncestorPDL2025}.
    \Cref{theorem: model checking} settles this problem, positively.

    Also, \cite[Footnote 16 of version v1]{figueiraCommonAncestorPDL2025} pointed out that the reduction for \kl{UNFOreg} given in \cite[Theorem 18]{jungQueryingUnaryNegation2018} has an error, because the reduction violates that formulas can use $1$ bit from one input vector in \TBSAT{1}
    (moreover, we cannot reduce to \TBSAT{q} for any fixed $q$, as the number of used bits cannot be bounded).
    In our reduction, we bound the number of used bits to $2$.
\end{remark}

\subsubsection{For \texorpdfstring{"GNTC"}{GNTC}}

\begin{lemma}\label{lemma: GNTC to UNTC}
    There is a polynomial-time reduction from the \kl{model checking problem} for \kl{GNTC} to the \kl{model checking problem} for "UNTC".
\end{lemma}
\begin{proof}
    By applying the same argument as the reduction \cite[Theorem 5.1]{baranyGuardedNegation2015} from the \kl{model checking problem} of "GNFO" to \kl[model checking problem]{that} of "UNFO".
    For each "guard" $\afml[1](\mul{v})$,
    let $P_{[\afml[1]]_{\mul{v}}}$ and $E_{[\afml[1]]_{\mul{v}}, j}$ be the \kl{fresh} unary \kl{relation name} and \kl{fresh} binary \kl{relation names}
    for encoding the relation $\set{\mul{j} \mid \ynstruc \modelsass{\mul{v} \mapsto \mul{j}} \afml}$,
    from the construction of \cite[Theorem 5.1]{baranyGuardedNegation2015}, where $\ynstruc$ is the given "structure".
    We then transform the given "GNTC formula" to an "UNTC formula" in the same way as \cite[Theorem 5.1]{baranyGuardedNegation2015}.
    For "transitive closure formulas" $\TC{\afml[1]_1(\mul{v}_1) \land \afml[1]_2(\mul{v}_2) \land \fml[2]}{\mul{z}} \mul{x}$ (where $\toSet{\mul{v}_1 \mul{v}_2} \subseteq \toSet{\mul{z}}$),
    we transform them as follows:
    \begin{align*}
    &\Big(\exists t_1t_2\ \bigwedge_{i=1}^{2} P_{[\afml[1]_i]_{\mul{v}_i}}\, t_i \land (\bigwedge_{i=1}^{2}\bigwedge_{j} E_{[\afml[1]_i]_{\mul{v}_i}, j}\, t_i \mul{v}_i\ith{j})\fmlsubst{\mul{z}}{\mul{x}} \\
    &{} \land \Big[\bigwedge_{i=1}^{2} P_{[\afml[1]_i]_{\mul{v}_i}}\,  t_i
    \land \exists \mul{z}~ (\bigwedge_{i=1}^{2}\bigwedge_{j} E_{[\afml[1]_i]_{\mul{v}_i}, j}\, t_i \mul{v}_i\ith{j}) \land \fml[2] \Big]_{t_1t_2}^* t_1t_2 \Big) {}\\
    &\lor \bigwedge_{i = 1}^{\ell} \mul{x}\ith{i} \EQ \mul{x}\ith{\ell+i},
    \end{align*}
    where
    \begin{itemize}
        \item $\ell$ is such that the length of $\mul{x}$ is $2\ell$, and
        \item $t_1$ and $t_2$ are \kl{fresh} \kl{variables}. \qedhere
    \end{itemize}
\end{proof}

\begin{proof}[Proof of \Cref{theorem: model checking} for "GNTC"]
    \proofcasethin{(Lower bound)}
    By the lower bound for "UNFO" \cite{segoufinUnaryNegation2013}.

    \proofcasethin{(Upper bound)}
    By the reduction above (\Cref{lemma: GNTC to UNTC}) with \Cref{theorem: model checking} for "UNTC".
\end{proof}

\subsection{Proof of {\Cref{theorem: model checking kFOTC kGNTC}}}\label{section: theorem: model checking kFOTC kGNTC}
In this section, we prove the following theorem.
\begin{theorem*}[Restatement of \Cref{theorem: model checking kFOTC kGNTC}]
\modelcheckingkFOTCkGNTC
\end{theorem*}
The proof proceeds almost in the same way as \Cref{section: theorem: model checking},
where we reduce to \WHERE{2, k} instead of \WHERE{2, 2}.

\subsubsection{For \texorpdfstring{$\tup{\kappa,\lambda,\xi}$-"genUNSQ"}{⟨𝜅,𝜆,𝜉⟩-genUNSQ}}
\AP
\phantomintro{genUNSQ formula}\phantomintro{genUNSQ sentence}
Similar to $\tup{\kappa, \lambda, \xi}$-"genUNTC",
we say that an "FO" with "squaring formula" $\fml$ is $\tup{\kappa, \lambda, \xi}$-""genUNSQ"" if
(i)
each subformula of the form $\lnot \fml[2]$ has at most $\kappa$ free variables: $\card\FV(\fml[2]) \le \kappa$, and
(ii)
for each subformula of the form $\SQ{\fml[2]}{\mul{u}} \mul{x}$,
its "arity" is at most $\lambda$: $\card \set{u : u \in \toSet{\mul{u}}} \le 2 \lambda$,
(iii)
each subformula of the form $\SQ{\fml[2]}{\mul{u}} \mul{x}$
has at most $\xi$ parameters: $\card(\FV(\fml[2]) \setminus \toSet{\mul{u}}) \le \xi$.

\begin{lemma}\label{lemma: model checking kUNSQ}
For every fixed integer $\kappa \ge 1$ and $\lambda, \xi \ge 0$,
the "model checking problem" is "PNPlog2"-complete
for $\tup{\kappa, \lambda, \xi}$-"genUNSQ".
\end{lemma}

\begin{proof}\leavevmode
\proofcasethin{(Lower bound)}
By the lower bound for "UNFO" \cite{segoufinUnaryNegation2013}.

\proofcasethin{(Upper bound)}
Let $k \defeq \max(\kappa, 2\lambda + \xi)$. 
We give a polynomial-time reduction to the "model checking problem" for \WHERE{2, k}.

For a $\tup{\kappa, \lambda, \xi}$-"genUNSQ formula" $\fml$ and a length-$k$ sequence of distinct "variables" $\mul{z} = z_1 \dots z_{k}$ such that $\FV(\fml) \subseteq \set{z_1, \dots, z_{k}}$,
the \WHERE{2, k} \kl(where){formula} $\reintro*\UNTCtoWHERE(\fml, \mul{z})$ is inductively defined as follows,
where $X^{(\fml, \mul{z})}$ is pairwise distinct for each $\fml$ and $\mul{z}$:

\proofcasethin{Case $\fml = \lnot \fml[2]$}
\begin{align*}
&\UNTCtoWHERE(\lnot \fml[2], \mul{z}) \quad\defeq\quad
\lnot X^{(\fml[2], \mul{z})} \mul{z} \mbox{\hspace{.5em}\textsf{where}\hspace{.5em}}  (X^{(\fml[2], \mul{z})} \mul{z} \leftarrow \UNTCtoWHERE(\fml[2],\mul{z})).
\end{align*}

\proofcasethin{Case $\fml = \SQ{\fml[2]}{\mul{u} \mul{v}} \mul{x} \mul{y}$}
\begin{align*}
&\UNTCtoWHERE(\SQ{\fml[2]}{\mul{u} \mul{v}} \mul{x} \mul{y}, \mul{z}) \quad\defeq\quad \exists \mul{z}' X^{(\fml[2], \mul{u} \mul{v} \mul{w})} \mul{x} \mul{z}' \mul{w} \land X^{(\fml[2], \mul{u} \mul{v} \mul{w})} \mul{z}' \mul{y} \mul{w}
\mbox{\hspace{.5em}\textsf{where}\hspace{.5em}}  (X^{(\fml[2], \mul{u} \mul{v} \mul{w})} \mul{u} \mul{v} \mul{w} \leftarrow \UNTCtoWHERE(\fml[2], \mul{u} \mul{v} \mul{w})),
\end{align*}
where $\mul{z}'$ is a sequence of "fresh" "variables", $\mul{w}$ is a sequence of "variables" that satisfies $\mul{w} \supseteq \FV(\fml[2]) \setminus \mul{u} \mul{v}$, and $\mul{u}\mul{v}\mul{w}$ is a pairwise distinct sequence of length $k$.

\proofcasethin{Otherwise}
$\fml$ is built from \kl{atoms} and \kl(genUNSQ){formulas} having at most $k$ \kl{free variables} using $\land$, $\lor$, and $\exists$.
Let $\fml[2]_1, \dots, \fml[2]_M$ be the maximal strict \kl{subformulas} of $\fml$ having at most $k$ \kl{free variables}.
For each $\ell \in \1{M}$,
let $\mul{z}_{\ell}$ be a length-$k$ sequence of distinct "variables" such that $\FV(\fml[2]_{\ell}) \subseteq \toSet{\mul{z}_{\ell}}$.
Then,
\begin{align*}
    & \UNTCtoWHERE(\fml, \mul{z}) \quad\defeq\quad \fml' \mbox{\hspace{.5em}\textsf{where}\hspace{.5em}} (X^{(\fml[2]_1,\mul{z}_1)} \mul{z}_1 \leftarrow \UNTCtoWHERE(\fml[2]_1,\mul{z}_1),\; \dots,\; X^{(\fml[2]_M,\mul{z}_M)} \mul{z}_M \leftarrow \UNTCtoWHERE(\fml[2]_M,\mul{z}_M)),
\end{align*}
where $\fml'$ is the \kl(EFO){formula} $\fml$ in which 
each $\fml[2]_{\ell}$ is replaced with $X^{(\fml[2]_{\ell}, \mul{z}_{\ell})} \mul{z}_{\ell}$.

We then have that $\fml[1]$ and $\UNTCtoWHERE(\fml[1], \mul{z})$ are "semantically equivalent",
by easy induction on $\fml[1]$.
Hence, by \Cref{lemma: WHERE}, this completes the proof.
\end{proof}

\subsubsection{For \texorpdfstring{$\tup{\kappa,\lambda,\xi}$-"genUNTC"}{ ⟨𝜅,𝜆,𝜉⟩-genUNTC}}

\begin{proof}[Proof of \Cref{theorem: model checking kFOTC kGNTC} for $\tup{\kappa,\lambda,\xi}$-"genUNTC"]\leavevmode
\proofcasethin{(Lower bound)}
By the lower bound for "UNFO" \cite{segoufinUnaryNegation2013}.

\proofcasethin{(Upper bound)}
We give a polynomial-time reduction to the "model checking problem" for $\tup{\kappa, \lambda, \xi}$-"genUNSQ".
Let $\ynstruc$ be a given "structure" and let $L \defeq \ceil{\log_2 \card \univ{\ynstruc}}$.
Let $\fml$ be a given $\tup{\kappa, \lambda, \xi}$-"genUNTC sentence".
We then let $\fml'$ be the $\tup{\kappa, \lambda, \xi}$-"genUNSQ sentence" obtained from $\fml$
by replacing each "subformula" of the form $\TC{\fml[2]}{\mul{u} \mul{v}} \mul{x} \mul{y}$ with the $\tup{\kappa, \lambda, \xi}$-"genUNSQ formula" $\SQgen{\fml[2]}{\le 2^{\lambda L}}{\mul{u} \mul{v}} \mul{x} \mul{y}$ (defined similarly to "UNSQ").
We then have that $\ynstruc \modelsass{\inter} \TC{\fml[2]}{\mul{u} \mul{v}} \mul{x} \mul{y}$ "iff" $\ynstruc \modelsass{\inter} \SQgen{\fml[2]}{\le 2^{\lambda L}}{\mul{u} \mul{v}} \mul{x} \mul{y}$,
because $\card \univ{\ynstruc}^{\lambda} \le 2^{\lambda L}$.
We thus have $\ynstruc \modelsnonass \fml$ "iff" $\ynstruc \modelsnonass \fml'$.
Hence, by \Cref{lemma: model checking kUNSQ}, this completes the proof.
\end{proof}

\subsubsection{For \texorpdfstring{$\tup{\kappa,\lambda,\xi}$-"genGNTC"}{⟨𝜅,𝜆,𝜉⟩-genGNTC}}

\begin{lemma}\label{lemma: kGNTC to kUNTC}
    There is a polynomial-time reduction from the \kl{model checking problem} for $\tup{\kappa,\lambda,\xi}$-"genGNTC" to the \kl{model checking problem} for $\tup{\kappa,\lambda,\xi}$-"genUNTC".
\end{lemma}
\begin{proof}
    Similar to the reduction for \Cref{lemma: GNTC to UNTC}.
    For each "atom" $\afml[1](\mul{v})$ with "free variables" $\mul{v}$,
    let $P_{[\afml[1]]_{\mul{v}}}$ and $E_{[\afml[1]]_{\mul{v}}, j}$ be the \kl{fresh} unary \kl{relation name} and \kl{fresh} binary \kl{relation names}
    for encoding the relation $\set{\mul{j} \mid \ynstruc \modelsass{\mul{v} \mapsto \mul{j}} \afml}$,
    from the construction of \cite[Theorem 5.1]{baranyGuardedNegation2015}, where $\ynstruc$ is the given "structure".
    We transform the given $\tup{\kappa,\lambda,\xi}$-"genGNTC formula" to a $\tup{\kappa,\lambda,\xi}$-"genUNTC formula" in the same way as \cite[Theorem 5.1]{baranyGuardedNegation2015}.
    Here, we transform
    $\TC{\bigwedge_{i = 1}^{\lambda} \afml[1]_i(\mul{v}_i) \land \bigwedge_{i = 1}^{\lambda} \afml[1]_{\lambda+i}(\mul{v}_{\lambda+i}) \land \fml[2]}{\mul{z}} \mul{x}$,
    where $\toSet{\mul{v}_i} \subseteq \toSet{\mul{z}}$,
    as follows:
    \begin{align*}
    &\Big(\exists \mul{t}\ \bigwedge_{i=1}^{2\lambda} P_{[\afml[1]_i]_{\mul{v}_i}} \mul{t}\ith{i} \land (\bigwedge_{i=1}^{2\lambda}\bigwedge_{j} E_{[\afml[1]_i]_{\mul{v}_i}, j}\; \mul{t}\ith{i} \mul{v}_i\ith{j})\fmlsubst{\mul{z}}{\mul{x}}\hspace{3em} \\
    &{} \land \Big[\bigwedge_{i=1}^{2\lambda} P_{[\afml[1]_i]_{\mul{v}_i}} \mul{t}\ith{i}
    \land \exists \mul{z}~ \bigwedge_{i=1}^{2\lambda}\bigwedge_{j} E_{[\afml[1]_i]_{\mul{v}_i}, j}\; \mul{t}\ith{i} \mul{v}_i\ith{j} \land \fml[2] \Big]_{\mul{t}}^* \mul{t} \Big)\\
    &{} \lor \bigwedge_{i = 1}^{\ell} \mul{x}\ith{i} \EQ \mul{x}\ith{\ell+i},
    \end{align*}
    where
    \begin{itemize}
        \item $\ell$ is the number satisfying that the length of $\mul{x}$ is $2\ell$, and
        \item $\mul{t}$ is a \kl{fresh} \kl{variable} sequence of length $2\lambda$.
    \end{itemize}
    We also transform
    $\bigwedge_{i = 1}^{\kappa} \afml[1]_i(\mul{v}_i) \land \lnot \fml[2](\mul{x})$,
    where $\toSet{\mul{v}_i} \subseteq \toSet{\mul{x}}$,
    as follows:
    \begin{align*}
    \exists \mul{t}\ \left( \bigwedge_{i=1}^{\kappa} P_{[\afml[1]_i]_{\mul{v}_i}} \mul{t}\ith{i} \land \bigwedge_{i=1}^{\kappa}\bigwedge_{j} E_{[\afml[1]_i]_{\mul{v}_i}, j}\; \mul{t}\ith{i} \mul{v}_i\ith{j} 
    \land \lnot \exists \mul{x}~ \bigwedge_{i=1}^{\kappa}\bigwedge_{j} E_{[\afml[1]_i]_{\mul{v}_i}, j}\; \mul{t}\ith{i} \mul{v}_i\ith{j} \land \fml[2](\mul{x}) \right),
    \end{align*}
    where
    \begin{itemize}
        \item $\mul{t}$ is a \kl{fresh} \kl{variable} sequence of length $\kappa$. \qedhere
    \end{itemize}
\end{proof}

\begin{proof}[Proof of \Cref{theorem: model checking kFOTC kGNTC} for $\tup{\kappa,\lambda,\xi}$-"genGNTC"]\leavevmode

    \proofcasethin{(Lower bound)}
    Because the "model checking problem" is already "PNPlog2"-hard for "UNFO" \cite{segoufinUnaryNegation2013}.

    \proofcasethin{(Upper bound)}
    By the reduction above (\Cref{lemma: kGNTC to kUNTC}) with \Cref{theorem: model checking kFOTC kGNTC} for $\tup{\kappa,\lambda,\xi}$-"genUNTC".
\end{proof}

\subsection{Note: Necessity of the conditions}\label{section: note: necessity of the conditions}

In this section, we note that the three fixed parameters $\kappa$, $\lambda$, and $\xi$ are necessary for the "PNPlog2" upper bound in \Cref{theorem: model checking kFOTC kGNTC}.

First, we note that the \kl{model checking problem} for full "TC" is "PSPACE"-complete,
which is a well-known result \cite{vardiComplexityRelationalQuery1982,berwangerFixedPointLogicsSolitaire2004}.
\begin{theorem}[{\cite{vardiComplexityRelationalQuery1982}}]\label{theorem: model checking TC}
    The \kl{model checking problem} for \kl{TC} is "PSPACE"-complete.
\end{theorem}
\begin{remark}
    In \Cref{theorem: model checking TC},
    to determine whether $\ynstruc \modelsass{\inter} \TC{\fml[2]}{\mul{u} \mul{v}} \mul{x} \mul{y}$ for a $k$-ary \kl{transitive closure formula},
    we consider finding a \kl{path} of length at most $\card\univ{\ynstruc}^{k}$ from $\inter(\mul{x})$ to $\inter(\mul{y})$.
    When $k$ is bounded ($k = 1$ in \cite[Theorem 4]{vardiComplexityRelationalQuery1982}),
    as $\card\univ{\ynstruc}^{k}$ is in polynomial,
    by an exhaustive search, we can check there exists a \kl{path} $\mul{a}_0, \dots, \mul{a}_{\ell}$ such that
    $\mul{a}_0 = \inter(\mul{x})$, $\mul{a}_{\ell} = \inter(\mul{y})$, and
    $\ynstruc \modelsass{\inter\intersubst{\mul{u} \mul{v}}{\mul{a}_{i-1}\mul{a}_{i}}} \fml[2]$ for $i \in \1{\ell}$.
    Even when $k$ is unbounded,
    we can still give a "PSPACE" algorithm using the doubling trick, like Savitch's theorem \cite{savitchRelationshipsNondeterministicDeterministic1970}.
\end{remark}

\subsubsection{Case \texorpdfstring{$\kappa = \omega$}{𝜅=𝜔}}
When $\kappa = \omega$, even after fixing the parameters $\lambda$ and $\xi$,
the "model checking problem" remains "PSPACE"-hard as follows:
\begin{proposition}\label{proposition: model checking kFOTC kGNTC omega _ _}
The "model checking problem" is "PSPACE"-complete
for both $\tup{\omega, 0, 0}$-"genUNTC" and $\tup{\omega, 0, 0}$-"genGNTC".
\end{proposition}%
\begin{proof}
The upper bound follows from \Cref{theorem: model checking TC}.
For the lower bound, by \Cref{lemma: kGNTC to kUNTC},
it suffices to show that the "model checking problem" for $\tup{\omega, 0, 0}$-"genUNTC" is "PSPACE"-hard.
Since $\tup{\omega, 0, 0}$-"genUNTC" coincides with first-order logic ("FO") syntactically (as we do not introduce 'nullary' "TC"),
the problem is "PSPACE"-hard, "eg", by encoding ""QBF"" (\reintro*\kl{quantified Boolean formula problem}) in "FO" \cite[Theorem 2]{vardiComplexityRelationalQuery1982}.
\end{proof}

\subsubsection{Case \texorpdfstring{$\lambda = \omega$}{𝜆=𝜔}}
When $\lambda = \omega$, even after fixing the parameters $\kappa$ and $\xi$,
the "model checking problem" remains "PSPACE"-hard as follows:
\begin{proposition}\label{proposition: model checking kFOTC kGNTC _ omega _}
The "model checking problem" is "PSPACE"-complete
for both $\tup{0, \omega, 0}$-"genUNTC" and $\tup{0, \omega, 0}$-"genGNTC".
\end{proposition}
\begin{proof}
The upper bound follows from \Cref{theorem: model checking TC}.
For the lower bound, by \Cref{lemma: kGNTC to kUNTC},
it suffices to show that the "model checking problem" for $\tup{0, \omega, 0}$-"genUNTC" is "PSPACE"-hard.
We give a reduction from the \AP\intro*\kl{intersection non-emptiness problem} for \AP\intro*\kl{deterministic finite automata} (\kl{DFAs}) \nointro(DFA){states}:
    given $k \in \nat$ and \kl{DFAs} $\automaton[1]_1, \dots, \automaton[1]_k$ over a finite set $\Sigma$ of \kl{letters},
     decide whether there exists a \kl{word} $\word \in \Sigma^*$ such that all $\automaton[1]_1, \dots, \automaton[1]_k$ accept $\word$.
    This problem is "PSPACE"-complete \cite[Lemma 3.2.3]{kozenLowerBoundsNatural1977}.
    For \kl{DFAs} $\automaton_1, \dots, \automaton_k$ ("Wlog", we can assume that the sets of \kl(DFA){states} are disjoint),
    we consider the \kl{structure} $\ynstruc_{\automaton_1, \dots, \automaton_k}$ where
    \begin{itemize}
        \item the \kl(structure){domain} is the union of \kl(DFA){states} of $\automaton_1, \dots, \automaton_k$,
        \item each binary \kl{relation name} $c \in \Sigma$ 
        expresses the union of the transition functions of $\automaton_i$ w.r.t.\ the \kl{letter} $c$,
        \item each unary \kl{relation name} $S_i$ expresses the singleton set indicating the initial \kl(DFA){state} of $\automaton_i$,
        \item each unary \kl{relation name} $A_i$ expresses the set indicating the acceptance \kl(DFA){states} of $\automaton_i$.
    \end{itemize}
    Then, all the automata $\automaton_1, \dots, \automaton_k$ accept a common \kl{word} iff
    \begin{align*}
    \ynstruc_{\automaton_1, \dots, \automaton_k}
    &~\modelsnonass~
    \exists \mul{x} \mul{y} ~ ( (\bigwedge_{i = 1}^{k} (S_i(\mul{x}\ith{i}) \land A_i(\mul{y}\ith{i}))) \land \left[\bigvee_{c \in \Sigma} \bigwedge_{i = 1}^{k} c(\mul{u}\ith{i},\mul{v}\ith{i}) \right]_{\mul{u}\mul{v}}^* \mul{x} \mul{y}).
    \end{align*}
    Hence, this completes the proof.
\end{proof}

\subsubsection{Case \texorpdfstring{$\xi = \omega$}{𝜉=𝜔}}
When $\xi = \omega$, even after fixing the parameters $\lambda$ and $\kappa$,
the "model checking problem" remains "PSPACE"-hard as follows:
\begin{proposition}\label{proposition: model checking kFOTC kGNTC _ _ omega}
The "model checking problem" is "PSPACE"-complete
for both $\tup{0, 1, \omega}$-"genUNTC" and $\tup{0, 1, \omega}$-"genGNTC".
\end{proposition}%

\begin{proof}
The upper bound follows from \Cref{theorem: model checking TC}.
For the lower bound,
by \Cref{lemma: kGNTC to kUNTC},
it suffices to show that the "model checking problem" for $\tup{0, 1, \omega}$-"genUNTC" is "PSPACE"-hard.
We reduce from "QBF" to this problem, similar to \Cref{proposition: model checking kFOTC kGNTC omega _ _}.

Let $\ynstruc$ be the fixed structure
$\left(\begin{tikzpicture}[baseline = -0.5ex]
    \graph[grow right = 1.cm, branch down = 2.5ex]{
    {1/{$\const{T}$}[vert]} -!- {2/{$\const{F}$}[vert]} -!- {3/{$\const{X}$}[vert]} 
    };
    \path (1) edge[earrow, pos = .4, ->] node[fill = white, inner sep = 1pt, font = \scriptsize] {$E$} (2);
    \path (2) edge[earrow, pos = .4, ->] node[fill = white, inner sep = 1pt, font = \scriptsize] {$E$} (3);
\end{tikzpicture} \right)$
with three unary relations
$\interpatom{\ynstruc}{\rsym_{\const{T}}} = \set{\const{T}},
\interpatom{\ynstruc}{\rsym_{\const{F}}} = \set{\const{F}},
\interpatom{\ynstruc}{\rsym_{\const{X}}} = \set{\const{X}}$ defined as in the figure.
Let $z_{\const{T}}, z_{\const{F}}, z_{\const{X}}$ be pairwise distinct "variables".
For a ""quantified Boolean formula"" $Q$ in negation normal form without "variables" $z_{\const{T}}$, $z_{\const{F}}$, $z_{\const{X}}$, we inductively define the $\tup{0, 1, \omega}$-"genUNTC formula" $Q^{\dagger}$ as follows:
\begin{align*}
    x^{\dagger} & ~\defeq~ x \EQ z_{\const{T}},&
    (\lnot x)^{\dagger} & ~\defeq~ x \EQ z_{\const{F}},\\
        (Q_1 \land Q_2)^{\dagger} & ~\defeq~ Q_1^{\dagger} \land Q_2^{\dagger},&
        (Q_1 \lor Q_2)^{\dagger} & ~\defeq~ Q_1^{\dagger} \lor Q_2^{\dagger},\\
    (\exists x Q)^{\dagger} & ~\defeq~ \exists x~ (x = z_{\const{T}} \lor x = z_{\const{F}}) \land Q^{\dagger},&
    (\forall x Q)^{\dagger} & ~\defeq~ \TC{Q^{\dagger} \land E(x, z)}{x,z} z_{\const{T}} z_{\const{X}},
\end{align*}
where $z_{\const{T}}, z_{\const{F}}, z_{\const{X}}$ are fresh variables not occurring in $Q$.
We then have that the "quantified Boolean formula" $Q$ is true
"iff" $\ynstruc \modelsass{z_{\const{T}}z_{\const{F}}z_{\const{X}} \mapsto \const{T} \const{F} \const{X}} Q^{\dagger}$ %
"iff" $\ynstruc \modelsnonass \exists z_{\const{T}}\exists z_{\const{F}}\exists z_{\const{X}}~(\rsym_{\const{T}}(z_{\const{T}}) \land \rsym_{\const{F}}(z_{\const{F}}) \land \rsym_{\const{X}}(z_{\const{X}}) \land Q^{\dagger})$.
\end{proof}

\paragraph*{Note on ``nullary negation'' first-order logic (Case $\kappa = \lambda = 0$)}
Exceptionally for $\tup{0, 0, \omega}$-\kl{genUNTC} (hence, also for $\tup{0, 0, \omega}$-\kl{genGNTC} via \Cref{lemma: kGNTC to kUNTC}),
the "model checking problem" is in "PNPlog2" as a corollary of \Cref{theorem: model checking kFOTC kGNTC}, because
$\tup{0, 0, \omega}$-\kl{genUNTC} coincides with $\tup{0, 0, 0}$-\kl{genUNTC}.
Moreover, we can view $\tup{0, 0, 0}$-\kl{genUNTC} also as the ``nullary negation'' first-order logic ("NNFO").
\AP
\phantomintro{NNFO}\phantomintro{NNFO sentence}%
We say that a "UNFO formula" $\fml$ is an ""NNFO formula"" if each negated "subformula" $\lnot \fml[2]$ of $\fml$ has no "free variables".
The "model checking problem" for "NNFO" is more precisely "PNPlog"-complete, from the result for "UNFO" of "negation depth" 1 \cite[Theorem 5.2]{segoufinUnaryNegation2013}, as follows.
\begin{proposition}\label{proposition: model checking kFOTC kGNTC _ 0 omega}
The "model checking problem" for both $\tup{0, 0, \omega}$-"genUNTC" ("ie", "NNFO") and $\tup{0, 0, \omega}$-"genGNTC" is "PNPlog"-complete.
\end{proposition}
\begin{proof}
\proofcasethin{(Lower bound)}
By the same reduction as \cite[Theorem 5.2]{segoufinUnaryNegation2013} (for "UNFO" of "negation depth" 1).
The sentence used in \cite[Theorem 5.2]{segoufinUnaryNegation2013} is a Boolean combination of conjunctive query sentences, hence in "NNFO".
Additionally, it is also "PNPlog"-hard for $\tup{0, 0, \omega}$-"genGNTC" by \Cref{lemma: kGNTC to kUNTC}.

\smallskip

\proofcasethin{(Upper bound)}
By \Cref{lemma: kGNTC to kUNTC}, it suffices to consider "NNFO".
We give a polynomial-time reduction to "UNFO" of "negation depth" 1.
Intuitively, we unfold the nesting of negations by replacing each negated "subformula" (\kl[UNFO sentence]{sentence}) $\fml[2]$ with a ``propositional variable'' represented as the "atom" $\const{T}(z_{\fml[2]})$ (as in the Tseitin transformation).

Let $\ynstruc$ be a "structure".
When $\card\univ{\ynstruc} = 1$, this case is easily solvable in polynomial time.
Below, we assume $\card\univ{\ynstruc} \ge 2$.
Let $\ynstruc'$ be the "structure" $\ynstruc$
in which a "fresh" unary "relation name" $\const{T}$ satisfying
$\interpatom{\ynstruc'}{\const{T}} \neq \emptyset$ and $\interpatom{\ynstruc'}{\const{T}} \neq \univ{\ynstruc'}$
("ie", $\interpatom{\ynstruc'}{\const{T}}$ is non-trivial).
For an "NNFO formula" $\fml$ (without the "relation name" $\const{T}$),
we inductively define $\fml^{\dagger} = \tup{\fml^{\dagger_1}, \fml^{\dagger_2}}$,
where $\fml^{\dagger_1}$ is a "UNFO formula" without negations
and $\fml^{\dagger_2}$ is a map from "variables" to "UNFO sentences" of "negation depth" 1,
as follows:
\begin{align*}
    (\rsym \mul{x})^{\dagger} &\defeq \tup{\rsym \mul{x}, \emptyset},&
    (\fml[2] \land \fml[3])^{\dagger} &\defeq \tup{\fml[2]^{\dagger_1} \land \fml[3]^{\dagger_1}, \fml[2]^{\dagger_2} \cup \fml[3]^{\dagger_2}},\\
    (\exists x \fml[2])^{\dagger} &\defeq \tup{\exists x \fml[2]^{\dagger_1}, \fml[2]^{\dagger_2}}, &
    (\fml[2] \lor \fml[3])^{\dagger} &\defeq \tup{\fml[2]^{\dagger_1} \lor \fml[3]^{\dagger_1}, \fml[2]^{\dagger_2} \cup \fml[3]^{\dagger_2}},\\
    (\lnot \fml[2])^{\dagger} &\defeq \tup{\const{T}(z_{\lnot \fml[2]}), \set{z_{\lnot \fml[2]} \mapsto \lnot \fml[2]^{\dagger_1}} \cup \fml[2]^{\dagger_2}},
\end{align*}
where $z_{\lnot \fml[2]}$ is a "fresh" "variable" for each "subformula" $\lnot \fml[2]$.
Then for each $\inter$, we have
$\ynstruc' \modelsass{\inter} \fml \leftrightarrow \exists z_1 \dots z_n (\fml^{\dagger_1} \land \bigwedge_{i =1}^n (\const{T}(z_i) \leftrightarrow \fml^{\dagger_2}_{z_i}))$ where $\fdom(\fml^{\dagger_2}) = \set{z_1, \dots, z_n}$,
$z_1, \dots, z_n$ are pairwise distinct, and $\fml^{\dagger_2}_{z_i}$ denotes the \kl[UNFO sentence]{sentence} $\fml^{\dagger_2}(z_i)$.
Hence, this completes the proof.
\end{proof}
\end{toappendix}
\section{Conclusion}
\label{sec:conclusion}
We have given a polynomial-time %
translation from \kl{GNTC} to \kl{UNTC} preserving  "satisfiability@satisfiable" and "finite-satisfiability@finitely-satisfiable"
(\Cref{thm:GNTC-UNTC}),
and hence the "satisfiability problem" for "GNTC" is in "2ExpTime" (\Cref{corollary: 2-EXPTIME GNTC}).
We have also given a "local checker@@UNTC" yielding a direct single exponential-time reduction from
"SAT-""UNTC" to the "non-emptiness problem" for "2APTAs" (\Cref{corollary: 2-EXPTIME UNTC}).
Furthermore, we have shown that the "model checking problems" for "UNTC", "GNTC", $\tup{\kappa, \lambda, \xi}$-"genUNTC", and $\tup{\kappa, \lambda, \xi}$-"genGNTC" are "PNPlog2"-complete (\Cref{section: model checking}).
A natural direction would be to extend "GNTC" while preserving the "2ExpTime" "satisfiability problem". For instance, we would like to explore the possibility of extending "GNTC" to the ""clique-guarded negation fragment"" \cite[Section 7]{baranyGuardedNegation2015}
or to 
unparameterized guarded "LFP" operators \cite{baranyGuardedNegation2015}, possibly by extending the "local checker@@UNTC".
Regarding the second extension,
it would also be interesting to find a fragment of "GNFP-UP" including both "GNTC" and "GNFP".

\newcommand{\pdepth}{\textit{pdepth}}%
As noted in \cite[p.~5]{benediktStepExpressivenessDecidable2016}, the canonical translation of "GNTC" into "GNFP-UP" yields formulas with unbounded parameter depth ($\pdepth$) -- a syntactic measure of the maximum number of nested parameter changes within a formula. However, we currently lack the model-theoretic tools to show that "GNTC" is not captured by any fragment of "GNFP-UP" of bounded $\pdepth$.%
\footnote{%
In particular, the strict $\pdepth$ hierarchy of \cite{benediktStepExpressivenessDecidable2016} is established via a complexity-based argument, separating the elementary satisfiability of bounded $\pdepth$ from the non-elementary satisfiability of unbounded $\pdepth$. Since the satisfiability problem for "GNTC" is already elementary, this argument cannot be applied. Further, $\text{GN}^k$-bisimulation games do not track the parameter nesting necessary to prove inexpressibility at a given $\pdepth$ level.%
}

We leave open the decidability of %
the "finite-satisfiability problem" of "GNTC". %
 While %
the "finite-satisfiability problem" of "GNTC"
is not harder than %
the "finite-satisfiability problem" of "UNTC"
(\Cref{thm:GNTC-UNTC}),
decidability is open for "UNTC" \cite[Open Question 1]{figueiraCommonAncestorPDL2025}
and even for loop-"PDL" \cite{daneckiPropositionalDynamicLogic1984} (see also \cite{gollerPDLIntersectionConverse2009,figueiraCommonAncestorPDL2025}).
However, the problem is known to be decidable, "eg", for "GNFP" \cite{baranyGuardedNegation2015} (see also \cite{segoufinUnaryNegation2013,bojanczykTwoWayAlternatingAutomata2002,baranyFiniteSatisfiabilityGuarded2012}) and "UNFO" with transitive relations \cite{danielskiFiniteSatisfiabilityUnary2019}.

\bibliography{main}

\end{document}